\documentclass[10pt, oneside]{article}

 \usepackage[bookmarks=true, linkbordercolor={0 0 1}]{hyperref}

\pdfoutput=1

\input ./combined_macros.sty

\newcommand{\defterm}[1]{\textbf{\emph{#1}}}

\usepackage{pdflscape}
\usetikzlibrary{shapes.geometric, arrows, positioning, matrix}

\tikzstyle{s16} = [rectangle, rounded corners, text centered, draw=black,fill=red!30]
\tikzstyle{s16chiral} = [s16, dashed]
\tikzstyle{s8} = [rectangle, rounded corners, text centered, draw=black,fill=orange!30]
\tikzstyle{s4} = [rectangle, rounded corners, text centered, draw=black,fill=yellow!30]
\tikzstyle{s2chiral} = [rectangle, dashed, rounded corners, text centered, draw=black,fill=green!30]
\tikzstyle{dimension} = [circle, text centered, text width=0.7cm, minimum height=0.7cm, draw=black]
\tikzstyle{arrow} = [thick,->,>=stealth]

\title{A Taxonomy of Twists of Supersymmetric Yang--Mills Theory}
\author{Chris Elliott\and Pavel Safronov \and Brian R. Williams}

\date{\today}

\begin{document}

\maketitle

\begin{abstract}
We give a complete classification of twists of supersymmetric Yang--Mills theories in dimensions~$2\leq n \leq 10$.
We formulate supersymmetric Yang--Mills theory classically using the BV formalism, and then we construct an action of the supersymmetry algebra using the language of $L_\infty$ algebras.  For each orbit in the space of square-zero supercharges in the supersymmetry algebra, under the action of the spin group and the group of R-symmetries, we give a description of the corresponding twisted theory.  These twists can be described in terms of mixed holomorphic-topological versions of Chern--Simons and BF theory. 
\end{abstract}

\pagestyle{intro}
\setcounter{tocdepth}{2}
\tableofcontents

\section*{Introduction} 
\label{sect:intro}
\addcontentsline{toc}{section}{Introduction}

In this paper we calculate supersymmetric twists of super Yang--Mills theories in dimension 2 through 10. Our main tools are the classical Batalin--Vilkovisky formalism, which eliminates the need for auxiliary fields to close the on-shell supersymmetry action, and a consistent use of dimensional reduction which allows us to deduce lower-dimensional statements from higher-dimensional statements.

\subsection*{Classical Field Theories}

Let us begin with an informal discussion of classical field theories. A classical field theory is usually defined in terms of the data of the space of fields $\cF$ equipped with an action functional. To incorporate gauge symmetries, one may either work with $\cF$ as a stack or, as in the BRST formalism, with $\cF$ as a $Q$-manifold, i.e. a graded manifold equipped with a square-zero vector field of cohomological degree $1$ (the BRST differential). In the Batalin--Vilkovisky \cite{BatalinVilkovisky} approach one considers instead the space of BV fields $\cE$, which is equipped with a $(-1)$-shifted symplectic structure; this may be modeled by a $QP$-manifold \cite{Schwarz}. Moreover, we assume that the $Q$-structure is Hamiltonian, i.e. that it is given by a Poisson bracket $\{S, -\}$ with respect to the BV action functional. Here $\cE$ is interpreted as modelling the derived critical locus of the action functional on $\cF$.

In this paper we follow the approach developed in the works of Costello and Gwilliam \cite{CostelloBook,Book1}. As the space of BV fields $\cE$ is an infinite-dimensional manifold, it is difficult to work with it directly (for instance, to make sense of a $(-1)$-shifted symplectic structure). Instead, we zoom in on the neighborhood of a point where $Q$ vanishes (i.e. we consider a given classical solution). We may then consider $\cE$ as the space of sections of a graded vector bundle $E\rightarrow M$ over the spacetime manifold $M$. This allows us to work with finite-dimensional objects throughout. Namely, a $(-1)$-shifted symplectic structure on $\cE$ boils down to a $(-1)$-shifted symplectic pairing $E\cong E^![-1]$, where $E^!=E^*\otimes \Dens_M$. We refer to Definition \ref{def:classicalfieldtheory} for the precise definition of a classical field theory in the BV formalism that we use. Moreover, with this definition we may talk about weak equivalences of classical field theories (a notion inaccessible with $QP$-manifolds) which are simply maps of classical field theories inducing a quasi-isomorphism on $\cE$. We call these \emph{perturbative equivalences} (see Definition \ref{def:perturbativeequivalence}) to emphasize that we are working in a formal neighborhood of a given classical solution. For simplicity, throughout the paper we ignore issues of unitarity: in other words, we always consider complexified bundles of fields.

\subsection*{Classical Supersymmetric Field Theories}

Now consider a classical field theory where the spacetime manifold is $M=\RR^n$, and where the theory is translation-invariant. Given the data of a spinorial representation $\Sigma$ equipped with a symmetric pairing $\Gamma\colon\Sym^2(\Sigma)\rightarrow V=\CC^n$, we may construct a super Lie algebra of supertranslations $\fA=\Pi\Sigma\oplus V$, where $\Pi$ indicates that $\Sigma$ is placed in odd $\ZZ/2\ZZ$-degree, with the only nonvanishing Lie bracket given by $\Gamma$. A supersymmetric classical field theory is then a translation-invariant classical field theory on $\RR^n$ where the translation action on the fields is extended to an action of the super Lie algebra $\fA$. In addition, we may consider an $R$-symmetry group $G_R$ that acts on $\Sigma$ preserving $\Gamma$ and the $\so(n)$-action and also compatibly on the classical field theory.

In most literature on supersymmetry one simply tries to build an action of $\fA$ on the space of ordinary fields $\cF$. However, one often runs into a problem that the supersymmetry action is only \emph{on-shell}: the map from $\fA$ to vector fields $\mathrm{Vect}(\cF)$ preserves Lie brackets only on the critical locus of the action functional. The usual solution is to enlarge the space of fields by adding auxiliary fields with no kinetic terms on which there is an honest (\emph{off-shell}) action of $\fA$. However, this choice may be not canonical. For instance, in 10d $\cN=(1, 0)$ super Yang--Mills one needs to break the Lorentz group $\SO(10)$ to $\Spin(7)\times\SO(2)$ to have an off-shell action of a subalgebra of $\fA$ where the odd part is 9-dimensional (instead of 16-dimensional) \cite{BaulieuBerkovitsBossardMartin}.

We instead take another approach pioneered by Baulieu, Bellon, Ouvry and Wallet \cite{BaulieuBV}. Namely, one may canonically extend the supersymmetry action from the space of ordinary fields $\cF$ to the space of BV fields $\cE$. The property of the action being on-shell now means that the map $\fA\rightarrow \mathrm{Vect}(\cE)$ preserves Lie brackets, but only up to homotopy. One may then try to incorporate these homotopies: to extend the Lie action to an $L_\infty$ action.  This contrasts with the auxiliary field approach of the previous paragraph, where one instead builds a resolution of the space of BV fields on which the supersymmetry Lie algebra acts strictly.

In this paper we consider supersymmetric Yang--Mills theories in dimensions 2 through 10. In dimensions 3 through 10 these may be obtained by dimensional reduction of the following theories: 10d $\cN=(1, 0)$ super Yang--Mills, 6d $\cN=(1, 0)$ super Yang--Mills, 4d $\cN=1$ super Yang--Mills and 3d $\cN=1$ super Yang--Mills. These theories depend on a choice of a Lie algebra $\fg$ equipped with a symmetric nondegenerate bilinear pairing. In addition, in dimensions 6, 4 and 3 we may add matter multiplets: in dimension 6 these depend on a choice of a symplectic $\fg$-representation (a hypermultiplet), in dimension 4 these depend on a choice of a $\fg$-representation (a chiral multiplet) and in dimension 3 these depend on a choice of an orthogonal $\fg$-representation. We do not consider superpotential, mass, or Fayet--Iliopoulos terms in this paper. Moreover, as we are working perturbatively, we ignore all topological terms ($\theta$-terms).

The on-shell supersymmetry of pure super Yang--Mills theories in these dimensions can be proven by using a well-known relationship between composition algebras (e.g. division algebras) and supersymmetry (see Section \ref{sect:compositionalgebras}) which goes back to the works \cite{Evans,KugoTownsend}. For instance, we may construct the 10d $\cN=(1, 0)$ supersymmetry from the algebra of octonions $\mathbb{O}$, 6d $\cN=(1, 0)$ supersymmetry from the algebra of quaternions $\mathbb{H}$, 4d $\cN=1$ supersymmetry from the complex numbers $\CC$ and 3d $\cN=1$ supersymmetry from the real numbers $\RR$. Our treatment follows the work of Baez and Huerta \cite{BaezHuerta} and we show how to extend the on-shell $\fA$-action to an $L_\infty$-action using these ideas. As a new result, we also construct an $L_\infty$-action on matter multiplets where the language of composition algebras turns out to be indispensable (see Section \ref{sect:mattermultipletSUSY}). Namely, for any real associative composition algebra $A_\RR$ we simply need a complex $\fg$-representation $P$ equipped with an $A_\RR$-module structure and a symmetric bilinear pairing. We have the following three cases:
\begin{itemize}
\item (\textbf{6d $\cN=(1, 0)$ supersymmetry}) For $A_\RR=\mathbb{H}$ $P$ is forced to take the form $U\otimes W_+$, where $U$ is a symplectic $\fg$-representation and $W_+$ is a 2d complex symplectic vector space (so that $\mathbb{H}\otimes_\RR\CC\cong \eend(W_+)$).
\item (\textbf{4d $\cN=1$ supersymmetry}) For $A_\RR=\mathbb{C}$ $P$ is forced to take the form $R\oplus R^*$, where $R$ is a $\fg$-representation.
\item (\textbf{3d $\cN=1$ supersymmetry}) For $A_\RR=\mathbb{R}$ we simply have an orthogonal $\fg$-representation $P$.
\end{itemize}

In addition to the dimensional reduction of these super Yang--Mills theories, there are also certain special super Yang--Mills theories with chiral supersymmetry in dimension 2: namely, 2d $\cN=(1, 0)$, $\cN=(2, 0)$ and $\cN=(4, 0)$ with matter as well as pure $\cN=(\cN_+, 0)$ theories for any $\cN_+$. We treat these separately (see Section \ref{sect:2dchiral}), but again the language of composition algebras turns out to be convenient.

\subsection*{Supersymmetric Twists}

The notion of supersymmetric twisting for a supersymmetric field theory was introduced by Witten \cite{WittenTQFT}, and further developed mathematically by Costello \cite{CostelloSUSY}. 
The definitions we use in this paper will follow our previous work \cite{ElliottSafronov}, so let us briefly recall the important terminology.

Suppose $Q$ is a square-zero supercharge, i.e. an odd element $Q\in\fA$ such that $[Q, Q]=0$. Then it gives rise to a square-zero odd symplectic vector field on the space of BV fields $\cE$. In particular, we may modify the differential on $\cE$ by the replacement $\mathrm{d}\mapsto \mathrm{d}+Q$. Working up to perturbative equivalence, this turns out to drastically simplify the theory as we will shortly see.

The original classical field theory carried a $\ZZ\times\ZZ/2\ZZ$-grading, where $\ZZ$ is the cohomological (in the physics literature: ghost number) grading and $\ZZ/2\ZZ$ is the fermionic grading. We see that $\mathrm{d}$ has bidegree $(1, 0)$ while $Q$ has bidegree $(0, 1)$. So, in general the twisted theory is only $\ZZ/2\ZZ$-graded (with respect to the total grading). To improve that, we may additionally consider a homomorphism $\alpha\colon \U(1)\rightarrow G_R$ into the R-symmetry group under which $Q$ has weight $1$ and such that the $\alpha$-grading modulo 2 coincides with the fermionic grading. Then the $\alpha$-grading gives rise to a $\ZZ$-grading on the twisted theory.

Finally, let us observe that the original classical field theory carried an action of $\Spin(n)$ by rotations of $\RR^n$. But since $Q$ is not preserved under $\Spin(n)$, this action does not survive in the twisted theory. To improve that, we may consider a group $G$ with a \emph{twisting homomorphism} $G\rightarrow \Spin(n)\times G_R$ under which $Q$ is a scalar. Given such a twisting homomorphism, the twisted theory carries a $G$-action.

To summarize, supersymmetric twisting consists of the following three steps:

\begin{enumerate}
\item Choose a square-zero supercharge $Q\in\Sigma$ and modify the differential of the theory as $\mathrm{d}\mapsto\mathrm{d}+Q$.

\item Choose a group $G$ together with a twisting homomorphism $G\rightarrow \Spin(n)\times G_R$ under which $Q$ is scalar. To remove redundancy, we will assume $G\rightarrow \Spin(n)$ is an embedding.

\item Choose a homomorphism $\alpha\colon \U(1)\rightarrow G_R$ under which $Q$ has weight $1$ and such that the $\alpha$-grading modulo 2 is the fermionic grading. This step may not be possible in general.
\end{enumerate}

A classification of possible square-zero supercharges $Q$ was previously done in \cite{ElliottSafronov} and in this paper we use that classification to calculate the twist of super Yang--Mills theories on $\RR^n$ in all dimensions.

\subsection*{Supersymmetric Twists and Supergravity}

In this paper we only consider the case of global supersymmetry for super Yang--Mills theories on $\RR^n$. 
In certain cases one may consider coupling of super Yang--Mills to supergravity in which case there is an interpretation of the twisting procedure as performing perturbation theory in a nontrivial supergravity background. 
Let us briefly explain this perspective.

A classical solution of supergravity consists, in particular, of the following data: a spacetime manifold $M$, a $\Spin(n)$-bundle $P_\Spin\rightarrow M$ equipped with a connection (spin connection), a $G_R$-bundle $P_R\rightarrow M$ equipped with a connection and a ghost for supertranslations $\eta\in \Gamma(M, (P_\Spin\times P_R)\times^{\Spin(n)\times G_R}\Sigma)$.  
The ghost $\eta$ is bosonic: it lives in bidegree $(-1,1)$ for the $\ZZ \times \ZZ/2\ZZ$-grading, so it makes sense to give it a non-zero value.  
If we couple super Yang--Mills to supergravity, then the super Yang--Mills fields become sections of the associated bundles to $P_\Spin\times P_R$.

We have the following supergravity analogs of the data $(Q, \phi, \alpha)$ for supersymmetric twisting:
\begin{itemize}
\item The supergravity analog of the choice of a square-zero supercharge $Q$ is the value of the ghost $\eta$.

\item The supergravity analog of the twisting homomorphism is a choice of $G$-bundle $P_G\rightarrow M$ with connection so that $P_\Spin\times P_R$ is induced via the homomorphism $G\rightarrow \Spin(n)\times G_R$.

\item The supergravity analog of $\alpha\colon \U(1)\rightarrow G_R$ is a choice of trivial $\U(1)$-subbundle in $G_R$ on which the connection restricts to zero.
\end{itemize}

\subsection*{Applications to Quantization}

The quantization of gauge theories is notoriously subtle and requires a rich theory of renormalization.  One attractive application of the descriptions of the twists of supersymmetric gauge theories that we provide is to study quantization in a setting where the machinery required for renormalization is much more rigid.

To rigorously study the quantization of supersymmetric Yang--Mills theory we can work with the mathematical theory of renormalization developed by Costello in \cite{CostelloBook}. This theory of renormalization can been used to study field theories with and without supersymmetry: for example in \cite{CostelloWittengenus, LiLi, BCOV1, ChanLeungLi, GradyLiLi}. In the context of (non-supersymmetric) Yang--Mills theory for instance, it is shown that this theory recovers asymptotic freedom by an explicit analysis of the local counterterms present in the four-dimensional gauge theory \cite{EWY}.  

In principle, the existence of local counterterms can be used to analyze the full untwisted supersymmetric gauge theories in a mathematically rigorous way.  In practice, however, our approach to renormalization does not provide any significant advantage over traditional approaches used in QFT.  However, a significant simplification happens at the level of the {\em twisted} supersymmetric Yang-Mills theories that we study in this work. To start with, for some examples (but not all), the twisted theory turns out to be a {\em topological field theory}.  This occurs whenever the bracket $[Q,-]$ with the twisting supercharge surjects onto the space of translations. The theory of renormalization for topological theories can be handled using configuration spaces \cite{Kontsevich, AxelrodSinger}. 

In the general setting of this paper, while not every twist results in a topological field theory, it does result in a theory in which some directions of spacetime behave topologically and the remaining directions behave holomorphically. For a mixed holomorphic-topological translation invariant field theory of this type on $\RR^n \times \CC^d$, this means that at least half of the linearly independent translation invariant vector fields act on the field theory in a BRST exact way.

Inspired by the work of Costello and Li in \cite{BCOV1} and Li in \cite{LiFeynman, LiVertex}, the foundations of renormalization for mixed holomorphic-topological field theories on Euclidean space has been developed in \cite{BWhol}. The key result is that the renormalization for mixed holomorphic-topological theories is extremely well-behaved from an analytic perspective. It is shown in the cited work that, to first order in $\hbar$, the renormalization of mixed holomorphic/topological theory is {\em finite}. Furthermore, in \cite{LiVertex}, it is shown that in real dimension two this holds to all orders in $\hbar$. 

These results yield a practical approach to the problem of mathematically characterizing the one-loop quantization of every twist of supersymmetric Yang--Mills theory. Furthermore, in all examples of theories obtained via twisting occurring in dimensions 8 and lower, not much is lost when asking for the one-loop quantization.  The twisted gauge theories here are all either equivalent to BF-type theories (see Section \ref{gen_BF_section}) or deformations of such theories by a holomorphic differential operator.  Such theories admit prequantizations (that is, they define families of effective field theories compatible under renormalization group flow), which are exact at one loop, meaning all higher order corrections vanish identically.  

From this starting point, the first natural problem would be to verify whether these one-loop exact prequantizations define actual quantizations of the classical twisted field theory.  That is, for each such theory, to compute the one-loop anomalies to the solution of the quantum master equation.  This problem comes in two parts: first, to compute the one-loop anomaly to quantization of the theory on flat space $\RR^n \times \CC^d$, and second -- in the case where we can use a twisting homomorphism to define the twisted theory on certain structured $(n+2d)$-manifolds, to calculate the corresponding one-loop anomaly on curved space (in other words, incorporating the computation of a gravitational anomaly). We plan to return to this question in future work.

\subsection*{The Relationship to Factorization Algebras}
In our previous paper \cite{ElliottSafronov}, we discussed supersymmetric twisting with an emphasis not on the classical fields of a supersymmetric field theory, but instead on their classical or quantum \emph{observables}.  The factorization algebra formalism of Costello and Gwilliam \cite{Book1, Book2} provides a model for the local structure of the observables in a general quantum field theory.  In brief, for every open subset $U \sub M$ of the spacetime manifold, one associates a (possibly $\ZZ \times \ZZ/2\ZZ$-graded) vector space $\obs(U)$ of local observables on $U$.  For any pair $U_1, U_2 \sub V$ of disjoint open subsets of an open set $V$, one associates a morphism
\[m_{U_1, U_2}^V \colon \obs(U_1) \otimes \obs(U_2)\to \obs(V),\]
thought of as an \emph{operator product} for local observables.  These products should vary smoothly as one varies the open subsets $U_1, U_2$ and $V$.  Starting with a classical field theory on $M$, defined using the BV formalism, one can build a factorization algebra $\obs^{\mr{cl}}$ modelling the classical observables of the field theory.  If the classical field theory carries the action of a group $G$, so does the associated factorization algebra.  Furthermore, Costello and Gwilliam develop techniques for the quantization of such algebras of classical observables, using the theory of renormalization as discussed in the previous section.

In \cite{ElliottSafronov} we studied the supersymmetric twisting procedure as applied to factorization algebras on $\RR^n$ with an action of a supersymmetry algebra.  If $Q$ is a topological supercharge, then the $Q$-twist $\obs^Q$ of a supersymmetric factorization algebra automatically satisfies a strong translation invariance condition: all translations must act homotopically trivially.  In good circumstances, we can say even more.  An \emph{$\bb E_n$-algebra} is an algebra over the operad of little $n$-disks; in the language above, this can be obtained from a factorization algebra for which homotopy equivalent configurations $U_1, U_2 \sub V$ induce homotopy equivalent products.
\begin{nonum}[{\cite[Theorem 3]{ElliottSafronov}}]
If $Q$ is a topological supercharge, and the operator $\obs^Q(B_r(0)) \to \obs^Q(B_R(0))$ associated to the inclusion of concentric balls is an equivalence, then the factorization algebra $\obs^Q$ has the canonical structure of an $\bb E_n$-algebra.
\end{nonum}
The hypothesis of the theorem is automatically satisfied, for example, for superconformal theories, and should be concretely checkable in examples.

In the present work we classify twists of classical field theories, to which one can associate twisted factorization algebras of classical -- and, if the appropriate anomalies vanish, quantum -- observables in the sense of our previous work.  In some (topological) examples, these define $\bb E_n$-algebras.  In other examples, where the twist is not fully topological, the twisted local observables will define higher analogues of vertex algebras (as in, for instance, \cite{GwilliamWilliamsKM}).

\subsection*{Summary of Twisted Super Yang--Mills Theories}

In this section we will summarize the main results of the paper presented in Part \ref{classification_part}, where we calculate twists of super Yang--Mills theories in dimensions 2 through 10.

Let us begin by explaining what we mean by ``calculation''. Recall that for a Lie algebra $\fg$ there is a $d$-dimensional topological BF theory defined on a $d$-dimensional spacetime manifold $M$ with the space of BV fields $\Omega^\bullet(M; \fg)[1]\oplus \Omega^\bullet(M; \fg^*)[d-2]$, where $\Omega^\bullet$ denotes the space of differential forms equipped with the de Rham differential $\mathrm{d}$. If $M$ is replaced by a complex manifold $X$, we may also consider its version with the space of fields $\Omega^{\bullet, \bullet}(X; \fg)[1]\oplus \Omega^{\bullet, \bullet}(X; \fg^*)[2\dim(X)-2]$, where $\Omega^{\bullet, \bullet}$ is the space of differential forms equipped with the Dolbeault differential. Finally, we have yet another version, a \emph{holomorphic BF theory}, with the space of BV fields $\Omega^{0, \bullet}(X; \fg)[1]\oplus \Omega^{\dim(X), \bullet}(X; \fg^*)[\dim(X)-2]$, again equipped with the Dolbeault differential. We will denote the space of fields in these three examples as $T^*[-1]\map(M_{\mathrm{dR}}, B\fg)$, $T^*[-1]\map(X_{\mathrm{Dol}}, B\fg)$ and $T^*[-1]\map(X, B\fg)$ respectively (the notation is explained in Section \ref{sect:FMPs}). We may also combine these three examples into what we call a \emph{generalized BF theory} with the spaces of fields $T^*[-1]\map(X\times Y_{\mathrm{Dol}}\times M_{\mathrm{dR}}, B\fg)$ (see Definition \ref{def:generalizedBF} for more details).

Let us also recall that if $\fg$ is equipped with a symmetric bilinear nondegenerate pairing, we also have a 3-dimensional topological Chern--Simons theory. If we forgo $\ZZ$-gradings and work with $\ZZ/2\ZZ$-gradings, we may also consider a topological Chern--Simons theory in any odd dimension (see \cite{AlekseevMnev} for a 1-dimensional version and \cite{BakGustavsson2} for a 5-dimensional version). Just like for the BF theory, we also have two other versions which may be combined into a generalized Chern--Simons theory. Another direction we can generalize in is to replace the Lie algebra $\fg$ by a dg Lie algebra, in which case the BF theory itself becomes a particular example of the Chern--Simons theory.

Our goal will then be to show that a particular twist of super Yang--Mills is equivalent to a given generalized Chern--Simons theory. We summarize our results in two forms. In Tables \ref{table_of_twists_16}, \ref{table_of_twists_8} and \ref{table_of_twists_4} we summarize all the possible twists of dimensional reductions of the 10d $\cN=(1, 0)$, 6d $\cN=(1, 0)$ and 4d $\cN=1$ super Yang--Mills theories respectively.  In Table \ref{table_of_twists_2d} we summarize the twists of 2d supersymmetric Yang-Mills theory. Before these tables, we will give a short description of each twisted theory in a more physical language, with references to where in the literature it was previously considered.

\begin{remarknonum}
Let us briefly discuss a twisting construction that our classification will not include.  Given a supercharge $Q$ whose square is not zero, but instead generates an $S^1$-action on spacetime, one can construct a $Q$-twisted theory on the locus of $S^1$-fixed points.  This is Nekrasov's $\Omega$-background construction \cite{Nekrasovthesis,NekrasovSW}, see also the discussion in \cite{CostelloOmega}.  In \cite{Nekrasovthesis}, 3d and 5d topological Chern--Simons theories, as well as a 4d mixed holomorphic-topological Chern--Simons theory are discussed as arriving in this context.  In \cite{CostelloOmega} the 5-dimensional theory, and a non-commutative version thereof, is shown to arise from this construction.
\end{remarknonum}

\textbf{Dimension 10}
\begin{itemize}
 \item \emph{$\mc N=(1,0)$ holomorphic twist.} The unique twisted super Yang--Mills theory in 10 dimensions is equivalent to the 5d holomorphic Chern--Simons theory defined on a Calabi--Yau 5-fold. Note that this theory is only $\ZZ/2\ZZ$-graded. This twist was first studied by Baulieu \cite{Baulieu}. As is well-known \cite{GSanomaly,CostelloLiAnomaly}, the theory has a one-loop anomaly and does not admit a quantization.
\end{itemize}

\textbf{Dimension 9}
\begin{itemize}
 \item \emph{$\mc N=1$ minimal twist.} The unique twisted super Yang--Mills theory in 9 dimensions is equivalent to a generalized version of the Chern--Simons theory defined on a product of a Calabi--Yau 4-fold and a real 1-manifold. Note that this theory is only $\ZZ/2\ZZ$-graded. Its classical solutions are $G$-bundles holomorphic along the Calabi--Yau manifold and flat along the 1-manifold.
\end{itemize}

\textbf{Dimension 8}
\begin{itemize}
 \item \emph{$\mc N=1$ holomorphic twist.} Super Yang--Mills theory in 8 dimensions admits three classes of twists.  The minimal twist, by a holomorphic (or, equivalently, pure) supercharge, is equivalent to a holomorphic version of the BF theory defined on a complex 4-fold.
 \item \emph{$\mc N=1$ intermediate twist.} The holomorphic twist admits a deformation to a twist by a rank 1 impure spinor. This theory is equivalent to a generalized version of the Chern--Simons theory defined on a product of a Calabi--Yau 3-fold and an oriented surface. Note that this theory is only $\ZZ/2\ZZ$-graded.
 \item \emph{$\mc N=1$ topological twist.} The holomorphic twist also admits a deformation to a topological twist defined on $\spin(7)$-manifolds. This theory is perturbatively trivial, in the sense that the classical BV complex is contractible. It was studied in \cite{AcharyaOLoughlinSpence,BaulieuKannoSinger}. The partition function of this theory counts $\spin(7)$-instantons modulo gauge \cite{Lewis,DonaldsonThomas,ReyesCarrion}. If we denote by $\Omega$ the Cayley 4-form on a $\spin(7)$-manifold $M$, then the classical solutions are given by principal $G_\RR$-bundles on $M$ (where $G_\RR$ is a compact Lie group) together with a connection $A$, such that its curvature $F_A$ satisfies the equation
\begin{equation} 
 F = \ast (\Omega\wedge F).
 \end{equation}
\end{itemize}

\textbf{Dimension 7}
\begin{itemize}
 \item \emph{$\mc N=1$ minimal twist.} The twists of super Yang--Mills theory in 7 dimensions arise by dimensionally reducing the twists in 8 dimensions.  The minimal twist, by a pure spinor, is equivalent to a generalized version of the BF theory defined on a product of a complex 3-fold and a real 1-manifold.
 \item \emph{$\mc N=1$ intermediate twist.} The minimal twist admits a deformation to a twist by a rank 1 impure spinor. This theory is equivalent to a generalized version of the Chern--Simons theory defined on a product of a Calabi--Yau surface and a real oriented 3-manifold. Note that this theory is only $\ZZ/2\ZZ$-graded.
 \item \emph{$\mc N=1$ topological twist.} The minimal twist also admits a deformation to a topological twist defined on $G_2$-manifolds. This theory is again perturbatively trivial. It was studied in \cite{AcharyaOLoughlinSpence, BaulieuKannoSinger}. The partition function of this theory counts $G_2$-monopoles modulo gauge \cite{DonaldsonSegal}. If we denote by $\psi$ the calibration 4-form on a $G_2$-manifold $M$, then the classical solutions are given by principal $G_\RR$-bundles $P\rightarrow M$ together with a connection $A$ and a section $\sigma\in\Gamma(M, \mathrm{ad} P)$ satisfying
\begin{equation}
\d_A\sigma = \ast(F\wedge \psi).
\end{equation}
\end{itemize}

\textbf{Dimension 6}
\begin{itemize}
 \item \emph{$\mc N=(1,0)$ and $\mc N=(1,1)$ holomorphic twist.} The holomorphic twist of the 6d $\cN=(1, 0)$ super Yang--Mills theory with matter valued in a symplectic $G$-representation $U$ is equivalent to the theory whose classical solutions are holomorphic maps from a Calabi--Yau 3-fold $X$ to the Hamiltonian reduction of $U$ (a holomorphic version of the gauged Rozansky--Witten model). In general, this theory is only $\ZZ/2\ZZ$-graded. If $U=T^*R$, the theory is $\ZZ$-graded and is the cotangent theory to the space of holomorphic maps from a complex 3-fold to $R/G$. 6d $\cN=(1, 1)$ super Yang--Mills corresponds to the special case $R=\fg$. This twist is studied in \cite{CostelloYagi,Butson}.
 \item \emph{$\mc N=(1,1)$ rank $(2, 2)$ twist.} In the $\mc N=(1,1)$ case there are two intermediate twists. The one by a supercharge of rank $(2, 2)$ is equivalent to a generalized version of the Chern--Simons theory defined on a product of a Calabi--Yau curve and a real oriented 4-manifold. Note that this theory is only $\ZZ/2\ZZ$-graded.
 \item \emph{$\mc N=(1,1)$ special rank $(1,1)$ twist.} The other intermediate twist, by a rank $(1, 1)$ supercharge, is equivalent to a generalized form of the BF theory defined on a product of a complex surface and a real surface.
 \item \emph{$\mc N=(1,1)$ topological twist.} The special rank $(1, 1)$ twist admits a deformation to a topological twist defined on Calabi--Yau 3-folds. This theory is perturbatively trivial. It was studied in \cite{AcharyaOLoughlinSpence, BaulieuKannoSinger}. The partition function of this theory counts solutions to the Donaldson--Thomas equations \cite{Thomas}. If we denote by $\omega$ the K\"ahler form on a Calabi--Yau 3-fold $M$, then the classical solutions are given by principal $G_\RR$-bundles $P\rightarrow M$ together with a connection $A$ and a 3-form $u\in\Omega^{0, 3}(M, \mathrm{ad} P\otimes_\RR \CC)$ satisfying
\begin{align}
F_{0, 2} + \ol\dd^*_A u &= 0 \\
F_{1, 1}\wedge \omega^2 + [u, \bar{u}] &= 0.
\end{align}
\end{itemize}

\textbf{Dimension 5}
\begin{itemize}
 \item \emph{$\mc N=1$ and $\mc N=2$ minimal twist.} The minimal twist of the 5d $\cN=1$ super Yang--Mills theory with matter valued in a symplectic $G$-representation $U$ is equivalent to a theory defined on a product of a Calabi--Yau surface $X$ and a real 1-manifold $M$. The classical solutions are given by maps from $X\times M$ to the symplectic reduction of $U$ holomorphic along $X$ and locally-constant along $M$. In general, this theory is only $\ZZ/2\ZZ$-graded. It was studied by K\"all\'en and Zabzine \cite{KallenZabzine}. If $U=T^* R$, the theory is $\ZZ$-graded and is the cotangent theory to the space of maps from $X\times M$ to $R/G$ holomorphic along $X$ and locally-constant along $M$. 5d $\cN=2$ super Yang--Mills corresponds to the special case $R=\fg$.
 \item \emph{$\mc N=2$ intermediate twist.} In the $\mc N=2$ case the minimal twist admits a deformation to an intermediate twist. The corresponding theory is equivalent to a generalized BF theory defined on a product of a complex curve $C$ and a real 3-manifold $M$. The twist was considered in \cite{ElliottPestun} in the case $M=S^1\times \Sigma$ for a Riemann surface $\Sigma$, where the moduli space can be viewed as a multiplicative version of the Hitchin system.
 \item \emph{$\mc N=2$ topological A twist.} The intermediate twist admits a deformation to two topological twists. One, which we refer to as the A-twist, arises by dimensionally reducing the topological twist of 6d $\mc N=(1, 1)$ super Yang--Mills theory. This theory is perturbatively trivial. It was studied by Qiu and Zabzine, \cite{QiuZabzine} (see also \cite{Anderson} for a discussion of the twisting homomorphism). The partition function of this theory counts solutions to the Haydys--Witten equations \cite{Haydys,WittenFivebranes}. Let $M$ be a $K$-contact manifold and denote by $R$ the Reeb vector field. The classical solutions are given by principal $G_\RR$-bundles $P\rightarrow M$ together with a connection $A$ and a section $B\in\Omega^2(M; \mathrm{ad} P)$ satisfying a self-duality equation $\iota_R \ast B = B$ which together satisfy the following equations (see \cite[equations (4) and (5)]{QiuZabzine}; we refer there for the explanation of the notation):
\begin{align}
\iota_R F - (d_A^* B)^H &= 0 \\
F^+_H - \frac{1}{4}B\times B - \frac{1}{2}\iota_R d_A B &= 0.
\end{align} 
 \item \emph{$\mc N=2$ topological B twist.} Finally, the other topological twist, associated to a rank 4 supercharge, can be identified with a 5d Chern--Simons theory defined on an oriented 5-manifold. Note that this theory is only $\ZZ/2\ZZ$-graded. This twist was identified in work of Geyer--M\"ulsch and of Bak--Gustavsson \cite{GeyerMuelsch, BakGustavsson1,BakGustavsson2}.
\end{itemize}

\textbf{Dimension 4}
\begin{itemize}
 \item \emph{$\mc N=1$ holomorphic twist.} The holomorphic twist of the 4d $\cN=1$ super Yang--Mills theory with matter valued in a $G$-representation $R$ is equivalent to the cotangent theory of the theory of holomorphic maps from a complex surface $X$ to $R/G$. This twist was studied by Johansen \cite{Johansen} (see also \cite{BaulieuTanzini,CostelloSUSY}).
 \item \emph{$\cN=2$ and $\cN=4$ holomorphic twist.} We may also consider holomorphic twists of 4d $\cN=2$ and 4d $\cN=4$ super Yang--Mills theories. The holomorphic twist of 4d $\cN=2$ super Yang--Mills with matter valued in a symplectic $G$-representation $U$ is equivalent to the cotangent theory of the theory of holomorphic maps from a Calabi--Yau surface $X$ to the Hamiltonian reduction of $U$. The d $\cN=4$ super Yang--Mills theory corresponds to the case $U=T^*\fg$ in which case the space of classical solutions is a $(-1)$-shifted cotangent bundle to the moduli stack of $G$-Higgs bundles on a complex surface $X$.
 \item \emph{$\mc N=2$ and $\mc N=4$ intermediate twist.} There is a deformation of the $\cN=2$ holomorphic twist which is equivalent to a theory of maps from a product of a Calabi--Yau curve $C$ and an oriented surface $\Sigma$ into the Hamiltonian reduction of $U$ which are holomorphic along $C$ and locally-constant along $\Sigma$. This twist was previously studied by Kapustin \cite{KapustinHolo}.
 \item \emph{$\mc N=2$ topological rank (2, 0) twist.} The $\cN=2$ holomorphic twist admits a deformation to a topological twist, the \emph{Donaldson twist}. This theory is perturbatively trivial. This theory was first considered in \cite{WittenTQFT}, and the coupling to matter was studied in \cite{AnselmiFre, AlvarezLabastida, HyunParkPark}. The partition function counts solutions to nonabelian Seiberg--Witten equations \cite{Pidstrigach}. Let $G_\RR$ be a compact Lie group and $U$ a quaternionic-unitary $G_\RR$-representation. In particular, $U$ carries a commuting $\SU(2)$-action given by unit quaternions. Suppose $M$ is a spin 4-manifold and let $P_{\Spin}\rightarrow M$ be the corresponding $\Spin(4)$-principal bundle. The classical solutions in this theory are given by principal $G_\RR$-bundles $P\rightarrow M$ together with a connection $A$ and a section $u\in\Gamma(M, (P\times P_{\Spin})\times^{G_\RR\times \Spin(4)} U)$ which together satisfy
\begin{align}
\sd \d_A u &= 0 \label{eq:SW1} \\
F^+ + \Phi(u) &= 0, \label{eq:SW2}
\end{align}
where $\Phi$ is the moment map and $\sd\d_A$ is the Dirac operator (we refer to \cite{Pidstrigach,HaydysDirac} for more details).

\item \emph{$\cN=4$ topological rank $(2, 0)$ twist.} The same twist may be considered for the 4d $\cN=4$ super Yang--Mills theory, in which case it has three compatible twisting homomorphisms \cite{Yamron}, i.e. there are three ways of interpreting the differential equations on arbitrary oriented 4-manifolds. First, considering the 4d $\cN=4$ super Yang--Mills theory as a 4d $\cN=2$ super Yang--Mills theory with $U=\fg\otimes\mathbb{H}$, we obtain a theory which counts solutions to the nonabelian Seiberg--Witten equations \eqref{eq:SW1}, \eqref{eq:SW2}. Another twisting homomorphism was studied by Vafa and Witten \cite{VafaWitten}. The corresponding theory counts solutions of the Vafa--Witten equations on an oriented 4-manifold $M$. The classical solutions are given by principal $G_\RR$-bundles $P\rightarrow M$ together with a connection $A$, a self-dual two-form $B\in\Omega^{2, +}(M, \mathrm{ad} P)$ and a section $C\in\Gamma(M, \mathrm{ad} P)$ which satisfy
\begin{align}
-\d_A C + \d_A^* B &= 0 \\
F^+ - \frac{1}{4}B\times B - \frac{1}{2}[C, B] &= 0.
\end{align}
Finally, the third twisting homomorphism was studied by Marcus \cite{Marcus} and Kapustin and Witten \cite{KapustinWitten}. The classical solutions are given by principal $G_\RR$-bundles $P\rightarrow M$ together with a connection $A$ and a one-form $\phi\in\Omega^1(M, \mathrm{ad} P)$ which satisfy
\begin{align}
(F-\phi\wedge \phi)^+ &= 0\\
(\d_A \phi)^- &= 0 \\
\d^*_A \phi &= 0.
\end{align}

 \item \emph{$\mc N=4$ topological B twist.} For the 4d $\cN=4$ super Yang--Mills theory there is a single topological twist which is not perturbatively trivial. It is equivalent to a topological BF theory defined on a 4-manifold (this theory corresponds to the value $t=\pm i$ of the family considered in \cite{KapustinWitten}).
 
 \item \emph{$\cN=4$ topological rank $(2, 2)$ twist.} The $\cN=4$ topological B twist admits a deformation to a perturbatively trivial theory. The corresponding deformation is parameterized by $s\in\CC^\times$, where $s=1$ is the topological B twist. Choosing a parameter $t\in\CC^\times$ satisfying $s=-t^2$, the theory counts solutions of the Kapustin--Witten equations on an oriented 4-manifold $M$. The classical solutions are given by principal $G_\RR$-bundles $P\rightarrow M$ together with a connection $A$ and a one-form $\phi\in\Omega^1(M, \mathrm{ad} P)$ which satisfy
\begin{align}
(F-\phi\wedge\phi +t \d_A \phi)^+ &= 0 \\
(F-\phi\wedge\phi -t^{-1}\d_A \phi)^- &= 0 \\
\d^*_A \phi &= 0.
\end{align}

\item \emph{$\cN=4$ topological rank $(2, 1)$ twist.} The $\cN=4$ intermediate twist also admits a deformation to a perturbatively trivial theory defined on K\"ahler surface $M$ given by twisting by a rank $(2, 1)$ supercharge. The corresponding equation is a deformation of the Kapustin--Witten equations using the K\"ahler form.

An analysis of topological twists of the 4d $\cN=4$ super Yang--Mills theory using similar techniques to this paper, but with the aim of obtaining the full derived stack of solutions to the equations of motion, rather than only the perturbative classical field theory, was carried out in \cite{ElliottYoo1}.
\end{itemize}

\textbf{Dimension 3}
\begin{itemize}
 \item \emph{$\cN=2$ minimal twist.} The minimal twist of the 3d $\cN=2$ super Yang--Mills theory with matter valued in a $G$-representation $R$ is defined on a product $C\times L$ of a complex curve $C$ and a real 1-manifold $L$. It is equivalent to the cotangent theory of the theory of maps from $C\times L$ to $R/G$ which are holomorphic along $C$ and locally-constant along $L$. This twist was studied in \cite{ACMV}.
 \item \emph{$\cN=4$ and $\cN=8$ minimal twist.} We may also consider the minimal twist of $\cN=4$ and $\cN=8$ super Yang--Mills theories. The minimal twist of the $\cN=4$ super Yang--Mills theory with matter valued in a symplectic $G$-representation $U$ is equivalent to the cotangent theory of the theory of maps from $C\times L$ to the Hamiltonian reduction of $U$ which are holomorphic along $C$ and locally-constant along $L$. The 3d $\cN=8$ theory corresponds to the case $U=T^*\fg$.

 \item \emph{$\mc N=4$ topological A twist.} In 3d $\cN=4$ super Yang--Mills theory we may consider a deformation of the minimal twist which gives rise to a perturbatively trivial topological theory defined on spin 3-manifolds. This twist was studied in \cite{BaulieuGrossman,BlauThompson1,Ohta}. From a mathematical point of view the space of states on a two-sphere is studied in \cite{BFN}. The partition function counts solutions to a 3-dimensional version of the nonabelian Seiberg--Witten equations \eqref{eq:SW1}, \eqref{eq:SW2}. Let $G_\RR$ be a compact Lie group and $U$ a quaternionic-unitary $G_\RR$-representation. Let $M$ be a spin 3-manifold and let $P_{\Spin}\rightarrow M$ be the corresponding $\Spin(3)$-principal bundle. The classical solutions in this theory are given by principal $G_\RR$-bundles $P\rightarrow M$ together with a connection $A$, a section $\sigma\in\Gamma(M, \mathrm{ad} P)$ and a section $u\in\Gamma(M, (P\times P_{\Spin})\times^{G_\RR\times \Spin(3)} U)$ which together satisfy
\begin{align}
\sd \d_A u + [\sigma, u]  &= 0 \label{eq:3dSW1} \\
F + \ast\d_A \sigma + \Phi(u) &= 0. \label{eq:3dSW2}
\end{align}

\item \emph{$\cN=8$ topological A twist.} We may regard the 3d $\cN=8$ super Yang--Mills theory as a 3d $\cN=4$ super Yang--Mills theory with matter valued in $U=\fg\otimes\mathbb{H}$. In particular, the partition function in the twisted theory counts solutions to the equations \eqref{eq:3dSW1}, \eqref{eq:3dSW2}. We may also consider a different twisting homomorphism obtained by dimensionally reducing the Vafa--Witten or Kapustin--Witten twisting homomorphism. The classical solutions in this theory are given by principal $G$-bundles $P\rightarrow M$ ($G$ is the complexification of the compact Lie group $G_\RR$) together with a connection $A$ and a section $\sigma\in\Gamma(M, \mathrm{ad} P)$ satisfying a complexified version of the Bogomolny equation:
\[
F + \ast \d_A \sigma = 0.
\]
The corresponding field theory in the formalism of extended topological field theories is studied in \cite{BZGN}.

 \item \emph{$\mc N=4$ and $\mc N=8$ topological B twist.} The minimal twist of the 3d $\cN=4$ super Yang--Mills theory also admits another deformation to a gauged version of the Rozansky--Witten model valued in the Hamiltonian reduction $U\ham G$ \cite{RozanskyWitten,BlauThompson2}.
 \end{itemize}

\textbf{Dimension 2}
\begin{itemize}
 \item \emph{$\mc N=(2,2), (4,4), (8,8)$ holomorphic twist.} There is a holomorphic twist in dimension 2 which is defined on complex curves $C$. The twist of 2d $\cN=(2, 2)$ super Yang--Mills theory with matter valued in a $G$-representation $R$ is equivalent to the cotangent theory to the theory of holomorphic maps from $C$ to $R/G$. The twist of 2d $\cN=(4, 4)$ super Yang--Mills theory with matter valued in a symplectic $G$-representation $U$ is equivalent to the cotangent theory to the theory of holomorphic maps from $C$ to the Hamiltonian reduction $U\ham G$. Finally, the case of 2d $\cN=(8, 8)$ super Yang--Mills theory corresponds to choosing $U=T^*\fg$.
 \item \emph{$\mc N=(2,2), (4,4), (8,8)$ topological A twist.} In each of these cases, the minimal twist can again be deformed to a topological theory in two inequivalent ways.  The first is a perturbatively trivial theory, the gauged A-model. We begin with a description of the twist of the 2d $\cN=(2, 2)$ super Yang--Mills theory with gauge group $G_\RR$ (a compact Lie group) and matter valued in a unitary $G_\RR$-representation $R$ equipped with a moment map $\Phi$. The partition function counts solutions to symplectic vortex equations \cite{CGMRS}. Let $\Sigma$ be an oriented surface equipped with an almost complex structure and a square root $S$ of the line bundle of densities $\Dens_\Sigma$. The classical solutions are given by principal $G_\RR$-bundles $P\rightarrow \Sigma$ equipped with a connection $A$ and a section $u\in\Gamma(M, ((P\times^{G_\RR} R)\otimes_\RR S)$ which satisfy
\begin{align}
\ol\dd_A u &= 0 \label{eq:symplecticvortex1} \\
F + \Phi(u) &= 0. \label{eq:symplecticvortex2}
\end{align}
Next, consider the twist of the 2d $\cN=(4, 4)$ super Yang--Mills theory with a complexified gauge group $G$ and matter valued in a complex symplectic $G$-representation $U$ equipped with a moment map $\Phi$. In this case the classical solutions are given by principal $G$-bundles $P\rightarrow \Sigma$ equipped with a connection $A$ and a section $u\in\Gamma(\Sigma, (P\times^G U)\otimes_\RR S)$ which satisfy a complexified version of \eqref{eq:symplecticvortex1}, \eqref{eq:symplecticvortex2}:
\begin{align}
\ol\dd_A u &= 0 \\
F + \Phi(u) &= 0.
\end{align}
Finally, consider the twist of the 2d $\cN=(8, 8)$ super Yang--Mills theory. The classical solutions are given by principal $G$-bundles $P\rightarrow \Sigma$ equipped with a connection $A$ and sections $u_1\in\Gamma(\Sigma, \mathrm{ad} P)$, $u_2\in\Omega^{1, 1}(\Sigma, \mathrm{coad} P)$ which satisfy
\begin{align}
\ol\dd_A u_1 &= 0 \\
\ol\dd_A u_2 &= 0 \\
F + 2(u_1, u_2) &= 0.
\end{align}

 \item \emph{$\mc N=(2,2), (4,4), (8,8)$ topological B twist.} The other topological twist gives rise to a gauged B-model. The twist of the 2d $\cN=(2, 2)$ super Yang--Mills theory with complexified gauge group $G$ and matter valued in a $G$-representation $R$ is equivalent to the cotangent theory to the theory of locally-constant maps from a surface $\Sigma$ to $R/G$. The twist of the 2d $\cN=(4, 4)$ super Yang--Mills theory with complexified gauge group $G$ and matter valued in a symplectic $G$-representation $U$ is equivalent to the cotangent theory of the theory of locally-constant maps from $\Sigma$ to the Hamiltonian reduction $U\ham G$. Finally, the case of the 2d $\cN=(8, 8)$ super Yang--Mills theory corresponds to choosing $U=T^*\fg$. The study of the topological twists of 2d $\mc N=(2,2)$ supersymmetric field theories goes back to the works of Eguchi and Yang \cite{EguchiYang} and Witten \cite{Wittenmirror}.
 \item \emph{$(\cN, 0)$ holomorphic twist.} Theories with chiral supersymmetry in 2 dimensions (i.e. with 2d $(\cN, 0)$ supersymmetry) only admit a holomorphic twist. The corresponding twisted theory is equivalent to a cotangent theory to the theory of holomorphic maps from a complex curve $C$ to $\fg^{\cN-2} / G$. Twisted 2-dimensional $(2,0)$ $\sigma$-models were first studied by Witten in \cite{Wittenmirror}, and can be used to obtain the chiral algebra of chiral differential operators \cite{WittenCDO}. 
\end{itemize}

\begin{table}[htbp]
 \centering
 \begin{tabular}{c|c|c|c|c}
 $d$ & $\mc N$ & Twist & Description & Invariant Directions \\
 \hline
 \multirow{2}*{10} & \multirow{2}*{$(1,0)$} & \multirow{2}*{\hyperref[sect:10dholomorphictwist]{Rank $(1,0)$}} & Holomorphic Chern--Simons Theory & \multirow{2}*{5 (holomorphic)} \\
 &&&$\mr{Map}(\CC^5, B\gg)$ &\\ \hline
 \multirow{2}*{9} & \multirow{2}*{1} & \multirow{2}*{\hyperref[sect:9dminimaltwist]{Rank 1}} & Generalized Chern--Simons Theory & \multirow{2}*{5 (minimal)} \\
 &&&$\mr{Map}(\CC^4 \times \RR_{\mr{dR}}, B\gg)$ &\\ \hline
 \multirow{6}*{8} & \multirow{6}*{1} & \multirow{2}*{\hyperref[sect:8dholomorphictwist] {Rank $(1,0)$ pure}} & Holomorphic BF Theory & \multirow{2}*{4 (holomorphic)} \\
 &&& $T^*[-1]\mr{Map}(\CC^4, B\gg)$ & \\ \cline{3-5}
 && \multirow{2}*{\hyperref[sect:8dpartiallytopologicaltwist]{Rank $(1,1)$}} & Generalized Chern--Simons Theory& \multirow{2}*{5} \\
 &&& $\mr{Map}(\CC^3 \times \RR^2_{\mr{dR}}, B\gg)$ & \\ \cline{3-5}
  && \multirow{2}*{\hyperref[sect:8dtopologicaltwist]{Rank $(1,0)$ impure}} & Perturbatively trivial ($\Spin(7)$ Instanton)& \multirow{2}*{8 (topological)} \\
 &&& $\mr{Map}(\CC^4, B\gg)_\mr{dR}$ & \\ \hline
  \multirow{6}*{7} & \multirow{6}*{1} & \multirow{2}*{\hyperref[sect:7dminimaltwist] {Rank $1$ pure}} & Generalized BF Theory & \multirow{2}*{4 (minimal)} \\
 &&& $T^*[-1]\mr{Map}(\CC^3 \times \RR_{\mr{dR}}, B\gg)$ & \\ \cline{3-5}
 && \multirow{2}*{\hyperref[sect:7dpartialtwist] {Rank $2$}} & Generalized Chern--Simons Theory& \multirow{2}*{5} \\
 &&& $\mr{Map}(\CC^2 \times \RR^3_{\mr{dR}}, B\gg)$ & \\ \cline{3-5}
  && \multirow{2}*{\hyperref[sect:7dtopologicaltwist]{Rank $1$ impure}} & Perturbatively trivial ($G_2$ Monopole)& \multirow{2}*{7 (topological)} \\
 &&& $\mr{Map}(\CC^3 \times \RR_{\mr{dR}}, B\gg)_\mr{dR}$ & \\ \hline
 \multirow{8}*{6} & \multirow{8}*{$(1,1)$} & \multirow{2}*{\hyperref[sect:6d11holomorphictwist]{Rank $(1,0)$}} & {Holomorphic BF Theory} & \multirow{2}*{3 (holomorphic)} \\
 &&& $T^*[-1]\mr{Map}(\CC^3, \gg/\gg)$ & \\ \cline{3-5}
 && \multirow{2}*{\hyperref[sect:6d11partialtwist]{Rank $(1,1)$ special}} & {Generalized BF Theory} & \multirow{2}*{4} \\
 &&& $T^*[-1]\mr{Map}(\CC^2 \times \RR^2_{\mr{dR}}, B\gg)$ & \\ \cline{3-5}
 && \multirow{2}*{\hyperref[sect:6drank22twist]{Rank $(2,2)$}} & {Generalized Chern--Simons Theory} & \multirow{2}*{5} \\
 &&& $\mr{Map}(\CC \times \RR^4_{\mr{dR}}, B\gg)$ & \\ \cline{3-5}
 && \multirow{2}*{\hyperref[sect:6d11topologicaltwist]{Rank $(1,1)$ generic}} & {Perturbatively trivial} & \multirow{2}*{6 (topological)} \\
 &&& $\mr{Map}(\CC^2 \times \RR^2_{\mr{dR}}, B\gg)_{\mr{dR}}$ & \\ \hline
 \multirow{8}*{5} & \multirow{8}*{$2$} & \multirow{2}*{\hyperref[sect:5dminimaltwist]{Rank $1$}} & {Generalized BF Theory} & \multirow{2}*{3 (minimal)} \\
 &&& $T^*[-1]\mr{Map}(\CC^2 \times \RR_{\mr{dR}}, \gg/\gg)$ & \\ \cline{3-5}
 && \multirow{2}*{\hyperref[sect:5dpartialtwist]{Rank $2$ special}} & {Generalized BF Theory} & \multirow{2}*{4} \\
 &&& $T^*[-1]\mr{Map}(\CC \times \RR^3_{\mr{dR}}, B\gg)$ & \\ \cline{3-5}
 && \multirow{2}*{\hyperref[sect:5drank4twist]{Rank $4$}} & {5d Chern--Simons Theory} & \multirow{2}*{5 (topological)} \\
 &&& $\map(\RR^5_{\mr{dR}}, B\fg)$ & \\ \cline{3-5}
 && \multirow{2}*{\hyperref[sect:5drank2topologicaltwist] {Rank $2$ generic}} & {Perturbatively trivial} & \multirow{2}*{5 (topological)} \\
 &&& $\mr{Map}(\CC \times \RR^3_{\mr{dR}}, B\gg)_{\mr{dR}}$ & \\ \hline
 \multirow{12}*{4} & \multirow{12}*{$4$} & \multirow{2}*{\hyperref[sect:4d4holomorphictwist] {Rank $(1,0)$}} & {Holomorphic BF Theory} & \multirow{2}*{2 (holomorphic)} \\
 &&& $T^*[-1]\mr{Map}(\CC^2_{\mr{Dol}}, B\gg)$ & \\ \cline{3-5}
 && \multirow{2}*{\hyperref[sect:4d4partialtwist] {Rank $(1,1)$}} & Generalized BF Theory & \multirow{2}*{3} \\
 &&& $T^*[-1]\mr{Map}(\CC_{\mr{Dol}} \times \RR^2_{\mr{dR}}, B\gg)$ & \\ \cline{3-5}
 && \multirow{2}*{\hyperref[sect:4dqgltwist] {Rank $(2,2)$ special}} & BF Theory  & \multirow{2}*{4 (topological)} \\
 &&& $T^*[-1]\mr{Map}(\RR^4_{\mr{dR}}, B\gg)$ & \\ \cline{3-5}
 && \multirow{2}*{\hyperref[sect:4d4partialtwist] {Rank $(2,1)$}} & {Perturbatively trivial} & \multirow{2}*{4 (topological)} \\
 &&& $\mr{Map}(\CC_{\mr{Dol}} \times \RR^2_{\mr{dR}}, B\gg)_{\mr{dR}}$ & \\ \cline{3-5}
 && \multirow{2}*{\hyperref[sect:4d4Atwist] {Rank $(2,0)$}} & {Perturbatively trivial } & \multirow{2}*{4 (topological)} \\
 &&& $\mr{Map}(\CC^2_{\mr{Dol}}, B\gg)_{\mr{dR}}$ & \\ \cline{3-5}
  && \multirow{2}*{\hyperref[sect:4dqgltwist] {Rank $(2,2)$ generic}} & {Perturbatively trivial} & \multirow{2}*{4 (topological)} \\
 &&& $\mr{Map}(\RR^4_{\mr{dR}}, B\gg)_{\mr{dR}}$ & \\ \hline
  \multirow{8}*{3} & \multirow{8}*{$8$} & \multirow{2}*{\hyperref[sect:3d8minimal_twist] {Rank $1$}} & {Generalized BF Theory} & \multirow{2}*{2 (minimal)} \\
 &&& $T^*[-1]\mr{Map}(\CC_{\rm Dol} \times \RR_{\mr{dR}}, \gg/\gg)$
 & \\ \cline{3-5}
 && \multirow{2}*{\hyperref[cor:3dN8Btwist] {Rank $2$ (B)}} & {BF Theory } & \multirow{2}*{3 (topological)} \\
 &&& $T^*[-1]\mr{Map}(\RR^3_{\mr{dR}}, \gg/\gg)$ & \\ \cline{3-5}
 && \multirow{2}*{\hyperref[cor:3dN8Atwist] {Rank $2$ (A)}} & {Perturbatively trivial} & \multirow{2}*{3 (topological)} \\
 &&& $\mr{Map}(\RR^3_{\mr{dR}}, \gg/\gg)_{\mr{dR}}$ & \\ \hline
 \end{tabular}
 \caption{Twists of Maximally Supersymmetric Pure Yang--Mills Theories with Lie algebra $\fg$ (16 supercharges).}
 \label{table_of_twists_16}
\end{table}

\begin{table}[!ht]
 \centering
 \begin{tabular}{c|c|c|c|c}
 $d$ & $\mc N$ & Twist & Description & Invariant Directions \\
 \hline
 \multirow{2}*{6} & \multirow{2}*{$(1,0)$} & \multirow{2}*{\hyperref[sect:6dholomorphictwist]{Rank $(1,0)$}} & {Holomorphic BF Theory coupled to a holomorphic symplectic boson} & \multirow{2}*{3 (holomorphic)} \\
 &&& $\Sect(\CC^3, (U\otimes K_{\CC^3}^{1/2}) \ham \fg)$ & \\ \hline
 \multirow{2}*{5} & \multirow{2}*{$1$} & \multirow{2}*{\hyperref[sect:5d1minimaltwist] {Rank $1$}} & {Generalized BF Theory coupled to a generalized symplectic boson} & \multirow{2}*{3 (minimal)} \\
 &&& $\Sect(\CC^2 \times \RR_{\mr{dR}}, (U\otimes K_{\CC^2}^{1/2}) \ham \fg)$ & \\ \hline
 \multirow{6}*{4} & \multirow{6}*{$2$} & \multirow{2}*{\hyperref[sect:4d_2_holomorphictwist] {Rank $(1,0)$}} & {Holomorphic BF Theory} & \multirow{2}*{2 (holomorphic)} \\
 &&& $T^*[-1]\mr{Sect}(\CC^2, (U \otimes K_{\CC^2}^{1/2}) \ham \fg)$ & \\ \cline{3-5}
 && \multirow{2}*{\hyperref[sect:4d_2_11] {Rank $(1,1)$}} & {Generalized BF Theory coupled to a generalized symplectic boson} & \multirow{2}*{3} \\
 &&& $\Sect(\CC \times \RR^2_{\mr{dR}}, (U\otimes K_\CC^{1/2}) \ham \fg)$  & \\ \cline{3-5}
 && \multirow{2}*{\hyperref[sect:4d2Donaldson] {Rank $(2,0)$}} & {Perturbatively trivial } & \multirow{2}*{4 (topological)} \\
 &&& $\mr{Sect}(\CC^2, (U\otimes K_{\CC^2}^{1/2}) \ham \fg)_{\mr{dR}}$ & \\ \hline
 \multirow{6}*{3} & \multirow{6}*{$4$} & \multirow{2}*{\hyperref[sect:3d_4_minimal_twist] {Rank $1$}} & {Generalized BF Theory coupled to a generalized symplectic boson} & \multirow{2}*{2 (minimal)} \\
 &&& $T^*[-1]\mr{Sect}(\CC \times \RR_{\mr{dR}}, (U\otimes K_\CC^{1/2}) \ham \fg)$  & \\ \cline{3-5}
 && \multirow{2}*{\hyperref[sect:3d_4_B_twist] {Rank $2$ (B)}} & {BF Theory coupled to a symplectic boson} & \multirow{2}*{3 (topological)} \\
 &&& $\map(\RR^3_{\mr{dR}}, U \ham \fg)$ & \\ \cline{3-5}
 && \multirow{2}*{\hyperref[sect:3d_4_A_twist] {Rank $2$ (A)}} & {Perturbatively trivial } & \multirow{2}*{3 (topological)} \\
 &&& $\mr{Sect}(\CC \times \RR_{\mr{dR}}, (U\otimes K_\CC^{1/2}) \ham \fg)_{\mr{dR}}$ & \\ \hline
  \end{tabular}
 \caption{Twists of Supersymmetric Yang--Mills Theories with gauge Lie algebra $\fg$ with a hypermultiplet valued in a symplectic representation $U$ (8 supercharges).}
 \label{table_of_twists_8}
\end{table}

\begin{table}[hbp]
 \centering
 \begin{tabular}{c|>{\centering}m{5ex}|c|>{\centering}m{65ex}|c}
 $d$ & $\mc N$ & Twist & Description & Invariant Directions \\
 \hline
 \multirow{2}*{4} & \multirow{2}*{$1$} & \multirow{2}{*}{\hyperref[sect:4d1holomorphictwist] {Rank $(1,0)$}} & Holomorphic BF Theory coupled to $R$-matter & \multirow{2}*{2 (holomorphic)} \\
 &&& $T^*[-1]\mr{Map}(\CC^2, R/\fg)$ & \\ \hline
 \multirow{2}*{3} & \multirow{2}*{$2$} & \multirow{2}*{\hyperref[sect:3dminimaltwist] {Rank $1$}} & {Generalized BF Theory coupled to $R$-matter} & \multirow{2}*{2 (minimal)} \\
 &&&  $T^*[-1]\mr{Map}(\CC \times \RR_{\mr{dR}}, R/\fg)$ & \\ \hline
  \end{tabular}
 \caption{Twists of Supersymmetric Yang--Mills Theories with gauge Lie algebra $\fg$ with a chiral multiplet valued in a representation $R$ (4 supercharges).}
 \label{table_of_twists_4}
\end{table}

\begin{table}[!ht]
 \centering
 \begin{tabular}{c|c|c|c}
 $\mc N$ & Twist & Description & Invariant Directions \\
 \hline
 \multirow{6}*{$(4,4)$} & \multirow{2}*{\hyperref[sect:2d44minimaltwist] {Rank $(1,0)$}} & {Holomorphic BF theory coupled to a holomorphic symplectic boson} & \multirow{2}*{1 (holomorphic)} \\
 && {$T^*[-1]\Sect(\CC, T[1]((U\otimes K_\CC^{1/2}) \ham \gg))$} & \\ \cline{2-4}
 & \multirow{2}*{\hyperref[sect:2d44Btwist] {Rank $(1,1)$ (B)}} & {Topological BF theory coupled to a holomorphic symplectic boson} & \multirow{2}*{2 (topological)} \\
 && {$T^*[-1]\mr{Map}(\RR^2_{\mr{dR}}, U \ham \gg)$} & \\  \cline{2-4}
 & \multirow{2}*{\hyperref[sect:2d44Atwist] {Rank $(1,1)$ (A)}} & {Perturbatively trivial (A-model)} & \multirow{2}*{2 (topological)} \\
 && {$\mr{Map}(\RR^2_{\mathrm{dR}}, (U \ham \gg)_{\mr{dR}})$} & \\  \hline
 \multirow{6}*{$(2,2)$} & \multirow{2}*{\hyperref[sect:2d22minimaltwist] {Rank $(1,0)$}} & {Holomorphic BF theory coupled to $R$ matter} & \multirow{2}*{1 (holomorphic)} \\
 && {$T^*[-1]\mr{Map}(\CC, T[1](R/\gg))$} & \\  \cline{2-4}
  & \multirow{2}*{\hyperref[sect:2d22Btwist] {Rank $(1,1)$ (B)}} & {Topological BF theory coupled to $R$ matter} & \multirow{2}*{2 (topological)} \\
 && {$T^*[-1]\mr{Map}(\RR^2_{\mr{dR}}, R/\gg)$} & \\  \cline{2-4}
 & \multirow{2}*{\hyperref[sect:2d22Atwist] {Rank $(1,1)$ (A)}} & {Perturbatively trivial (A-model)} & \multirow{2}*{2 (topological)} \\
 && {$T^*[-1]\mr{Map}(\CC, (R/\gg)_{\mr{dR}})$} & \\  \hline
 \multirow{2}*{$(\mc N_+, 0)$} & \multirow{2}*{\hyperref[sect:2dN0minimaltwist] {Rank $(1,0)$}} & {Holomorphic BF theory coupled to $\mc N_+ - 2$ free fermions} & \multirow{2}*{1 (holomorphic)} \\
 && {$T^*[-1]\mr{Sect}(\CC, (\fg^{\mc N_+ -2} \otimes K_{\CC}^{1/2}) / \gg)$} & \\  \hline
 \multirow{2}*{$(4,0)$} & \multirow{2}*{\hyperref[sect:2d40minimaltwist] {Rank $(1,0)$}} & {Holomorphic BF theory coupled to a holomorphic symplectic boson} & \multirow{2}*{1 (holomorphic)} \\
 && {$T^*[-1]\mr{Sect}(\CC, (U \otimes K_\CC^{1/2}) \ham \gg)$} & \\  \hline
 \multirow{2}*{$(2,0)$} & \multirow{2}*{\hyperref[sect:2d20minimaltwist] {Rank $(1,0)$}} & {Holomorphic BF theory coupled to $R$ matter} & \multirow{2}*{1 (holomorphic)} \\
 && {$T^*[-1]\mr{Map}(\CC, R/\gg)$} & \\ \hline
 \end{tabular}
 \caption{Twists of Supersymmetric Yang--Mills Theories in two dimensions with gauge group $G$.  When $\mc N=(0,2)$ and $(2,2)$ the theory includes a chiral multiplet valued in a representation $R$. When $\mc N=(0,4)$ and $(4,4)$ the theory includes a hypermultiplet valued in a symplectic representation $U$.  We can promote the supersymmetry to $\mc N=(8,8)$ when $U = T^*\gg$, but no new twists occur.}
 \label{table_of_twists_2d}
\end{table}

\clearpage
\begin{figure}[hbp]
\begin{tikzpicture}
  \matrix (mat) [nodes in empty cells, minimum width=3.5ex, minimum height=3.5ex, column sep=1.8ex,row sep=6ex]{
   \node (corner) {}; & \node (d10) {10}; & \node (d9) {9}; & \node (d8) {8}; & \node (d7) {7}; & \node (d6) {6}; & \node {$\qquad \quad \ 5$}; && \node (d5) {};& \node {$\qquad \ 4$}; && \node (d4) {}; & &\node (d3) {3}; & \node (d2) {}; &  \\
  \node {8}; &    &   & \node[s16] (88) {$(1,0)_A$};  &   &   &   &&   &&&&   &&&   \\
  \node {7}; &    &   &   & \node[s16] (77) {$1_A$};  &   &   &&   &&&&   &&&   \\
  \node {6}; &    &   &   &   & \node[s16] (66) {$(1,1)_A$};  &   &&   &&&&   &&&   \\
  \node {5}; & \node[s16] (105) {$(1,0)$};  & \node[s16] (95) {$1$};  & \node[s16] (85) {$(1,1)$};  & \node[s16] (75) {$2$}; & \node[s16]  (65){$(2,2)$};  & \node[s16] (55B) {$4$}; & \node[s16] (55A) {$2_A$};  &   &&&&   &&&   \\
  \node {4}; &    &   & \node[s16] (84) {$(1,0)_B$};  & \node[s16] (74) {$1_B$};  & \node[s16] (64) {$(1,1)_B$};  & \node[s16] (54) {$2_B$};  && \node[s16] (44B) {$(2,2)_B$}; & \node[s16] (44K) {$(2,1)$}; & \node[s8] (44A) {$(2,0)$}; & \node[s16] (44g) {$(2,2)_A$};  &   &&&   \\
  \node {3}; &    &   &   &   & \node[s8] (63) {$(1,0)$};  & \node[s8] (53) {$1$};  && \node[s8] (43) {$(1,1)$};  &&&& \node[s8] (33B) {$2_B$}; & \node[s8] (33A) {$2_A$}; &  &   \\
  \node {2}; &    &   &   &   &   &   && \node[s4](42) {$(1,0)$};  &&&& \node[s4] (32) {$1$};  &&   \\
  \node (i1) {1}; &   &  &  &  &  &   & &   &&& &   && &   \\};
  
  \draw[thick] (corner.south west) -- (d2.south west);
  \draw[thick] (corner.north east) -- (i1.south east);
  
  \draw[dotted] (0.5,-7) -- (0.5,7);
  \draw[dotted] (6.4,-7) -- (6.4,7);
  \draw[dotted] (-2,-7) -- (-2,7);
  \draw[dotted] (-3.8,-7) -- (-3.8,7);
  \draw[dotted] (-4.8,-7) -- (-4.8,7);
  \draw[dotted] (-6.3,-7) -- (-6.3,7);
  \draw[dotted] (-7.2,-7) -- (-7.2,7);
  
  \draw[arrow] (105) -- (95);
  \draw[arrow] (95) -- (85);
  \draw[arrow] (88) -- (77);
  \draw[arrow] (95) -- (84);
  \draw[arrow] (85) -- (75);
  \draw[arrow] (85) -- (74);
  \draw[arrow] (84) -- (74);
  \draw[arrow] (77) -- (66);
  \draw[arrow] (75) -- (65);
  \draw[arrow] (75) -- (64);
  \draw[arrow] (74) -- (64);
  \draw[arrow] (74) -- (63);
  \draw[arrow] (66) -- (55A);
  \draw[arrow] (65) -- (55B);
  \draw[arrow] (65) -- (54);
  \draw[arrow] (64) -- (54);
  \draw[arrow] (64) -- (53);
  \draw[arrow] (63) -- (53);
  \draw[arrow] (55A) -- (44g);
  \draw[arrow] (55B) -- (44B);
  \draw[arrow] (54) -- (44B);
  \draw[arrow] (54) -- (43);
  \draw[arrow] (53) -- (43);
  \draw[arrow] (53) -- (42);
  \draw[arrow] (44B) -- (33B);
  \draw[arrow] (44A) -- (33A);
  \draw[arrow] (44K) -- (33A);
  \draw[arrow] (43) -- (33B);
  \draw[arrow] (43) -- (32);
  \draw[arrow] (42) -- (32);
  \draw[arrow] (44g) -- (33A);
  
\end{tikzpicture}
\caption{This figure shows the orbits of square-zero supercharges in each dimension, and how they relate to one another under dimensional reduction.  The labels indicate each orbit: the number refers to the rank, and the subscript indicates the situations where the supercharges of a given rank split into multiple orbits.  Each column is labelled by a dimension, and each row by the number of invariant directions of the supercharge.  Colours indicate the maximal supersymmetry algebra where the given supercharges live, so red indicates supercharges defined in algebras with 16 supercharges, orange those with 8 supercharges, and yellow those with 4 supercharges.  There is an arrow whenever one twist dimensionally reduces to another twist one dimension lower.}
\label{fig:superchargeorbits}
\end{figure}
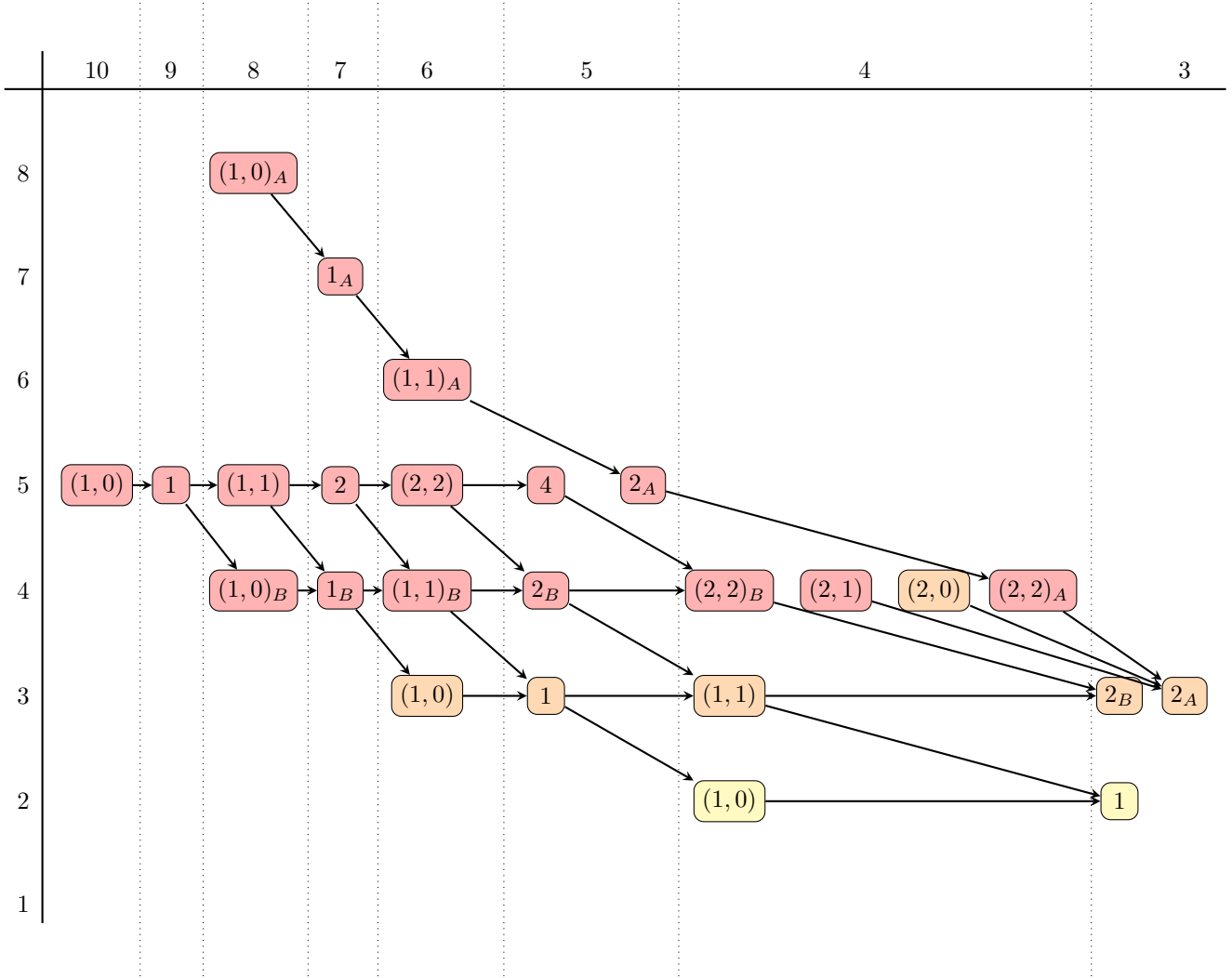

\subsection*{Outline of the Paper}
The remainder of the paper is divided into two parts.  In Part \ref{formalism_part} we set up the formalism that we will use when we study supersymmetric gauge theories and their twists.  The first main ingredient is the Batalin-Vilkovisky formalism for classical field theory (Section \ref{BV_section}).  The other main ingredient is the systematic study of supersymmetry algebras and supersymmetric action functionals using normed division algebras (Section \ref{sect:susy}).  We use this formalism to prove in Section \ref{sect:SYM} that super Yang--Mills theories with matter in dimensions 10, 6, 4 and 3 are in fact supersymmetric, meaning that there is a well-defined $L_\infty$ action of the supersymmetry algebra on the classical BV theory in question.  We introduce the idea of dimensional reduction (Section \ref{dim_red_section}) for classical field theories to show that supersymmetry action are well-defined in lower dimensions.

In Part \ref{classification_part} of the paper, we produce the classification of supersymmetric Yang--Mills theories in dimensions 2 to 10 systematically.  We start with dimension 10 and work down by dimensional reduction.  Each subsection is divided by the number of supersymmetries, and the orbits of square-zero supercharges by which we can twist.  Twisted theories are characterized up to perturbative equivalence, including the residual Lorentz symmetry acting on each twisted theory.

\subsection*{Acknowledgements}
We would like to thank Kevin Costello, Owen Gwilliam, Justin Hilburn and Philsang Yoo for helpful discussions during the preparation of this paper.  
We are also grateful to Dylan Butson for kind discussion of his related forthcoming work, specifically in reference to supersymmetric theories in dimension $6$.  
The research of C.E. on this project has received funding from the European Research Council (ERC) under the European Union's Horizon 2020 research and innovation programme (QUASIFT grant agreement 677368).
The research of P.S. was supported by the NCCR SwissMAP grant of the Swiss National Science Foundation.
The research of B.R.W. was supported by Northeastern University and National Science Foundation Award DMS-1645877.

\pagestyle{standard}
\part{Supersymmetric Gauge Theory} \label{formalism_part}

\section{The BV-BRST Formalism} \label{BV_section}

In this section we will set up the homological formalism in which we study classical field theory: the BV-BRST formalism.  Much of the material in this section is not original.  We refer the reader to \cite{CostelloBook, Book2} for more details on this perspective.  We will conclude the section by describing a number of fundamental examples of classical field theories that are highly structured: mixed holomorphic-topological theories.  We will also discuss the concept of \emph{dimensional reduction} of a classical field theory on $M$ along a fibration $M \to N$.  We will use the idea of dimensional reduction to construct many of the supersymmetric field theories which we will consider in the next section.  

\subsection{Conventions}
Throughout the paper we will frequently study objects, for instance vector bundles, equipped with a $\ZZ\times\ZZ/2\ZZ$-grading. \emph{Degree} will refer to the first (cohomological) grading and \emph{odd} or \emph{even} to the second (fermionic) grading.  We will write $\Pi E$ to denote $E$ placed in odd $\ZZ/2\ZZ$ degree.  For an element $x$ we denote by $|x|\in\ZZ/2\ZZ$ the total degree.

Given a vector bundle $E\rightarrow M$ we denote by $\cE$ the topological vector space of smooth sections of $E$ and by $\cE_c$ the topological vector space of smooth compactly supported sections.
We denote by $\cO(\cE)$ (respectively $\cO(\cE_c)$) the completed algebra of symmetric functions on $\cE$ (respectively $\cE_c)$. 
We denote by $\oloc(\cE)$ the space of local functionals on $\cE$ (see \cite[Definition 4.5.1.1]{Book2}). An element of $\oloc(\cE)$ will be denoted symbolically by an expression of the form
\[\int_M f (\phi, \phi', \dots),\]
where $f$ is a density on $M$ depending on infinite jets of sections of $E$. Note, however, that the integral here is a formal symbol. 
The space of local functionals can be viewed as a subspace
\[
\oloc(\cE) \subset \cO(\cE_c)
\]
where the integral symbol makes sense in earnest when applied to sections which are compactly supported.
We denote by $\oloc^+(\cE)\subset \oloc(\cE)$ the subspace of local functionals which are at least cubic.

Given two vector bundles $E, F$ on $M$ we can also make sense of the space of local functionals from $E$ to $F$.
By definition, this is 
\[
{\rm Fun}_{\rm loc}(\cE, \cF) = \prod_{n \geq 0} {\rm PolyDiff}(\cE^{\times n}, \cF)_{S_n}
\]
where ${\rm PolyDiff}(\cE^{\times n}, \cF)$ denotes the space of polydifferential operators, and we take coinvariants for the obvious symmetric group action.
When $\cF = \cE$, we refer to ${\rm Fun}_{\rm loc}(\cE, \cE)$ as the space of local vector fields on $E$. 
There is a natural Lie bracket on ${\rm Fun}_{\rm loc}(\cE, \cE)$ and a canonical action of this Lie algebra on local functionals. 

\subsection{Formal Moduli Problems and Classical Field Theories}
\label{sect:FMPs}

The classical BV (Batalin-Vilkovisky) formalism \cite{BatalinVilkovisky} is a model for classical field theory from the Lagrangian perspective.  In brief, the classical BV formalism produces a local model for the critical locus of an action functional, but considered in the derived sense.  That is, given a space $\mc F$ of fields and an action functional with derivative $\d S$, one considers not just the usual locus in $\mc F$ of fields with $\d S(\phi) = 0$, but the derived intersection $\mr{dCrit}(S) = \mc F \cap^h_{T^*\mc F} \Gamma_{\d S}$ of the zero section in $T^*\mc F$ with the graph of $\d S$.  The formalism we describe below can be interpreted as an abstract formalism for modelling the tangent complex at a point to a derived critical locus $\mr{dCrit}(S)$ as a formal moduli problem.

Recall that a formal moduli problem is a functor from connective dg Artinian algebras $(R, \mathfrak{m})$ to simplicial sets which satisfies a derived version of Schlessinger's condition. We refer to \cite{DAGX,PridhamFMP,Toen} for more details.

For instance, if $\fg$ is an $L_\infty$ algebra, we have a formal moduli problem $B\fg$ defined by
\[(B\fg)(R, \mathfrak{m})=\mathrm{MC}(\fg\otimes \mathfrak{m}),\]
where $\mathrm{MC}(\fg\otimes \mathfrak{m})$ is the simplicial set of Maurer--Cartan elements. The main result of \cite{DAGX,PridhamFMP} is that the functor $B$ defines an equivalence of $\infty$-categories between $L_\infty$ algebras and formal moduli problems. The inverse functor sends a formal moduli problem $\cM$ to the $L_\infty$-algebra $T_{\cM, \ast}[-1]$, the shifted tangent complex of $\cM$ at the basepoint. This important result will serve as motivation for the main definition in this section (Definition \ref{def:classicalfieldtheory}).

Let $V$ be a $\fg$-representation. Then we may construct an $L_\infty$ algebra
\[L_{V, \fg}=\fg\oplus V[-1]\]
with the only nontrivial brackets coming from the $L_\infty$ brackets on $\fg$ and the action map of $\fg$ on $V$. We introduce the notation
\[V/\fg:= B L_{V, \fg}.\]

\begin{example}
If $\fg$ is an $L_\infty$ algebra, it has an adjoint representation $\fg$. We define the \defterm{$n$-shifted tangent bundle} of $B\fg$ to be
\[T[n] B\fg = \fg[n+1] / \fg.\]
\label{ex:tangentBg}
\end{example}

\begin{example}
Suppose $\fg$ is an $L_\infty$ algebra which is bounded as a complex and has finite-dimensional graded pieces. Then $\fg^*$ is a coadjoint representation of $\fg$. We define the \defterm{$n$-shifted cotangent bundle} of $B\fg$ to be
\[T^*[n] B\fg = \fg^*[n-1] / \fg.\]
\label{ex:cotangentBg}
\end{example}

\begin{definition}
Let $\fg$ be an $L_\infty$ algebra. A \defterm{$\Gm$-action} on a formal moduli problem $B\fg$ is a weight grading $\fg = \bigoplus_m \fg(m)$ compatible with the $L_\infty$ structure.
\end{definition}

\begin{example}
Suppose $\fg$ is an $L_\infty$ algebra and $V$ is a $\fg$-representation. Then $V/\fg$ carries a $\Gm$-action: the underlying $L_\infty$ algebra $\fg\oplus V[-1]$ carries a grading where $\fg$ has weight $0$ and $V[-1]$ has weight $1$. For instance, $T[n] B\fg$ and $T^*[n] B\fg$ carry $\Gm$-actions.
\end{example}

\begin{example}
Suppose $\fg$ is a dg Lie algebra and $U$ a $\fg$-representation equipped with an $n$-shifted symplectic pairing $U\otimes U\rightarrow \CC[n]$. Consider the dg Lie algebra
\[\fh = \fg\oplus U[-1]\oplus \fg^*[n-2]\]
with the brackets $\fg\otimes \fg\rightarrow \fg$ given by the Lie bracket on $\fg$, $\fg\otimes U\rightarrow U$ given by the $\fg$-action on $U$, $\fg\otimes \fg^*\rightarrow \fg^*$ given by the coadjoint action and $\mu\colon U\otimes U\rightarrow \fg^*[d-1]$ defined by $(\mu(v, w), x)_\fg = ([x, v], w)_U$. The dg Lie algebra $\fh$ carries nondegenerate invariant symmetric pairing of cohomological degree $n-2$ given by pairing $\fg$ and $\fg^*$ and pairing $U$ with itself. We denote
\[U\ham \fg := B\fh.\]
This formal moduli problem is equipped with a $\Gm$-action where $\fg$ has weight 0, $U[-1]$ has weight $1$ and $\fg^*[n-2]$ has weight $2$. This may be thought of as an infinitesimal version of the Hamiltonian reduction of $U$ by the $\fg$-action.
\end{example}

Now suppose $L$ is a local $L_\infty$ algebra on a manifold $M$. For every open subset $U\subset M$ we have a formal moduli problem
\[(B L)(U) = B L(U),\]
i.e. $L$ defines a presheaf $B L$ of formal moduli problems on $M$. The following definition was introduced in \cite[Definition 4.1.3.3]{Book2}.

\begin{definition}
A \defterm{local formal moduli problem on $M$} is a presheaf of formal moduli problems on $M$ represented by a local $L_\infty$ algebra.
\end{definition}

\begin{remark}
In \cite{Book2} an extra assumption of ellipticity is required for the local $L_\infty$ algebras considered. It is only relevant for quantization, which we do not consider in this paper, so for simplicity we will not require ellipticity (though in fact all examples we consider will end up being elliptic).
\end{remark}

Given a local formal moduli problem $\cM=B L$, we may consider the space of local functionals which is defined as
\[\oloc(\cM) := \oloc(\cL[1]).\]
The local $L_\infty$ structure on $\cL$ induces a Chevalley--Eilenberg differential on $\oloc(\cM)$.

\begin{example}
Let $X,Y$ be complex manifolds, let $M$ be a smooth manifold and let $B\fg$ be a formal moduli problem represented by an $L_\infty$ algebra $\fg$. Then we may define the following local formal moduli problem on $X\times Y\times M$. Let $\Omega^{0, \bullet}_X$ be the graded vector bundle of $(0, n)$-forms on $X$, $\Omega^{\bullet, \bullet}_Y$ be the graded vector bundle of $(p, q)$-forms on $Y$ and $\Omega^\bullet_M$ a graded vector bundle of differential forms on $M$. We may then consider a graded vector bundle
\[L = \Omega^{0,\bullet}_X\otimes \Omega^{\bullet, \bullet}_Y\otimes \Omega^\bullet_M\otimes \fg\]
on $X\times Y\times M$. It carries a differential given by the sum $\ol\dd_X + \ol\dd_Y + \d_M + \d_\fg$. It also carries a local $L_\infty$ structure which uses the $L_\infty$ structure on $\fg$ and the wedge product of differential forms. We then define
\[\map(X\times Y_{\mathrm{Dol}}\times M_{\mathrm{dR}}, B\fg) := B L.\]
\label{ex:mappingspace}
\end{example}

\begin{remark}
A smooth complex algebraic variety $X$ gives rise to derived stacks $X_{\mathrm{Dol}}$ and $X_{\mathrm{dR}}$ defined by Simpson \cite{Simpson,PTVV}. So, given smooth complex algebraic varieties $X,Y,M$ and a derived stack $F$ we may consider the mapping stack
\[\map(X\times Y_{\mathrm{Dol}}\times M_{\mathrm{dR}}, F).\]
Example \ref{ex:mappingspace} is an analogous construction in the world of formal moduli problems.
\end{remark}

\begin{example}
Consider $X, Y, M, \fg$ as in Example \ref{ex:mappingspace} and suppose $E$ is a line bundle on $X\times Y\times M$, equipped with a holomorphic structure along $X\times Y$ and a flat connection along $M$. Moreover, assume $B\fg$ carries a $\Gm$-action. We then have a local $L_\infty$ algebra
\[L = \bigoplus_m \Omega^{0,\bullet}_X\otimes \Omega^{\bullet, \bullet}_Y\otimes \Omega^\bullet_M\otimes \fg(m)\otimes E^{\otimes m}\]
on $X\times Y\times M$. We define
\[\Sect(X\times Y_{\mathrm{Dol}}\times M_{\mathrm{dR}}, B\fg\times_{\Gm} L) := B L.\]
\end{example}

As in Examples \ref{ex:tangentBg} and \ref{ex:cotangentBg}, we may define shifted tangent and cotangent bundles of a local formal moduli problem which give more examples.

\begin{prop}
Consider $X,Y,M, \fg$ as in Example \ref{ex:mappingspace}. The local formal moduli problem $T^*[n]\map(X\times Y_{\mathrm{Dol}}\times M_{\mathrm{dR}}, B\fg)$ is isomorphic to the local formal moduli problem
\[\Sect\left(X\times Y_{\mathrm{Dol}}\times M_{\mathrm{dR}}, T^*[n+\dim(X)+2\dim(Y)+\dim(M)] B\fg\times_{\Gm} (K_X\otimes \Dens_M)\right),\]
where $K_X$ is the canonical bundle of $X$ and $\Dens_M$ is the line bundle of densities on $M$.
\label{prop:mapintocotangent}
\end{prop}
\begin{proof}
The claim follows from the following isomorphisms of graded vector bundles:
\begin{align*}
\left(\Omega^{0, \bullet}_X\right)^!&\cong \Omega^{0, \bullet}_X\otimes K_X[\dim(X)] \\
\left(\Omega^{\bullet, \bullet}_Y\right)^!&\cong \Omega^{\bullet, \bullet}_Y[2\dim(Y)] \\
\left(\Omega^\bullet_M\right)^!&\cong \Omega^\bullet_M\otimes \Dens_M[\dim(M)].
\end{align*}
\end{proof}

\begin{corollary}
Suppose $X,Y,M,\fg$ are as in Example \ref{ex:mappingspace} and, moreover, that $X$ is equipped with a holomorphic volume form and $M$ is oriented. Then
\[T^*[n]\map(X\times Y_{\mathrm{Dol}}\times M_{\mathrm{dR}}, B\fg)\cong \map(X\times Y_{\mathrm{Dol}}\times M_{\mathrm{dR}}, T^*[n+\dim(X)+2\dim(Y)+\dim(M)] B\fg).\]
\label{cor:mapintocotangent}
\end{corollary}

Given a local formal moduli problem, we may talk about shifted symplectic structures \cite{PTVV} on it. In this paper we will only be interested in a strict notion as follows.

\begin{definition}
Let $\cM$ be a local formal moduli problem on $M$ represented by a local $L_\infty$ algebra $L$. A \defterm{strict $n$-shifted symplectic structure} on $\cM$ is a pairing $\omega\colon L\otimes L\rightarrow \Dens_M[n-2]$ satisfying the following conditions:
\begin{enumerate}
\item It is fiberwise nondegenerate.

\item It is graded skew symmetric.

\item The pairing $\cL_c\otimes \cL_c \rightarrow \CC$ defined by
\[\alpha\otimes \beta \mapsto \int_M \omega( \alpha, \beta)\]
is an invariant pairing on the $L_\infty$ algebra $\cL_c$.
\end{enumerate}
\end{definition}

We can now state a concise definition of a classical field theory in the BV formalism.

\begin{definition}
A \defterm{classical BV field theory} (or, simply, classical field theory) is a local formal moduli problem on the spacetime manifold $M$ equipped with a strict $(-1)$-shifted symplectic structure.
\label{def:classicalfieldtheory}
\end{definition}

Given a local formal moduli problem $\cM = BL$ equipped with a strict $n$-shifted symplectic structure, the space of local functionals $\oloc(\cM)$ is equipped with a Poisson bracket (see \cite[Chapter 5.3]{CostelloBook}) 
\[
\{-,-\}\colon \oloc(\cM) \times \oloc(\cM) \to \oloc(\cM)[-n]
\]
This bracket is a graded version of the so-called Soloviev bracket \cite{Soloviev} defined on the $\infty$-jets, as described in \cite[Section 4]{GetzlerBracket}. 

We explain how to define the Poisson bracket in our context. Write $E = L[1]$ for convenience. First note that there is a linear map
\[
\d_{\mr{dR}} \colon \oloc(\cE) \to {\rm Fun}_{\rm loc}(\cE, \cE^!) 
\]
defined as follows. 
A local functional $F \in  \oloc(\cE)$ can be written as an equivalence class of a sum of densities of the form
\[
D_1(-) \cdots D_n(-) \Omega
\]
where $D_i$ is a differential operator $D_i \colon \cE \to C^\infty_M$ and $\Omega$ is a density on $M$. 
Without loss of generality, suppose $F$ is of this form. 
Then, we can view $F$ as a functional in $\cO(\cE_c)$ by the assignment
\[
\phi \mapsto \int_M D_1(\phi) \cdots D_n(\phi) \Omega
\]
where $\phi$ denotes a compactly supported section. 
Define the symmetric multilinear map
\begin{align*}
\d_{\mr{dR}} F  \colon  \cE_c^{\times (n-1)} & \to  \cE^\vee \\
 (\phi_1, \ldots, \phi_{n-1}) & \mapsto  D_1(\phi_1) \cdots D_{n-1}(\phi_{n-1}) D_{n} (-) + \{{\rm symmetric\;terms}\} .
\end{align*}
Integrating by parts, we see that for any $(n-1)$-tuple $(\phi_1, \ldots, \phi_{n-1}) \in \cE_c^{n-1}$ the linear functional $(\d_{\mr{dR}} F) (\phi_1,\ldots, \phi_{n-1})$ is an element of $\cE^!$. 
This implies that $\d_{\mr{dR}} F \in {\rm Fun}_{\rm loc}(\cE, \cE^!)$.

The non-degenerate pairing $\omega$ determines a bundle isomorphism $\omega \colon E \cong E^! [n]$ and hence an isomorphism of local functions
\[
\omega \colon {\rm Fun}_{\rm loc}(\cE, \cE^!) \cong {\rm Fun}_{\rm loc}(\cE, \cE[-n]) .
\]
We recognize the right hand side as the space of local vector fields placed in a shifted cohomological degree.
In total, we see that a local functional $F$ determines a local vector field by applying this isomorphism to $\d_{\mr{dR}}F$:
\[
X_F := \omega \circ \d_{\mr{dR}} (F) \in  {\rm Fun}_{\rm loc}(\cE, \cE[-n])  .
\]
This is the Hamiltonian vector field corresponding to $F$.  We can now define the Poisson bracket.

\begin{definition}
The Poisson bracket between local functionals $F, G$ is defined by
\[\{F, G\} = X_F (G).\]
\end{definition}

The Poisson bracket enjoys the graded skew symmetry property
\[\{F, G\} = (-1)^{|F| |G|} \{G, F\}\]
as well as the graded Jacobi identity.

The differential on a local $L_\infty$ algebra $L$ is given by a differential operator $Q_{\mathrm{BV}}\colon L\rightarrow L$. The structure of the $L_\infty$ brackets can be encoded into its potential. In the same way, the structure of a local $L_\infty$ algebra $L$ together with an $n$-shifted symplectic structure on $BL$ may be encoded into the \defterm{action functional} $S\in\oloc(BL)$ of cohomological degree $1-n$ such that
\[S = \frac{1}{2}\int_M\omega(e, Q_{BV} e) + I,\]
where $e\in L$ and $I\in\oloc(BL)$ is at least cubic. Moreover, the action functional satisfies the \defterm{classical master equation}
\[\{S, S\} = 0.\]
We refer to \cite[Proposition 5.4.0.2]{Book2} for this construction.

Given a classical field theory represented by a local $L_\infty$ algebra $L$ on $M$, as in Definition \ref{def:classicalfieldtheory}, we call $E=L[1]$ the \defterm{bundle of BV fields}, and we call the complex $(E, Q_{\mr{BV}})$ the \defterm{classical BV complex}. We call the Poisson bracket on $\oloc(\cM)$ the \defterm{BV bracket}. It will sometimes be convenient to think of a classical field theory as a quadruple $(E, \omega, Q_{\mr{BV}}, I)$ consisting of the bundle of BV fields equipped with a $(-1)$-shifted symplectic pairing $\omega$, a classical BV differential $Q_{\mr{BV}}$ and an interaction functional $I$.  We characterize such data in the following way.

\begin{definition}
A \defterm{free BV theory} on a manifold $M$ is the data of:
\begin{itemize}
\item a finite rank graded vector bundle $E \to M$ equipped with an even differential operator of cohomological degree $+1$
\[
Q_{\mr{BV}} \colon \cE \to \cE [1] 
\]
such that $(1)$: $Q_{\mr{BV}}^2 = 0$ and $(2)$: the pair $(\cE , Q_{\mr{BV}})$ is an elliptic complex;
\item a map of bundles
\[
\omega\colon E \otimes E \to \Dens_M [-1]
\]
that is
\begin{enumerate}
\item[$(1)$] fiberwise nondegenerate,
\item[$(2)$] graded skew symmetric, and
\item[$(3)$] satisfies $\int_M \omega(e_0, Q_{\mr{BV}} e_1) = (-1)^{|e_0|} \int_M \omega(Q_{\mr{BV}} e_0, e_1)$ where $e_i$ are compactly supported sections of $E$ .
\end{enumerate}
\end{itemize}
\end{definition}

The datum of a classical BV field theory as in Definition \ref{def:classicalfieldtheory} is equivalent to the datum of a free BV theory $(E, Q, \omega)$ equipped with an even functional
\[I \in \oloc^+(\cE)\]
of cohomological degree zero satisfying the Maurer-Cartan equation
\[Q_{\mr{BV}} I + \frac{1}{2} \{I,I\} = 0,\]
under the identification of the BV action as
\[S = \frac{1}{2} \int_M \omega(e, Q_{\mr{BV}} e) + I\in \oloc(E).\]

\begin{example} \label{def:cotangent_type}
Let $\cM=BL$ be a local formal moduli problem. Then the $(-1)$-shifted cotangent bundle $T^*[-1]\cM$ carries a natural $(-1)$-shifted symplectic structure. Indeed, $T^*[-1]\cM = B(L\oplus L^![-3])$ and we simply pair $L$ and $L^!$. Classical field theories arising via this construction are called \defterm{cotangent type} theories.
\end{example}

We will also consider $\CC[t]$-families of classical field theories.
\begin{definition} \label{family_of_BV_theories_def}
A \defterm{$\CC[t]$-family of classical field theories} is a graded bundle of locally-free $\CC[t]$-modules $L$ on $M$ equipped with a structure of a $\CC[t]$-linear local $L_\infty$ algebra and a $\CC[t]$-linear $(-1)$-shifted symplectic structure $\omega\colon L\otimes_{\CC[t]} L\rightarrow \CC[t]\otimes \Dens_M[-3]$.
\end{definition}

We will consider $\CC[t]$-families of classical field theories where $L = \CC[t]\otimes L_0$ and the pairing
\[\omega\colon L\otimes_{\CC[t]} L\rightarrow \CC[t]\otimes \Dens_M[-3]\]
comes from a pairing
\[\omega_0\colon L_0\otimes L_0\rightarrow \Dens_M[-3].\]
In this case the local $L_\infty$ structure is encoded in a $t$-dependent action functional $S$.

\begin{remark}
Besides the $\ZZ$-graded classical field theories defined above, we may consider the following variants of the above definition:
\begin{itemize}
\item A $\ZZ\times \ZZ/2\ZZ$-graded local $L_\infty$ algebra is a $\ZZ$-graded local $L_\infty$ algebra $L$ equipped with an extra $\ZZ/2\ZZ$-grading (the \defterm{fermionic grading}) with respect to which all operations are even. An $n$-shifted symplectic structure on a $\ZZ\times\ZZ/2\ZZ$-graded local formal moduli problem $BL$ is a pairing $L\otimes L\rightarrow \Dens_M[n-2]$, which is even with respect to the fermionic grading.

\item A $\ZZ/2\ZZ$-graded classical field theory is defined in the same way as a $\ZZ$-graded classical field theory where we only consider the cohomological grading modulo 2.
\end{itemize}
\end{remark}

\subsection{Perturbative Equivalence of Classical Field Theories}

Next, we formulate the notion of a morphism, and an equivalence, of classical BV theories. 

\begin{definition}
A \defterm{morphism} $\Phi\colon (E, \omega, S)\rightsquigarrow (E', \omega', S')$ of classical field theories over the same manifold $M$ is a collection $\Phi =\sum_{n\geq 1}^\infty \Phi_n$ of poly-differential operators $\Phi_n\colon \Sym^n(E)\rightarrow E'$ that intertwine the pairings $\omega, \omega'$ and define an $L_\infty$ map $\cE[-1]\rightarrow \cE'[-1]$. A morphism is a \defterm{perturbative equivalence} if the map $\Phi_1\colon (\cE, Q_{BV})\rightarrow (\cE', Q'_{BV})$ is a quasi-isomorphism. A classical field theory is \defterm{perturbatively trivial} if it is perturbatively equivalent to the zero theory ($E = 0$).
\label{def:perturbativeequivalence}
\end{definition}

The interpretation of this definition is that $\Phi$ is a non-linear map between the bundles of BV fields, and $\Phi_n$ is its $n^{\text{th}}$ Taylor coefficient.

We will now describe two primitive examples of equivalences of classical field theories that will be useful in simplifying twisted theories. 
First, we consider the process of eliminating an auxiliary field.

\begin{prop}\label{prop:integrateoutfield}
Fix a volume form $\dvol_M$ on $M$. Suppose $(E, \omega, S)$ is a classical field theory, where $E\cong E_0\oplus (\cO_M\oplus \Dens_M[-1])$ with the symplectic pairing $\omega$ given by a sum of a symplectic pairing $\omega_0$ on $E_0$ and the standard symplectic pairing on the second summand. Denote by $\phi$ a section of $\cO_M$ and by $\phi^*$ a section of $\Dens_M[-1]$. Suppose the BV action is
\[S = S_0 + \frac{1}{2} \int \dvol_M (\phi^2 - 2\phi S_1),\]
where $S_0$ is a local functional independent of $\phi,\phi^*$ and $S_1$ is a $\cO_M$-valued polydifferential operator which is independent of $\phi$. 
Then the theory $(E, \omega, S)$ is perturbatively equivalent to the theory $(E_0, \omega_0, S')$ with the BV action $S' = S_0 - S_1^2/2$, where we set $\phi = S_1$ and $\phi^* = 0$.
\end{prop}
\begin{proof}
Concretely, suppose that the linear part of $S_1$ is given by an operator $Q_1$, and that the interacting part of $S_1$ is given by a functional $I_1 = \sum_{n=1}^\infty I_1^{(n)}$.  
The desired equivalence $\Phi \colon (E, \omega, S) \to (E_0, \omega_0, S')$ is given by the natural projection $\Phi = \Phi_1 \colon E \to E_0$. 
The quasi-inverse $\Psi \colon (E_0, \omega_0, S') \to (E, \omega, S)$ is defined as follows.
First $\Psi_1(e) = (e, Q_1(e), 0) \in E$.
For $n > 1$, define
\begin{align*}
\Psi_n \colon \sym^n(E_0) &\to E \\
e_1 \otimes \cdots \otimes e_n &\mapsto (0, I_1^{(n)}(e_1, \ldots, e_n), 0).
\end{align*}
These $\Psi_n$ manifestly intertwine the pairings $\omega$ and $\omega'$. To see that they intertwine the action functionals, we observe that 
\begin{align*}
S(\Psi(e)) &= S(e, S_1(e), 0) \\
&= S_0(e) + \frac 12 \dvol_M \int (S_1(e)^2 - 2 S_1(e)^2) \\
&= S_0(e) - \frac 12 S_1(e)^2 \\
&= S'(e).
\end{align*}
\end{proof}

\begin{remark}
In terms of the classical BV complex, this proposition has the following interpretation.  We consider theories where the classical BV complex is of the form
\[\xymatrix{
\cdots & \ul{0} & \ul{1} & \cdots \\
\cdots \ar[r] & E_0^0 \ar[r]^{Q_0} \ar@{.>}[dr] & E_0^1 \ar[r] &\cdots \\
&\cO_M \ar^-{{\rm dvol}}[r] \ar@{.>}[ur] &\dens_M. &
}\]
The bottom map multiplies a function by the volume element.
The dotted arrows are induced from $S_1$.
The proposition implies that we can replace this with a quasi-isomorphic cochain complex consisting of only the first line, provided we make a suitable modification of the classical action functional. 
\end{remark}

We may also eliminate pairs of fields as follows.

\begin{prop} 
\label{prop:BRSTdoublet}
Let $(E_0, \omega_0, S_0)$ be a classical BV theory and let $F \to M$ be a graded vector bundle.
Consider the theory $(E, \omega, S)$ with underlying graded vector bundle
\[
E = E_0 \oplus \left(F \oplus F^! [-1]\right) \oplus \left(F^! \oplus F[-1] \right)
\]
whose sections we denote by $e_0 + \phi + \phi^* + \psi + \psi^*$ according to the above decomposition. 
The shifted symplectic form $\omega$ is given by the sum of $\omega_0$ and the standard degree $+1$ pairings between $F, F^! [-1]$ and $F^!, F[-1]$. 
Suppose further that the local functional
\[
S = S_0 + \int \phi \psi^* - \int \phi I_\phi - \int \psi^* I_{\psi^*} - \int \phi^* I_{\phi^*} - \int \psi I_{\psi}
\]
satisfies the classical master equation, where $I_{\phi}, I_{\psi^*}, I_{\phi^*}, I_{\psi}$ are polydifferential operators on fields valued in $F^!$, $F$, $F$, $F^!$ respectively, and which are independent of $\phi$ and $\psi^*$.
Then the classical BV theory $(E, \omega, S)$ is perturbatively equivalent to the BV theory $(E_0, \omega_0, S')$ where $S'$ is given by setting $\phi = I_{\psi^*}, \phi^* = 0$ and $\psi^* = I_{\phi}, \psi = 0$ in the original action functional $S$. 
\end{prop}
\begin{proof}
Concretely, we will write $\sum_{n \ge 1} I^{(n)}_\phi$ and $\sum_{n \ge 1} I^{(n)}_{\psi^*}$ for the Taylor expansions of $I_\phi$ and $I_{\psi^*}$ respectively.  
The desired equivalence $\Phi \colon (E, \omega, S) \to (E_0, \omega_0, S')$ is given by the natural projection $\Phi = \Phi_1 \colon E \to E_0$. 
The quasi-inverse $\Psi \colon (E_0, \omega_0, S') \to (E, \omega, S)$ is defined as follows.
The linear term is $\Psi_1(e) = (e, I^{(1)}_{\phi}(e),0,0, I^{(1)}_{\psi^*}(e))$, and for $n > 1$ we have 
\[\Psi_n(e_1\otimes \cdots \otimes e_n) = (0, I^{(n)}_{\phi}(e_1, \ldots, e_n), 0, I^{(n)}_{\psi^*}(e_1, \ldots, e_n)).\]

The maps $\Psi_n$ manifestly intertwine the pairings on $E_0$ and $E$, since the image of $\Psi_n$ lands in an isotropic summand of the $E_1\oplus E_1^![-1]\oplus E_1^!\oplus E_1[-1]$ part of $E$.  
Also, by construction, the $\Psi_n$ intertwine the action functionals, since
\begin{align*}
S(F(e)) &= S_0(e) + \frac{1}{2} \int_M \omega(I_{\psi^*}(e), I_\phi(e) - I_\phi(e)) + \omega(I_\phi (e), I_{\psi^*}(e) - I_{\psi^*}(e)) \\
&= S'(e).
\end{align*}
\end{proof}

\begin{remark}
For the classical BV theory $(E, \omega, S)$ as in the proposition, the linearized BV differential defines the following cochain complex of fields:
\[\xymatrix{
\cdots & \ul{-1} & \ul{0} & \ul{1} & \ul{2} & \cdots \\
\cdots \ar[r] & E_0^{-1}\ar[r] \ar@{.>}[dr] \ar@{.>}[ddr] &E_0^0 \ar@{.>}[ddr] \ar@{.>}[dr]  \ar[r]  &E_0^1 \ar[r] &E_0^2 \ar[r] &\cdots \\
&&F_\phi \ar^{~}[dr]  \ar@{.>}[ur] &F^!_{\phi*}  \ar@{.>}[ur] && \\
&&F^!_{\psi^*} \ar^{~}[ur]  \ar@{.>}[uur]  &F_\psi  \ar@{.>}[uur] && \\
}\]
where the subscripts match the notation for the fields in the statement above \footnote{Note that we are writing $F$ as if it is concentrated in a single cohomological degree, but the proposition applies for any graded vector bundle as in the statement of the proposition.}.  The top line is the underlying cochain complex of the theory with fields $E_0$. The arrows $F_\phi \to F_\psi$ and $F^!_{\psi^*} \to F^!_{\phi^*}$ are given by the identity. The dotted arrows represent terms in the differential arising from $I_\phi, I_{\psi^*}, I_{\phi^*}, I_{\psi}$. The above proposition implies we can replace this cochain complex of fields with a quasi-isomorphic complex consisting of only the first line, provided we make a suitable modification of the classical action functional. 
\end{remark}
\begin{remark}
We will call the pair $(\phi, \psi)$ satisfying the conditions of the previous proposition a \defterm{trivial BRST doublet}.
\end{remark}

\subsection{Symmetries in the Classical BV Formalism} \label{symmetry_section}
In this section we define what it means for a (super) Lie algebra to act on a classical field theory  (see also \cite[Chapter 11]{Book2} for a related discussion). Let $(E, \omega, S)$ be a classical field theory and $\fg$ a super Lie algebra.  We will define $\gg$-equivariant local observables in the classical field theory by introducing $\gg$-valued background fields into our classical field theory and extending the action functional to a functional that involves these background fields, but still satisfies the classical master equation.  We begin by defining an appropriate version of the Chevalley--Eilenberg cochain complex.

\begin{definition}
The \defterm{Chevalley--Eilenberg complex} for the Lie algebra $\gg$, with coefficients in $\oloc(\cE)$, is defined as follows.  Consider the graded vector space
\[C^\bullet(\fg, \oloc(\cE)) = \bigoplus_n \hom(\wedge^n \fg, \oloc(\cE))[-n]\]
parameterizing multilinear maps $f\colon \fg^{\otimes n}\rightarrow \oloc(\cE)$ that satisfy the antisymmetry property
\[f(x_1, \dots, x_i, x_{i+1}, \dots, x_n) = (-1)^{|x_1||x_2|+1} f(x_1, \dots, x_{i+1}, x_i, \dots, x_n)\]
where $x_j \in \fg$.  The Chevalley-Eilenberg differential is given, following the sign conventions of \cite{SafronovCoisoInt}, by the formula
\[(\d_{\mr{CE}} f)(x_1, \dots, x_n) = \sum_{i < j}(-1)^{|x_i| \sum_{p=1}^{i-1} |x_p| + |x_j| \sum_{p=1,p\neq i}^{j-1} |x_p| +i+j+|f|} f([x_i, x_j], x_1, \dots, \widehat{x}_i, \dots, \widehat{x}_j, \dots, x_n).\]
The complex is additionally equipped with a degree $+1$ BV bracket via the formula
\[\{f, g\}(x_1, \dots, x_{k+l}) = \sum_{\sigma\in S_{k, l}} \mathrm{sgn}(\sigma) (-1)^{\epsilon+\epsilon_1} \{f(x_{\sigma(1)}, \dots, x_{\sigma(k)}), g(x_{\sigma(k+1)}, \dots, x_{\sigma(k+l)})\},\]
where $S_{k, l}$ is the set of $(k, l)$-shuffles, $\epsilon$ is the usual Koszul sign and
\[\epsilon_1 = |g|k + \sum_{i=1}^k |x_{\sigma(i)}|(l+|g|).\]
\end{definition}

The operator $Q_{\mr{BV}}$ on $\oloc(\cE)$ extends $C^\bu(\fg)$-linearly to an operator on $C^\bullet(\fg, \oloc(\cE))$ by the rule
\[
(Q_{\mr{BV}} f)(x_1,\ldots, x_n) = Q_{\mr{BV}} f(x_1,\ldots, x_n) 
\]
where $f \colon \fg^{\otimes n} \to \oloc(\cE)$. 
The differentials $\d_{\mr{CE}}$ and $Q_{\mr{BV}}$ are compatible in the sense that $(\d_{\mr{CE}} + Q_{\mr{BV}})^2 = 0$ making $C^\bu(\fg, \oloc(\cE))$ into a cochain complex with total differential $\d_{\mr{CE}} + Q_{\mr{BV}}$. 
Via the BV bracket, the shift of this cochain complex $C^\bu(\fg, \oloc(\cE))[-1]$ is a dg Lie algebra.  This shifted cotangent complex will model equivariant local observables in our classical field theory, but to finish defining the $\gg$ action we must define the equivariant version of the classical interaction.  This is defined as follows.

\begin{definition}
\label{infinitesimal_action_def}
Let $(E, \omega, S)$ be a classical field theory. An \defterm{action} of a super Lie algebra $\fg$ on $(E, \omega, S)$ is an element 
\[S_\fg = \sum_{k\geq 0} S_{\fg}^{(k)}\in C^\bullet(\fg, \oloc(\cE))\]
of cohomological degree zero, where $S_\fg^{(k)}$ is a multilinear map $\fg^{\otimes k} \to \oloc(\cE)$, that satisfies the following three conditions:
\begin{itemize}
\item[(a)] $S_\fg^{(0)} = S$.
\item[(b)] For each $k \geq 1$ and $x_1, \ldots, x_k \in \fg$ the local functional $S^{(k)}_\fg (x_1,\ldots, x_k)$ is at least quadratic in the fields.
\item[(c)] $S_\fg$ satisfies the Maurer--Cartan equation:
\[\d_{\mr{CE}} S_\fg + \frac{1}{2} \{S_\fg, S_\fg\} = 0.\]
\end{itemize}
\end{definition}

\begin{remark}
We have defined an action of a Lie algebra on a classical field theory in terms of a Noether current $S_\fg$.
Such data gives rise to an $L_\infty$ action of $\fg$ on the space of fields $\cE$ in the following way.
By the Maurer-Cartan equation, the operator $\d_{\mr{CE}} + \{S_\fg, -\}$ 
defines a differential on the graded vector space $\cO(\fg[1] \oplus \cE)$. 
By assumption that the Noether current is at least quadratic in the fields, we see that this differential defines a family of maps
\[
\fg^{\otimes k} \otimes \cE^{\otimes \ell} \to \cE
\]
combining to give $\cE$ the structure of an $L_\infty$-module for $\fg$.
\end{remark}

\begin{remark}
We have seen that a classical BV theory can also be presented in terms of a BV differential $Q_{\mr{BV}}$ and an interaction $I$ satisfying the Maurer-Cartan equation
\[Q_{\mr{BV}} I + \frac{1}{2} \{I,I\} = 0 .\]
One can also formulate actions of a Lie algebra on a classical theory in these terms. 
The data of an action of a Lie algebra $\fg$ on a classical field theory $(E, \omega, S)$ is equivalent to the choice of a local interaction functional 
\[I_\fg = I + \sum_{k\geq 1} S_{\fg}^{(k)} \text{ in } C^\bullet(\fg, \oloc(\cE)),\]
satisfying the Maurer-Cartan equation
\[(\d_{\mr{CE}} + Q_{\mr{BV}}) I_\fg + \frac{1}{2} \{I_\fg, I_\fg\} = 0.\]
\end{remark}

We may also define actions of supergroups on classical field theories.  The action of a supergroup $G$ is more data than the action of a super Lie algebra $\gg$: it includes the infinitesimal action of the Lie algebra $\gg$, along with an action of $G$ on the fields exponentiating this infinitesimal action.  That is, we make the following definition.

\begin{definition}
\label{group_action_def}
Let $(E, \omega, S)$ be a classical field theory, and let $G$ be a supergroup acting on spacetime $M$. An \defterm{action} of $G$ on $(E, \omega, S)$ is given by the following data:
\begin{itemize}
\item An action of $G$ on $\cE$ compatible with the $G$-action on $M$.

\item An action $S_\fg$ of its super Lie algebra $\fg$ with $S^{(k)}_\fg = 0$ for $k\geq 2$ 
\end{itemize}
These are required to satisfy the following conditions:
\begin{itemize}
\item The $G$-action on $\cE$ preserves the symplectic pairing $\omega$ and the action functional $S$.

\item For every $x\in\fg$, the vector field $X_{{S^{(1)}_\fg}(x)}$ on $\cE$ coincides with the infinitesimal action of $\gg$ on $\cE$.
\end{itemize}
\end{definition}

\begin{remark}
While we allow for $L_\infty$ actions of Lie algebras, we only consider strict actions of Lie groups.
\end{remark}

\subsection{From BRST to BV}
We will now explain how to build classical BV theories from more traditional data: that of the \emph{usual} fields of a classical field theory, together with the usual action functional and the action of gauge transformations.  These data can be packaged into what is known as a BRST theory, where fermionic fields (referred to as ghosts) are introduced to generate the infinitesimal gauge transformations, in the following way.

\begin{definition}
A \defterm{classical BRST theory} on a manifold $M$ consists of the following data:
\begin{itemize}
\item A local formal moduli problem $\cM$ represented by a local $L_\infty$ algebra $L$.
\item A local functional $S_{\mr{BRST}} \in \oloc(\cM)$ of polynomial degree $\geq 2$.
\end{itemize}
Together, these data must satisfy the equation
\[Q_{\mr{BRST}} S_{\mr{BRST}} = 0,\]
where $Q_{\mr{BRST}}$ is the Chevalley--Eilenberg differential defined by the local $L_\infty$ structure on $L$.
\end{definition}

We call $L[1]$ the \defterm{space of BRST fields}.

\begin{remark}
In the most typical examples, the bundle $L[1]$ is concentrated in $\ZZ$-degrees $-1$ and 0.  In this case, sections in degree 0 are thought of as physical fields, and ghosts -- sections in degree $-1$ -- are thought of as generators of the infinitesimal gauge symmetry.  The action of gauge transformations on fields is then encoded by the Lie structure.
\end{remark}

Given a classical BRST theory $\cM$, we may consider the cotangent theory $T^*[-1]\cM$ whose action functional we denote by $S_{\rm anti}$.
Let us denote the pullback of $S_{\mathrm{BRST}}$ along the obvious projection $T^*[-1]\cM\rightarrow \cM$ by the same letter. Then the equation $Q_{\mathrm{BRST}} S_{\mathrm{BRST}} = 0$ in $\oloc(\cM)$ implies that $\{S_{\rm anti}, S_{\mathrm{BRST}}\} = 0$ in $\oloc(T^*[-1]\cM)$.

\begin{definition}
Let $(\cM, S_{\mr{BRST}})$ be a classical BRST theory. The \defterm{associated classical BV theory} is the $(-1)$-shifted cotangent bundle $T^*[-1]\cM$ equipped with its natural $(-1)$-shifted symplectic structure and the action functional $S_{\mathrm{BRST}} + S_{\rm anti}$.
\end{definition}

In the case where $S_{\mr{BRST}} = 0$, the associated BV theory is of cotangent type, as in Example \ref{def:cotangent_type}.

\begin{remark}
In general, multiple BRST theories can give rise to the same BV theory.  A BV theory $(E, \omega, S)$ is of cotangent type as long as there is \emph{some} $F$ with $S_{\mr{BRST}} = 0$ producing the given theory using the construction above.  Theories of cotangent type can still have interesting, non-trivial action functionals, encoded by the $L_\infty$ structure on $F$.
\end{remark}

\begin{remark}
It is common to refer to coordinates on $\cM$ as \defterm{fields} and coordinates along the cotangent direction of $T^*[-1]\cM$ as \defterm{antifields}. Given a field $\phi$ the corresponding antifield will be denoted by $\phi^*$.
\end{remark}

\subsection{Examples of Classical Field Theories}

In this section we give some examples of classical field theories we will use in our classification of twisted supersymmetric field theories. All theories we consider in this section are $\ZZ$-graded.

\subsubsection{Generalized BF Theory} \label{gen_BF_section}

Our first example will generalize the fundamental example of \emph{BF theory} to a not entirely topological context.  Ordinarily, BF theory describes the classical BRST theory on a $d$-manifold $M$ with fields given by a $G$-gauge field $A$ and a $\gg$-valued $(d-2)$-form $B$, with action functional
\[S(A,B) = \int_M \langle B \wedge F_A \rangle.\]
This theory is, in fact, of cotangent type, where the base of the cotangent includes $A$ and its antifield, and the fiber includes $B$ and its antifield.  This basic setup can be generalized to a setting where $M$ need not be entirely topological, and where $\gg$ may be a more general $L_\infty$ algebra, in the following way.

\begin{definition}
Let $X$ and $Y$ be complex manifolds and $M$ a smooth manifold. Fix an $L_\infty$ algebra $\fg$. The \defterm{generalized BF theory} is the classical field theory
\[T^*[-1]\map(X\times Y_{\mathrm{Dol}}\times M_{\mathrm{dR}}, B\fg).\]
\label{def:generalizedBF}
\end{definition}

Let us unpack the definition. Let $d = \dim_\CC(X) + 2\dim_\CC(Y) + \dim(M)$. Then the bundle of BV fields is
\[E = \Omega^{0,\bu}_X \otimes \Omega^{\bu,\bu}_Y \otimes \Omega^\bu_M \otimes \fg[1]\oplus \Omega^{\dim(X),\bu}_X \otimes \Omega^{\bu,\bu}_Y \otimes \Omega^\bu_M\otimes \fg^*[d-2],\]
where we denote the two fields by $A$ and $B$. The BV action is
\[S = \int_{X\times Y\times M} \langle B\wedge (\dbar_X + \dbar_Y + \d_{\dR, M}) A\rangle + \sum_{k\geq 1}\frac{1}{k!} \int_{X\times Y\times M} \langle B\wedge \ell_k(A, \dots, A)\rangle,\]
where $\langle -, -\rangle$ is the natural pairing between $\fg^*$ and $\fg$ and where $\ell_k$ denotes the $k^{\text{th}}$ component of the $L_\infty$ structure on $\fg$.

\begin{example}
For $X=Y=\pt$ and $\fg$ an ordinary Lie algebra we recover the usual topological BF theory with BV action functional
\[S = \int_M \left\langle B\wedge \left(\d_\dR A+ \frac{1}{2}[A\wedge A]\right)\right\rangle.\]
\end{example}

We will see many BF theories as the output when we twist supersymmetric gauge theories.
In fact, a special case of the definition above also arises when twisting theories of matter.  We will refer to as this as a generalized $\beta\gamma$ system, where the definition will extend that of the usual 2d $\beta \gamma$ system. 

\begin{definition}
Let $X$ and $Y$ be complex manifolds and let $M$ be a smooth manifold. Fix a complex vector space $V$. The \defterm{generalized $\beta\gamma$ system} is the classical field theory
\[T^*[-1]\map(X\times Y_{\mathrm{Dol}}\times M_{\mathrm{dR}}, V).\]
\end{definition}

The following is obvious from the definition.

\begin{prop}
Let $\fg$ be a dg Lie algebra, and consider the generalized BF theory on $\RR^{2n_1+2n_2+n_3}$ for Lie algebra $\fg$ with the space of BV fields
\[\cE = T^*[-1]\map(\CC^{n_1}\times (\CC^{n_2})_{\mr{Dol}}\times (\RR^{n_3})_{\mr{dR}}, B\fg).\]
Then it carries an action of $\U(n_1)\times \U(n_2)\times \mathrm{O}(n_3)$ given by the pullback action on differential forms on $\CC^{n_1}\times \CC^{n_2}\times \RR^{n_3}$.
\label{prop:BFrotationaction}
\end{prop}

\begin{remark}
In fact, the $\mathrm{O}(n_3)$-action given by the previous proposition extends to a homotopically trivial action in the sense of \cite[Section 2.4]{ElliottSafronov}.
\end{remark}

\subsubsection{Generalized Chern--Simons Theory} \label{gen_CS_section}
The next class of examples of classical BV theories will be generalizations of Chern--Simons theory. Unlike the example of the generalized BF theory, these theories are not generally of cotangent type.

\begin{definition}
Let $X$ and $Y$ be complex manifolds and $M$ a smooth oriented manifold. Fix an $L_\infty$ algebra $\fg$. We assume $X$ is equipped with a holomorphic volume form $\Omega_X \in\Omega^{\dim(X), 0}(X)$ and $\fg$ is equipped with a nondegenerate invariant symmetric pairing $\langle-, -\rangle\colon \fg\otimes\fg\rightarrow \CC[\dim_\CC(X) + 2\dim_\CC(Y) + \dim(M) - 3]$. The \defterm{generalized Chern--Simons theory} is the classical field theory
\[\map(X\times Y_{\mathrm{Dol}}\times M_{\mathrm{dR}}, B\fg),\]
whose underlying local $L_\infty$ algebra
\[\Omega^{0,\bu}_X \otimes \Omega^{\bu,\bu}_Y \otimes \Omega^\bu_M \otimes \fg\]
is equipped with a shifted symplectic structure coming from the nondegenerate pairing on $\fg$ and the wedge product of forms on $X\times Y\times M$.
\label{def:generalizedCS}
\end{definition}

We may also consider a $\ZZ/2\ZZ$-graded version of the above theory where $\fg$ is merely $\ZZ/2\ZZ$-graded. If we assume that $B\fg$ carries a $\Gm$-action, we may define the generalized Chern--Simons theory without choosing a holomorphic volume form on $X$ (see also \cite{GinzburgRozenblyum}).

\begin{definition}
Let $X,Y$ be complex manifolds and $M$ a smooth oriented manifold. Fix an integer $m$ and suppose $X$ is equipped with an $m$-th root of the canonical bundle $K_X^{1/m}$. Let $\fg$ be an $L_\infty$ algebra equipped with a $\ZZ$-grading $\fg=\bigoplus_n \fg(n)$ and equipped with a symmetric pairing $\langle -, -\rangle$ as before, which has weight $m$ with respect to the grading. Then the generalized Chern--Simons theory is the classical field theory
\[\Sect(X\times Y_{\mathrm{Dol}}\times M_{\mathrm{dR}}, B\fg\times_{\Gm} K_X^{1/m}).\]
\label{def:sectChernSimons}
\end{definition}

\begin{example}
For $X=Y=\pt$, $M$ a 3-manifold and $\fg$ an ordinary Lie algebra we recover the usual 3-dimensional Chern--Simons theory with the BV action
\[S = \int_M \left(\frac{1}{2}\langle A\wedge \d_{\dR} A\rangle + \frac{1}{6}\langle A\wedge [A\wedge A]\rangle\right).\]
More generally, if $X=Y=\pt$ and $M$ is any $d$-dimensional manifold where $d$ is odd, we recover $d$-dimensional Chern--Simons theory.  This has the same BV action, where now $A$ is a (not necessarily homogeneous) differential form on $M$.  If $d$ is not 3 this theory is only $\ZZ/2\ZZ$-graded.
\end{example}

\begin{example}
For $Y=M=\pt$, $X$ a Calabi-Yau 3-fold and $\fg$ an ordinary Lie algebra we recover the holomorphic Chern--Simons theory with the BV action
\[S = \int_X \Omega_X\wedge \left(\frac{1}{2}\langle A\wedge \dbar A\rangle + \frac{1}{6}\langle A\wedge [A\wedge A]\rangle\right).\]
As in the previous example, this still makes sense if $X$ is a Calabi-Yau $d$-fold with $d$ odd, as a $\ZZ/2\ZZ$-graded theory.
\end{example}

\begin{example}
Using Corollary \ref{cor:mapintocotangent} we have an isomorphism
\[T^*[-1]\map(X\times Y_{\mathrm{Dol}}\times M_{\mathrm{dR}}, B\fg)\cong \map(X\times Y_{\mathrm{Dol}}\times M_{\mathrm{dR}}, T^*[d-1]B\fg),\]
so the generalized BF theory may be considered as a particular example of a generalized Chern--Simons theory.
\label{ex:CSBF}
\end{example}

\begin{example}
Let $X,Y,M$ be as before and denote $d = \dim_\CC(X) + 2\dim_\CC(Y) + \dim(M)$. Suppose $\fg$ is a dg Lie algebra and $U$ is a $\fg$-representation equipped with a $(d-1)$-shifted symplectic structure $U\otimes U\rightarrow  \CC[d-1]$. Then we may consider the generalized Chern--Simons theory
\[\map(X\times Y_{\mathrm{Dol}}\times M_{\mathrm{dR}}, U \ham \fg).\]
\label{ex:CSHamiltonianreduction}
\end{example}

\begin{example}
Consider the setting of Example \ref{ex:CSHamiltonianreduction}. The formal moduli problem $U\ham \fg$ carries a natural $\Gm$-action which acts on $U$ with weight $1$. With respect to this weight grading, the $(d-1)$-shifted symplectic structure on $U\ham \fg$ has weight $2$, so we may define the corresponding generalized Chern--Simons theory with the space of fields
\[\Sect(X\times Y_{\mr{Dol}}\times M_{\mr{dR}}, U\ham \fg\times_{\Gm} K_X^{1/2}).\]
We will also use the notation
\[\Sect(X\times Y_{\mr{Dol}}\times M_{\mr{dR}}, (U\otimes K_X^{1/2})\ham \fg)\]
for the same theory.
\end{example}

\begin{example}
There is a special case of this, which one might label the ``holomorphic symplectic boson'' \cite[Definition 4.8]{SWSuperconformal} which is the theory
\[\Sect(X, (U\otimes K_X^{1/2})\ham\fg) . \]
\end{example}

As with generalized BF theory, generalized Chern--Simons theory carries a natural rotation action by linear automorphism groups of spacetime.

\begin{prop}
Suppose $\fg$ is a dg Lie algebra equipped with a nondegenerate invariant symmetric pairing of degree $n_1+2n_2+n_3-3$. Consider the generalized Chern--Simons theory
\[\cE = \map(\CC^{n_1}\times (\CC^{n_2})_{\mr{Dol}}\times (\RR^{n_3})_{\mr{dR}}, B\fg).\]
Then it carries an action of $\SU(n_1)\times \U(n_2)\times \SO(n_3)$ given by the pullback action on differential forms on $\CC^{n_1}\times \CC^{n_2}\times \RR^{n_3}$.
\end{prop}

We may slightly enhance the previous proposition if we are in the setting of Definition \ref{def:sectChernSimons}. Define the unitary metalinear group to be
\[\MU(n) = \U(n)\times_{\U(1)} \U(1),\]
where $\U(n)\rightarrow \U(1)$ is the determinant map and $\U(1)\rightarrow \U(1)$ is the map $z\mapsto z^2$. We denote by $\det^{1/2}\colon \MU(n)\rightarrow \U(1)$ the projection on the second factor; this may be thought of as a square root of the determinant representation of $\U(n)$. The natural $\U(n)$-action on $\CC^n$ lifts to an $\MU(n)$-action on the bundle of half-densities $K_{\CC^n}^{1/2}\rightarrow \CC^n$.

\begin{prop}
Suppose $\fg$ is a dg Lie algebra and $U$ a $\fg$-representation equipped with a $(n_1+2n_2+n_3-1)$-shifted symplectic structure. Consider the generalized Chern--Simons theory
\[\Sect(\CC^{n_1}\times (\CC^{n_2})_{\mr{Dol}}\times (\RR^{n_3})_{\mr{dR}}, (U\otimes K_{\CC^{n_1}}^{1/2})\ham\fg).\]
Then it carries an action of $\MU(n_1)\times \U(n_2)\times \SO(n_3)$ given by the pullback action on differential forms on $\CC^{n_1}\times \CC^{n_2}\times \RR^{n_3}$.
\end{prop}

\subsubsection{Generalized Hodge Theory}
Generalized BF theories can be naturally deformed to theories which are perturbatively trivial, but which arise as shadows of non-trivial non-perturbative theories.  These will often appear as topological twists of supersymmetric field theories, the most famous example being the 2d A-model. By a deformation we will mean a $\CC[t]$-family of classical BV theories which reduce to the given theory at $t=0$.

Given an $L_\infty$ algebra $\fg$ we denote by $\fg_{\Hod}$ the $\CC[t]$-linear $L_\infty$ algebra
\[\fg_{\Hod} = \CC[t]\otimes (\fg\oplus \fg[1])\]
with the $L_\infty$ brackets coming from the $L_\infty$ brackets on $\fg$ in the first term, where we consider $\fg[1]$ as the adjoint representation of $\fg$. The differential is given by the original differential on $\fg$ plus the term $t\id$ from the second summand to the first summand. We define $\fg_{\mathrm{dR}}$ to be the value of $\fg_{\Hod}$ at $t=1$. Note that the underlying complex is contractible.

\begin{remark}
The terminology comes from Simpson's Hodge stack \cite{Simpson}: if $X$ is a smooth scheme, one can define a derived stack $X_{\mr{Hod}}$ over $\bb A^1$, where the fiber at $0 \in \bb A^1$ is the Dolbeault stack $X_{\mr{Dol}}$ of $X$ and the fiber at a non-zero point is equivalent to the de Rham stack $X_{\mr{dR}}$ of $X$, which has a contractible tangent complex.
\end{remark}

If $\fg$ carries a nondegenerate invariant symmetric pairing of degree $d$, so does $\fg_{\Hod}$.

\begin{definition} \label{Hodge_family_def}
Let $X$ and $Y$ be complex manifolds and let $M$ be a smooth oriented manifold. Fix an $L_\infty$ algebra $\fg$. We assume $X$ is equipped with a holomorphic volume form $\Omega_X \in\Omega^{\dim(X), 0}(X)$ and $\fg$ is equipped with a nondegenerate invariant symmetric pairing $\langle-, -\rangle\colon \fg\otimes\fg\rightarrow \CC[\dim_\CC(X) + 2\dim_\CC(Y) + \dim(M) - 3]$. The \defterm{generalized Hodge theory} is the $\CC[t]$-family of classical BV theories, as in Definition \ref{family_of_BV_theories_def}, given by the generalized Chern--Simons theory
\[\map(X\times Y_{\mr{Dol}}\times M_{\mr{dR}}, B \fg_\Hod).\]
\end{definition}

\begin{prop}
The $t=0$ specialization of the generalized Hodge theory $\map(X\times Y_{\mr{Dol}}\times M_{\mr{dR}}, B \fg_\Hod)$ is isomorphic to the generalized BF theory $T^*[-1]\map(X\times Y_{\mr{Dol}}\times M_{\mr{dR}}, B \fg)$. The specialization of the generalized Hodge theory at $t\neq 0$ is perturbatively trivial.
\label{prop:Hodgetheoryspecialization}
\end{prop}

\begin{proof}
At $t=0$ we get
\[\left.\fg_\Hod\right|_{t=0}\cong \fg\oplus \fg[1]\cong \fg\oplus \fg^*[d-1],\]
where we use the symmetric bilinear pairing on $\fg$ in the second isomorphism. The first claim then follows from Example \ref{ex:CSBF}.

At $t\neq 0$ the $L_\infty$ algebra $\fg_\Hod$ becomes acyclic, which proves the second claim.
\end{proof}

\subsection{Dimensional Reduction} \label{dim_red_section}

In this section we formulate the procedure of dimensional reduction of a classical field theory. Fix a submersion $p\colon M\rightarrow N$ equipped with a fiberwise volume form, i.e. an isomorphism $p^*\Dens_N\cong \Dens_M$.  The idea is that the \emph{dimensional reduction} of a classical field theory on $M$ along the submersion $p$ is the theory obtained by restricting to those fields which are constant along the fibers of $p$.  We will begin with an abstract definition of dimensional reduction, then prove that if $M = N \times \RR^k$, and we consider field theories which are translation invariant along the fiber, then this procedure is well-defined.

\begin{definition}
We say that a classical field theory $(E_N, \omega_N, S_N)$ on $N$ is a \defterm{dimensional reduction} along $p$ of the classical field theory $(E_M, \omega_M, S_M)$ on a manifold $M$ if one is given the data of a linear isomorphism $p^* E_N\cong E_M$ of the bundles of BV fields satisfying the following conditions:
\begin{itemize}
\item The diagram
\[
\xymatrix{
p^* E_N\otimes p^* E_N \ar^{\omega_N}[r] \ar^{\sim}[d] & p^*\Dens_N[-1] \ar^{\sim}[d] \\
E_M\otimes E_M \ar^{\omega_M}[r] & \Dens_M[-1]
}
\]
is commutative.

\item Under the map $p^*\colon \cE_N\rightarrow \cE_M$ we have $p^* S_M = S_N$.
\end{itemize}
\end{definition}

We have an obvious notion of isomorphisms of dimensional reductions: these are linear isomorphisms of classical field theories on $N$ which are compatible with the isomorphisms $p^* E_N\cong E_M$. Thus, the collection of dimensional reductions of a given classical field theory on $M$ forms a groupoid.

\begin{prop}
Suppose $(E_M, \omega_M, S_M)$ is a classical field theory on $M$ and $p\colon M\rightarrow N$ is a homotopy equivalence. Then the groupoid of dimensional reductions of $(E_M, \omega_M, S_M)$ is either contractible or empty.

Suppose $M=N\times \RR$ and choose a translation-invariant density along the $\RR$ direction. If the original classical field theory is translation-invariant along the $\RR$ direction, dimensional reductions exist, i.e. the groupoid is non-empty.
\label{prop:dimensionalreductionunique}
\end{prop}
\begin{proof} \textbf{Uniqueness}. We begin by showing that any two dimensional reductions are isomorphic and moreover that such an isomorphism is unique if it exists. Since $p\colon M\rightarrow N$ is a homotopy equivalence, the functor $p^*$ establishes an equivalence between the category of graded vector bundles on $N$ and on $M$. In a similar way, $p^*$ establishes an equivalence between the category of graded vector bundles $E_N$ on $N$ equipped with a nondegenerate pairing $E_N\otimes E_N\rightarrow \Dens_N[-1]$ and a similar category for $M$.

Since $\cE_N\rightarrow \cE_M$ is injective, the condition $p^* S_M = S_N$ uniquely determines $S_N$.

\textbf{Existence}. Now suppose $(E_M, \omega_M, S_M)$ is translation-invariant along the $\RR$ direction. Translation invariance provides the descent datum to construct the bundle of fields $E_N$ on $N$ equipped with a nondegenerate pairing $\omega_N$. Moreover, the restriction of $S_M$ under $\cE_N\hookrightarrow \cE_M$ is independent of the $\RR$ factor by translation invariance, so $S_N=p^* S_M$ is again a local functional.
\end{proof}

\begin{remark}
Therefore, it makes sense to talk about ``the'' dimensional reduction of a classical field theory along the projection $p \colon N \times \RR \to N$: there exists a dimensional reduction which is unique up to a canonical isomorphism.
\end{remark}

We will now describe dimensional reductions of some of the previously discussed BV theories.
Throughout this section we let $X$ and $Y$ be complex manifolds and $M$ be a smooth manifold. 
We focus on BV theories described as formal mapping spaces whose sources are formal manifolds of the form 
\begin{equation}\label{eqn: generalized}
X \times Y_{\rm Dol} \times M_{\rm dR} .
\end{equation}

We first consider the case in which $M$ is of the form $M' \times \RR$ and we reduce along the projection
\[p_{\rm dR} \colon X \times Y \times (M' \times \RR) \to X \times Y \times M' .\]
In this case, we will only need to know the dimensional reduction of generalized Chern--Simons theory.

\begin{prop} \label{CS_to_BF_diml_red_prop}
Fix an $L_\infty$ algebra $\fg$ equipped with a nondegenerate invariant pairing as in Definition \ref{def:generalizedCS} and consider the corresponding generalized Chern--Simons theory
\[\map(X\times Y_{\mr{Dol}}\times (M'\times \RR)_{\mr{dR}}, B\fg).\]
Its dimensional reduction along the projection $p_{\rm dR}$ is equivalent to the generalized BF theory
\[
T^*[-1] {\rm Map}(X\times Y_{\mr{Dol}}\times M'_{\mr{dR}}, B\fg) .
\]
\end{prop}

\begin{proof}
To simplify the notation in the proof, we assume $X,Y,M'=\pt$, though the argument in the general case is identical. Then $\fg$ carries a $(-2)$-shifted pairing $\langle-,-\rangle$. In particular, the generalized BF theory
\[T^*[-1] \map(\pt, B \fg) = T^*[-1] B \fg\]
has the bundle of BV fields $\fg[1] \oplus \fg^*[-2]$. We may identify it with $\fg[1]\oplus \fg$, where the pairing $\omega_N$ pairs the two factors using $\langle-,-\rangle$.

We may identify $p^*(\fg[1]\oplus \fg)\cong \Omega^\bullet_{\RR}\otimes \fg[1]$ as vector bundles on $\RR$. 
Under this identification the integration pairing $\omega_M$ on differential forms reduces to the pairing $\omega_N$. The de Rham differential vanishes on translation-invariant forms, which shows compatibility of dimensional reduction with the differentials $Q_{\mr{BV}}$. Finally, in both cases the interaction term comes from the $L_\infty$ structure on $\fg$.
\end{proof}

Next, we consider dimensional reduction along a holomorphic direction. 
First, we set up some notation.

Let $V_\RR$ be a real vector space equipped with a nondegenerate symmetric bilinear pairing and an orientation. 
The symmetric bilinear pairing trivializes $\det(V_\RR)^{\otimes 2}$ and the orientation allows us to obtain a trivialization of $\det(V_\RR)$, i.e. a real volume form. 
We denote by $V=V_\RR\otimes_{\RR}\CC$ its complexification. 
Note that $V$ inherits a nondegenerate Hermitian form from the symmetric bilinear pairing on $V_\RR$. 
Also, since
\[\det(V)\cong \det(V_\RR)\otimes_{\RR}\CC\]
the real volume form on $V_\RR$ determines a complex volume form on $V$. 

Complexification yields a group homomorphism
\begin{equation}
\SO(V_\RR)\longrightarrow \SU(V)
\label{eq:SOtoSU}
\end{equation}
such that the real projection $\Re\colon V\rightarrow V_\RR$ is $\SO(V_\RR)$-equivariant.

As in Equation (\ref{eqn: generalized}), we assume that $X$ is a complex manifold of the form $X' \times V$.  We now consider the dimensional reduction along the map
\[
p_{\dbar}\colon (X' \times V) \times Y \times M \to X' \times Y \times (M \times V_\RR)
\]
induced by ${\rm Re}\colon V \to V_{\RR}$.

\begin{prop} \label{CS_diml_red_prop}
Let $X,Y,M,\fg$ be as before and $V_\RR, V$ as above. Fix an $L_\infty$ algebra $\fg$ equipped with a nondegenerate invariant pairing as in Definition \ref{def:generalizedCS} and consider the generalized Chern--Simons theory
\[\map((X'\times V)\times Y_{\mr{Dol}}\times M_{\mr{dR}}, B\fg).\]
Its dimensional reduction along the projection $p_{\dbar}$ is equivalent to the generalized Chern--Simons theory
\[\map(X'\times Y_{\mr{Dol}}\times (M\times V_\RR)_{\mr{dR}}, B\fg).\]
The equivalence is $\SO(V_\RR)$-equivariant.
\end{prop}
\begin{proof}
We may assume $X',Y,M=\pt$ as in the previous proof.

We have an isomorphism of vector bundles on $V$:
\[\Omega^{0, \bullet}_V\cong \Sym_\CC^\bullet(\underline{V}[-1])\]
and similarly an isomorphism of bundles on $V_\RR$:
\[\Omega^\bullet_{V_\RR}\cong \Sym_\RR^\bullet(\underline{V}_\RR[-1])\otimes_\RR\CC.\]

Under the composition
\[\Omega^\bullet(V_\RR; \CC)\xto{p^*} \Omega^\bullet(V)\rightarrow \Omega^{0, \bullet}(V)\]
defined by pulling back forms along $p$ and projecting onto $(0,\bu)$-forms, the map $\Omega_V\wedge (-)\colon \Omega^{0, \bullet}_V \to \Omega^{{\rm dim}(V),\bu}_V \to \Omega^{\dim(V),\dim(V)}_V = {\rm Dens}_V$ given by projection onto the top component, reduces to the map $\Omega^{\bullet}(V_\RR;\CC)\rightarrow \Omega^{\dim(V)}_{V_\RR} = {\rm Dens}_{V_\RR}$ which also projects onto the top component. 
This shows that the BV pairings of the original and dimensionally reduced theory are compatible. 
\end{proof}

As a corollary of this result we obtain the reduction of generalized Hodge theory.

\begin{corollary}
Let $X', Y, M, V_\RR, V$ be as before. 
Fix an $L_\infty$ algebra $\fg$ equipped with a non-degenerate pairing as in Definition \ref{Hodge_family_def} and consider the generalized Hodge theory
\[\map((X'\times V)\times Y_{\mr{Dol}}\times M_{\mr{dR}}, B\fg_\Hod).\]
Its dimensional reduction along the projection $p_{\dbar}$ is equivalent to the generalized Hodge theory
\[\map(X'\times Y_{\mr{Dol}}\times (M\times V_\RR)_{\mr{dR}}, B\fg_\Hod).\]
The equivalence is $\SO(V_\RR)$-equivariant.
\label{cor:Hodgeholomorphicreduction}
\end{corollary}

We have the following analogs of Propositions \ref{CS_to_BF_diml_red_prop} and \ref{CS_diml_red_prop} for the generalized BF theory.

\begin{prop}
Let $X, Y, M', \fg$ be as in Proposition \ref{CS_to_BF_diml_red_prop}. The dimensional reduction of the generalized BF theory
\[T^*[-1]\map(X\times Y_{\mr{Dol}}\times (M'\times \RR)_{\mr{dR}}, B\fg)\]
along the projection $p_{\rm dR}$ is equivalent to the generalized BF theory
\[
T^*[-1] \map(X\times Y_{\mr{Dol}}\times M'_{\mr{dR}}, \fg/\fg).
\]
\label{prop:BFdeRhamreduction}
\end{prop}

\begin{prop}
Let $X', Y, M, V_\RR, V$ be as in \ref{CS_diml_red_prop}. Fix an $L_\infty$ algebra $\fg$ and consider the generalized BF theory
\[T^*[-1] \map((X'\times V)\times Y_{\mr{Dol}}\times M_{\mr{dR}}, B\fg).\]
Its dimensional reduction along the projection $p_{\dbar}$ is equivalent to the generalized BF theory
\[T^*[-1] \map(X'\times Y_{\mr{Dol}}\times (M\times V_\RR)_{\mr{dR}}, B\fg).\]
This equivalence is $\SO(V_\RR)$-equivariant.
\label{prop:BFholomorphicreduction}
\end{prop}

Finally, we take $Y$ to be a complex manifold of the form $Y' \times \CC$, where $V$ is complex $k$-dimensional and consider the projection
\[
p_{\rm Dol} \colon X \times (Y' \times \CC) \times M \to X \times Y' \times (M \times \RR) 
\]
induced by ${\rm Re} \colon \CC \to \RR$.

\begin{prop}
\label{prop:BFdolbeaultreduction}
Let $X, M, Y'$ be as above.
Fix an $L_\infty$ algebra $\fg$ and consider the generalized BF theory
\[
T^*[-1] {\rm Map}(X \times (Y' \times \CC)_{\rm Dol} \times M_{\rm dR} , B \fg) .
\]
Its dimensional reduction along $p_{\rm Dol}$ is equivalent to the generalized BF theory
\[
T^*[-1] {\rm Map} (X \times Y'_{\rm Dol} \times (M \times \RR)_{\rm dR}, \fg / \fg)
\]
This equivalence is $\SO(V_\RR)$-equivariant.
\end{prop}
\begin{proof}
Note that there is an isomorphism of $L_\infty$ algebras
\[
\Omega^{\bu, \bu} (Y' \times \CC ; \fg) = \Omega^{\bu,\bu}(Y') \otimes \Omega^{0,\bu}(\CC) \otimes \fg[\epsilon]
\]
where $\epsilon$ is a parameter of degree $+1$.
The result then follows from Proposition \ref{prop:BFholomorphicreduction} applied to the $L_\infty$ algebra $\fg[\epsilon]$. 
\end{proof}

\section{Supersymmetry} \label{sect:susy}
Having set up the formalism behind classical field theories in the BV and BRST formalisms, we will introduce the other main formal ingredient of this paper: the supersymmetry action.  So, we will discuss the classification of supersymmetry algebras, the notion of a supersymmetric field theory, and the idea of a \emph{twist} of a supersymmetric field theory, extending work of the first two authors in \cite{ElliottSafronov}.  We will introduce the classification of supersymmetry algebras using the division algebra perspective of Baez and Huerta \cite{BaezHuerta}, which will be useful for the classification of super Yang--Mills theories in Section \ref{sect:SYM} below.

\subsection{Spinors}
\label{sect:spinors}

In this paper we will extensively use the theory of spinors. Let $V$ be a complex vector space equipped with a nondegenerate symmetric bilinear pairing. Recall that the Clifford algebra $\Cl(V)$ is defined to be the quotient of the tensor algebra on $V$ by the relation
\[v_1 v_2 + v_2 v_1 = 2(v_1, v_2).\]
Consider a $\ZZ/2\ZZ$-graded Clifford module $M=M^+\oplus M^-$. Denote the Clifford action by $\rho(v)\in\eend(M)$. We assume the Clifford module is equipped with a nondegenerate pairing $(-, -)\colon M^+\otimes M^-\rightarrow \CC$ such that
\[(\rho(v) Q_1, Q_2) = (Q_1, \rho(v) Q_2)\]
for any $Q_1, Q_2\in M$ and $v\in V$. From now on we denote $M^+=\Sigma$ and $M^-=\Sigma^*$. We define the $\Gamma$-pairings
\[\Gamma\colon\Sym^2(\Sigma)\longrightarrow V,\qquad \Gamma\colon\Sym^2(\Sigma^*)\longrightarrow V\]
by
\begin{equation}
(\rho(v) Q_1, Q_2) = (v, \Gamma(Q_1, Q_2))
\label{eq:Gammaspinorpairing}
\end{equation}

We have a subset $\Spin(V)\subset \Cl(V)$, so $\Sigma$ and $\Sigma^*$ are representations of the spin group. Moreover, the Clifford action and the maps $\Gamma$ are $\Spin(V)$-equivariant.

We may identify $\wedge^2(V)\rightarrow \so(V)$ via
\[\omega\mapsto (w\mapsto -2\iota_{(w, -)} \omega).\]
This gives rise to an action map
\[\wedge^2(V)\otimes \Sigma\longrightarrow \Sigma\]
of two-forms on spinors.

Consider the map $q\colon \wedge^\bullet(V)\rightarrow \Cl(V)$ given by antisymmetrization, so that, for instance,
\begin{equation}
q(v_1\wedge v_2) = v_1v_2 - (v_1, v_2).
\label{eq:quantizationtwoforms}
\end{equation}

The resulting action $\wedge^2(V)\otimes \Sigma\rightarrow \Sigma$ then coincides with the original action of $\so(V)$ on the spinorial representation $\Sigma$, so that $\so(V)$-equivariance of $\Gamma$ gives the following.

\begin{prop}
For $X\in\wedge^2(V)$ and $Q_1,Q_2\in\Sigma$ we have
\[\Gamma(Q_1, \rho(X) Q_2) + \Gamma(Q_2, \rho(X) Q_1) = -2\iota_{\Gamma(Q_1, Q_2)} X.\]
\label{prop:cliffordactionproperty1}
\end{prop}

We may extend the discussion to the case of Riemannian manifolds $N$, where we replace $V$ by the vector bundle $TN$. Given a bundle of Clifford modules $M=\Sigma\oplus \Sigma^*$ as before we have the associated Dirac operator
\[\sd{\d}\colon \Gamma(N, \Sigma)\rightarrow \Gamma(N, \Sigma^*).\]

From \eqref{eq:quantizationtwoforms} we get the following property.

\begin{prop}
Suppose $Q_1,Q_2\in\Sigma$ and $\lambda\in\Gamma(N, \Sigma)$. Then
\[\sd{\d}\rho(\Gamma(Q_1, \lambda)) Q_2 = \rho(d \Gamma(Q_1, \lambda)) Q_2 + (Q_1, \sd{\d} \lambda) Q_2.\]
\label{prop:cliffordactionproperty2}
\end{prop}

Finally, we have the following important compatibility between the Clifford action of differential forms and the Dirac operator proved in \cite[equation 7.6]{Snygg}.

\begin{prop}
Suppose $Q\in\Sigma$ and $X\in\Omega^p(N)$. Then
\[\sd{\d}(\rho(X) Q) = \rho(\d X)Q + (-1)^{n(1+p)}\rho(\ast \d\ast X) Q.\]
\label{prop:cliffordactionproperty3}
\end{prop}

Note that both Proposition \ref{prop:cliffordactionproperty2} and Proposition \ref{prop:cliffordactionproperty3} extend to the case when $\lambda$ and $X$ respectively are twisted by a vector bundle and $\sd{\d}$ is the corresponding twisted Dirac operator.

\subsection{Supersymmetry Algebras} \label{sect:susyalgebras}
In this section we will recall the framework for supersymmetry algebras and their classification following Deligne \cite{DeligneSpinors} and our previous work \cite{ElliottSafronov}, we refer there for more details.

Let $V_\RR$ be a finite-dimensional real vector space of dimension $n=\dim_\RR(V_\RR)$ equipped with an orientation and a nondegenerate symmetric bilinear form. Denote by $V=V_\RR\otimes_\RR\CC$ its complexification. Consider the Lie algebra $\so(V)$. Let us recall the following facts:
\begin{itemize}
\item If $n$ is odd, $\so(V)$ has a distinguished fundamental representation called the \defterm{spin} representation $S$.

\item If $n$ is even, $\so(V)$ has a pair of distinguished fundamental representations called the \defterm{semi-spin} representations $S_+$ and $S_-$.
\end{itemize}

\begin{definition}
A \defterm{spinorial representation} $\Sigma$ is a finite sum of spin or semi-spin representations of $\so(V)$.
\end{definition}

So, in odd dimensions we have $\Sigma=S\otimes W$ and in even dimensions we have $\Sigma=S_+\otimes W_+\oplus S_-\otimes W_-$, where $W$ is a finite dimensional multiplicity space.

We have an embedding $\U(n)\subset \SO(2n, \RR)$ which lifts to an embedding $\MU(n)\subset \Spin(2n, \RR)$. If we denote by $L$ the standard $n$-dimensional representation of $\U(n)$, then the semi-spin representations of $\Spin(2n, \RR)$ restrict to $\MU(n)$ as
\[S_+\cong \det(L)^{-1/2}\otimes \wedge^{\mathrm{even}} L,\qquad S_-\cong \det(L)^{-1/2}\otimes \wedge^{\mathrm{odd}} L.\]

\begin{definition}
Fix a spinorial representation $\Sigma$ and a nondegenerate $\so(V)$-equivariant pairing $\Gamma\colon \sym^2(\Sigma)\rightarrow V$. The \defterm{supertranslation Lie algebra} is the $\so(V)$-equivariant super Lie algebra $\fA = \Pi\Sigma\oplus V$ whose only nontrivial bracket is given by $\Gamma$.
\end{definition}

For a given spinorial representation, the pairing $\Gamma$ is typically unique up to a scale, so a supertranslation Lie algebra is specified by fixing a spinorial representation. In turn, a spinorial representation is determined by the dimension of the multiplicity space, so we will talk about $\mc{N}$ or $(\mc{N}_+, \mc{N}_-)$ supertranslation Lie algebras, where the numbers are specified as follows.
\begin{itemize}
\item If $n\equiv 0, 1, 3, 4\pmod 8$, we let $\mc{N} = \dim(W)$.

\item If $n\equiv 2 \pmod 8$, we let $\mc{N}_{\pm}=\dim(W_{\pm})$.

\item If $n\equiv 5, 7\pmod 8$, we let $2\mc{N} = \dim(W)$.

\item If $n\equiv 6\pmod 8$, we let $2\mc{N}_{\pm} = \dim(W_{\pm})$.
\end{itemize}

Fix the following data:
\begin{itemize}
\item A spinorial representation $\Sigma$ of $\so(V)$.

\item An $\so(V)$-equivariant nondegenerate pairing $\Gamma\colon \sym^2(\Sigma)\rightarrow V$.

\item A Lie group $G_R$, the \defterm{group of $R$-symmetries}, which acts on $\Sigma$ by $\so(V)$-equivariant automorphisms preserving $\Gamma$.
\end{itemize}

Note that the supertranslation Lie algebra $\fA$ is a $\Spin(V_\RR)\times G_R$-equivariant super Lie algebra.  We will sometimes want to refer to the infinitesimal version of this action.
\begin{definition}
Let $\fA$ be a supertranslation algebra.  The corresponding \defterm{supersymmetry algebra} is the super Lie algebra $(\so(V) \oplus \gg_R) \ltimes \fA$.
\end{definition}

We will now define the fundamental notion of a supersymmetric field theory.  Consider a spacetime manifold $M=V_\RR$. Let $\ISO(V_\RR) = \Spin(V_\RR)\ltimes V_\RR$ be the \defterm{Poincar\'{e} group} which acts by affine transformations on $M$.

\begin{definition}
\label{dfn: super}
A classical field theory $(E, S, \omega)$ is \defterm{supersymmetric} if $E\rightarrow M$ is an $\ISO(V_\RR)\times G_R$-equivariant vector bundle and the infinitesimal strict action of the translation Lie algebra $V$ on the classical theory is extended to a $\Spin(V_\RR)\times G_R$-equivariant $L_\infty$ action of the supertranslation Lie algebra $\fA$ on the classical theory.
\end{definition}

\subsection{Composition Algebras and Minimal Supersymmetry}
\label{sect:compositionalgebras}

We will now recall a relationship between certain ``minimal'' supersymmetry algebras and composition algebras. Our treatment will essentially follow that of Baez and Huerta \cite{BaezHuerta}.
This section provides the representation theoretic underpinning for theories with minimally supersymmetric matter in dimensions $\geq 3$.

Let $A$ be a unital (possibly non-associative) complex algebra equipped with an antiinvolution $\sigma\colon A\rightarrow A$. We make the following assumptions:
\begin{enumerate}
\item The map $\Re(x)=x\mapsto (x + \sigma(x))/2$ defines a projector onto the subspace of $A$ spanned by the unit.

\item By the previous assumption we have a quadratic form $x\sigma(x)\colon A\rightarrow \CC$. We assume that it is nondegenerate.
\end{enumerate}

In fact, the data of the antiinvolution $\sigma$ may equivalently be encoded in the data of a non-degenerate multiplicative norm $x\mapsto x\sigma(x)$, i.e. $A$ is a real composition algebra \cite[Chapter 1.3]{SpringerVeldkamp}.

For a $2\times 2$-matrix $M$ with entries in $A$ we define its hermitian adjoint $M^\dagger$ by transposing the matrix and applying $\sigma$ to the entries. 
Let $V$ to be the complex vector space of $2\times 2$ Hermitian matrices with values in $A$.
Note that $\dim(V) = \dim(A) + 2$. 
The space $V$ carries a nondegenerate quadratic form given by $M\mapsto -\det(M)$. 
Moreover, it carries an orthogonal involution $\widetilde{M} = M - \tr(M)\cdot 1$.
Let ${\rm C}\ell_V$ be the resulting Clifford algebra.

Given a left $A$-module, we may turn it into a right $A$-module via the antiinvolution $\sigma\colon A\rightarrow A$. Since $A$ is a Frobenius algebra, we have the following basic construction 

\begin{lemma}
Suppose $M$ is a left $A$-module and $N$ a right $A$-module equipped with a nondegenerate pairing $(-, -)\colon M\otimes N\rightarrow \CC$ satisfying $(am, n) = (m, na)$ for every $a\in A$, $m\in M$ and $n\in N$. Then there is a unique map $(-, -)^A\colon M\otimes N\rightarrow A$ of $(A, A)$-bimodules whose real part is $(-, -)$.
\label{lm:extendpairing}
\end{lemma}

Consider the left $A$-module $\Sigma=A\oplus A$ and the right $A$-module $\Sigma^* = A\oplus A$.
We will equip the sum $\Sigma \oplus \Sigma^*$ with an action of the Clifford algebra ${\rm C}\ell_V$.  
The action on the two summands $\Sigma$ and $\Sigma^*$ will be different, hence the different notation.  
Define the action maps
\[\rho \colon V \otimes \Sigma\rightarrow \Sigma^*,\qquad \rho \colon V\otimes \Sigma^*\rightarrow \Sigma\]
by
\[\rho(M)Q = M Q,\qquad \rho(M)Q = \widetilde{M} Q.\]

The following is proved in \cite[Proposition 6]{BaezHuerta}.

\begin{prop}
The action maps $V\otimes \Sigma\rightarrow \Sigma^*$ and $V\otimes \Sigma^*\rightarrow \Sigma$ satisfy the Clifford relation
\[\rho(M)\rho(M) = -\det(M)\cdot 1.\]
\end{prop}

As a result, $\Sigma\oplus \Sigma^*$ forms a $\ZZ/2\ZZ$-graded Clifford module.

We have a nondegenerate scalar spinorial pairing
\[\Sigma\otimes \Sigma^*\longrightarrow \CC\]
given by
\[(Q_1, Q_2) = \Re(Q_1^\dagger Q_2).\]
It obviously satisfies $(Q_1 \sigma(a), Q_2) = (Q_1, Q_2 a)$. The extension to an $A$-valued pairing provided by Lemma \ref{lm:extendpairing} is given by
\[(Q_1, Q_2)^A = Q_1^\dagger Q_2.\]

By duality we obtain maps $\Gamma\colon\Sym^2(\Sigma)\rightarrow V$ and $\Gamma\colon \Sym^2(\Sigma^*)\rightarrow V$ given, respectively, by
\[\Gamma(Q_1, Q_2) = \widetilde{Q_1Q_2^\dagger + Q_2Q_1^\dagger},\qquad \Gamma(Q_1, Q_2) = Q_1Q_2^\dagger + Q_2Q_1^\dagger.\]

We will now state two important properties of $\Gamma$ and the spinorial pairing. The following statement was proved in \cite[Theorem 11]{BaezHuerta} (see also \cite{Schray} for the case $\dim(V)=10$).

\begin{theorem}
Suppose $A$ is alternative, i.e. $a\otimes b\otimes c\mapsto (ab)c - a(bc)$ is completely antisymmetric. For $Q_1, Q_2, Q_3\in\Sigma$ we have
\[\rho(\Gamma(Q_1, Q_2))Q_3 + \rho(\Gamma(Q_2, Q_3))Q_1 + \rho(\Gamma(Q_2, Q_3))Q_1 = 0.\]
\label{thm:3psi}
\end{theorem}

If we moreover assume $A$ is associative, there is a relationship between the scalar spinorial pairing and $\Gamma$.

\begin{theorem}
Suppose $A$ is associative. For $Q_1, Q_2\in\Sigma$ and $Q_3\in\Sigma^*$ we have
\[Q_1(Q_2, Q_3)^A + Q_2(Q_1, Q_3)^A = \rho(\Gamma(Q_1, Q_2)) Q_3.\]
\label{thm:matter3psi}
\end{theorem}
\begin{proof}
The right-hand side is
\[(Q_1Q_2^\dagger + Q_2Q_1^\dagger)Q_3\]
which by associativity can be rewritten as
\[Q_1(Q_2^\dagger Q_3) + Q_2(Q_1^\dagger Q_3)\]
which is the left-hand side.
\end{proof}

Note that since $\dim(V) = \dim(A) + 2$, the above constructions are only valid in dimensions $\geq 3$.
We will be interested in the following examples:
\begin{enumerate}
\item (\textbf{3d $\mc N=1$ supersymmetry}) $A = \CC$. 
Here, $\dim(V) = 3$ and $\Sigma$ is the spin representation of $\Spin(3, \CC)$.

\item (\textbf{4d $\mc N=1$ supersymmetry}) $A = \CC\otimes_{\RR}\CC\cong \CC[x]/(x^2+1)$ with $\sigma(x) = -x$. Moreover, $\dim(V) = 4$ and $\Sigma=S_+\oplus S_-$ is the sum of semi-spin representations of $\Spin(4;\CC)$.

\item (\textbf{6d $\mc N=(1, 0)$ supersymmetry}) $A=\mathbb{H}\otimes_{\RR}\CC\cong \eend(W_+)$, where $W_+$ is a two-dimensional symplectic vector space with $\sigma$ given by the dual operator. Moreover, $\dim(V) = 6$ and $\Sigma = S_+\otimes W_+$ is the sum of two copies of a semi-spin representation of $\Spin(6;\CC)$.

\item (\textbf{10d $\mc N=(1, 0)$ supersymmetry}) $A=\mathbb{O}\otimes_{\RR}\CC$. We have $\dim(V) = 10$ and $\Sigma = S_+$ is a semi-spin representation of $\Spin(10;\CC)$.
\end{enumerate}

All four examples are alternative, while the first three examples are also associative.

\subsection{Two-dimensional Chiral Supersymmetry} \label{sect:2dchiral}

In the previous section we have related composition algebras to minimal supersymmetry algebras in dimensions 3, 4, 6 and 10. 
In this section we explain a different relationship between composition algebras and supersymmetry algebras, this time in the case of 2d $\cN=(\cN_+, 0)$ supersymmetry.

Recall that in the case $\dim(V) = 2$ we have two one-dimensional semi-spin representations $S_+, S_-$. Moreover, we have an isomorphism
\[V\cong S_+^{\otimes 2}\oplus S_-^{\otimes 2}\]
and a pairing $( -, -)\colon S_+\otimes S_-\rightarrow \CC$, both of which are $\so(V)$-equivariant. We denote the embeddings $S_{\pm}^{\otimes 2}\subset V$ by $\Gamma_{\pm}$, so that
\begin{equation}
(\Gamma_+(s, s), \Gamma_-(f, f)) = 2(s, f)^2.
\label{eq:2dvectorpairing}
\end{equation}

Let $W$ be a complex vector space of dimension $\cN_+$ equipped with a nondegenerate symmetric bilinear pairing. We consider the spinorial representation
\[\Sigma = S_+\otimes W\]
and its dual
\[\Sigma^* = S_-\otimes W.\]
The Clifford action $\rho \colon V\otimes S_+\rightarrow S_-$ is defined so that
\[\rho(\Gamma_-(f, f)) s = 2(s, f) f\]
and similarly for $V\otimes S_-\rightarrow S_+$.

\begin{prop}
For $v,w\in V$ and $s\in S_+\oplus S_-$ we have
\[\rho(v)\rho(w)s + \rho(w)\rho(v) s = 2(v, w) s.\]
\end{prop}
\begin{proof}
It is enough to prove the claim for $s\in S_+$, $w\in S_+^{\otimes 2}$ and $v\in S_-^{\otimes 2}$. Assume $w=\Gamma_+(s, s)$ and $v = \Gamma_-(f, f)$ for $f\in S_-$. Then we have
\begin{align*}
\rho(\Gamma_+(s, s)) \rho(\Gamma_-(f, f)) s &= 2(s, f) \rho(\Gamma_+(s, s)) f \\
&= 4(s, f)^2 s.
\end{align*}
But by \eqref{eq:2dvectorpairing} we have
\[(\Gamma_+(s, s), \Gamma_-(f, f)) = 2(s, f)^2\]
which proves the claim.
\end{proof}

The Clifford action $V\otimes S_{\pm}\rightarrow S_{\mp}$ extends in an obvious way to a Clifford action $V\otimes \Sigma\rightarrow \Sigma^*$ and $V\otimes \Sigma^*\rightarrow \Sigma$ given by the identity on the $W$ component. Thus, $\Sigma\oplus \Sigma^*$ is a $\ZZ/2\ZZ$-graded Clifford module.

The spaces $\Sigma, \Sigma^*$ are equipped with $\so(V)$-equivariant nondegenerate pairings $\Gamma\colon \Sym^2(\Sigma)\rightarrow V$ defined by
\[\Gamma(s\otimes q_1, s\otimes q_2) = \Gamma_+(s, s) (q_1, q_2)\]
and $\Gamma\colon \Sym^2(\Sigma^*)\rightarrow V$ defined similarly.

\begin{prop}
For any $v\in V$ and $Q_1, Q_2\in\Sigma$ or $Q_1, Q_2\in\Sigma^*$ we have
\[(v, \Gamma(Q_1, Q_2)) = (\rho(v) Q_1, Q_2).\]
\end{prop}
\begin{proof}
It is enough to prove the claim with $Q_1, Q_2\in\Sigma$. Assume $v = \Gamma_-(f, f)$ for some $f\in S_-$, $Q_1 = s\otimes q_1$ and $Q_2 = s\otimes q_2$. Then the left-hand side is
\[(\Gamma_-(f, f), \Gamma_+(s, s)) (q_1, q_2) = 2(s, f)^2 (q_1, q_2).\]
The right-hand side is
\begin{align*}
(\rho(\Gamma_-(f, f) s\otimes q_1, s\otimes q_2) &= (2(s, f) f\otimes q_1, s\otimes q_2) \\
&= 2(s, f)^2 (q_1, q_2).
\end{align*}
\end{proof}

An important property of two-dimensional chiral supersymmetry is the following analog of Theorem \ref{thm:3psi}.

\begin{theorem}
For $Q_1, Q_2, Q_3\in\Sigma$ we have
\[\rho(\Gamma(Q_1, Q_2))Q_3 = 0.\]
\label{thm:2d3psi}
\end{theorem}
\begin{proof}
Indeed, $\Gamma(Q_1, Q_2)$ lies in $S_+^{\otimes 2}\subset V$, but the nonzero Clifford action is given by
\[S_-^{\otimes 2}\otimes (S_+\otimes W)\longrightarrow S_-\otimes W.\]
\end{proof}

We will now fix a composition algebra $A$ with an antiinvolution $\sigma$ as in Section \ref{sect:compositionalgebras} and set $W = A$. The nondegenerate symmetric bilinear pairing $a_1, a_2\mapsto \Re(a_1\sigma(a_2))$ on $A$ endows $W$ with a pairing. Both $\Sigma$ and $\Sigma^*$ are right $A$-modules and the Clifford actions $V\otimes \Sigma\rightarrow \Sigma^*$ and $V\otimes \Sigma^*\rightarrow \Sigma$ are maps of right $A$-modules.

Since $\Sigma$ and $\Sigma^*$ are right $A$-modules, by Lemma \ref{lm:extendpairing} we may extend the scalar spinorial pairing to an $A$-valued pairing $\Sigma\otimes \Sigma^*\rightarrow A$ by
\[(s_1\otimes q_1, s_2\otimes q_2)^A = (s_1, s_2)\sigma(q_1) q_2.\]
We now give the analogue of Theorem \ref{thm:matter3psi} in the 2d chiral setting.

\begin{theorem}
For $Q_1, Q_2\in\Sigma$ and $Q_3\in\Sigma^*$ we have
\[Q_1(Q_2, Q_3)^A + Q_2(Q_1, Q_3)^A = \rho(\Gamma(Q_1, Q_2))Q_3.\]
\label{thm:2dmatter3psi}
\end{theorem}
\begin{proof}
Pick basis elements $s\in S_+$ and $f\in S_-$, so that
\[Q_1 = s\otimes q_1,\qquad Q_2 = s\otimes q_2,\qquad Q_3 = f\otimes q_3.\]
The right-hand side is
\begin{align*}
(q_1, q_2) \rho(\Gamma_+(s, s)) f\otimes q_3 &= 2(s, f) s\otimes (q_1, q_2) q_3 \\
&= (s, f) s\otimes (q_1\sigma(q_2) + q_2\sigma(q_1)) q_3.
\end{align*}
We have
\[Q_1(Q_2, Q_3)^A = s\otimes q_1 (s, f) (\sigma(q_2)q_3),\]
so the left-hand side is
\[s(s, f)\otimes (q_1(\sigma(q_2) q_3) + q_2(\sigma(q_1)q_3)).\]
By associativity of $A$ the two expressions are equal.
\end{proof}

\subsection{Supersymmetric Twisting}
The idea of twisting, originally developed by Witten \cite{WittenTQFT}, is to modify a classical BV theory by deforming the differential $Q_{\mr{BV}}$ by the action of a square-zero fermionic symmetry. 

\begin{definition}
A \defterm{square-zero supercharge} is a nonzero element $Q\in\Sigma$ such that $\Gamma(Q, Q)=0$. Its \defterm{number of invariant directions} is the dimension of the image of $\Gamma(Q, -)\colon \Sigma\rightarrow V$.
\end{definition}

It is shown in \cite[Proposition 3.25]{ElliottSafronov} that the number $d$ of invariant directions is at least $n/2$. We will use the following adjectives for square-zero supercharges depending on $d$:
\begin{itemize}
\item A supercharge $Q$ is \defterm{topological} if $d = n$.

\item A supercharge $Q$ is \defterm{holomorphic} if $n$ is even and $d=n/2$.

\item A supercharge $Q$ is \defterm{minimal} if $n$ is odd and $d=(n+1)/2$.
\end{itemize}

In the intermediate case we refer to $Q$ as a \defterm{holomorphic-topological} (alternatively, partially topological) supercharge. The collection of all square-zero supercharges in dimensions 2 through 10 (where one restricts to supersymmetries with at most 16 supercharges) was studied in \cite{ElliottSafronov} and \cite{EagerSaberiWalcher}. In particular, the orbits of square-zero supercharges under the $R$-symmetry group, $\Spin(V)$ and the obvious scaling action of $\CC^\times$ are shown in Figure \ref{fig:superchargeorbits}.

Let $(E, S, \omega)$ be a supersymmetric classical field theory. Recall, this means we have a Maurer-Cartan element 
\[
S_{\fA} = S + \sum_{k \geq 1} S_{\fA}^{(k)} \in C^\bu(\fA, \oloc(\cE))
\]
where $S_{\fA}^{(k)}\colon \fA^{\otimes k} \to \oloc(\cE)$ as in Definition \ref{infinitesimal_action_def} and the classical field theory has an action of the $R$-symmetry group $G_R$.

\begin{definition} \label{def:twisting}
Suppose $(E, S, \omega)$ is a supersymmetric classical field theory and $Q$ a square-zero supercharge. The \defterm{$Q$-twisted classical field theory} is the $\ZZ/2\ZZ$-graded classical field theory with the same bundle of $BV$ fields and symplectic pairing $\omega$, but with the BV action
\[S^Q = S + \sum_{k \geq 1} S_{\fA}^{(k)}(Q, \dots, Q).\]
\end{definition}

Given additional data, we may enhance the classical field theory.

\begin{definition}
Let $Q\in\Sigma$ be a square-zero supercharge. A homomorphism $\alpha\colon U(1)\rightarrow G_R$ is \defterm{compatible} with $Q$ if $Q$ has weight 1 and the $\alpha$-weight mod 2 on $E$ coincides with the fermionic grading.
\end{definition}

Given such an $\alpha$ we may consider a new $\ZZ$-grading on $E$ given by the sum of the cohomological grading and the grading given by $\alpha$. The map $S_{\fA}^{(k)}\colon \fA^{\otimes k} \to \oloc(\cE)$ is $G_R$-equivariant, so the element $S_{\fA}^{(k)}(Q, \dots, Q)$ has $\alpha$-weight $k$. But it also has cohomological degree $-k$. In other words, the twisted action $S^Q$ has total degree zero, so $(E, S^Q, \omega)$ is a $\ZZ$-graded classical field theory.

\begin{definition}
Let $Q\in\Sigma$ be a square-zero supercharge and suppose $\iota \colon G\rightarrow \Spin(V_\RR)$ is a fixed group homomorphism. A \defterm{twisting homomorphism} is a homomorphism $\phi\colon G\rightarrow G_R$ such that $Q$ is preserved under the product $(\iota, \phi) \colon G\rightarrow \Spin(V_\RR)\times G_R$.
\end{definition}

The classical field theory $(E, S, \omega)$ carries a $\Spin(V_\RR)\times G_R$-action. However, the $Q$-twisted theory $(E, S^Q, \omega)$ does not in general carry a $\Spin(V_\RR)\times G_R$-action since the elements $S_{\fA}^{(k)}(Q, \dots, Q)$ are not in general invariant under $\Spin(V_\RR)\times G_R$. However, given a twisting homomorphism $\phi$ we see that $S_{\fA}^{(k)}(Q, \dots, Q)$ is preserved under $G$, so $(E, S^Q, \omega)$ carries a $G$-action.

\subsection{Dimensional Reduction of Supersymmetric Theories}

Suppose $V_\RR=\RR^n$ as before and choose a subspace $W_\RR\subset V_\RR$, so that $V_\RR = W_\RR\oplus W^\perp_\RR$. We denote $W=W_{\RR}\otimes_{\RR} \CC$.

Fix a spinorial representation $\Sigma$ of $\so(V)$, a nondegenerate pairing $\Gamma_V\colon \sym^2(\Sigma)\rightarrow V$ and a group of $R$-symmetries $G_V$.  This datum generates a supersymmetry algebra, which we will denote by $\mf A$.  We have a natural embedding
\[\so(W)\oplus \so(W^{\perp})\subset \so(V),\]
so $\Sigma$ restricts to a spinorial $\so(W)$ representation. We define the dimensionally reduced $\Gamma$-pairing as the composite
\[\Gamma_W\colon \sym^2(\Sigma)\xrightarrow{\Gamma_V} V\rightarrow W,\]
where the last map is the orthogonal projection onto $W$. Finally, we have a new $R$-symmetry group
\[G_W = G_V\times \Spin(W^\perp_\RR).\]
This datum generates a supersymmetry algebra $\mf A'$ in dimension $\mr{dim}(W_\RR)$ as defined in Section \ref{sect:susyalgebras}.

Recall from Proposition \ref{prop:dimensionalreductionunique} that the dimensional reduction of a classical field theory along the projection $p\colon V_\RR\rightarrow W_\RR$ exists and is unique. We have the following generalization of this statement to supersymmetric theories.

\begin{prop} \label{prop:susydimlred}
Suppose $(E, \omega, S)$ is an $\mf A$-supersymmetric classical field theory on $V_\RR$. Then its dimensional reduction along the projection $p\colon V_\RR\rightarrow W_\RR$ has a unique $\mf A'$-supersymmetric structure, compatible with the supersymmetry on $V_\RR$ in the sense that $p^*S^{(i)}_{V_\RR} = S^{(i)}_{W_\RR}$.
\end{prop}

\begin{proof}
This follows from the proof of Proposition \ref{prop:dimensionalreductionunique} by coupling the theory $(E, \omega, S)$ to auxiliary fields generated by the representation $\Sigma$.
\end{proof}

The following proposition is an immediate consequence of Proposition \ref{prop:susydimlred} and Definition \ref{def:twisting}.

\begin{prop}
Fix a square-zero supercharge $Q$ and a compatible homomorphism $\alpha\colon \U(1)\rightarrow G_R$. Then the dimensional reduction of the twist of the classical field theory $E$ is isomorphic to the twist of the dimensional reduction of $E$.
\label{prop:twistdimensionalreduction}
\end{prop}

\section{Supersymmetric Yang--Mills Theories} \label{sect:SYM}

In this section we construct supersymmetry algebra action on super Yang--Mills theories. We have the following versions of super Yang--Mills theory depending on $\dim(\Sigma)$:
\begin{itemize}
\item (16 supercharges). This theory exists in dimensions 2 through 10 and depends on a Lie algebra $\fg$.

\item (8 supercharges). This theory exists in dimensions 2 through 6 and depends on a Lie algebra $\fg$ together with a symplectic $\fg$-representation $U$.

\item (4 supercharges). This theory exists in dimensions 2 through 4 and depends on a Lie algebra $\fg$ together with a $\fg$-representation $R$.

\item (2 supercharges). This theory exists in dimensions 2 through 3 and depends on a Lie algebra $\fg$ together with an orthogonal $\fg$-representation $P$.
\end{itemize}

There are a few additional possibilities that occur in dimension 2.

\begin{itemize}
\item ($\cN_+$ supercharges, chiral supersymmetry). This theory exists in dimension 2 and depends on a Lie algebra $\fg$.

\item (4 supercharges, chiral supersymmetry). This theory exists in dimension 2 and depends on a Lie algebra $\fg$ together with a symplectic $\fg$-representation $U$.

\item (2 supercharges, chiral supersymmetry). This theory exists in dimension 2 and depends on a Lie algebra $\fg$ together with a $\fg$-representation $R$.

\item (1 supercharge, chiral supersymmetry). This theory exists in dimension 2 and depends on a Lie algebra $\fg$ together with an orthogonal $\fg$-representation $P$.
\end{itemize}

In each case the lower-dimensional theories are obtained by dimensional reduction from the theory in the highest dimension: for instance, 7d $\mc N=1$ super Yang--Mills (16 supercharges) is obtained by dimensional reduction from 10d $\mc N=(1, 0)$ super Yang--Mills. So, it will be enough to construct the supersymmetry action in these highest-dimensional theories.

\subsection{Super Yang--Mills Theory: Pure Gauge Theory}
\label{sect:gaugemultipletSUSY}

We begin with a description of certain pure supersymmetric Yang--Mills theories. Let $V_\RR = \RR^n$ be a real vector space of dimension $n$ equipped with a nondegenerate symmetric bilinear pairing and let $V$ be its complexification. 
Fix a $\ZZ/2\ZZ$-graded Clifford module $\Sigma\oplus \Sigma^*\rightarrow \CC$ with the associated $\Gamma$-pairings
\[\Gamma\colon \Sym^2(\Sigma)\rightarrow V,\qquad \Gamma\colon \Sym^2(\Sigma^*)\rightarrow V\]
defined as in Section \ref{sect:spinors}. We make the following assumption on this setup.

\begin{assumption}
For $Q_1, Q_2, Q_3\in\Sigma$ we have
\[\rho(\Gamma(Q_1, Q_2))Q_3 + \rho(\Gamma(Q_2, Q_3))Q_1 + \rho(\Gamma(Q_3, Q_1))Q_2 = 0.\]
\label{assumption:3psi}
\end{assumption}

\begin{itemize}
\item (\textbf{2d $\cN=(\cN_+, 0)$ supersymmetry}) We have $\dim(V) = 2$ and $\Sigma = S_+\otimes W$ for some complex vector space $W$ equipped with a nondegenerate symmetric bilinear pairing. Assumption \ref{assumption:3psi} is satisfied by Theorem \ref{thm:2d3psi}.

\item (\textbf{3d $\cN=1$ supersymmetry}) We have $\dim(V) = 3$ and $\Sigma = S$. Assumption \ref{assumption:3psi} is satisfied by Theorem \ref{thm:3psi}.

\item (\textbf{4d $\cN=1$ supersymmetry}) We have $\dim(V) = 4$ and $\Sigma = S_+\oplus S_-$. Assumption \ref{assumption:3psi} is satisfied by Theorem \ref{thm:3psi}.

\item (\textbf{6d $\cN=(1, 0)$ supersymmetry}) We have $\dim(V) = 6$ and $\Sigma = S_+\otimes W_+$ for a two-dimensional complex symplectic vector space $W_+$. Assumption \ref{assumption:3psi} is satisfied by Theorem \ref{thm:3psi}.

\item (\textbf{10d $\cN=(1, 0)$ supersymmetry}) We have $\dim(V) = 10$ and $\Sigma = S_+$. Assumption \ref{assumption:3psi} is satisfied by Theorem \ref{thm:3psi}.
\end{itemize}

Let $\fg$ be a Lie algebra equipped with a nondegenerate symmetric bilinear pairing. The BRST fields of the Yang--Mills theory are as follows:
\begin{itemize}
\item A connection $A \in \Omega^1(V_\RR; \fg)$ on the trivial bundle.

\item A spinor $\lambda \in \Gamma(V_\RR; \Pi \Sigma \otimes \fg)$.

\item A ghost field $c\in\Omega^0(V_\RR; \fg[1])$.
\end{itemize}

Denote by $F_A = \d A + \frac{1}{2}[A\wedge A]$ the curvature of $A$ and let $\sd\d_A$ be the twisted Dirac operator obtained from $\Gamma$ (see Section \ref{sect:spinors}).

\begin{definition}
\label{def:sym}
The BRST theory for classical supersymmetric Yang--Mills theory has underlying $\ZZ \times \ZZ/2\ZZ$-graded bundle:
\[
F_{\rm gauge} = \Omega^1(V_\RR; \fg) \oplus \Gamma(V_\RR; \Pi \Sigma \otimes \fg) \oplus \Omega^0(V_\RR; \fg[1])
\]
whose sections we denote by $(A, \lambda, c)$.  
The dg Lie structure on $F_{\rm gauge}[-1]$ has differential given by the de Rham differential $\d \colon \Omega^0(V_{\RR} ; \fg) \to \Omega^1(V_{\RR} ; \fg)$ and bracket
\begin{align*}
[-,-]  \colon  \Omega^0(V_{\RR} ; \fg) \otimes \left(\Omega^1(V_{\RR} ; \fg) \oplus \Gamma(V_\RR ; \Sigma \otimes \fg) \oplus \Omega^0(V_{\RR} ; \fg) \right) & \to  \Omega^1(V_{\RR} ; \fg) \oplus \Gamma(V_\RR ; \Sigma \otimes \fg) \oplus \Omega^0(V_\RR ; \fg)
\end{align*}
defined by $[c, A + \lambda + c'] = [c, A] + [c, \lambda] + [c,c']$.
The BRST action is defined by
\[
S_{\mr{BRST}}(A, \lambda) = \int_{V_\RR}\dvol \left( -\frac{1}{4} (F_A, F_A) + \frac{1}{2}(\lambda, \sd \d_A \lambda)\right).
\]
\end{definition}

The BV theory of supersymmetric Yang--Mills is the BV theory associated to this BRST theory. 
By definition, the fields are identified with sections of the bundle $T^*[-1] F_{\rm gauge} = F_{\rm gauge} \oplus F_{\rm gauge}^! [-1]$. 
If we denote by $(A^*, \lambda^*,c^*)$ the anti-fields, the full BV action takes the form:
\begin{equation}
S_{\gauge} = \int_{V_\RR}\dvol \left( -\frac{1}{4} (F_A, F_A) + \frac{1}{2}(\lambda, \sd\d_A \lambda) + (\d_A c, A^*) + ([c, \lambda], \lambda^*) + \frac{1}{2}([c, c], c^*)\right).
\label{eq:YMBVaction}
\end{equation}

To simplify the notation, the pairing on $\fg$ from now on will be implicit.

The Poincar\'e group acts, in the sense of Definition \ref{group_action_def}, on Yang--Mills theory on $\RR^n$. Indeed, there is an obvious Poincar\'e action on fields where we use that $\Sigma$ is a representation of $\Spin(V_\RR)$. The corresponding Hamiltonian is given by
\begin{equation}
S_{\gauge}^{(1)}(v) = \int_{V_\RR}\dvol\left( (L_{v}A, A^*) - (v.\lambda, \lambda^*) - (v.c)c^*\right),
\label{eq:YMPoincareaction}
\end{equation}
for $v\in\mf{iso}(V)$, where $v.\lambda$ contains both a derivative and the $\so(V)$ action on $\Sigma$.

We will now construct an $L_\infty$ action of the super Lie algebra $\mf{A}$ on the theory. Following Definition \ref{infinitesimal_action_def}, we have to prescribe a collection of functionals $S_{\gauge}^{(1)}, S_{\gauge}^{(2)}, \dots$, where $S_{\gauge}^{(k)}\colon \mf{A}^{\otimes k}\rightarrow \oloc(\cE)$, together satisfying the classical master equation. The supersymmetry action we construct will extend the Poincar\'{e} action from \eqref{eq:YMPoincareaction}, so we just have to specify the values of $S_{\gauge}^{(k)}$ on the supersymmetry generators in $\Sigma$. The action of supersymmetry is given by a linear and a quadratic functional
\begin{align}
S_{\gauge}^{(1)}(Q) &= \int_{V_\RR}\dvol\left( -(\Gamma(Q, \lambda), A^*) + \frac{1}{2}(\rho(F_A)Q, \lambda^*)\right) \label{eq:gaugeI1} \\
S_{\gauge}^{(2)}(Q_1, Q_2) &= \int_{V_\RR}\dvol\left( \frac{1}{4}(\Gamma(Q_1, Q_2), \Gamma(\lambda^*, \lambda^*)) - \frac{1}{2}(Q_1, \lambda^*)(Q_2, \lambda^*) - \iota_{\Gamma(Q_1, Q_2)} A c^*\right). \label{eq:gaugeI2}
\end{align}

The following theorem summarizes the fact that super Yang--Mills theory is indeed supersymmetric in the sense of Definition \ref{dfn: super}. 

\begin{theorem}\label{thm:gaugemultipletSUSY}
The functional $S_{\gauge, \fA} = S_{\gauge} + S_{\gauge}^{(1)} + S_{\gauge}^{(2)} \in C^\bu(\mf{A}, \oloc(\cE_{\gauge}))$ satisfies the classical master equation
\[\d_{\mr{CE}} \left(S_{\gauge, \fA}\right) + \frac{1}{2}\{S_{\gauge, \fA}, S_{\gauge, \fA}\} = 0.\]
\end{theorem}

This result implies that the pure gauge sector of super Yang--Mills theory carries an $L_\infty$ action by the super Lie algebra $\fA$.
We consider coupling to matter in the next section.

The rest of the section will be devoted to the proof of the above theorem. The classical master equation decomposes into the following equations:
\begin{align*}
\{S_{\gauge}, S_{\gauge}^{(1)}\} &= 0 \\
\{S_{\gauge}, S_{\gauge}^{(2)}\} + \d_{\mr{CE}} S_{\gauge}^{(1)} + \frac{1}{2}\{S_{\gauge}^{(1)}, S_{\gauge}^{(1)}\} &= 0 \\
\d_{\mr{CE}} S_{\gauge}^{(2)} + \{S_{\gauge}^{(1)}, S_{\gauge}^{(2)}\} &= 0 \\
\{S_{\gauge}^{(2)}, S_{\gauge}^{(2)}\} &= 0.
\end{align*}

Note that the last equation is automatically satisfied since $S_{\gauge}^{(2)}$ is independent of $\lambda$ and $c$. The rest of the claims will be proved in a sequence of Lemmas. To simplify the expressions, we drop the integrals from our notation.

\begin{lemma}
For each $Q\in \Sigma$, one has $\{S_{\gauge}, S_{\gauge}^{(1)}(Q)\} = 0$.
\label{lm:gaugemultiplet1}
\end{lemma}
\begin{proof}
Let us decompose $S_{\gauge} = \sum_{i=1}^5 S_{\gauge, i}$ into individual summands of Equation (\ref{eq:YMBVaction}).

The first term gives
\begin{align*}
\{S_{\gauge, 1}, S_{\gauge}^{(1)}(Q)\} &= -\frac{1}{2} (\d_A \Gamma(Q, \lambda), F_A)\\
&= -\frac{1}{2}(-1)^{n-1} (\ast \d_A \ast F_A, \Gamma(Q, \lambda)).
\end{align*}

The second term gives
\begin{align*}
\{S_{\gauge, 2}, S_{\gauge}^{(1)}(Q)\} &= -\frac{1}{2}(\lambda, \rho(\Gamma(Q, \lambda))\lambda) + \frac{1}{2}(\rho(F_A) Q, \sd{\d}_A\lambda) \\
&= -\frac{1}{2}(\Gamma(Q, \lambda), \Gamma(\lambda, \lambda))  - \frac{1}{2}(\lambda, \sd{\d}_A(\rho(F_A) Q)) \\
&= -\frac{1}{2}(\Gamma(Q, \lambda), \Gamma(\lambda, \lambda)) - \frac{1}{2} (-1)^n (\lambda, \rho(\ast \d_A\ast F_A) Q),
\end{align*}
where we have used Proposition \ref{prop:cliffordactionproperty3} and the Bianchi identity in the last line.

By \eqref{eq:Gammaspinorpairing} and Assumption \ref{assumption:3psi} we have $(\Gamma(Q, \lambda), \Gamma(\lambda, \lambda)) = 0$, so $\{S_1 + S_2, S_{\gauge}^{(1)}(Q)\} = 0$.

Finally, $\{S_{\gauge,3} + S_{\gauge,4} + S_{\gauge, 5}, S_{\gauge}^{(1)}(Q)\} = 0$ due to gauge-invariance of $S_{\gauge}^{(1)}(Q)$.
\end{proof}

\begin{remark}
The previous Lemma expresses the fact that the pure super Yang--Mills on-shell action on BRST fields is supersymmetric; this was proven by Baez and Huerta in \cite{BaezHuerta}, and our proof essentially follows the proof in loc. cit.
\end{remark}

\begin{lemma}
One has
\[\{S_{\gauge}, S_{\gauge}^{(2)}\} + \d_{\mr{CE}} S_{\gauge}^{(1)} + \frac{1}{2}\{S_{\gauge}^{(1)}, S_{\gauge}^{(1)}\} = 0.\]
\label{lm:gaugemultiplet2}
\end{lemma}
\begin{proof}
Evaluating the equation
\[\{S_{\gauge}, S_{\gauge}^{(2)}\} + \d_{\mr{CE}} S_{\gauge}^{(1)} + \frac{1}{2}\{S_{\gauge}^{(1)}, S_{\gauge}^{(1)}\} = 0\]
on $v_1, v_2\in \mf{iso}(V)$, the claim reduces to the fact that \eqref{eq:YMPoincareaction} defines a strict Lie action. Evaluating it on $v\in\mf{iso}(V)$ and $Q\in\Sigma$, the claim reduces to the fact that $S_{\gauge}^{(1)}$ is Poincar\'{e}-invariant. So, the only nontrivial check is for $Q_1,Q_2\in\Sigma$.

The individual terms are
\begin{enumerate}
\item\begin{align*}
\frac{1}{2}\{S_{\gauge}^{(1)}, S_{\gauge}^{(1)}\}(Q_1, Q_2) = &-\{S_{\gauge}^{(1)}(Q_1), S_{\gauge}^{(1)}(Q_2)\} \\
= &-\frac{1}{2}(\rho(\d_A \Gamma(Q_1, \lambda)) Q_2, \lambda^*) + \frac{1}{2}(\Gamma(Q_2, \rho(F_A)Q_1), A^*) \\
&-\frac{1}{2}(\rho(\d_A \Gamma(Q_2, \lambda)) Q_1, \lambda^*) + \frac{1}{2}(\Gamma(Q_1, \rho(F_A)Q_2), A^*),
\end{align*}
\item\begin{align*}
(\d_{\mr{CE}} S_{\gauge}^{(1)})(Q_1, Q_2) &= S_{\gauge}^{(1)}(\Gamma(Q_1, Q_2)) \\
&= (L_{\Gamma(Q_1, Q_2)}(A), A^*) -(\Gamma(Q_1, Q_2).\lambda, \lambda^*) - (\Gamma(Q_1, Q_2).c) c^*,
\end{align*}
\item 
\begin{align*}
\{S_{\gauge}, S_{\gauge}^{(2)}(Q_1, Q_2)\} = &-\frac{1}{2}(Q_2, \lambda^*)(Q_1, \sd{\d}_A \lambda + [c, \lambda^*]) - \frac{1}{2}(Q_1, \lambda^*)(Q_2, \sd{\d}_A \lambda + [c, \lambda^*]) \\
&+\frac{1}{2}(\Gamma(Q_1, Q_2), \Gamma(\lambda^*, \sd{\d}_A \lambda + [c, \lambda^*])) + \iota_{\Gamma(Q_1, Q_2)}(\d_A c) c^* - (\d_A\iota_{\Gamma(Q_1, Q_2)} A, A^*) \\
&+ ([\lambda, \iota_{\Gamma(Q_1, Q_2)}A], \lambda^*) - [\iota_{\Gamma(Q_1, Q_2)} A, c] c^*.
\end{align*}
\end{enumerate}
The total coefficient in front of $A^*$ is
\[\frac{1}{2}\Gamma(Q_1, \rho(F_A)Q_2) + \frac{1}{2}\Gamma(Q_2, \rho(F_A) Q_1) + L_{\Gamma(Q_1, Q_2)} A - \d_A \iota_{\Gamma(Q_1, Q_2)} A.\]
Using Proposition \ref{prop:cliffordactionproperty1} we get that the sum of the first two terms is $-\iota_{\Gamma(Q_1, Q_2)}F_A$ which cancels the last two terms.

The total coefficient in front of $c^*$ is
\[-\Gamma(Q_1, Q_2).c + \iota_{\Gamma(Q_1, Q_2)}(\d_A c) - [\iota_{\Gamma(Q_1, Q_2)} A, c] = 0.\]

The total coefficient in front of $\lambda^*$ is
\begin{align*}
&-\frac{1}{2}\rho(\d_A \Gamma(Q_1, \lambda))Q_2 -\frac{1}{2}\rho(\d_A \Gamma(Q_2, \lambda))Q_1 - \Gamma(Q_1, Q_2).\lambda \\
&+ \frac{1}{2}\rho(\Gamma(Q_1, Q_2))\sd{\d}_A\lambda - \frac{1}{2}(Q_2, \sd{\d}_A\lambda) Q_1 -\frac{1}{2}(Q_1, \sd{\d}_A\lambda) Q_2 + [\lambda, (\Gamma(Q_1, Q_2), A)]
\end{align*}
Using Proposition \ref{prop:cliffordactionproperty2} the first, second, fifth and sixth terms combine to
\[-\frac{1}{2} \sd{\d}_A\rho(\Gamma(Q_1, \lambda))Q_2 - \frac{1}{2} \sd{\d}_A\rho(\Gamma(Q_2, \lambda)) Q_1\]
which is equal to $\frac{1}{2} \sd{\d}_A \rho(\Gamma(Q_1, Q_2))\lambda$ by Assumption \ref{assumption:3psi}. Using the Clifford relation this term cancels the rest of the terms.
\end{proof}

Evaluating the equation
\[\d_{\mr{CE}} S_{\gauge}^{(2)} + \{S_{\gauge}^{(1)}, S_{\gauge}^{(2)}\} = 0\]
on $v_1, v_2, v_3\in\mf{iso}(V)$ or on $v_1, v_2\in\mf{iso}(V)$ and $Q\in\Sigma$ we automatically get zero. Evaluating it on $v\in\mf{iso}(V)$ and $Q_1, Q_2\in\Sigma$ we get Poincar\'{e}-invariance of $S_{\gauge}^{(2)}$.

\begin{lemma}
One has
\[\{S_{\gauge}^{(1)}, S_{\gauge}^{(2)}\}(Q_1, Q_2, Q_3) = 0\]
for every $Q_1, Q_2, Q_3\in\Sigma$.
\label{lm:gaugemultiplet3}
\end{lemma}
\begin{proof}
We have
\begin{align*}
\{S_{\gauge}^{(1)}(Q_1), S_{\gauge}^{(2)}(Q_2, Q_3)\} = &-\iota_{\Gamma(Q_2, Q_3)}\Gamma(Q_1, \lambda) c^* - \frac{1}{2} (\Gamma(Q_2, Q_3), \Gamma(\rho(A^*) Q_1, \lambda^*)) \\
&+ \frac{1}{2}(Q_2, \rho(A^*)Q_1)(Q_3, \lambda^*) + \frac{1}{2}(Q_3, \rho(A^*) Q_1)(Q_2, \lambda^*).
\end{align*}

$\{S_{\gauge}^{(1)}, S_{\gauge}^{(2)}\}(Q_1, Q_2, Q_3)$ is obtained by cyclically symmetrizing the above expression. By Assumption \ref{assumption:3psi} the cyclic symmetrization of the term with $c^*$ is zero. The Clifford relation implies that
\begin{align*}
\frac{1}{2} (\Gamma(Q_2, Q_3), \Gamma(\rho(A^*) Q_1, \lambda^*)) &= -\frac{1}{2}(\Gamma(Q_2, Q_3), \Gamma(\rho(A^*)\lambda^*, Q_1)) + (\Gamma(Q_2, Q_3), A^*) (Q_1, \lambda^*) \\
&= -\frac{1}{2}(\rho(\Gamma(Q_2, Q_3)) Q_1, \rho(A^*)\lambda^*) + (\Gamma(Q_2, Q_3), A^*) (Q_1, \lambda^*).
\end{align*}
Therefore, again using Assumption \ref{assumption:3psi} we see that the cyclic symmetrization of the terms with $A^*$ vanishes.
\end{proof}

\subsection{Coupling to Matter Multiplets}
\label{sect:mattermultipletSUSY}

In this section we describe the coupling of super Yang--Mills theory to matter valued in a $\fg$-representation $P$, i.e. the supersymmetric gauged linear $\sigma$-models. Our description of the supersymmetry of the matter multiplet is inspired by the presentation of the supersymmetric nonlinear $\sigma$-models by Deligne and Freed in \cite[Chapter 3]{DeligneFreed}.

Consider as before $V_\RR$ and a Clifford module $\Sigma\oplus \Sigma^*$ satisfying Assumption \ref{assumption:3psi}. In addition, fix a complex associative composition algebra $A$ equipped with an antiinvolution $\sigma$ as in Section \ref{sect:compositionalgebras}. Suppose $\Sigma\oplus \Sigma^*$ carries a compatible right $A$-module structure. Let $(-, -)^A\colon \Sigma\otimes \Sigma^*\rightarrow A$ be the corresponding $A$-valued pairing given by Lemma \ref{lm:extendpairing}. We make the following additional assumption.

\begin{assumption}
For $Q_1, Q_2\in\Sigma$ and $Q_3\in\Sigma^*$ we have
\[Q_1(Q_2, Q_3)^A + Q_2(Q_1, Q_3)^A = \rho(\Gamma(Q_1, Q_2))Q_3.\]
\label{assumption:matter3psi}
\end{assumption}

Explicitly, we consider the following examples of theories of matter with minimal supersymmetry. 

The first three examples concern chiral two-dimensional supersymmetry, where $\Sigma = S^{2 \rm d}_+ \otimes A$, with $S_+^{2 \rm d}$ the one-dimensional positive complex semi-spin representation of $\Spin(2;\CC)$:

\begin{itemize}
\item (\textbf{2d $\cN=(1, 0)$ supersymmetry}) $A = \CC$. Assumption \ref{assumption:matter3psi} is satisfied by Theorem \ref{thm:2dmatter3psi}.

\item (\textbf{2d $\cN=(2, 0)$ supersymmetry}) $A = \CC[x]/(x^2+1)$. Assumption \ref{assumption:matter3psi} is satisfied by Theorem \ref{thm:2dmatter3psi}.

\item (\textbf{2d $\cN=(4, 0)$ supersymmetry}) $A = \eend(Z)$ where $Z$ is a 2-dimensional symplectic vector space. Assumption \ref{assumption:matter3psi} is satisfied by Theorem \ref{thm:2dmatter3psi}.
\end{itemize}

In dimensions $\geq 3$ we have the following examples, where $\Sigma = A \oplus A$:

\begin{itemize}

\item (\textbf{3d $\cN=1$ supersymmetry}) $A = \CC$. Assumption \ref{assumption:matter3psi} is satisfied by Theorem \ref{thm:matter3psi}.

\item (\textbf{4d $\cN=1$ supersymmetry}) $A = \CC[x]/(x^2+1)$. Assumption \ref{assumption:matter3psi} is satisfied by Theorem \ref{thm:matter3psi}.

\item (\textbf{6d $\cN=(1, 0)$ supersymmetry}) $A = \eend(Z)$. Assumption \ref{assumption:matter3psi} is satisfied by Theorem \ref{thm:matter3psi}.
\end{itemize}

Let $P$ be a left $A$-module equipped with a $\CC$-valued nondegenerate symmetric bilinear pairing such that
\[(av, w) = (v, \sigma(a)w).\]
Moreover, assume $P$ carries a $\fg$-action commuting with the $A$-module structure and preserving the bilinear pairing. Explicitly, for $A=\CC, \CC\otimes_\RR\CC, \mathbb{H}\otimes_\RR\CC$ we get the following data:
\begin{itemize}
\item $A=\CC$. We are looking for a $\fg$-representation $P$ equipped with a nondegenerate symmetric bilinear pairing.

\item $A=\CC[x]/(x^2+1)$. A left $A$-module $P$ splits as $P=P_+\oplus P_-$, where $x$ acts as $\pm i$ on $P_{\pm}$. Note that with respect to the right $A$-action $x$ acts as $\mp i$ on $P_{\pm}$. So, the symmetric bilinear pairing identifies $P_+\cong P_-^*$. In other words, the data boils down to a $\fg$-representation $R$, so that $P = R\oplus R^*$.

\item $A=\eend(Z)$. A left $A$-module is necessarily of the form $P\cong Z\otimes U$, where $A$ just acts on $Z$. Compatibility of the orthogonal pairing on $P$ with the $A$-action implies that it is given by a product of the symplectic pairing on $Z$ and a symplectic pairing on $U$. So, the data boils down a symplectic $\fg$-representation $U$.
\end{itemize}

We are going to construct a theory on $V_\RR$ describing a matter multiplet valued in $P$. The BRST fields are given as follows:
\begin{itemize}
\item a scalar $\phi\in\Gamma(V_\RR; P)$;
\item a spinor $\psi\in\Gamma(V_\RR; \Pi \Sigma^*\otimes_A P)$.
\end{itemize}
As usual, we denote the antifields by $\phi^*\in\Gamma(V_\RR; \Pi P)$ and $\psi^*\in\Gamma(V_\RR; \Sigma\otimes_A P)$.

We extend the pairings on $P$ and between $\Sigma$ and $\Sigma^*$ to a pairing between $\Sigma\otimes_A P$ and $\Sigma^*\otimes_A P$ in the following way. Given $\sum_i \tilde{s}_i\otimes v_i\in \Sigma^*\otimes_A P$ and $\sum_j s_j\otimes w_j\in\Sigma\otimes_A P$, their pairing is
\begin{equation}
\sum_{i, j} \Re((v_i, w_j)^A (s_j, \tilde{s}_i)^A),
\label{eq:spinorialmatterpairing}
\end{equation}
where we extend both pairings to $A$-valued pairings using Lemma \ref{lm:extendpairing}. We may also extend the $\Gamma$-pairing to a map
\[\Gamma\colon \Sym^2(\Sigma^*\otimes_A P)\rightarrow V\]
defined by the property
\[(v, \Gamma(\psi_1, \psi_2)) = (\psi_1, \rho(v) \psi_2),\qquad v\in V,\ \psi_i\in \Sigma^*\otimes_A P.\]

The BV action for the matter multiplet is
\begin{equation}\label{matteraction}
S_{\matter} = \int_{V_\RR} \dvol \left(\frac{1}{2}  (\d_A \phi, \d_A \phi) + (\psi , \sd \d_A \psi) + 2 (\lambda\phi, \psi) + (c\psi, \psi^*) - (c\phi, \phi^*)\right),
\end{equation}
where we use the pairing \eqref{eq:spinorialmatterpairing} in the second term.

It is Poincar\'{e}-invariant with the corresponding Hamiltonian
\begin{equation}
S_{\matter}^{(1)}(v) = \int_{V_\RR}\dvol\left( (L_{v}A, \phi^*) - (v.\psi, \psi^*)\right),
\label{eq:matterPoincareaction}
\end{equation}
for $v\in\mf{iso}(V)$.

The action of supersymmetry is given by a linear and quadratic functional:
\begin{align}
S_{\matter}^{(1)} (Q) & = \int_{V_\RR}\dvol\left( ((Q, \psi), \phi^*) + \frac{1}{2}(\rho(\d_A \phi) Q, \psi^*)\right) \label{eq:matterI1} \\
S_{\matter}^{(2)} (Q_1 , Q_2) & = \frac{1}{4}\int_{V_\RR} \dvol(\Gamma(Q_1, Q_2) , \Gamma(\psi^*, \psi^*)) \label{eq:matterI2}
\end{align}
where $Q, Q_1,Q_2 \in \Sigma$.

We consider the full action of the super Yang--Mills theory
\[S_{\rm SYM} = S_{\gauge} + S_{\matter}\]
where $S_{\gauge}$ is the BV action of the pure gauge sector of super Yang--Mills theory defined in Equation \eqref{eq:YMBVaction}.

The action by supersymmetry on the full theory is encoded by the $\fA$-dependent functionals
\[S^{(1)} = S^{(1)}_{\gauge} + S^{(1)}_{\matter},\qquad S^{(2)} = S^{(2)}_{\gauge} + S^{(2)}_{\matter}\]
where $S^{(1)}_\gauge, S^{(2)}_\gauge$ are as in Equations (\ref{eq:gaugeI1}), (\ref{eq:gaugeI2}).
The following result states that these functionals encode an off-shell action of the supersymmetry algebra.

\begin{theorem}
The functional $S_{\fA} = S_{\rm SYM} + S^{(1)} + S^{(2)}$ satisfies the classical master equation
\begin{equation}
\label{CMEYM}
\d_{\mr{CE}} S_\fA + \frac{1}{2} \{S_\fA, S_\fA\} = 0 .
\end{equation}
Thus, according to Definition \ref{infinitesimal_action_def}, the functional $S_{\fA}$ defines an elliptic $L_\infty$ action of the super Lie algebra $\fA$ on super Yang--Mills theory and so super Yang--Mills theory is supersymmetric
\label{thm:YMSUSY}
\end{theorem}

The rest of this section is devoted to the proof of Theorem \ref{thm:YMSUSY}. 
Notice that when we take the matter to be valued in a trivial representation for the Lie algebra $\fg$, the result reduces to Theorem \ref{thm:gaugemultipletSUSY}.
We may therefore restrict our attention to the terms in \eqref{CMEYM} which involve fields of the matter multiplet. 
Consequently, the classical master equation \eqref{CMEYM} decomposes into the following set of equations:
\begin{equation}
\label{CMEYM2}
\begin{array}{rrrrrr}
\{S_{\rm SYM} , S^{(1)}\} & = & 0 \\ 
\{S_{\matter}, S^{(2)}\} + \d_{\mr{CE}} S_{\matter}^{(1)} + \{S_{\gauge}^{(1)}, S_{\matter}^{(1)}\} + \frac{1}{2} \{S_{\matter}^{(1)}, S_{\matter}^{(1)}\} & = & 0 \\
\d_{\mr{CE}} S_{\matter}^{(2)} + \{S_{\matter}^{(1)}, S_{\matter}^{(2)}\} & =& 0 \\
\{S_{\matter}^{(2)}, S_{\matter}^{(2)}\} & =& 0
\end{array}
\end{equation}

The last equation is automatically satisfied since $S^{(2)}$ is independent of the fields $\phi, \psi, A, \lambda$.

The first equation in \eqref{CMEYM2} states that the classical action is supersymmetric.

\begin{lemma} \label{lem:YM1}
One has $\{S_{\rm SYM}, S^{(1)}\} (Q) = 0$ for all $Q \in \Sigma$. 
\end{lemma}
\begin{proof}
Let us decompose $S_{\matter}=\sum_{i=1}^5 S_{\matter, i}$ into the individual summands in Equation \eqref{matteraction}. 

The first term gives
\begin{align*}
\{S_{\matter,1}, S^{(1)}(Q)\} &= -(\d_A\phi, \d_A(Q, \psi)) + (\Gamma(Q, \lambda)\phi, \d_A\phi) \\
&= \d_A^* \d_A\phi (Q, \psi) + (\Gamma(Q, \lambda)\phi, \d_A\phi).
\end{align*}

The second term gives
\begin{align*}
\{S_{\matter,2}, S^{(1)}(Q)\} &= -(\psi, \sd \d_A \rho(\d_A \phi) Q) - (\psi, \rho(\Gamma(Q, \lambda))\psi) \\
&= -(\psi, \rho(F_A) Q)\phi - \d_A^*\d_A \phi(\psi, Q) - (\psi, \rho(\Gamma(Q, \lambda))\psi),
\end{align*}
where we have used Proposition \ref{prop:cliffordactionproperty3} in the second line.

The third term gives
\begin{align*}
\{S_{\matter,3}, S^{(1)}(Q)\} &= ((\rho(F_A)Q) \phi, \psi) + 2(\lambda(Q, \psi), \psi) - (\lambda\phi, \rho(\d_A \phi) Q) \\
&= ((\rho(F_A)Q) \phi, \psi) + (\rho(\Gamma(Q, \lambda))\psi, \psi) - (\Gamma(\lambda\phi, Q), \d_A \phi),
\end{align*}
where we have used Assumption \ref{assumption:matter3psi} in the middle term and \eqref{eq:Gammaspinorpairing} in the last term. It is then obvious that
\[\{S_{\matter,1} + S_{\matter,2} + S_{\matter,3}, S^{(1)}(Q)\} = 0.\]

Finally, the terms $\{S_{\matter,4} + S_{\matter,5} + S_{\gauge,3}, S^{(1)}(Q)\}$ are zero due to gauge-invariance of $S^{(1)}(Q)$, while the rest of the terms are zero by Lemma \ref{lm:gaugemultiplet1}.
\end{proof}

Next, we move on to the second equation in \eqref{CMEYM2}.

\begin{lemma} 
One has
\begin{equation}\label{CMEYM3}
\{S_{\matter}, S^{(2)}\} + \d_{\mr{CE}} S_{\matter}^{(1)} + \{S_{\gauge}^{(1)}, S_{\matter}^{(1)}\} + \frac{1}{2} \{S_{\matter}^{(1)}, S_{\matter}^{(1)}\} = 0 .
\end{equation}
\end{lemma}
\begin{proof}
Evaluating expression \eqref{CMEYM3} on $v_1,v_2 \in \mf{iso}(V)$ reduces to the claim that \eqref{eq:matterPoincareaction} defines a strict Lie action. Evaluating on $v \in \mf{iso}(V)$ and $Q \in \Sigma$, the claim reduces to the fact that $S^{(1)}$ is Poincar\'{e}-invariant. So, the only nontrivial term to check is the evaluation on $Q_1,Q_2 \in \Sigma$. 

The individual terms are:
\begin{enumerate}
\item \begin{align*}
\frac{1}{2}\{S_{\matter}^{(1)}, S_{\matter}^{(1)}\}(Q_1, Q_2) =& -\{S_{\matter}^{(1)}(Q_1), S_{\matter}^{(1)}(Q_2)\} \\
=&-\frac{1}{2}(Q_1, \rho(\d_A \phi) Q_2)\phi^* - \frac{1}{2}(Q_2, \rho(\d_A \phi)Q_1)\phi^* \\ &  +\frac{1}{2}(\rho(\d(Q_1,\psi))Q_2, \psi^*) +\frac{1}{2} (\rho(\d(Q_2, \psi)) Q_1, \psi^*),
\end{align*}

\item \begin{align*}
\{S_{\gauge}^{(1)}, S_{\matter}^{(1)}\}(Q_1, Q_2) = &-\{S_{\gauge}^{(1)}(Q_1), S_{\matter}^{(1)}(Q_2)\} - \{S_{\gauge}^{(1)}(Q_2), S_{\matter}^{(1)}(Q_1)\} \\
=& -\frac{1}{2}(\rho(\Gamma(Q_1, \lambda))(\phi Q_2), \psi^*) - \frac{1}{2}(\rho(\Gamma(Q_2, \lambda))(\phi Q_1), \psi^*),
\end{align*}

\item
\[
(\d_{\mr{CE}}S_{\matter}^{(1)})(Q_1,Q_2) = L_{\Gamma(Q_1,Q_2)} (\phi)\phi^* - (\Gamma(Q_1,Q_2).\psi, \psi^*),
\]

\item
\[
\{S_{\matter}, S^{(2)}(Q_1,Q_2)\} =  \frac{1}{2} \Gamma(Q_1,Q_2) \Gamma(\psi^*, \sd \d_A \psi - 2 \lambda\phi + c\psi^*) - ((\iota_{\Gamma(Q_1, Q_2)} A)\psi, \psi^*) + ((\iota_{\Gamma(Q_1, Q_2)} A)\phi, \phi^*)
\]
\end{enumerate}

We first collect all terms in equation \eqref{CMEYM3} proportional to $\phi^*$:
\[
-\frac{1}{2} (Q_1, \rho(\d_A\phi) Q_2) - \frac{1}{2} (Q_2, \rho(\d_A \phi) Q_1) + L_{\Gamma(Q_1,Q_2)} \phi + (\iota_{\Gamma(Q_1, Q_2)} A)\phi.
\]

By \eqref{eq:Gammaspinorpairing} we observe that the first two terms cancel with the last two terms.

Next, we collect all terms in equation \eqref{CMEYM3} containing $\psi^*$ and $\psi$:
\begin{equation}\label{psistar}
\frac{1}{2} \sd \d_A Q_2(Q_1 \psi)^A + \frac{1}{2} \sd \d_A Q_1(Q_2, \partial_i\psi)^A - \Gamma(Q_1,Q_2) . \psi + \frac{1}{2} \rho(\Gamma(Q_1,Q_2)) \sd \d_A \psi - (\iota_{\Gamma(Q_1, Q_2)} A)\psi.
\end{equation}
Applying Assumption \ref{assumption:matter3psi} to $Q_3 = \psi$, the first two terms become $\frac{1}{2} \sd \d_A \rho(\Gamma(Q_1,Q_2) \psi)$. 
Finally, by the Clifford identity the sum of this term with the fourth term in \eqref{psistar} is precisely $\iota_{\Gamma(Q_1,Q_2)}\d_A \psi$ which cancels the remaining terms.
\end{proof}

\begin{lemma}
One has
\[\{S_{\matter}^{(1)}, S_{\matter}^{(2)}\}(Q_1, Q_2, Q_3) = 0\]
for every $Q_1, Q_2, Q_3\in \Sigma$.
\end{lemma}
\begin{proof}
We have
\begin{align*}
\{S_{\matter}^{(1)}(Q_1), S_{\matter}^{(2)}(Q_2, Q_3)\} & = \frac{1}{2} (\Gamma(Q_2, Q_3), \Gamma(\psi^*, \phi^* Q_1)) \\ & = (\psi^*, \phi^* \rho(\Gamma(Q_2,Q_3)) Q_1)  .
\end{align*}
The expression
$\{S_{\matter}^{(1)}, S_{\matter}^{(2)}\}(Q_1,Q_2,Q_3)$ is obtained by cyclically symmetrizing the above expression. By Assumption \ref{assumption:3psi} the cyclic symmetrization is identically zero.
\end{proof}

\part{Classification of Twists} \label{classification_part}

In the following sections we fix a complex Lie algebra $\fg$ equipped with a symmetric bilinear invariant nondegenerate pairing, which should be thought of as the complexified Lie algebra of the gauge group.

\section{Dimension 10}

The odd part of the $10$-dimensional supersymmetry algebra is 
\[
\Sigma\cong S_+\otimes W_+\oplus S_-\otimes W_-,
\] 
where $S_+, S_-$ are the 16-dimensional semi-spin representations of $\Spin(10, \CC)$, and where $W_+$ and $W_-$ are complex vector spaces equipped with nondegenerate symmetric bilinear pairings. 

There are supersymmetric Yang--Mills theories with $\cN=(1, 0)$ or $\cN=(0, 1)$ supersymmetry. 
We concentrate on the first case, the second case being identical. 
So, we fix $W_+=\CC$ and $W_- = 0$.

\subsection{\texorpdfstring{$\cN=(1, 0)$}{N=(1,0)} Super Yang--Mills Theory}

We consider $\cN=(1, 0)$ super Yang--Mills theory on $M=\RR^{10}$ with the Euclidean metric.

This theory admits a unique twist:
\begin{itemize}
\item A square-zero supercharge $Q\neq 0\in\Sigma$ has 5 invariant directions and does not admit a compatible homomorphism $\alpha$. So, it gives rise to a $\ZZ/2\ZZ$-graded holomorphic theory. Such a supercharge is stabilized by $G=\SU(5)\subset \Spin(10, \CC)$.
\end{itemize}

\subsubsection{Holomorphic Twist}
\label{sect:10dholomorphictwist}

Let $Q\in\Sigma$ be a square-zero supercharge, which we will fix for the rest of this section. 
The image of $\Gamma(Q, -)\colon \Sigma\rightarrow V$ is a complex Lagrangian subspace $L\subset V$. Denote by $\sigma\colon V\rightarrow V$ the complex conjugation induced by the real structure $V = V_\RR\otimes_\RR\CC$. Since the bilinear form on $V_\RR$ is positive-definite, $L\cap \sigma(L) = 0$. In other words, $L$ defines a (linear) complex structure on $V_\RR$. Moreover, we may canonically identify $\sigma(L)\cong L^*$.

\begin{remark}
It is important here that we are working in the complexified setting.  While it makes sense to study a real form of 10d supersymmetric Yang--Mills theory in Lorentzian signature associated to a real Lie algebra $\gg_\RR$, where the fermions are valued in the Majorana-Weyl spinor bundle, in this real supersymmetry algebra there are no square-zero supercharges. Indeed, a square-zero supercharge induces a 5-dimensional isotropic subspace of $\RR^{10}$, which only exists in the split signature $5+5$ (in which case there is no real structure for the Weyl spinor representation).  So necessarily, twists of maximal super Yang--Mills theory cannot be compatible with any unitary structure.
\end{remark}

Let $\ML(L) = \ML(5)$ be the metalinear group of $L$. Under the embedding $\ML(L)\subset \Spin(V)$, the semi-spin representation $\Sigma=S_+$ decomposes as
\[\Sigma = \det(L)^{1/2} \oplus \wedge^2 L^*\otimes \det(L)^{1/2}\oplus L\otimes \det(L)^{-1/2}.\]
$Q\in\Sigma$ lies in the first summand, so the choice of $Q$ is equivalent to the choice of a (linear) K\"ahler structure $L$ on $V_\RR$ together with a complex half-density on $L$. The square of this half-density defines a Calabi--Yau structure on $M$.

We will now rewrite the fields and the action in terms of the Calabi--Yau structure. Let $\omega\in\Omega^{1, 1}(M)$ be the K\"ahler form, $\Omega\in\Omega^{5, 0}(M)$ the holomorphic volume form and $\Lambda\colon \Omega^{p+1, q+1}(M)\rightarrow \Omega^{p, q}(M)$ the dual Lefschetz operator. We denote the real volume form on $M$ by
\[\dvol = \frac{\omega^5}{5!}.\]

The vector representation decomposes as
\[\Omega^1(M)\cong \Omega^{1, 0}(M)\oplus \Omega^{0, 1}(M),\]
the semi-spin representation $S_+$ decomposes as
\[\Omega^0(M; S_+)\cong \Omega^{1, 0}(M)\oplus \Omega^{0, 2}(M)\oplus \Omega^0(M)\]
and the semi-spin representation $S_-$ decomposes as
\[\Omega^0(M; S_-)\cong \Omega^{0, 1}(M)\oplus \Omega^{2, 0}(M)\oplus \Omega^0(M).\]
Under this decomposition the scalar pairing $S_+\otimes S_-\rightarrow \CC$ corresponds to the wedge product of individual components post-composed with $\Lambda$. Under the above decompositions the Clifford multiplication of a vector $A = A_{1, 0} + A_{0, 1}$ and a spinor $\lambda = \rho + B + \chi\in S_+$ is given by
\[\rho(A)\lambda = (A_{0, 1}\chi + \Lambda(A_{1, 0}\wedge B), A_{1, 0}\wedge \rho + \ast(A_{0, 1}\wedge B\wedge\Omega), \Lambda(A_{0, 1}\wedge \rho))\in S_-.\]

\vspace{-10pt}
\paragraph{Fields:} The BRST fields are given by:
\begin{itemize}
\item Gauge fields $A_{1, 0}\in\Omega^{1, 0}(M; \gg)$, $A_{0, 1}\in\Omega^{0, 1}(M; \gg)$.
\item Fermions $\rho\in\Omega^{1, 0}(M; \Pi\gg)$, $B\in\Omega^{0, 2}(M; \Pi\gg)$, $\chi\in\Omega^0(M; \Pi\gg)$.
\item A ghost field $c\in\Omega^0(M; \gg)[1]$.
\end{itemize}
We denote their antifields by $A_{1, 0}^*, A_{0, 1}^*, \rho^*, B^*, \chi^*, c^*$.

The BV action of the theory is obtained from \eqref{eq:YMBVaction} by decomposing it in terms of the above fields. To write it out we will need an expression for the Hodge star operator on K\"{a}hler manifolds, see \cite[Proposition 1.2.31]{Huybrechts}.

\begin{prop}
Let $(M, \omega)$ be a K\"{a}hler $d$-fold and decompose
\[\Omega^2(M) = \Omega^{2, 0}(M)\oplus \Omega^{0, 2}(M)\oplus (\CC\omega\oplus \Omega^{1, 1}_\perp(M)).\]
Then
\begin{enumerate}
\item The spaces $\Omega^{2, 0}(M)\oplus \Omega^{0, 2}(M)$, $\CC\omega$ and $\Omega^{1, 1}_\perp(M)$ are mutually orthogonal.

\item For a form $\alpha\in\Omega^{2, 0}(M)\oplus \Omega^{0, 2}(M)$ we have
\[\ast \alpha = \frac{1}{(d-2)!} \alpha\wedge \omega^{d-2}.\]

\item For $\alpha\in \Omega^{1, 1}_\perp(M)$ we have
\[\ast\alpha = -\frac{1}{(d-2)!} \alpha\wedge \omega^{d-2}.\]

\item For $\alpha\in\CC\omega$ we have
\[\ast \alpha = \frac{1}{(d-1)!} \alpha\wedge \omega^{d-2}.\]
\end{enumerate}
\end{prop}

\begin{corollary} \label{Kahler_YM_term_cor}
Let $M$ be a K\"{a}hler $d$-fold and $F = F_{2, 0} + F_{1, 1} + F_{0, 2}$ a two-form. Then
\[F\wedge \ast F + \frac{1}{(d-2)!} F\wedge F\wedge \omega^{d-2} = \left(4(F_{2, 0}, F_{0, 2}) + (\Lambda F_{1, 1})^2\right) \frac{\omega^d}{d!}.\]
\end{corollary}

Since we are working near the trivial connection, the topological term $\int F\wedge F\wedge \omega^3$ is exact, so we will drop it. The BV action of the twisted theory $S_{\mr{twist}}$ is then the sum of the following terms:
\begin{align}
S_{\mr{BRST}} &= \int \dvol \left(-(F_{2, 0}, F_{0, 2}) - \frac{1}{4}(\Lambda F_{1, 1})^2 + \chi \Lambda(\ol\dd_{A_{0, 1}}\rho) + (B, \dd_{A_{1, 0}} \rho)\right)  + \frac{1}{2}B\wedge \ol\dd_{A_{0, 1}} B\wedge\Omega \label{eq:10dBRST} \\
S_{\mr{anti}} &= \int \dvol\left((\dd_{A_{1, 0}} c, A_{1, 0}^*) + (\ol\dd_{A_{0, 1}} c, A_{0, 1}^*) + ([\rho, c], \rho^*) + [\chi, c]\chi^* + ([B, c], B^*) + \frac{1}{2}[c, c]c^*\right) \label{eq:10danti} \\
S^{(1)}(Q) &= \int \dvol\left(-(\rho, A_{1, 0}^*) + (F_{0, 2}, B^*) + \frac{1}{2}\Lambda F_{1, 1} \chi^*\right) \\
S^{(2)}(Q) &= -\frac{1}{4}\int \dvol (\chi^*)^2.
\end{align}
The first two terms comprise the action of supersymmetric Yang--Mills theory $S_{\rm SYM} = S_{\rm gauge}$ of Definition \ref{def:sym} written with respect to the $\ML(L) = \ML(5)$ decomposition. 
The last two terms arise from the twisting procedure.

\begin{theorem}
The holomorphic twist of 10d $\mc N=(1, 0)$ super Yang--Mills on $M=\RR^{10}$ is perturbatively equivalent to holomorphic Chern--Simons theory on $M\cong\CC^5$ with the space of fields $\map(M, B\gg)$. Moreover, the equivalence is $\SU(5)$-equivariant.
\label{thm:10dholomorphictwist}
\end{theorem}
\begin{proof}
First, we may eliminate $\chi$ and $\chi^*$ using Proposition \ref{prop:integrateoutfield}. So, the above theory described by $S_{\rm twist}$ is perturbatively equivalent to the theory without the fields $\chi$ and $\chi^*$ with the BV action
\begin{align*}
S_{-\chi} &= \int \dvol \left(-(F_{2, 0}, F_{0, 2}) + (B, \dd_{A_{1, 0}} \rho)\right)  + \frac{1}{2}B\wedge \ol\dd_{A_{0, 1}} B\wedge\Omega \\
&+ \int \dvol\left((\dd_{A_{1, 0}} c, A_{1, 0}^*) + (\ol\dd_{A_{0, 1}} c, A_{0, 1}^*) + ([\rho, c], \rho^*) + ([B, c], B^*) + \frac{1}{2}[c, c]c^*\right) \\
&+ \int \dvol(-(\rho, A_{1, 0}^*) + (F_{0, 2}, B^*)).
\end{align*}

Next, we have a term $\int \dvol \rho\wedge A_{1, 0}^*$ in the action, i.e. $(\rho, A_{1, 0})$ is a trivial BRST doublet, so by Proposition \ref{prop:BRSTdoublet} we may remove it. 
The above theory described by $S_{-\chi}$ is perturbatively equivalent to the theory without fields $\rho,\rho^*,A_{1,0},A_{1,0}^*$ and with the BV action
\[
S_0 = \int \frac{1}{2}B\wedge \ol\dd_{A_{0, 1}} B\wedge\Omega + \dvol\left((\ol\dd_{A_{0, 1}} c, A_{0, 1}^*) + ([B, c], B^*) + \frac{1}{2}[c, c]c^* + (F_{0, 2}, B^*)\right)
\]
Up to rescaling of the antifields by the rule
\[
\alpha^* \mapsto \Omega^{-1} e^{\omega} (\alpha^*)
\]
the fields and action coincide precisely with those of holomorphic Chern--Simons theory (see Section \ref{gen_CS_section}).
\end{proof}

\begin{remark}
A similar claim was previously proved by Baulieu \cite{Baulieu} by adding an auxiliary field to 10d $\mc N=(1, 0)$ super Yang--Mills.
\end{remark}

\section{Dimension 9} \label{9d_section}

The odd part of the $9$-dimensional supersymmetry algebra is 
\[
\Sigma\cong S\otimes W,
\] 
where $S$ is the 16-dimensional spin representation of $\Spin(9, \CC)$ and $W$ is a complex vector space equipped with a nondegenerate symmetric bilinear pairing. 

There is a supersymmetric Yang--Mills theory with $\mc N=1$ supersymmetry, so we fix $W = \CC$.

\subsection{\texorpdfstring{$\cN=1$}{N=1} Super Yang--Mills Theory}

We consider $\cN=1$ super Yang--Mills theory on $M=\RR^9$ equipped with the Euclidean metric.

This theory admits a unique twist:
\begin{itemize}
\item A square-zero supercharge $Q\neq 0\in\Sigma$ has 5 invariant directions and does not admit a compatible homomorphism $\alpha$. So, it gives rise to a $\ZZ/2\ZZ$-graded holomorphic theory. Such a supercharge is stabilized by $G=\SU(4)\subset \Spin(9, \CC)$.
\end{itemize}

We may identify the odd part of the 9d $\mc N=1$ supersymmetry algebra with the odd part of the 10d $\mc N=(1, 0)$ supersymmetry algebra. Under this identification a supercharge $Q$ squares to zero in 9d iff it squares to zero in 10d.

\subsubsection{Minimal Twist}
\label{sect:9dminimaltwist}

Let $Q\in\Sigma$ be a square-zero supercharge. Denote the image of $\Gamma(Q, -)\colon \Sigma\rightarrow V$ by $L^\perp\subset V$. Its orthogonal complement $L$ is maximal isotropic and $L^{\perp}/L$ is one-dimensional. Since the bilinear form on $V_\RR$ is positive-definite, $L\cap \sigma(L) = 0$. Moreover, $N = L^{\perp}\cap \sigma(L^{\perp})\subset V$ is a $\sigma$-stable one-dimensional subspace, we let $N_\RR$ be the $\sigma$-invariants of $N$. Therefore, we get a decomposition
\[V = L\oplus \sigma(L)\oplus N,\]
where $L^{\perp} = L\oplus N$.

Under the embedding $\ML(L)\subset \Spin(V)$ the spin representation $\Sigma=S$ decomposes as
\[\Sigma = \wedge^\bullet L\otimes \det(L)^{-1/2}\]
and the supercharge $Q$ lies in the one-dimensional subspace $\det(L)^{1/2}\subset \Sigma$. Therefore, the choice of $Q$ is equivalent to the choice of a one-dimensional subspace $N_\RR\subset V_\RR$ and a complex structure on $V_\RR/N_\RR$ together with a complex half-density.

It will be convenient to perform a computation of the twist in a slightly more general setting which will be useful for lower-dimensional computations.

Suppose $L$ is a complex vector space equipped with a Hermitian structure and a complex half-density. Suppose $N_\RR=\RR^{5-\dim(L)}$ equipped with a Euclidean metric and a spin structure. Denote by $N=N_\RR\otimes_\RR\CC$ its complexification which carries a complex half-density. Let $V_\RR = L\times N$ (a 10-dimensional real vector space). By the results of Section \ref{sect:10dholomorphictwist}, there is a canonical square-zero supercharge $Q\in\Sigma$ determined by the complex structure on $L\times N$ and a complex half-density.

The dimensional reduction of 10d super Yang--Mills on $L\times N$ along $\Re\colon N\rightarrow N_\RR$ is by definition the $(5+\dim(L))$-dimensional super Yang--Mills on $L\times N_\RR$. Since $N\cong N_\RR\oplus iN_\RR$, this theory carries an action of the $R$-symmetry group $G_R=\Spin(N_\RR)$. We consider a twisting homomorphism $\phi\colon \SU(L)\times \Spin(N_\RR)\rightarrow G_R=\Spin(N_\RR)$ given by the projection onto the second factor under which $Q$ is preserved.

\begin{theorem}
The twist of $(5+\dim(L))$-dimensional super Yang--Mills on $L\times N_\RR$ by $Q$ is perturbatively equivalent to the generalized Chern--Simons theory with the space of fields $\map(L\times (N_\RR)_{\mr{dR}}, B\fg)$. Moreover, the equivalence is $\SU(L)\times \Spin(N_\RR)$-equivariant.
\label{thm:10dCSreduction}
\end{theorem}
\begin{proof}
By Theorem \ref{thm:10dholomorphictwist} the twist of 10d $\cN=1$ super Yang--Mills on $L\times N$ by $Q$ is perturbatively equivalent to the holomorphic Chern--Simons theory. Moreover, the equivalence is $\SU(L)\times \SU(N)$-equivariant. By Proposition \ref{CS_diml_red_prop} we get that the dimensional reduction of holomorphic Chern--Simons on $L\times N$ along $\Re\colon N\rightarrow N_\RR$ is isomorphic to the generalized Chern--Simons theory with the space of fields $\map(L\times N_\RR, B\fg)$ and this isomorphism is $\SU(L)\times \SO(N_\RR)$-equivariant, where $\SO(N_\RR)$ acts on $N$ via the homomorphism \eqref{eq:SOtoSU}.

Therefore, we just need to establish that the $\Spin(N_\RR)$-action on the twisted $(5+\dim(L))$-dimensional super Yang--Mills obtained using the twisting homomorphism coincides with the $\Spin(N_\RR)$-action on the generalized Chern--Simons theory. The $\Spin(N_\RR)$-action on the fields of $(5+\dim(L))$-dimensional super Yang--Mills is obtained via the homomorphism
\[\Spin(N_\RR)\xrightarrow{\mathrm{diagonal}}\Spin(N_\RR)\times \Spin(N_\RR)\hookrightarrow \Spin(N_\RR\oplus N_\RR),\]
where the diagonal embedding comes from the identity map to the partial Lorentz group $\Spin(N_\RR)$ and the twisting homomorphism, i.e. the identity map, to the $R$-symmetry group $G_R=\Spin(N_\RR)$. The $\SO(N_\RR)$-action on the fields of the generalized Chern--Simons theory is given by the composite
\[\SO(N_\RR)\xrightarrow{\eqref{eq:SOtoSU}} \SU(N)\longrightarrow \SO(N_\RR\oplus N_\RR)\]

The claim then follows from the commutativity of the diagram
\[
\xymatrix@C=2cm{
\SO(N_\RR) \ar^-{\mathrm{diagonal}}[r] \ar^{\eqref{eq:SOtoSU}}[d] & \SO(N_\RR)\times \SO(N_\RR) \ar[d] \\
\SU(N) \ar[r] & \SO(N_\RR\oplus N_\RR).
}
\]
\end{proof}

We will now concentrate on the 9-dimensional case.

\begin{theorem}
The minimal twist of 9d $\mc N=1$ super Yang--Mills on $M=\CC^4\times \RR$ is perturbatively equivalent to the generalized Chern--Simons theory with the space of fields $\map(\CC^4\times \RR_{\mr{dR}}, B\gg)$. Moreover, the equivalence is $\SU(4)$-equivariant.
\label{thm:9dminimaltwist}
\end{theorem}
\begin{proof}
Any square-zero supercharge in the 9-dimensional supersymmetry algebra is square-zero in the 10-dimensional supersymmetry algebra. The claim follows from Theorem \ref{thm:10dCSreduction} applied to $L=\CC^4$.
\end{proof}

\section{Dimension 8}

The odd part of the $8$-dimensional supersymmetry algebra is 
\[
\Sigma\cong S_+\otimes W\oplus S_-\otimes W^*,
\] 
where $S_+, S_-$ are the $8$-dimensional semi-spin representations of $\Spin(8, \CC)$ and $W$ is a complex vector space. 
The semi-spin representations carry nondegenerate symmetric bilinear pairings $S_{\pm}\otimes S_{\pm}\rightarrow \CC$. 

There is a supersymmetric Yang--Mills theory with $\mc N=1$ supersymmetry, so we fix $W = \CC$.

\subsection{\texorpdfstring{$\cN=1$}{N=1} Super Yang--Mills Theory}

We consider $\cN=1$ super Yang--Mills theory on $M=\RR^8$ with the Euclidean metric. It admits R-symmetry group $G_R=\Spin(2; \CC)$ which acts with weight $1/2$ on $W$ and weight $-1/2$ on $W^*$.

This theory admits three twists by the following supercharges.
\begin{itemize}
\item Supercharges $(Q, 0)$ and $(0, Q)$ with $(Q, Q)_{S_\pm} = 0$. These are holomorphic. Moreover, we have an embedding $\alpha\colon \U(1)\hookrightarrow \Spin(2, \CC)$ under which they have weight 1, so they give rise to a $\ZZ$-graded holomorphic theory. Such a supercharge is stabilized by $G=\SU(4)\subset \Spin(8, \CC)$. We have a twisting homomorphism $\phi\colon \MU(4)\xrightarrow{\det^{1/2}} \U(1)\xrightarrow{\alpha} G_R$, so the twisted theory carries an action of $\MU(4)$.

\item Supercharges $(Q, 0)$ and $(0, Q)$ with $(Q, Q)_{S_\pm}\neq 0$. These are topological. As before, we may choose a compatible homomorphism $\alpha$, so they give rise to a $\ZZ$-graded topological theory. Such a supercharge is stabilized by $\Spin(7, \RR)\subset \Spin(8, \CC)$.

\item Square-zero supercharges $(Q_+, Q_-)$ where both $Q_{\pm}$ are nonzero. These have 5 invariant directions and do not admit a compatible homomorphism $\alpha$, so they give rise to a $\ZZ/2\ZZ$-graded theory. We have $(Q_\pm, Q_\pm)_{S_\pm} = 0$. The supercharge $Q_+$ is stabilized by $\SU(4)\subset \Spin(8, \RR)$. The supercharge $Q_-$ is stabilized by $\SU(3)\subset \SU(4)\subset \Spin(8, \RR)$. We have a twisting homomorphism $\phi\colon \SU(3)\times \Spin(2, \RR)\rightarrow G_R=\Spin(2, \CC)$ given by projection onto the second factor, so the twisted theory in fact carries an action of $\SU(3)\times \Spin(2, \RR)$.
\end{itemize}

\subsubsection{Holomorphic Twist}
\label{sect:8dholomorphictwist}

Suppose $Q\in S_+$ such that $(Q, Q)_{S_+}=0$. As in Section \ref{sect:10dholomorphictwist}, the data of such $Q$ is equivalent to the data of a K\"ahler structure $L$ on $V_\RR$ together with a complex half-density on $L$.

We consider the twisting homomorphism $\det^{1/2}\colon \MU(4)\rightarrow \Spin(2, \CC)$ under which $Q$ becomes scalar. Moreover, we have an embedding $\alpha\colon \U(1)\subset\Spin(2, \CC)$, so the theory is $\ZZ$-graded and carries an $\MU(4)$-action. In fact, this action will manifestly factor through $\U(4)$.

\vspace{-10pt}
\paragraph{Fields:} The BRST fields are given by:
\begin{itemize}
\item Gauge fields $A_{1, 0}\in\Omega^{1, 0}(M; \gg)$ and $A_{0, 1}\in\Omega^{0, 1}(M; \gg)$.
\item Scalar fields $a\in\Omega^{4,0}(M; \gg)[2]$ and $\wt a\in\Omega^{0, 4}(M; \gg)[-2]$.
\item Fermions $\chi\in\Omega^0(M; \gg)[-1]$, $B\in\Omega^{0, 2}(M; \gg)[-1]$, $\wt \chi\in\Omega^{0, 4}(M; \gg)[-1]$, $\rho\in\Omega^{1, 0}(M; \gg)[1]$ and $C\in\Omega^{3,0}(M; \gg)[1]$.
\item A ghost field $c\in \Omega^0(M; \gg)[1]$.
\end{itemize}

\begin{theorem}
The holomorphic twist of 8d $\cN=1$ super Yang--Mills on $M=\RR^8$ is perturbatively equivalent to the holomorphic BF theory on $M\cong\CC^4$ with the space of fields $T^*[-1]\map(M, B\fg)$. Moreover, the equivalence is $\U(4)$-equivariant.
\label{thm:8dholomorphictwist}
\end{theorem}
\begin{proof}
8d $\cN=1$ super Yang--Mills theory is obtained by dimensionally reducing 10d $\cN=1$ super Yang--Mills theory. Under dimensional reduction the 10d fields from Section \ref{sect:10dholomorphictwist} decompose as follows:
\begin{align*}
A_{1, 0}&\rightsquigarrow A_{1, 0}+ \wt a \\
A_{0, 1}&\rightsquigarrow A_{0, 1}+ a \\
\rho&\rightsquigarrow \rho+ \wt\chi \\
B&\rightsquigarrow B+ C.
\end{align*}

The claim about the underlying $\ZZ/2\ZZ$-graded theories follows by applying dimensional reduction (Proposition \ref{CS_to_BF_diml_red_prop}) to the computation of the minimal twist of 9d $\cN=1$ super Yang--Mills (Theorem \ref{thm:9dminimaltwist}). We are left to check that the equivalence respects  the gradings and the $\U(4)$-action. Indeed, the equivalence given by Theorem \ref{thm:9dminimaltwist} eliminates fields $A_{1, 0}, \wt a, \rho, \chi, \wt\chi$ and hence the underlying local $L_\infty$ algebra after the twist becomes
\[
\xymatrix@R=0.5cm@C=0.5cm{
&\Omega^0(\CC^4; \gg)_c \ar[r] &\Omega^{0,1}(\CC^4; \gg)_{A_{0,1}} \ar[r] &\Omega^{0,2}(\CC^4; \gg)_{B} \ar[r] &\Omega^{0,3}(\CC^4; \gg)_{C^*} \ar[r] &\Omega^{0,4}(\CC^4; \gg)_{a^*} \\
&&\ar@{}^{\bigoplus}[r] &&& \\
\Omega^{4,0}(\CC^4; \gg)_{a} \ar[r] &\Omega^{4,1}(\CC^4; \gg)_C \ar[r] &\Omega^{4,2}(\CC^4; \gg)_{B^*} \ar[r] &\Omega^{4,3}(\CC^4; \gg)_{A_{0,1}^*} \ar[r] &\Omega^{4,4}(\CC^4; \gg)_{c^*}&
}
\]

concentrated in cohomological degrees $-1, \dots, 4$. These fields have the same degrees as in the holomorphic BF theory.
\end{proof}

\subsubsection{Topological Twist}
\label{sect:8dtopologicaltwist}

Next we discuss the case of the topological twist. We are going to prove that it is perturbatively trivial. In fact, it will be useful to study a degeneration of the topological twist to a holomorphic twist and describe the corresponding family of twisted theories.

Let $V_\RR = \RR^8$ and $V=V_\RR\otimes_\RR\CC$. Fix a K\"ahler structure on $\RR^8$ and denote by $L\subset V$ the $i$-eigenspace of the complex structure. Moreover, fix a complex volume form on $L$. Under $\SU(L)\subset \Spin(V)$ the semi-spin representation $S_+$ decomposes as
\[S_+\cong \CC Q_0\oplus \wedge^2 L\oplus \CC \Bar{Q}_0.\]
The scalar spinorial pairing $S_+\otimes S_+\rightarrow \CC$ is given by pairing the outer terms with each other and $\wedge^2L$ with itself using the complex volume form on $L$. Consider a family of square-zero supercharges
\begin{equation}
Q_t = Q_0 + t\Bar{Q}_0\in S_+
\label{eq:8dHodgefamily}
\end{equation}
for $t\in\CC$. We have
\[(Q_t, Q_t) = t,\]
so at $t=0$ we have a holomorphic supercharge and $t\neq 0$ we have a topological supercharge.

We will use the notation for fields of 8d $\cN=1$ super Yang--Mills from Section \ref{sect:8dholomorphictwist}. Using the Calabi--Yau structure we will regard $C$ as an element of $\Omega^{0, 1}(\CC^4;\fg)[1]$, $\wt\chi$ as an element of $\Omega^0(\CC^4;\fg)[-1]$, and $a$ as an element of $\Omega^0(\CC^4 ; \fg)$. First, we are going to write the functionals \eqref{eq:gaugeI1} and \eqref{eq:gaugeI2} in terms of these fields.

\begin{prop}
The functionals $S^{(1)}$ and $S^{(2)}$ (see \eqref{eq:gaugeI1} and \eqref{eq:gaugeI2}) in terms of the fields of 8d $\cN=1$ super Yang--Mills are
\begin{align*}
S^{(1)}(Q_t) &= \int \dvol\left(-(\rho, A_{1, 0}^*) - t(C, A_{0, 1}^*) - (\wt \chi + t\chi)\wt a^* + (F_{0, 2}, B^*)  + (\ol\dd_{A_{0, 1}} a, C^*)\right) \\
&+ \int\dvol\left((t\dd_{A_{1, 0}} \Tilde{a}, \rho^*) + \frac{1}{2}\Lambda F_{1, 1}(\chi^* - t\wt\chi^*) + \frac{1}{2}[a, \wt a](\chi^*+t\wt\chi^*)\right) + t\Omega^{-1} F_{2, 0}\wedge B^* \\
S^{(2)}(Q_t) &= \int \dvol \left(t\chi^*\wt \chi^* -\frac{1}{4}(\chi^* + t\wt\chi^*)^2 + tac^*\right) + \frac{t}{2}\Omega^{-1}B^*\wedge B^* .
\end{align*}
\end{prop}

The action of the twisted 8d super Yang--Mills theory is given by
\[S_{Q_t} = S_{\rm SYM} + S^{(1)}(Q_t) + S^{(2)}(Q_t),\]
where $S_{\mathrm{BRST}}$ and $S_{\mathrm{anti}}$ are given by \eqref{eq:10dBRST} and \eqref{eq:10danti} respectively.

We have a homomorphism $\alpha\colon \U(1)\rightarrow G_R=\Spin(2, \RR)$ with respect to which $Q_t$ has weight $1$, so the $Q_t$-twisted theory will have a $\ZZ$-grading.

\begin{theorem}
The twist of 8d $\cN=1$ super Yang--Mills with respect to the family $Q_t$ of square-zero supercharges is perturbatively equivalent to the holomorphic Hodge theory $\map(\CC^4, B\fg_{\mr{Hod}})$. Moreover, this equivalence is $\SU(4)$-equivariant.
\label{thm:8dHodgetwist}
\end{theorem}
\begin{proof}
The proof proceeds as in the proof of Theorem \ref{thm:10dholomorphictwist} with slight modifications.

Observe that the quadruple of fields $\{\chi^*, \chi, \wt\chi^*, \wt\chi\}$ has the same Poisson brackets as the quadruple $\{\chi^*-t\wt\chi^*, \chi, \wt\chi^*, \wt\chi + t\chi\}$. Therefore, we may eliminate the fields $\chi^*-t\wt\chi^*, \chi$ using Proposition \ref{prop:integrateoutfield}. We then have trivial BRST doublets $\{\wt\chi + t\chi, \wt a\}$ and $\{\rho, A_{1, 0}\}$ which may be eliminated using Proposition \ref{prop:BRSTdoublet}. We are left with the action
\[S_{\mr{BF}} + \int\dvol\left(-t(C, A_{0, 1}^*) + ta c^*\right) + \frac{t}{2} \Omega^{-1} B^*\wedge B^*,\]
where $S_{\mr{BF}}$ is the action functional of the holomorphic twist at $t=0$. Since the extra terms are quadratic in the fields, the claim is reduced to a comparison of the underlying local $L_\infty$ algebra of the twisted theory and that of the holomorphic Hodge theory. The former is given by (cf. the proof of Theorem \ref{thm:8dholomorphictwist})
\[
\xymatrix@C=0.5cm{
&\Omega^0(\CC^4; \gg)_c \ar[r] &\Omega^{0,1}(\CC^4; \gg)_{A_{0,1}} \ar[r] &\Omega^{0,2}(\CC^4; \gg)_{B} \ar[r] &\Omega^{0,3}(\CC^4; \gg)_{C^*} \ar[r] &\Omega^{0,4}(\CC^4; \gg)_{a^*} \\
\Omega^0(\CC^4; \gg)_a \ar[r] \ar@{-->}^{t\id}[ur] & \Omega^{0,1}(\CC^4; \gg)_C \ar[r] \ar@{-->}^{t\id}[ur] &\Omega^{0,2}(\CC^4; \gg)_{B^*} \ar[r] \ar@{-->}^{t\id}[ur] &\Omega^{0,3}(\CC^4; \gg)_{A_{0,1}^*} \ar[r] \ar@{-->}^{t\id}[ur] &\Omega^{0,4}(\CC^4; \gg)_{c^*} \ar@{-->}^{t\id}[ur] &
}
\]
which is exactly the local $L_\infty$ algebra of the holomorphic Hodge theory.
\end{proof}

\begin{corollary}
The topological twist of 8d $\cN=1$ super Yang--Mills is perturbatively trivial.
\label{cor:8dtopologicaltwist}
\end{corollary}
\begin{proof}
The topological twist of 8d $\cN=1$ super Yang--Mills is the twist by $Q_t$ with $t\neq 0$. By Theorem \ref{thm:8dHodgetwist} it is equivalent to the $t\neq 0$ specialization of the holomorphic Hodge theory which by Proposition \ref{prop:Hodgetheoryspecialization} is perturbatively trivial.
\end{proof}

\subsubsection{Partially Topological Twist}
\label{sect:8dpartiallytopologicaltwist}

Finally we discuss the case of the partially topological supercharge $(Q_+, Q_-)\in\Sigma$. We consider the twisting homomorphism $\phi\colon \SU(3)\times \Spin(2, \RR)\rightarrow G_R=\Spin(2, \CC)$ given by projection on the second factor.

\begin{theorem}
The partially topological twist of 8d $\cN=1$ super Yang--Mills is perturbatively equivalent to the generalized Chern--Simons theory with the space of fields $\map(\CC^3\times \RR^2_{\mr{dR}}, B\fg)$. Moreover, the equivalence is $\SU(3)\times \Spin(2, \RR)$-equivariant.
\label{thm:8dpartiallytopologicaltwist}
\end{theorem}
\begin{proof}
Since $Q_+$ and $Q_-$ satisfy $(Q_\pm, Q_\pm)_{S_\pm} = 0$, they lift to a square-zero supercharge in the 10-dimensional supersymmetry algebra. The claim follows from Theorem \ref{thm:10dCSreduction} applied to $L=\CC^3$.
\end{proof}

\section{Dimension 7} \label{7d_section}

The odd part of the $7$-dimensional supersymmetry algebra is
\[
\Sigma\cong S\otimes W
\]
where $S$ is the 8-dimensional spin representation of $\Spin(7, \CC)$, and $W$ is a complex symplectic vector space.  The spin representation carries a nondegenerate symmetric bilinear pairing $S\otimes S\rightarrow \CC$. 

There is a supersymmetric Yang--Mills theory with $\mc N=1$ supersymmetry, so we fix $W = \CC^2$.

\subsection{\texorpdfstring{$\cN=1$}{N=1} Super Yang--Mills Theory}

We consider $\cN=1$ super Yang--Mills theory on $M=\RR^7$ with the Euclidean metric. It admits $R$-symmetry group $G_R=\Spin(3, \CC)$ acting on $W$ by its two-dimensional spin representation.

This theory admits three twists, by the following classes of supercharge.
\begin{itemize}
\item Rank 1 supercharges $Q=\alpha\otimes w\in S\otimes W$, where $(\alpha, \alpha)_S = 0$. These are minimal, i.e. the number of invariant directions is 4. We have a homomorphism $\alpha\colon \U(1)\rightarrow G_R=\Spin(3, \CC)$ under which they have weight $1$. We also have a twisting homomorphism $\phi\colon \MU(3)\xrightarrow{\det^{1/2}} \U(1)\rightarrow \Spin(3, \CC)$, so the twisted theory is $\ZZ$-graded and carries an action of $\MU(3)$.

\item Rank 1 supercharges $Q = \alpha \otimes w \in S \otimes W$, where $(\alpha, \alpha)_S\neq 0$. These are topological and stabilized by $G_2\subset \Spin(7, \CC)$. We have a homomorphism $\alpha\colon \U(1)\rightarrow G_R$ under which they have weight $1$.

\item Square-zero supercharges $Q$ of rank 2. These have 5 invariant directions and do not admit a compatible homomorphism $\alpha$, so they give rise to a $\ZZ/2\ZZ$-graded theory. We have a twisting homomorphism $\phi\colon \SU(2)\times \Spin(3, \RR)\rightarrow G_R=\Spin(3, \CC)$ given by projection onto the second factor, so the theory carries an action of $\SU(2)\times\Spin(3, \RR)$.
\end{itemize}

\subsubsection{Minimal Twist}
\label{sect:7dminimaltwist}

Denote the image of $\Gamma(Q, -)\colon \Sigma\rightarrow V$ by $L^\perp$, so that its orthogonal complement $L\subset V$ is a 3-dimensional isotropic subspace. As in Section \ref{sect:9dminimaltwist}, the data of a partially topological supercharge is equivalent to the choice of a one-dimensional subspace $N_\RR\subset V_\RR$ and a complex structure on $V_\RR/N_\RR$ together with a half-density.

It will be convenient to perform a computation of the twist in a slightly more general setting which will be useful for lower-dimensional computations, similarly to Theorem \ref{thm:10dCSreduction}.

Suppose $L$ is a complex vector space equipped with a Hermitian structure. Suppose $N_\RR=\RR^{4-\dim(L)}$ equipped with a Euclidean metric and a spin structure. Denote by $N=N_\RR\otimes_\RR\CC$ its complexification, which carries a complex half-density. Let $V_\RR = L\times N$ (an 8-dimensional real vector space). We set $W^8=\det(L)^{-1/2}$ as a representation of $\MU(L)$. By the results of Section \ref{sect:8dholomorphictwist}, there is a canonical square-zero supercharge $Q\in\Sigma$ determined by the complex structure on $L\times N$ and a complex half-density on $N$.

The dimensional reduction of 8d $\cN=1$ super Yang--Mills on $L\times N$ along $\Re\colon N\rightarrow N_\RR$ is by definition the $(4+\dim(L))$-dimensional super Yang--Mills on $L\times N_\RR$ with 16 supercharges. Since $N\cong N_\RR\oplus iN_\RR$, this theory carries an action of the $R$-symmetry group $G_R=\Spin(2, \CC)\times \Spin(N_\RR)$. We consider the $\ZZ$-grading $\alpha\colon \U(1)\rightarrow G_R$ given by the embedding into the first copy of $\U(1)$ and a twisting homomorphism
\[\phi\colon \MU(L)\times \Spin(N_\RR)\xrightarrow{\det^{1/2}\times \id} G_R=\Spin(2, \RR)\times \Spin(N_\RR).\]

\begin{theorem}
The twist of $(4+\dim(L))$-dimensional super Yang--Mills on $L\times N_\RR$ by $Q$ is perturbatively equivalent to the generalized BF theory with the space of fields $T^*[-1] \map(L\times (N_\RR)_{\mr{dR}}, B\fg)$. Moreover, the equivalence is $\MU(L)\times \Spin(N_\RR)$-equivariant.
\label{thm:8dBFreduction}
\end{theorem}
\begin{proof}
By Theorem \ref{thm:8dholomorphictwist} the twist of 8d $\cN=1$ super Yang--Mills on $L\times N$ by $Q$ is perturbatively equivalent to the holomorphic BF theory. Moreover, the equivalence is $\MU(L)\times \SU(N)$-equivariant. By Proposition \ref{prop:BFholomorphicreduction} we get that the dimensional reduction of holomorphic BF theory on $L\times N$ along $\Re\colon N\rightarrow N_\RR$ is isomorphic to the generalized BF theory with the space of fields $T^*[-1]\map(L\times N_\RR, B\fg)$ and this isomorphism is $\MU(L)\times \SO(N_\RR)$-equivariant, where $\SO(N_\RR)$ acts on $N$ via the homomorphism \eqref{eq:SOtoSU}.
\end{proof}

The 7-dimensional result immediately follows.

\begin{theorem}
The minimal twist of 7d $\cN=1$ super Yang--Mills on $M=\CC^3\times \RR$ is perturbatively equivalent to the generalized BF theory with the space of fields $T^*[-1]\map(\CC^3\times \RR_{\mr{dR}}, B\fg)$. Moreover, the equivalence is $\U(3)$-equivariant.
\label{thm:7dminimaltwist}
\end{theorem}

\subsubsection{Topological Twist}
\label{sect:7dtopologicaltwist}

Next we study the topological twist. As in the case of the minimal twist, we will perform a computation applicable in lower dimensions as well.

Let $L$ be a complex vector space equipped with a Hermitian structure and a complex half-density. Suppose $N_\RR = \RR^{4-\dim(L)}$ equipped with a Euclidean metric and a spin structure. Denote by $N=N_\RR\otimes_\RR\CC$ its complexification which carries a complex half-density. Let $V_\RR=L\times N$, a real 8-dimensional vector space equipped with a complex structure and a complex half-density. Using results of Section \ref{sect:8dtopologicaltwist} we obtain a family $Q_t$ of 8d square-zero supercharges.

The dimensional reduction of 8d $\cN=1$ super Yang--Mills on $L\times N$ along $\Re\colon N\rightarrow N_\RR$ gives $(4+\dim(L))$-dimensional super Yang--Mills on $L\times N_\RR$ with 16 supercharges. The $R$-symmetry group is $G_R=\Spin(N_\RR\oplus \RR^2)$ and we consider the grading given by the homomorphism $\alpha\colon\U(1)\hookrightarrow \Spin(2, \RR)\times \Spin(N_\RR)\subset G_R$ given by embedding into the first factor. We consider the twisting homomorphism given by $\Spin(N_\RR)\hookrightarrow \Spin(2, \RR)\times \Spin(N_\RR)\subset G_R$ given by embedding into the second factor.

\begin{theorem}
The twist of $(4+\dim(L))$-dimensional super Yang--Mills on $L\times N_\RR$ by $Q_t$ is perturbatively equivalent to the generalized Hodge theory with the space of fields $\map(L\times (N_\RR)_{\mr{dR}}, B\fg_\Hod)$. Moreover, the equivalence is $\SU(L)\times \Spin(N_\RR)$-equivariant.
\label{thm:8dHodgereduction}
\end{theorem}
\begin{proof}
By Theorem \ref{thm:8dHodgetwist} the twist of 8d $\cN=1$ super Yang--Mills on $L\times N$ by $Q_t$ is perturbatively equivalent to the holomorphic Hodge theory. Moreover, the equivalence is $\SU(L)\times \SU(N)$-equivariant. The claim then follows from Corollary \ref{cor:Hodgeholomorphicreduction}.
\end{proof}

Now let $L=\CC^3$ and $N_\RR=\RR$. Dimensionally reducing the family $Q_t$ along $L\times N\rightarrow L\times N_\RR$ we obtain a family of supercharges which are topological for $t\neq 0$ and has 4 invariant directions for $t=0$. In other words, at $t=0$ we get the minimal twist and at $t\neq 0$ the topological twist.

\begin{theorem}
The twist of 7d $\cN=1$ super Yang--Mills with respect to the family $Q_t$ of square-zero supercharges is perturbatively equivalent to the generalized Hodge theory $\map(\CC^3\times \RR_{\mr{dR}}, B\fg_{\Hod})$. Moreover, this equivalence is $\SU(3)$-equivariant.
\label{thm:7dHodgetwist}
\end{theorem}

\begin{corollary}
The topological twist of 7d $\cN=1$ super Yang--Mills is perturbatively trivial.
\end{corollary}

\subsubsection{Partially Topological Twist}
\label{sect:7dpartialtwist}

Finally, we discuss the case of a partially topological twist. We consider the twisting homomorphism $\phi\colon \SU(3)\times \Spin(3, \RR)\rightarrow G_R=\Spin(3, \CC)$ given by projection on the second factor.

\begin{theorem}
The partially topological twist of 7d $\cN=1$ super Yang--Mills is perturbatively equivalent to the generalized Chern--Simons theory with the space of fields $\map(\CC^2\times \RR^3_{\mr{dR}}, B\fg)$. Moreover, the equivalence is $\SU(2)\times \Spin(3, \RR)$-equivariant.
\label{thm:7dpartiallytopologicaltwist}
\end{theorem}
\begin{proof}
Any partially topological supercharge in the 7-dimensional supersymmetry algebra lifts to a square-zero supercharge in the 10-dimensional supersymmetry algebra. The claim follows from Theorem \ref{thm:10dCSreduction} applied to $L=\CC^2$.
\end{proof}

\section{Dimension 6}

The odd part of the $6$-dimensional supersymmetry algebra is
\[
\Sigma\cong S_+\otimes W_+\oplus S_-\otimes W_-
\] 
where $S_+, S_-$ are the 4-dimensional semi-spin representations of $\Spin(6 , \CC)\cong \SL(4, \CC)$, and $W_+, W_-$ are complex symplectic vector spaces. 
We have $\Spin(6;\CC)$-invariant isomorphisms $S_+\cong S_-^*$.

There are Yang--Mills theories with $\cN=(1, 0)$ and $\cN=(1, 1)$ supersymmetry, which we consider separately.

\subsection{\texorpdfstring{$\cN=(1, 0)$}{N=(1, 0)} Super Yang--Mills Theory}

The general setup for $\cN=(1, 0)$ super Yang--Mills is described in Section \ref{sect:SYM} which we now recall. Let $U$ be a complex symplectic $\fg$-representation. We consider $\cN=(1, 0)$ super Yang--Mills theory on $M=\RR^6$ with the Euclidean metric. We fix $W_- = 0$ and $W_+=\CC^2$ equipped with a symplectic structure. The $R$-symmetry group depends on the type of the representation $U$:
\begin{itemize}
\item In general, the theory admits an $R$-symmetry group $G_R=\SL(2, \CC)$ with $W_+$ the two-dimensional defining representation.

\item If $U = T^* R=R\oplus R^*$ for a $\fg$-representation $R$, then the theory admits an $R$-symmetry group $G_R=\SL(2, \CC)\times \GL(1, \CC)$, where $\GL(1, \CC)$ acts trivially on $W_+$, with weight $1$ on $R$ and with weight $-1$ on $R^*$.
\end{itemize}

This theory admits a unique twist:
\begin{itemize}
\item A square-zero supercharge $Q\neq 0\in\Sigma$ has 3 invariant directions, so it gives rise to a holomorphic theory. If the representation $U$ is of cotangent type, we have a compatible homomorphism $\alpha\colon \U(1)\rightarrow G_R=\SL(2, \CC)\times \GL(1, \CC)$ given by a diagonal embedding, so in this case we get a $\ZZ$-grading. Such a supercharge is stabilized by $G=\SU(3)\subset \Spin(6, \CC)$. We have a twisting homomorphism $\phi\colon\MU(3)\xrightarrow{\det^{1/2}}\U(1)\hookrightarrow \SL(2, \CC)$, so the twisted theory carries an $\MU(3)$-action.
\end{itemize}

\subsubsection{Holomorphic Twist}
\label{sect:6dholomorphictwist}

Consider a nonzero $Q\in S_+\otimes W_+$ which we fix for the remainder of this section. Since $\wedge^2(S_+)\cong V$, the square-zero condition is equivalent to the condition that $Q$ has rank 1, i.e. $Q=q_+\otimes w_1\in S_+\otimes W_+$. We will also fix $w_2\in W_+$ such that $(w_1, w_2) = 1$.

As in Section \ref{sect:10dholomorphictwist}, the data of $q_+$ is equivalent to the data of a K\"ahler structure $L$ on $V_\RR$ together with a complex half-density on $L$.

Under the embedding $\MU(L)\subset \Spin(V_\RR)$, the semi-spin representations $S_+, S_-$ decompose as
\[S_+ = \det(L)^{1/2}\oplus L\otimes \det(L)^{-1/2},\qquad S_- = \det(L)^{-1/2} \oplus L^*\otimes \det(L)^{1/2},\]
where $q_+\in S_+$ lies in the first summand.

We fix an embedding $\U(1)\subset \SL(2, \CC)$ under which $w_1\in W_+$ has weight $1$. Under the composite
\[\phi\colon\MU(3)\xrightarrow{\det^{1/2}}\U(1)\subset \SL(2, \CC)\]
we obtain that $W_+\cong \det(L)^{-1/2} w_1\oplus \det(L)^{1/2} w_2$.

We will now rewrite the fields using the twisting homomorphism $\phi$ from $\MU(3)$, where we denote by $K$ the canonical bundle of $L=\CC^3$.

\vspace{-10pt}
\paragraph{Fields:} The BRST fields are given by:
\begin{itemize}
\item Gauge fields $A_{1, 0} \in \Omega^{1, 0}(M; \fg)$, $A_{0, 1}\in\Omega^{0, 1}(M; \fg)$.
\item Gauge fermions $\chi\in\Omega^0(M; \Pi\fg)$, $\xi\in\Omega^{3, 0}(M; \Pi\fg)$, $B\in\Omega^{0, 2}(M; \Pi\fg)$, $\rho\in\Omega^{1, 0}(M; \Pi\fg)$.
\item Matter bosons $\nu\in\Omega^0(M; U\otimes K^{-1/2})$, $\phi\in\Omega^0(M; U\otimes K^{1/2})$.
\item Matter fermions $\psi\in \Omega^{0, 1}(M; \Pi U\otimes K^{1/2})$, $\wt\nu\in\Omega^0(M; \Pi U\otimes K^{-1/2})$.
\item A ghost field $c\in \Omega^0(M; \gg)[1]$.
\end{itemize}

Let $\omega\in\Omega^{1, 1}(M)$ be the K\"ahler form. We denote the real volume form on $M$ by
\[\dvol = \frac{\omega^3}{6}.\]

Using Corollary \ref{Kahler_YM_term_cor}, the BV action $S_{\mr{twist}}$ of the $Q$-twisted theory consists of the sum of the following terms:

\begin{align*}
S_{\gauge} &= \int \dvol \left(-(F_{2, 0}, F_{0, 2}) - \frac{1}{4}(\Lambda F_{1, 1})^2\right) + \\
&\qquad + \frac 12 \left( \omega(B \wedge \dd_{A_{1,0}} \rho) + \omega^2\chi\Lambda(\ol \dd_{A_{0,1}} \rho) - \omega(\rho \wedge \dd_{A_{1,0}} B) + \xi \ol \dd_{A_{0,1}} B \right)  \\
S_{\matter} &= \int \bigg( \dvol ((\dd_{A_{1,0}}\nu, \ol \dd_{A_{0,1}} \phi) + (\dd_{A_{1,0}}\phi, \ol \dd_{A_{0,1}} \nu)) + 2\omega^2\wedge(\wt\nu \ol \dd_{A_{1,0}} \psi) + \psi \wedge \ol \dd_{A_{0,1}} \psi + \\
&\qquad + 2\dvol(([\xi,\nu],\wt\nu) + ([\chi, \phi], \wt\nu) )   \bigg) \\
S_{\mr{anti}} &= \int \dvol \bigg( (\dd_{A_{1, 0}} c, A_{1, 0}^*) + (\ol\dd_{A_{0, 1}} c, A_{0, 1}^*)  +  [\xi, c] \xi^* + [\chi, c]\chi^* +[\rho, c] \rho^* + \\
&\qquad + \frac{1}{2}[c, c]c^* + [\nu, c] \nu^* + [\phi, c] \phi^* + [\psi, c] \psi^* + [\wt\nu, c] \wt\nu^* \bigg) + [B, c]\wedge B^* \\
S^{(1)}_{\gauge} (Q) &=  \int \dvol \left( - (\rho, A_{1,0}^*) + \frac 12 (F_{0,2},B^*) + \frac 12 (F_{1,1}, \chi^*)  \right)\\
S^{(1)}_{\matter} (Q) &=  \int \dvol \left( (\wt\nu, \nu^*) + \frac 12( \ol \dd_{A_{0,1}} \phi, \psi^*) \right)\\
S^{(2)}_{\gauge} (Q) &= -\frac{1}{4}\int \dvol (\chi^*)^2.
\end{align*}

\begin{theorem}
The holomorphic twist of 6d $\cN=(1, 0)$ super Yang--Mills on $M=\RR^6$ with matter valued in a symplectic $\fg$-representation $U$ is perturbatively equivalent to the holomorphic Chern--Simons theory on $M\cong \CC^3$ with the space of fields $\Sect(M, (U\otimes K_M^{1/2})\ham \fg)$. Moreover, the equivalence is $\MU(3)$-equivariant.
\label{thm:6dholomorphictwist}
\end{theorem}
\begin{proof}
The proof of this theorem is very similar to the proof of Theorem \ref{thm:10dholomorphictwist}. First, we eliminate the fields $\chi$ and $\chi^*$ using Proposition \ref{prop:integrateoutfield}.  We then observe that the action includes the terms $\int \dvol  (\rho, A_{1,0}^*)$ and $\int \dvol (\wt\nu, \nu^*)$.  In other words, the pairs $(\rho, A_{1,0})$ and $(\nu, \wt\nu)$ form trivial BRST doublets, which can be eliminated using Proposition \ref{prop:BRSTdoublet}. The twisted theory $S_{\rm twist}$ is therefore perturbatively equivalent to the theory with the BV action 
\begin{align*}
& \int \bigg(\xi \ol \dd_{A_{0,1}} B  + \psi \wedge \ol \dd_{A_{0,1}} \psi \bigg) + \dvol \left(  \frac 12 (F_{0,2}, B^*) + \frac 12 (\ol \dd_{A_{0,1}} \phi, \psi^*)\right) \\
 &+ \int \dvol\left((\ol\dd_{A_{0, 1}} c, A_{0, 1}^*) + ([\xi, c], \xi^*) + \frac{1}{2}[c, c]c^* +  [\phi, c] \phi^* + ([\psi, c],\psi^*)\right) + [B, c] \wedge B^* .
\end{align*}
Up to rescaling the antifields, this is the action functional of the required theory.
\end{proof}

If $U=T^*R=R\oplus R^*$, the R-symmetry group is enhanced to $G_R=\SL(2, \CC)\times \GL(1, \CC)$. We have a homomorphism $\alpha\colon \U(1)\hookrightarrow G_R=\SL(2, \CC)\times \U(1)$ given by the diagonal embedding which is compatible with the holomorphic supercharge. We may also use a new twisting homomorphism
\[\Tilde{\phi}\colon\MU(3)\xrightarrow{\det^{1/2}}\U(1)\xrightarrow{\alpha} G_R.\]

With these modifications the BRST fields are given by:
\begin{itemize}
\item Gauge fields $A_{1, 0} \in \Omega^{1, 0}(M; \fg)$, $A_{0, 1}\in\Omega^{0, 1}(M; \fg)$.
\item Gauge fermions $\chi\in\Omega^0(M; \fg)[-1]$, $\xi\in\Omega^{3, 0}(M; \Pi\fg)[1]$, $B\in\Omega^{0, 2}(M; \fg)[-1]$, $\rho\in\Omega^{1, 0}(M; \Pi\fg)[1]$.
\item Matter bosons $\nu\in\Omega^0(M; R^*\oplus R\otimes K^{-1}[-2])$, $\phi\in\Omega^0(M; R\oplus R^*\otimes K[2])$.
\item Matter fermions $\psi\in \Omega^{0, 1}(M; R[-1]\oplus R^*\otimes K[1])$, $\wt\nu\in\Omega^0(M; R^*[1]\oplus R\otimes K^{-1}[-1])$.
\item A ghost field $c\in \Omega^0(M; \gg)[1]$.
\end{itemize}

Note that the $\MU(3)$-action on the fields factors through $\U(3)$ since only integer powers of $K$ occur. By comparing the degrees and the transformation rules of the fields in Theorem \ref{thm:6dholomorphictwist} we obtain the following statement.

\begin{theorem}
The holomorphic twist of 6d $\cN=(1, 0)$ super Yang--Mills on $M=\RR^6$ with matter valued in a $\fg$-representation $U=T^*R=R\oplus R^*$ is perturbatively equivalent to the holomorphic BF theory on $M\cong \CC^3$ with space of fields $T^*[-1]\map(M, R/\fg)$. Moreover, the equivalence is $\U(3)$-equivariant.
\label{thm:6dholomorphictwistgraded}
\end{theorem}

\subsection{\texorpdfstring{$\cN=(1, 1)$}{N=(1, 1)} Super Yang--Mills Theory}

The 6d $\cN=(1, 1)$ super Yang--Mills theory is obtained by dimensional reduction from the 10d $\cN=(1, 0)$ super Yang--Mills. It admits $R$-symmetry group $G_R=\Spin(4, \CC)$ under which $W_+, W_-$ are the two semi-spin representations.

Given an element $Q\in S_+\otimes W_+\oplus S_-\otimes W_-$ we denote by $W^*_{Q_\pm}\subset S_\pm$ the images of $Q$. We classify square-zero supercharges according to the ranks of these spaces:
\begin{itemize}
\item Rank $(1, 0)$ and $(0, 1)$. These automatically square to zero and are holomorphic.  Such supercharges factor through a copy of the $\mc N=(1,0)$ (respectively, $\cN=(0, 1)$) supersymmetry algebra. They admit a twisting homomorphism from $\MU(3)$ and a $\ZZ$-grading $\alpha\colon \U(1)\rightarrow G_R$.

\item Rank $(1, 1)$ and $\langle W^*_{Q_+}, W^*_{Q_-}\rangle = 0$. These automatically square to zero and have 4 invariant directions. There is a $\ZZ$-grading $\alpha\colon \U(1)\hookrightarrow G_R=\SL(2, \CC)\times \SL(2, \CC)$ given by the diagonal embedding and a twisting homomorphism $\phi\colon \MU(2)\times \Spin(2, \RR)\xrightarrow{\det^{1/2}\times \id}\U(1)\xrightarrow{\alpha} G_R$.

\item Rank $(1, 1)$ and $\langle W^*_{Q_+}, W^*_{Q_-}\rangle \neq 0$. These automatically square to zero and are topological. Such supercharges are stabilized by $\SU(3)\subset \Spin(6, \CC)$ and have a $\ZZ$-grading $\alpha\colon \U(1)\rightarrow G_R=\SL(2, \CC)\times \SL(2, \CC)$ given by the diagonal embedding.

\item Rank $(2, 2)$. The square-zero supercharges have 5 invariant directions and give rise to a $\ZZ/2\ZZ$-graded theory. The twisting homomorphism is given by the obvious embedding $\phi\colon \Spin(4, \RR)\rightarrow G_R=\Spin(4, \CC)$, so the twisted theory carries a $\Spin(4, \RR)$-action.
\end{itemize}

\subsubsection{Holomorphic Twist}
\label{sect:6d11holomorphictwist}

The 6d $\cN=(1, 1)$ super Yang--Mills theory viewed as a $\cN=(1, 0)$ supersymmetric theory coincides with the 6d $\cN=(1, 0)$ super Yang--Mills with matter in the representation $U=T^*\fg=\fg\oplus \fg^*$. Under this isomorphism the R-symmetry group $\SL(2, \CC)\times \GL(1, \CC)$ of 6d $\cN=(1, 0)$ super Yang--Mills is a subgroup of the R-symmetry group $\SL(2, \CC)\times \SL(2, \CC)$ of 6d $\cN=(1, 1)$ super Yang--Mills. In particular, from Theorem \ref{thm:6dholomorphictwistgraded} we obtain the following statement.

\begin{theorem}
The holomorphic twist of 6d $\cN=(1, 1)$ super Yang--Mills on $M=\RR^6$ is perturbatively equivalent to the holomorphic BF theory on $M\cong \CC^3$ with the space of fields $T^*[-1]\map(M, \fg/\fg)$. Moreover, the equivalence is $\U(3)$-equivariant.
\label{thm:6d11holomorphictwist}
\end{theorem}

\subsubsection{Rank \texorpdfstring{$(1, 1)$}{(1,1)} Partially Topological Twist}
\label{sect:6d11partialtwist}

Let $L=\CC^2$ equipped with a Hermitian structure, $N_\RR=\RR^2$ equipped with a Euclidean structure and $N=N_\RR\otimes_\RR\CC$ its complexification. Consider the 8-dimensional spacetime $V^8_\RR=L\times N$ and the 6-dimensional spacetime $V^6_\RR=L\times N_\RR$. Under the projection $V^8_\RR\rightarrow V^6_\RR$ a holomorphic square-zero supercharge $Q$ in 8 dimensions dimensionally reduces to a rank $(1, 1)$ partially topological square-zero supercharge in 6 dimensions. Therefore, from Theorem \ref{thm:8dBFreduction} we obtain the following statement.

\begin{theorem}
The rank $(1, 1)$ partially topological twist of 6d $\cN=(1, 1)$ super Yang--Mills is perturbatively equivalent to the generalized BF theory with the space of fields $T^*[-1]\map(\CC^2\times \RR^2_{\mathrm{dR}}, B\fg)$. Moreover, the equivalence is $\MU(2)\times\Spin(2, \RR)$-equivariant.
\end{theorem}

\subsubsection{Topological Twist}
\label{sect:6d11topologicaltwist}

Let $L=\CC^2$ equipped with a Hermitian structure and a complex half-density, $N_\RR=\RR^2$ equipped with a Euclidean structure and $N=N_\RR\otimes_\RR\CC$ its complexification. Consider the 8-dimensional spacetime $V^8_\RR=L\times N$ and the 6-dimensional spacetime $V^6_\RR=L\times N_\RR$. Under the projection $V^8_\RR\rightarrow V^6_\RR$ the family $Q_t$ of 8-dimensional square-zero supercharges given by equation \eqref{eq:8dHodgefamily} dimensionally reduces to a family of square-zero supercharges which are topological for $t\neq 0$ and have 4 invariant directions at $t=0$. So, we get a rank $(1, 1)$ partially topological twist at $t=0$ and a rank $(1, 1)$ topological twist at $t\neq 0$. Therefore, from Theorem \ref{thm:8dHodgereduction} we obtain the following statement.

\begin{theorem}
The twist of 6d $\cN=(1, 1)$ super Yang--Mills with respect to the family $Q_t$ of square-zero supercharges is perturbatively equivalent to the generalized Hodge theory $\map(\CC^2\times \RR^2_{\mathrm{dR}}, B\fg_{\Hod})$. Moreover, this equivalence is $\SU(2)\times \Spin(2, \RR)$-equivariant.
\label{thm:6dHodgetwist}
\end{theorem}

\begin{corollary}
The topological twist of 6d $\cN=(1, 1)$ super Yang--Mills is perturbatively trivial.
\end{corollary}

\subsubsection{Rank \texorpdfstring{$(2, 2)$}{(2,2)} Twist}
\label{sect:6drank22twist}

We consider $L=\CC$ equipped with a Hermitian structure and a complex half-density. From Theorem \ref{thm:10dCSreduction} we obtain the following statement.

\begin{theorem}
The rank $(2, 2)$ twist of 6d $\cN=(1, 1)$ super Yang--Mills is perturbatively equivalent to the generalized Chern--Simons theory with the space of fields $\map(\CC\times \RR^4_{\mr{dR}}, B\fg)$. Moreover, the equivalence is $\Spin(4, \RR)$-equivariant.
\label{thm:6drank22}
\end{theorem}

\section{Dimension 5}

The odd part of the $5$-dimensional supersymmetry algebra is 
\[
\Sigma\cong S\otimes W
\]
where $S$ is the 4-dimensional spin representation of $\Spin(5, \CC)\cong \Sp(4, \CC)$ and $W$ is a complex symplectic vector space. 
The spin representation carries a symplectic paring $S\otimes S\rightarrow \CC$.

There are Yang--Mills theories with $\cN=1$ or $\cN=2$ supersymmetry, where $2 \cN = \dim(W)$, which we consider separately.

\subsection{\texorpdfstring{$\mc N = 1$}{N=1} Super Yang--Mills Theory}
\label{5d_1_section}

We fix $W = \CC^2$ equipped with a symplectic structure. Let $U$ be a symplectic $\fg$-representation. The 5d $\cN=1$ super Yang--Mills theory is obtained by a dimensional reduction from the 6d $\cN=(1, 0)$ super Yang--Mills theory.

The $R$-symmetry group coincides with the $R$-symmetry group in 6 dimensions:
\begin{itemize}
\item For a general $U$ the $R$-symmetry group is $G_R=\SL(2, \CC)$ with $W$ the two-dimensional defining representation.

\item For $U=T^* R$ the $R$-symmetry group is $G_R=\SL(2, \CC)\times \GL(1, \CC)$, where $\GL(1, \CC)$ acts trivially on $W$, with weight $1$ on $R$ and with weight $-1$ on $R^*$.
\end{itemize}

The theory admits a unique twist:
\begin{itemize}
\item A square-zero supercharge $Q\neq 0\in\Sigma$ has 3 invariant directions. There is a twisting homomorphism $\phi\colon \MU(2)\xrightarrow{\det^{1/2}}\U(1)\hookrightarrow \SL(2, \CC)$, so the twisted theory carries an $\MU(2)$-action.
\end{itemize}

\subsubsection{Minimal Twist}
\label{sect:5d1minimaltwist}

A square-zero supercharge $Q$ has rank 1, i.e. $Q=q_+\otimes w_1$ for some $w_1\in W$. Choose a complementary element $w_2\in W$ such that $(w_1, w_2) = 1$. We have a twisting homomorphism
\[\phi\colon\MU(2)\xrightarrow{\det^{1/2}}\U(1)\subset \SL(2, \CC)\]
such that $W\cong \det(L)^{-1/2} w_1\oplus \det(L)^{1/2} w_2$.

As in Section \ref{sect:9dminimaltwist}, the data of $q_+$ is equivalent to the choice of a one-dimensional subspace $N_\RR\subset V_\RR$ and a complex structure on $V_\RR/N_\RR$ together with a complex half-density.

We will perform a computation of the twist in a more general setting which will be useful for lower-dimensional computations.

Suppose $L$ is a complex vector space equipped with a Hermitian structure. Suppose $N_\RR=\RR^{3-\dim(L)}$ equipped with a Euclidean metric and a spin structure. Denote by $N=N_\RR\otimes_\RR\CC$ its complexification which carries a complex half-density. Let $V_\RR=L\times N$ (a 6-dimensional real vector space). We set $W^6_+=\det(L)^{-1/2}w_1\oplus \det(L)^{1/2} w_2$ as a representation of $\MU(L)$. By the results of Section \ref{sect:6dholomorphictwist}, there is a canonical square-zero supercharge $Q=q_+\otimes w_1\in\Sigma$ determined by the complex structure on $L\times N$ and a complex half-density on $N$.

The dimensional reduction of 6d super Yang--Mills on $L\times N$ along $\Re\colon N\rightarrow N_\RR$ is by definition the $(3+\dim(L))$-dimensional super Yang--Mills on $L\times N_\RR$ with 8 supercharges. Since $N\cong N_\RR\oplus iN_\RR$, this theory carries an action of the $R$-symmetry group $G_R=\SL(2, \CC)\times \Spin(N_\RR)$. We consider a twisting homomorphism
\[\phi\colon \MU(L)\times \Spin(N_\RR)\rightarrow G_R=\SL(2, \CC)\times \Spin(N_\RR)\]
whose first component is $\MU(L)\rightarrow \SU(2)$ as before and the second component is the identity.

\begin{theorem}
The twist by $Q$ of $(3+\dim(L))$-dimensional super Yang--Mills on $L\times N_\RR$ with 8 supercharges and matter valued in a symplectic $\fg$-representation $U$ is perturbatively equivalent to the generalized Chern--Simons theory with the space of fields $\Sect(L\times (N_\RR)_{\mr{dR}}, (U\otimes K_L^{1/2})\ham \fg)$. Moreover, the equivalence is $\MU(L)\times \Spin(N_\RR)$-equivariant.
\label{thm:6dCSreduction}
\end{theorem}

\begin{proof}
By Theorem \ref{thm:6dholomorphictwist} the twist of 6d $\cN=(1, 0)$ super Yang--Mills on $L\times N$ by $Q$ is perturbatively equivalent to the holomorphic Chern--Simons theory with the space of fields $\Sect(L\times N, (U\otimes K_L^{1/2})\ham \fg)$. By Proposition \ref{CS_diml_red_prop} we get that the dimensional reduction of holomorphic Chern--Simons on $L\times N$ along $\Re\colon N\rightarrow N_\RR$ is isomorphic to the generalized Chern--Simons theory with the space of fields $\Sect(L\times N_\RR, (U\otimes K_L^{1/2})\ham\fg)$ and this isomorphism is $\MU(L)\times \SO(N_\RR)$-equivariant, where $\SO(N_\RR)$ acts on $N$ via the homomorphism \eqref{eq:SOtoSU}.
\end{proof}

We will now concentrate on the 5-dimensional case.

\begin{theorem}
The minimal twist of 5d $\cN=1$ super Yang--Mills on $M=\CC^2\times \RR$ with matter valued in a symplectic $\fg$-representation $U$ is perturbatively equivalent to the generalized Chern--Simons theory with the space of fields $\Sect(\CC^2\times \RR_{\mathrm{dR}}, (U\otimes K_{\CC^2}^{1/2})\ham \fg)$. Moreover, the equivalence is $\MU(2)$-equivariant.
\label{thm:5dminimaltwist}
\end{theorem}

If $U$ is of cotangent type, we may enhance the $R$-symmetry group to $G_R=\SL(2, \CC)\times \GL(1, \CC)\times \Spin(N_\RR)$. We have a homomorphism $\alpha\colon \U(1)\rightarrow \SL(2, \CC)\times \GL(1, \CC)\times \Spin(N_\RR)$ given by the diagonal embedding into the first two components. We also use a new twisting homomorphism
\[\wt\phi\colon \MU(L)\times \Spin(N_\RR)\xrightarrow{\det^{1/2}\times \id}\U(1)\times \Spin(N_\RR)\xrightarrow{\alpha\times\id}G_R.\]

\begin{theorem}
The minimal twist of 5d $\cN=1$ super Yang--Mills on $M=\CC^2\times \RR$ with matter valued in the $\fg$-representation $U=T^* R=R\oplus R^*$ is perturbatively equivalent to the generalized BF theory with the space of fields $T^*[-1]\map(\CC^2\times \RR_{\mathrm{dR}}, R/\fg)$. Moreover, the equivalence is $\U(2)$-equivariant.
\label{thm:5dminimaltwistgraded}
\end{theorem}

\subsection{\texorpdfstring{$\mc N = 2$}{N=2} Super Yang--Mills Theory} \label{5d_2_section}

The 5d $\cN=2$ super Yang--Mills theory is obtained by dimensional reduction from the 10d $\cN=(1, 0)$ super Yang--Mills. It admits $R$-symmetry group $G_R=\Spin(5, \CC)$ under which $W$ is the 4-dimensional spin representation equipped with a symplectic pairing.

An element $Q\in S\otimes W$ gives rise to maps $S^*\rightarrow W$ and $W^*\rightarrow S$. Both $S$ and $W$ are 4-dimensional symplectic vector spaces and the classification of supercharges will use their relative position.

This theory admits twists by the following four classes of supercharge.
\begin{itemize}
\item Rank 1. These automatically square to zero and have 3 invariant directions. Such supercharges come from the $\cN=1$ supersymmetry algebra. They admit a twisting homomorphism from $\MU(2)$ and a $\ZZ$-grading $\alpha\colon \U(1)\rightarrow G_R$.

\item Rank 2, where the image of $W^*\rightarrow S$ is Lagrangian. These automatically square to zero and have 4 invariant directions. There is a $\ZZ$-grading $\alpha\colon \U(1)\rightarrow G_R$ and a twisting homomorphism $\phi\colon \Spin(2, \RR)\times \Spin(3, \RR)\hookrightarrow G_R=\Spin(5, \CC)$.

\item Rank 2, where the image of $W^*\rightarrow S$ is symplectic. These automatically square to zero and are topological. There is a $\ZZ$-grading $\alpha\colon \U(1)\rightarrow G_R$ and a twisting homomorphism $\phi\colon \Spin(4, \RR)\hookrightarrow G_R=\Spin(5, \CC)$.

\item Rank 4. These are topological and do not admit a compatible homomorphism $\alpha$. There is the obvious twisting homomorphism $\phi\colon \Spin(5, \RR)\rightarrow \Spin(5, \CC)$.
\end{itemize}

\subsubsection{Minimal Twist}
\label{sect:5dminimaltwist}

The 5d $\cN=2$ super Yang--Mills viewed as a $\cN=1$ supersymmetry theory coincides with the 5d $\cN=1$ super Yang--Mills with matter valued in the representation $U=T^*\fg=\fg\oplus \fg^*$. From Theorem \ref{thm:5dminimaltwistgraded} we obtain the following statement.

\begin{theorem}
The minimal twist of 5d $\cN=2$ super Yang--Mills on $M=\CC^2\times \RR$ is perturbatively equivalent to the generalized BF theory with the space of fields $T^*[-1]\map(\CC^2\times \RR_{\mathrm{dR}}, \fg/\fg)$. Moreover, the equivalence is $\U(2)$-equivariant.
\end{theorem}

\subsubsection{Rank 2 Partially Topological Twist}
\label{sect:5dpartialtwist}

Let $L = \CC$ equipped with a Hermitian structure, $N_\RR=\RR^3$ equipped with a Euclidean structure and $N=N_\RR\otimes_\RR\CC$ its complexification. Consider the 8-dimensional spacetime $V^8_\RR=L\times N$ and the 5-dimensional spacetime $V^5_\RR=L\times N_\RR$. Under the projection $V^8_\RR\rightarrow V^5_\RR$ a holomorphic square-zero supercharge $Q$ in 8 dimensions dimensionally reduces to a rank $2$ partially topological square-zero supercharge in 5 dimensions. Therefore, from Theorem \ref{thm:8dBFreduction} we obtain the following statement.

\begin{theorem}
The rank $2$ partially topological twist of 5d $\cN=2$ super Yang--Mills is perturbatively equivalent to the generalized BF theory with the space of fields $T^*[-1]\map(\CC\times \RR^3_{\mathrm{dR}}, B\fg)$. Moreover, the equivalence is $\Spin(2, \RR)\times\Spin(3, \RR)$-equivariant.
\end{theorem}

\subsubsection{Rank 2 Topological Twist}
\label{sect:5drank2topologicaltwist}

Let $L=\CC$ equipped with a Hermitian structure and a complex half density, $N_\RR=\RR^3$ equipped with a Euclidean structure and $N=N_\RR\otimes_\RR\CC$ its complexification. Consider the 8-dimensional spacetime $V_\RR=L\times N$. Under the projection $\Re\colon N\rightarrow N_\RR$ the family $Q_t$ of 8-dimensional square-zero supercharges given by \eqref{eq:8dHodgefamily} dimensionally reduces to a family of 5-dimensional square-zero supercharges, which at $t\neq 0$ are topological at and $t=0$ have 4 invariant directions. Since they admit a compatible $\ZZ$-grading, at $t=0$ we obtain the rank 2 partially topological twist and at $t\neq 0$ we obtain the rank 2 topological twist. Therefore, from Theorem \ref{thm:8dHodgereduction} we obtain the following statement.

\begin{theorem}
The twist of 5d $\cN=2$ super Yang--Mills with respect to the family $Q_t$ of square-zero supercharges is perturbatively equivalent to the generalized Hodge theory $\map(\CC\times \RR^3_{\mathrm{dR}}, B\fg_{\Hod})$. Moreover, this equivalence is $\Spin(3)$-equivariant.
\label{thm:5dHodgetwist}
\end{theorem}

\begin{corollary}
The rank 2 topological twist of 5d $\cN=2$ super Yang--Mills is perturbatively trivial.
\label{cor:5drank2topologicaltwist}
\end{corollary}

\subsubsection{Rank 4 Twist}
\label{sect:5drank4twist}

We consider $V_\RR=\RR^5$ equipped with a Euclidean structure and, as before, let $V=V_\RR\otimes_\RR\CC$ be its complexification. $V$ carries a Hermitian structure and a half-density, so by the results of Section \ref{sect:10dholomorphictwist} we obtain a square-zero supercharge $Q$. Under the projection $\Re\colon V\rightarrow V_\RR$ the supercharge $Q$ dimensionally reduces to a rank 4 supercharge in 5 dimensions. Therefore, from Theorem \ref{thm:10dCSreduction} we obtain the following statement.

\begin{theorem}
The rank $4$ twist of 5d $\cN=2$ super Yang--Mills is perturbatively equivalent to the topological Chern--Simons theory with the space of fields $\map(\RR^5_{\mr{dR}}, B\fg)$. Moreover, the equivalence is $\Spin(5, \RR)$-equivariant.
\end{theorem}
 
\section{Dimension 4}

The odd part of the $4$-dimensional supersymmetry algebra is
\[
\Sigma\cong S_+\otimes W\oplus S_-\otimes W^*
\] 
where $S_+, S_-$ are the 2-dimensional semi-spin representations of $\Spin(4, \CC)\cong \SL(2, \CC) \times \SL(2, \CC)$ and $W$ is a complex vector space. 
The semi-spin representations carry symplectic pairings $S_\pm\otimes S_\pm\rightarrow \CC$.

There are Yang--Mills theories with $\cN= \dim(W) = 1,2,4$ supersymmetry, which we consider separately.

\subsection{\texorpdfstring{$\cN=1$}{N=1} Super Yang--Mills Theory} \label{sect:4d_1_section}
The general setup for $\cN=1$ super Yang--Mills is described in Section \ref{sect:SYM}, which we now recall.  Let $R$ be a complex $\fg$-representation. We consider $\cN=1$ super Yang--Mills theory on $M=\RR^4$ with the Euclidean metric. The theory admits an $R$-symmetry group $G_R=\GL(1, \CC)$ which acts on $W=\CC$ with weight $1$.

\vspace{-10pt}
\paragraph{Fields:} The BRST fields are given by:
\begin{itemize}
\item Gauge field $A \in \Omega^1(M; \fg)$.
\item Gauge fermions $(\lambda_+, \lambda_-) \in \Omega^0(M; \Pi S_+ \otimes \fg \oplus \Pi S_- \otimes \fg)$.
\item Matter bosons $(\Bar{\phi}, \phi) \in \Omega^0(M; R \oplus R^*)$.
\item Matter fermions $(\psi_-,\psi_+) \in \Omega^0(M; \Pi S_+ \otimes R^* \oplus \Pi S_- \otimes R)$.
\item A ghost field $A_0\in \Omega^0(M; \gg)[1]$.
\end{itemize}

The $R$-symmetry acts with weight $\pm 1$ on $\lambda_\pm, \psi_\pm$.

The theory admits a unique twist:
\begin{itemize}
\item Elements $Q\in S_+\oplus S_-$ of rank $(1, 0)$ or rank $(0, 1)$. Such supercharges are automatically square-zero and are holomorphic. 
We have a compatible twisting homomorphism 
\[
\phi\colon \MU(2)\xrightarrow{\det^{1/2}}\U(1)\hookrightarrow G_R 
\]
with the second arrow the natural embedding.
The twist is $\ZZ$-graded with homomorphism $\alpha \colon \U(1) \hookrightarrow G_R$ given by the natural embedding.
\end{itemize}

\subsubsection{Holomorphic Twist}
\label{sect:4d1holomorphictwist}

Choose a complex structure $L$ on $V_\RR$. Under the embedding $\MU(L) = \MU(2) \subset \Spin(V_\RR)$, the semi-spin representations decompose as
\[
S_+ = \det(L)^{-1/2} \oplus \det(L)^{1/2},\qquad S_- = L \otimes \det(L)^{-1/2} .
\]
Consider the twisting homomorphism $\phi\colon\MU(2)\rightarrow G_R$ under which $W = \det(L)^{-1/2}$. Then the spinorial representation becomes
\[\Sigma = (\CC Q\oplus \det(L)^{-1})\oplus L.\]

The embedding $\alpha\colon \U(1)\hookrightarrow G_R$ makes $Q$ weight $1$.

We first decompose the fields of the 4-dimensional $\mc N=1$ theory with respect to $\MU(2)$.
\vspace{-10pt}
\paragraph{Fields:} The BRST fields are given by:
\begin{itemize}
\item Gauge fields $A_{1, 0} \in \Omega^{1, 0}(M; \fg)$, $A_{0, 1}\in\Omega^{0, 1}(M; \fg)$.
\item Gauge fermions $\lambda_0\in\Omega^0(M; \fg)[-1]$, $A_{0, 2}\in\Omega^{0, 2}(M; \fg)[-1]$, $\lambda_{1, 0}\in\Omega^{1, 0}(M; \fg)[1]$.
\item Matter bosons $\phi\in\Omega^0(M; R^*)$, $\gamma_0\in\Omega^0(M; R)$.
\item Matter fermions $\psi_0\in\Omega^0(M; R^*)[1]$, $\beta_{2, 0}\in\Omega^{2, 0}(M; R^*)[1]$, $\gamma_{0, 1}\in\Omega^{0, 1}(M; R)[-1]$.
\item A ghost field $A_0\in \Omega^0(M; \gg)[1]$.
\end{itemize}

Let $\omega\in\Omega^{1, 1}(M)$ be the K\"ahler form. We denote the real volume form on $M$ by
\[\dvol = \frac{\omega^2}{2}.\]

Using Corollary \ref{Kahler_YM_term_cor}, the BV action $S_{\mr{BV}}$ of the $Q$-twisted theory consists of the sum of the following terms:

\begin{align*}
S_{\mr{gauge}} &= \int\dvol\left(-(F_{2, 0}, F_{0, 2}) - \frac{1}{4}(\Lambda F_{1, 1})^2\right) + \frac 12 \left((\lambda_{1,0} \wedge \dd_{A_{1,0}} A_{0,2}) + \omega (\lambda_{1,0} \wedge \dbar_{A_{0,1}} \lambda_0) \right)  \\
S_{\mr{matter}} &= \int \dvol \bigg((\dd_{A_{1,0}}\phi , \ol \dd_{A_{0,1}} \gamma_0) + (\dd_{A_{1,0}}\gamma_0, \ol \dd_{A_{0,1}} \phi) + ([A_{0,2}, \beta_{2,0}], \phi) \bigg) + \beta_{2,0} \dbar_{A_{0,1}} \gamma_{0,1}  + \\
&\qquad + 2 (\omega \wedge ([\lambda_{1,0}, \gamma_{0,1}] ,\gamma_{0})  ) + \omega (\psi_0 \partial_{A_{1,0}} \gamma_{0,1}) \\
S_{\mr{anti}} &= \int \dvol \bigg((\dd_{A_{1, 0}} A_0, A_{1,0}^*) + (\ol\dd_{A_{0, 1}} A_0, A_{0,1}^*) \bigg) +  \\
& \qquad + \dvol \bigg([A_0, \lambda_0] \lambda_0^* + [A_0, A_{0,2}] A_{0,2}^* \bigg) + [\lambda_{1,0}, A_0] \wedge \lambda_{1,0}^* + \dvol \frac{1}{2}[A_0, A_0]A_0^*\\ 
&\qquad  + \dvol \bigg( [\phi, A_0] \phi^* + [\gamma_0, A_0] \gamma_0^*  + [\psi_0, A_0] \psi_0^* + [\beta_{2,0}, A_0] \beta_{2,0}^* \bigg) + [\gamma_{0,1}, A_0]\wedge \gamma_{0,1}^*\\
S^{(1)}_{\mr{gauge}} &=  \int  \dvol - (\lambda_{1,0} A_{1,0}^*)  + \frac 12  \left(F_{0,2}\wedge A_{0,2}^*+ \omega\wedge F_{1,1} \lambda_0^*  \right) \\
S^{(1)}_{\mr{matter}} &=  \int  \dvol \left( \psi_0 \phi^* + \frac 12 ( \ol \dd_{A_{0,1}} \gamma_{0} , \gamma_{0,1}^*) \right)\\
S^{(2)}_{\mr{gauge}} &= -\frac{1}{4}\int \dvol(\lambda_0^*)^2.
\end{align*}

Notice that a priori the theory is only $\MU(2)$-equivariant, but manifestly descends to a $\U(2)$-equivariant theory. 

\begin{theorem}[See also \cite{SWchar}] \label{4d_minimal_twist_thm}
The holomorphic twist of 4d $\cN=1$ super Yang--Mills with matter valued in a $\fg$-representation-representation $R$ is perturbatively equivalent to holomorphic BF theory with the space of fields $T^*[-1]\map(\CC^2, R/\fg)$. 
Moreover, the equivalence is $\U(2)$-equivariant.
\end{theorem}

\begin{proof}
The proof of this theorem is very similar to the proof of Theorem \ref{thm:10dholomorphictwist}.  
First, we eliminate the fields $\lambda_0$ and $\lambda_0^*$ using Proposition \ref{prop:integrateoutfield}.  
We then observe that the action includes the terms $\int \lambda_{1,0} \wedge A_{1,0}^*$ and  $\int \psi_0\phi^*$. Thus, the two pairs $(\lambda_{1,0}, A_{1,0})$ and $(\phi, \psi_0)$ form trivial BRST doublets, which can be eliminated using Proposition \ref{prop:BRSTdoublet}.

The twisted theory is therefore perturbatively equivalent to the theory with BV action 
\begin{align*}
S_{\rm BV} & = \int \dvol \bigg((A_{0,1}^* \ol \dd_{A_{0,1}} A_{0}) + (\dbar_{A_{0,1}} \gamma_0, \gamma_{0,1}^*)\bigg) + F_{0,2} \wedge A_{0,2}^* + \beta_{2,0} \dbar_{A_{0,1}} \gamma_{0,1} \\ & \qquad + \dvol \bigg([A_0, A_{0,2}] A_{0,2}^* + \frac{1}{2} [A_0,A_0] A_0^* + [A_0, \gamma_0] \gamma_0^* + [A_0, \gamma_{0,1}] \gamma_{0,1}^* + [A_0, \beta_{2,0}^*] \beta_{2,0}\bigg)
\end{align*}

Up to rescaling the antifields, this is the action functional of the required theory, where $A_{0,\bu}, A_{0,\bu}^*$ comprise the fields of holomorphic BF theory with $B_{2,\bu} = A_{0,\bu}^*$, and the remaining fields comprise the fields of the $\beta\gamma$ system with $\beta_{2,1} = \gamma_{0,1}^*$, $\beta_{2,2} = \gamma_0^*$, and $\gamma_{0,2} = \beta_{2,0}^*$.
\end{proof}

\subsection{\texorpdfstring{$\cN=2$}{N=2} Super Yang--Mills Theory} \label{4d_2_section}

The 4d $\cN=2$ super Yang--Mills theory is obtained by a dimensional reduction from the 6d $\cN=(1, 0)$ super Yang--Mills theory with matter valued in a symplectic $\fg$-representation $U$. Let $W$ be a two-dimensional complex vector space. The theory admits the $R$-symmetry group $G_R=\SL(2;\CC)\times \GL(1, \CC)$, where $\GL(1, \CC)$ acts on $W$ with weight $1$.

\vspace{-10pt}
\paragraph{Fields:} The BRST fields are given by:
\begin{itemize}
\item Gauge field $A \in \Omega^1(M; \fg)$.
\item Scalar fields $a, \wt a\in\Omega^0(M;\fg)$.
\item Gauge fermions $(\lambda_+, \lambda_-) \in \Omega^0(M; \Pi S_+ \otimes W\otimes \fg \oplus \Pi S_- \otimes W^*\otimes \fg)$.
\item Matter boson $\phi\in\Omega^0(M; U\otimes W)$.
\item Matter fermions $(\psi_-, \psi_+)\in\Omega^0(M; \Pi S_+\otimes U\oplus \Pi S_-\otimes U)$.
\item A ghost field $c\in \Omega^0(M; \gg)[1]$.
\end{itemize}

The subgroup $\GL(1, \CC)\subset G_R$ has the following action on fields: weight $2$ on $a$, weight $-2$ on $\wt a$ and weight $\pm 1$ on $\lambda_\pm,\psi_\pm$.

If the representation $U$ is $T^*R = R\oplus R^*$, the $R$-symmetry group is enhanced to $G_R=\SL(2)\times \GL(1, \CC)\times \GL(1, \CC)$, where the last $\GL(1, \CC)$ acts with weight $1$ on $R$ and weight $-1$ on $R^*$.

There are three classes of square-zero supercharge in the 4d $\mc N=2$ supersymmetry algebra, distinguished by the ranks of the two summands $(Q_+,Q_-) \in S_+ \otimes W \oplus S_- \otimes W^*$:
\begin{itemize}
 \item Rank $(1,0)$ and $(0,1)$ supercharges automatically square to zero.  
 The corresponding twists are holomorphic.
 Such twists factor through a copy of the $\cN = 1$ supersymmetry algebra. As before, they admit a $\ZZ$-grading and a twisting homomorphism from $\MU(2)$.
 \item Rank $(2,0)$ and $(0,2)$ supercharges also automatically square to zero. The corresponding twists are topological (the \emph{Donaldson twist}). There is a twisting homomorphism from $\MU(2)$ and a compatible homomorphism $\alpha\colon \U(1)\rightarrow G_R$.
 \item Rank $(1,1)$ square-zero supercharges have three invariant directions. There is a twisting homomorphism from $\Spin(2, \RR)\times \Spin(2, \RR)\subset \Spin(4, \RR)$. For a general $U$ there is no compatible homomorphism $\alpha\colon \U(1)\rightarrow G_R$.
\end{itemize}

\subsubsection{Holomorphic Twist}
\label{sect:4d_2_holomorphictwist}

Choose a basis for $W$ given by $\{w_1, w_2\}$, where $(w_1, w_2) = 1$, and for concreteness we take $Q=q_+ \otimes w_1$ for some nonzero vector $q_+ \in S_+$. Denote by $L\subset V$ the image of $\Gamma(Q, -)\colon S_-\rightarrow V$. Under the embedding $\MU(L) \subset \Spin(V_\RR)$, the semi-spin representations decompose as
\[
S_+ = \det(L)^{-1/2} \oplus \det(L)^{1/2},\qquad S_- = L \otimes \det(L)^{-1/2} .
\]

Recall that the $R$-symmetry group is $G_R = \SL(2, \CC) \times \GL(1, \CC)$. For any integer $n\in\ZZ$ consider the homomorphism 
\[
\begin{array}{ccccc}
\alpha_n \colon & \U(1) & \rightarrow & \SL(2, \CC)\times \GL(1, \CC)\\
&  z & \mapsto & \left({\rm diag}(z^{2n}, z^{-2n}) , z^{-2n+1} \right)
\end{array}
\]
under which $w_1$ has weight $1$ and $w_2$ has weight $-4n+1$.

We consider the twisting homomorphism 
\[
\phi \colon \MU(2) \xto{\det^{1/2}} \U(1) \rightarrow G_R
\]
under which we have an $\MU(2)$-identification $W=\det(L)^{-1/2}w_1 \oplus \det(L)^{1/2}w_2$, so that
\[S_+\otimes W\cong \CC Q\oplus \det(L)^{-1}\oplus \CC\oplus \det(L).\]

\vspace{-10pt}
\paragraph{Fields:} The BRST fields are given by:
\begin{itemize}
\item Gauge fields $A_{1, 0}\in\Omega^{1, 0}(M; \gg)$ and $A_{0, 1}\in\Omega^{0, 1}(M; \gg)$.
\item Scalar fields $\Tilde{a} \in \Omega^{0}(M; \gg)[4n-2]$ and $a \in\Omega^{0}(M; \gg)[-4n+2]$
\item Gauge fermions $\chi \in \Omega^0(M ; \gg)[-1]$, $\xi \in \Omega^{2,0}(M ; \gg)[4n-1]$, $B \in \Omega^{0,2}(M ; \gg)[-1]$, $b \in \Omega^{0,1}(M ; \gg)[-4n+1]$, $\rho \in \Omega^{1,0}(M ; \gg)[1]$,  $\Tilde{\chi} \in \Omega^0(M ; \gg)[4n-1]$.
\item Matter bosons $\nu\in\Omega^0(M; U\otimes K_M^{-1/2})[-2n]$, $\phi\in\Omega^0(M; U\otimes K_M^{1/2})[2n]$.
\item Matter fermions $\psi \in \Omega^{0,1} (M ;  U\otimes K_M^{1/2})[2n-1]$, $\varsigma \in \Omega^{2,0} (M ; U\otimes K_M^{-1/2})[-2n+1]$, $\Tilde{\nu} \in \Omega^{0,2}(M ; U\otimes K_M^{-1/2})[-2n+1]$.
\item A ghost field $c\in \Omega^0(M; \gg)[1]$.
\end{itemize}

\begin{theorem}
Fix the homomorphism $\alpha = \alpha_n$. The holomorphic twist of 4d $\cN=2$ super Yang--Mills with matter valued in a symplectic $\fg$-representation $U$ is perturbatively equivalent to the holomorphic BF theory with the space of fields $T^*[-1] \Sect(M, (U\otimes K_M^{1/2}[2n])\ham \fg)$. Moreover, the equivalence is $\MU(2)$-equivariant.
\label{thm:4d2holomorphictwist}
\end{theorem}
\begin{proof}
4d $\cN=2$ super Yang--Mills theory is obtained by dimensionally reducing 6d $\cN=(1,0)$ super Yang--Mills theory. Under dimensional reduction the 6d fields from Section \ref{sect:6dholomorphictwist} decompose as follows:
\begin{align*}
A_{1, 0}&\rightsquigarrow A_{1, 0}+ \Tilde{a} \\
A_{0, 1}&\rightsquigarrow A_{0, 1}+ a \\
B&\rightsquigarrow B+ b \\
\rho&\rightsquigarrow \rho+ \Tilde{\chi} \\
\psi & \rightsquigarrow \psi+ \varsigma . \\
\end{align*}

The claim about the underlying $\ZZ/2\ZZ$-graded $\MU(2)$-equivariant theories follows by applying dimensional reduction (Proposition \ref{CS_to_BF_diml_red_prop}) to the minimal twist of 5d $\cN=1$ super Yang--Mills theory (Theorem \ref{thm:5dminimaltwist}).

Next, we check that the equivalence respects  the gradings. Indeed, the equivalence given by Theorem \ref{thm:5dminimaltwist} eliminates the fields $A_{1, 0}, \Tilde{a}, \rho, \chi$ and $\Tilde{\chi}$. The rest of the fields organize into the following collections:
\begin{align*}
c+A_{0, 1}+B&\in\Omega^{0, \bullet}(M; \fg)[1] \\
B^*+A^*_{0, 1}+ c^*&\in\Omega^{2, \bullet}(M; \fg) \\
\\
a + b + \xi^*&\in\Omega^{0, \bullet}(M; \fg)[2-4n] \\
\xi + b^* + a^*&\in\Omega^{2, \bullet}(M; \fg)[4n-1] \\
\\
\phi + \psi + \varsigma^*&\in\Omega^{0, \bullet}(M; U\otimes K_M^{1/2})[2n] \\
\varsigma + \psi^* + \phi^* &\in\Omega^{2, \bullet}(M; U\otimes K_M^{-1/2})[1-2n]
\end{align*}

These fields have the same degrees as in the holomorphic BF theory.
\end{proof}

\subsubsection{Rank \texorpdfstring{$(2,0)$}{(2,0)} Topological Twist}
\label{sect:4d2Donaldson}

Next we discuss the case of the topological twist. As in Section \ref{sect:8dtopologicaltwist} it will be useful to consider a family of topological supercharges degenerating to a rank $(1, 0)$ holomorphic supercharge.

Consider the same twisting homomorphism $\phi\colon \MU(2)\rightarrow G_R$ as in Section \ref{sect:4d_2_holomorphictwist} and $\alpha=\alpha_0\colon \U(1)\rightarrow G_R$. With respect to the $\MU(2)$-action we have a decomposition
\[S_+\otimes W\cong \CC Q_0 \oplus \det(L)^{-1}\oplus \CC \Bar{Q}_0\oplus \det(L).\]

Consider a family of supercharges
\begin{equation} 
\label{eq:4dHodgefamily}
Q_t = Q_0 + t \Bar{Q}_0
\end{equation}
where $t \in \CC$. When $t \ne 0$, this supercharge is of rank $(2,0)$, while at $t = 0$ it reduces to the holomorphic supercharge from the previous section.

\begin{remark}
With respect to $\alpha_n\colon \U(1)\rightarrow G_R$ the supercharge $Q_0$ has weight $1$, while $\Bar{Q}_0$ has weight $-4n+1$. So, requiring $Q_t$ to have weight $1$ forces us to choose $n=0$.
\end{remark}

We will use the notation for fields from Section \ref{sect:4d_2_holomorphictwist}. First, we are going to write the functionals \eqref{eq:gaugeI1}, \eqref{eq:gaugeI2}, \eqref{eq:matterI1}, \eqref{eq:matterI2} in terms of these fields.

\begin{prop}
\label{4d_donaldson_susy_prop}
Suppose $Q_t$ is the rank $(2,0)$ supercharge of \ref{eq:4dHodgefamily}.
The $\MU(2)$ decomposition of the functionals $S_{\rm gauge}^{(1)}, S^{(1)}_{\rm matter}, S^{(2)}_{\rm gauge}, S^{(2)}_{\rm matter}$ (see \eqref{eq:gaugeI1},  \eqref{eq:gaugeI2}, \eqref{eq:matterI1}, \eqref{eq:matterI2}) in terms of the fields of 4d $\cN=2$ super Yang--Mills theory are
\begin{align*}
S_{\rm gauge}^{(1)}(Q_t) &= \int \dvol\left(-(\rho, A_{1, 0}^*) - t(b, A_{0, 1}^*) - (\wt \chi + t\chi)\wt a^*\right) \\
& + \int\dvol\left((F_{0, 2}, B^*) + (\ol\dd_{A_{0, 1}}a, b^*) + \frac{1}{2}\Lambda(F_{1, 1}  + [a, \wt a])\chi^* + [\phi, \phi]\xi^* \right) \\
&+ \int\dvol\left( t\Omega^{-1} F_{2, 0}\wedge B^* + (t\dd_{A_{1, 0}} a, \rho^*)\right) \\
S^{(2)}_{\rm gauge} (Q_t) &= \int \dvol \left(t\chi^*\wt \chi^* + \frac{t}{2} \xi^* B^* -\frac{1}{4}(\chi^* + t\wt\chi^*)^2 + t a c^*\right) \\
S_{\rm matter}^{(1)}(Q_t) &= \int \dvol \left((\Tilde{\nu}, \nu^*) + t (\varsigma, \phi^*) + \frac{1}{2} (\dbar_{A_{0,1}} \phi, \psi^*) + [\nu, a]\varsigma^* \right) \\
S_{\rm matter}^{(2)}(Q_t) &= \int \dvol \frac t4  (\psi^*, \psi^*)
\end{align*}
\end{prop}

\begin{theorem}
The twist of 4d $\cN=2$ super Yang--Mills with respect to the family $Q_t$ of square-zero supercharges is perturbatively equivalent to the holomorphic Hodge theory $\Sect\left(\CC^2, \left((U\otimes K_{\CC^2}^{1/2})\ham \fg\right)_{\Hod}\right)$. Moreover, this equivalence is $\MU(2)$-equivariant.
\label{thm:4dDonaldsontwist}
\end{theorem}
\begin{proof}
The proof proceeds as in the proof of Theorem \ref{thm:10dholomorphictwist} with slight modifications.

Observe that the quadruple of fields $\{\chi^*, \chi, \wt\chi^*, \wt\chi\}$ has the same Poisson brackets as the quadruple $\{\chi^*-t\wt\chi^*, \chi, \wt\chi^*, \wt\chi + t\chi\}$. Therefore, we may eliminate the fields $\chi^*-t\wt\chi^*, \chi$ using Proposition \ref{prop:integrateoutfield}. We then have trivial BRST doublets $\{\wt\chi + t\chi, \wt a\}$, $\{\nu, \wt \nu\}$ and $\{\rho, A_{1, 0}\}$ which may be eliminated using Proposition \ref{prop:BRSTdoublet}. We are left with the action
\begin{equation}\label{eqn:(2,0)}
S_{\mr{BF}} + t \int\dvol\left(-(b, A_{0, 1}^*) + a c^* + \frac{1}{2} \xi^* B^* + (\varsigma, \phi^*) + \frac 14 (\psi^*, \psi^*) \right),
\end{equation}
where $S_{\mr{BF}}$ is the action functional of the holomorphic twist at $t=0$ found in the previous section. 
Since the extra terms are quadratic in the fields, the claim is reduced to a comparison of the underlying local $L_\infty$ algebra of the twisted theory and that of the holomorphic Hodge theory. The former is given by (cf. the proof of Theorem \ref{thm:4d2holomorphictwist})
\[
\xymatrix@R=0.5cm@C=0cm{
\ul{-1} & \ul{0} & \ul{1} & \ul{2} & \ul{3} & \ul{4} \\ 
& \Omega^0(\CC^2; \gg)_c \ar[r] &\Omega^{0,1}(\CC^2; \gg)_{A_{0,1}} \ar[r] &\Omega^{0,2}(\CC^2; \gg)_{B} \\
\Omega^0(\CC^2; \gg)_a  \ar@{-->}^{t\id}[ur] \ar[r] &\Omega^{0,1}(\CC^2; \gg)_{b}  \ar@{-->}^{t\id}[ur] \ar[r] &\Omega^{0,2}(\CC^2; \gg)_{\xi^*}  \ar@{-->}^{t\id}[ur] \\
& &  \Omega^{2,0}(\CC^2 ; \gg)_{B^*}  \ar@{-->}^{t\id}[dr] \ar[r] & \Omega^{2,1}(\CC^2 ; \gg)_{A_{0,1}^*} \ar[r]  \ar@{-->}^{t\id}[dr] & \Omega^{2,2}(\CC^2 ; \gg)_{c^*}  \ar@{-->}^{t\id}[dr] \\
& & & \Omega^{2,0}(\CC^2; \gg)_\xi \ar[r] &\Omega^{2,1}(\CC^2; \gg)_{b^*} \ar[r] &\Omega^{2,2}(\CC^2; \gg)_{a^*}  \\
& & \Omega^0(\CC^2 ; U\otimes K^{1/2})_\phi \ar[r] & \Omega^{0,1}(\CC^2 ; U\otimes K^{1/2})_\psi \ar[r] & \Omega^{0,2}(\CC^2 ; U\otimes K^{1/2})_{\varsigma^*} \\
& \Omega^{2,0}(\CC^2 ; U\otimes K^{-1/2})_{\varsigma} \ar[r]  \ar@{-->}^{t\id}[ur] & \Omega^{2,1}(\CC^2 ; U\otimes K^{-1/2})_{\psi^*} \ar[r]  \ar@{-->}^{t\id}[ur]  & \Omega^{2,2}(\CC^2 ; U\otimes K^{-1/2})_{\phi^*}  \ar@{-->}^{t\id}[ur]
}
\]
which is exactly the local $L_\infty$ algebra of the holomorphic Hodge theory.
\end{proof}

\begin{corollary}
The rank $(2,0)$ topological twist of 4d $\cN=2$ super Yang--Mills is perturbatively trivial.
\label{cor:4dDonaldsontwist}
\end{corollary}
\begin{proof}
The topological twist of 4d $\cN=2$ super Yang--Mills is the twist by $Q_t$ with $t\neq 0$. By Theorem \ref{thm:4dDonaldsontwist} it is equivalent to the $t\neq 0$ specialization of the holomorphic Hodge theory which by Proposition \ref{prop:Hodgetheoryspecialization} is perturbatively trivial.
\end{proof}

\subsubsection{Rank \texorpdfstring{$(1,1)$}{(1,1)} Twist}
\label{sect:4d_2_11}

We finally consider the twist with respect to a rank $(1,1)$ supercharge.
In this case, the twist is compatible with the group $G = \Spin(2, \RR) \times \Spin(2, \RR)$. 
We denote each factor by $\Spin(2, \RR)_i$, $i=1,2$.  
The twisting homomorphism is
\[
\phi \colon \Spin(2, \RR)_1 \times \Spin(2, \RR)_2 \to G_R=\SL(W)\times \GL(1, \CC),
\]
where on the first factor $\Spin(2, \RR)_1 \cong \U(1)\hookrightarrow \SL(2, \CC)$ is given by the diagonal embedding and on the second factor $\Spin(2, \RR)_2 \hookrightarrow \GL(1, \CC)$ is the obvious inclusion. 
If we denote by $S_{\pm, i}$, $i=1,2$ the semi-spin representations of the factor $\Spin(2, \RR)_i$, we have
\[W \cong S_{+,2}\otimes (S_{+,1}\oplus S_{-,1}).\]
The semi-spin representations of $\Spin(4, \RR)$ decompose with respect to $\Spin(2, \RR)_1 \times \Spin(2, \RR)_2\subset \Spin(4, \RR)$ as
\[S_+\cong S_{+,1} \otimes S_{+,2} \oplus S_{-,1} \otimes S_{-,2},\qquad S_-\cong S_{+,1} \otimes S_{-,2} \oplus S_{-,1} \otimes S_{+,2}.\]

So, both $S_+\otimes W$ and $S_-\otimes W^*$ contain a trivial one-dimensional subspace and hence we obtain a rank $(1, 1)$ square-zero supercharge.

\begin{theorem} \label{thm:4d_11_twist}
The rank $(1,1)$ partially topological twist of 4d $\mc N=2$ super Yang--Mills theory is perturbatively equivalent to the generalized Chern--Simons theory with space of fields $\mr{Sect}(\CC \times \RR^2_{\mr{dR}}, (U \otimes K_\CC^{1/2}) \ham \gg)$.  
Moreover, this equivalence is $\U(1) \times \Spin(2, \RR)$-equivariant.
\end{theorem}
\begin{proof}
Any square-zero supercharge of rank $(1, 1)$ lifts to a rank $1$ supercharge in the 5d $\cN=1$ supersymmetry algebra.
The result then follows from Theorem \ref{thm:6dCSreduction} applied to $L = \CC$. 
\end{proof}

If $U = T^* R$ is of cotangent type, we may enhance the $R$-symmetry group to $G_R=\SL(W)\times \GL(1, \CC)\times \GL(1, \CC)$. The last $\GL(1, \CC)$ acts trivially on $W$, by weight $+1$ on $R$ and weight $-1$ on $R^*$.

We have a homomorphism $\alpha\colon \U(1)\rightarrow  \SL(W)\times \GL(1, \CC) \times \GL(1, \CC)$ given by the diagonal embedding into the first and the third components. We also use a new twisting homomorphism $\wt\phi\colon \Spin(2, \RR)_1 \times \Spin(2, \RR)_2\rightarrow G_R$ given by composing $\phi$ with the obvious homomorphism from the first factor $\Spin(2, \RR)_1$ to the last $\GL(1, \CC)$ factor in $G_R$ (cf. the definition of $\alpha$ and $\wt\phi$ in Section \ref{sect:6dholomorphictwist}).

\begin{theorem}\label{thm:4d_11_twistgraded}
The minimal twist of 4d $\cN=2$ super Yang--Mills on $M=\CC \times \RR^2$ with matter valued in the $\fg$-representation $U=T^* R=R\oplus R^*$ is perturbatively equivalent to the generalized BF theory with the space of fields $T^*[-1]\map(\CC \times \RR^2_{\mathrm{dR}}, R/\fg)$. Moreover, the equivalence is $\U(1)\times \Spin(2, \RR)$-equivariant.
\end{theorem}

\subsection{\texorpdfstring{$\cN=4$}{N=4} Super Yang--Mills Theory} \label{4d_4_section}

The 4d $\cN=4$ super Yang--Mills theory is obtained by dimensional reduction from the 10d $\cN=(1, 0)$ super Yang--Mills. It admits $R$-symmetry group $G_R=\Spin(6; \CC)\cong \SL(4, \CC)$ under which $W \cong S_+^6$ is the positive six-dimensional semi-spin representation and $W^* \cong S_-^6$. 

Let us decompose the rotation group as $\Spin(4, \RR)\cong \SU(2)_+\times \SU(2)_-$. The classification of orbits of square-zero supercharges in the 4d $\mc N=4$ supersymmetry algebra is the most interesting among the examples we consider in this paper. We have the following classes. 
\begin{itemize}
\item Rank $(1,0)$ and $(0,1)$ supercharges automatically square to zero. The corresponding twists are holomorphic. Such twists factor through a copy of the $\cN = 1$ supersymmetry algebra. As before, they admit a $\ZZ$-grading and a twisting homomorphism from $\MU(2)$.

\item Rank $(2, 0)$ and $(0, 2)$ supercharges automatically square to zero. The corresponding twists are topological. Such twists factor through a copy of the $\cN = 2$ supersymmetry algebra. They admit the following twisting homomorphisms:
 \begin{enumerate}
  \item The half twisting homomorphism $\phi_{\mr{1/2}} \colon \SU(2)_+ \times \SU(2)_- \to \SL(4,\CC)$ given by $(A,B) \mapsto \mr{diag}(A, 1, 1)$. This is the twisting homomorphism that comes from the $\cN=2$ supersymmetry algebra.
  \item The Kapustin--Witten twisting homomorphism $\phi_{\mr{KW}} \colon \SU(2)_+ \times \SU(2)_- \to \SL(4,\CC)$ given by $(A,B) \mapsto \mr{diag}(A, B)$.
  \item The Vafa--Witten twisting homomorphism $\phi_{\mr{VW}} \colon \SU(2)_+ \times \SU(2)_- \to \SL(4, \CC)$ given by $(A,B) \mapsto \mr{diag}(A, A)$.
 \end{enumerate}
 
 \item Rank $(1,1)$ supercharges.  Such supercharges have three invariant directions. They factor through a copy of the $\cN=2$ supersymmetry algebra. As before, they admit a twisting homomorphism from $\Spin(2, \RR)\times \Spin(2, \RR)$ and admit a $\ZZ$-grading.

 \item Rank $(2,1)$ and $(1,2)$ supercharges.  Such supercharges are topological, compatible with a twisting homomorphism from $\MU(2)$ and admit a $\ZZ$-grading.
 
 \item Rank $(2,2)$ square-zero supercharges are topological. They correspond to a choice of an exact sequence
 \[0\rightarrow S_+^*\rightarrow W\rightarrow S_-\rightarrow 0.\]
 Since $S_+$, $S_-$ and $W$ all carry volume forms, the space of such square-zero supercharges is parameterized by a continuous parameter $s\in\CC^\times$ given by the ratio of the isomorphism $\det(W)\cong \det(S_+)^*\otimes \det(S_-)$ induced by $Q$ and the isomorphism induced by the volume forms. These supercharges admit a $\ZZ$-grading and are compatible with the twisting homomorphism $\phi_{\mathrm{KW}}\colon \Spin(4, \RR)\rightarrow G_R$.
\end{itemize}

\subsubsection{Holomorphic Twist}
\label{sect:4d4holomorphictwist}

Let $L$ be a complex structure on $V_\RR$. Consider a twisting homomorphism $\MU(L)\rightarrow G_R=\Spin(6; \CC)$ under which $W$ decomposes as
\[W = L\otimes \det(L)^{-1/2}\oplus \det(L)^{-1/2} w_1\oplus \det(L)^{1/2} w_2.\]
In particular,
\[S_+\otimes W\cong L\oplus L^*\oplus \det(L)\oplus \CC \oplus \CC\oplus \det(L)^{-1}\]
and we consider the supercharge $Q\in S_+\otimes W$ of rank $(1, 0)$ contained in the scalar summand which spans the subspace $\CC w_1\subset W$.

We consider a homomorphism $\alpha\colon \U(1)\rightarrow G_R$ under which $L\otimes \det(L)^{-1/2}\subset W$ has weight $-1$ and $w_1,w_2$ have weight $1$. In particular, $Q$ has $\alpha$-weight $1$.

\vspace{-10pt}
\paragraph{Fields:} In the notation of Section \ref{sect:4d_2_holomorphictwist}, the BRST fields are given by:
\begin{itemize}
\item gauge bosons $A_{1, 0}\in\Omega^{1, 0}(M; \gg)$ and $A_{0, 1}\in\Omega^{0, 1}(M; \gg)$, $\Tilde{a} \in \Omega^{0}(M; \gg)[-2]$ and $a \in\Omega^{0}(M; \gg)[2]$;
\item gauge fermions $\chi \in \Omega^0(M ; \gg)[-1]$, $\xi \in \Omega^{2,0}(M ; \gg)[-1]$, $B \in \Omega^{0,2}(M ; \gg)[-1]$, $b \in \Omega^{0,1}(M ; \gg)[1]$, $\rho \in \Omega^{1,0}(M ; \gg)[1]$,  $\Tilde{\chi} \in \Omega^0(M ; \gg)[-1]$.
\item matter bosons $\nu\in\Omega^{0, 1}(M; \fg)$, $\phi\in\Omega^{1,0}(M; \fg)$.
\item matter fermions $\psi \in \Omega^{1,1} (M;  \fg)[-1]$, $\varsigma \in \Omega^{1,0} (M ; \fg)[1]$, $\Tilde{\nu} \in \Omega^{0,1}(M; \fg)[1]$.
\item a ghost field $c\in \Omega^0(M; \gg)[1]$.
\end{itemize}

Note that the $\MU(2)$-action on fields factors through a $\U(2)$-action.

\begin{theorem}
The holomorphic twist of 4d $\cN=4$ super Yang--Mills on $M=\RR^4$ is perturbatively equivalent to the BF theory with the space of fields $T^*[-1] \map(\CC^2_{\rm Dol} , B \fg)$. Moreover, the equivalence is $\U(2)$-equivariant.
\label{thm:4d4holomorphictwist}
\end{theorem}
\begin{proof}
The 4d $\cN=4$ super Yang--Mills theory viewed as a $\cN=2$ theory is 4d $\cN=2$ Yang--Mills theory with matter valued in $U=T^*\fg$. Under this correspondence $\alpha\colon \U(1)\rightarrow G_R$ defined above coincides with $\alpha_0$ from Section \ref{sect:4d_2_holomorphictwist}. From Theorem \ref{thm:4d2holomorphictwist} we obtain that the twist is given by $T^*[-1] \map(\CC^2 , (T^*\fg\otimes K_{\CC^2}^{1/2})\ham \fg)$ as a $\ZZ$-graded theory.

Note, however, that the twisting homomorphism used in Section \ref{sect:4d_2_holomorphictwist} differs from the twisting homomorphism defined above. In particular, this equivalence is not $\U(2)$-equivariant. In the present case the fields organize into the following collections:
\begin{align*}
c+A_{0, 1}+B&\in\Omega^{0, \bullet}(M; \fg)[1] \\
\phi + \psi + \varsigma^*&\in\Omega^{1, \bullet}(M; \fg) \\
\xi + b^* + a^*&\in\Omega^{2, \bullet}(M; \fg)[-1] \\
\\
a + b + \xi^*&\in\Omega^{0, \bullet}(M; \fg)[2] \\
\varsigma + \psi^* + \phi^* &\in\Omega^{1, \bullet}(M; \fg)[1] \\
B^*+A^*_{0, 1}+ c^*&\in\Omega^{2, \bullet}(M; \fg)
\end{align*}
These are exactly the fields in $T^*[-1] \map(\CC^2_{\rm Dol} , B \fg)$.
\end{proof}

\subsubsection{Rank \texorpdfstring{$(2,0)$}{(2,0)} Topological Twist}
\label{sect:4d4Atwist}

Next we look at the case of the twist by a rank $(2,0)$ supercharge. 
As in Section \ref{sect:4d2Donaldson}, it will be useful to consider a family of topological supercharges which degenerate to the rank $(1,0)$ supercharge we just discussed. 

We use the same twisting homomorphism $\phi \colon \MU(2) \to \Spin(6;\CC)$ and twisting datum $\alpha \colon U(1) \to G_R$ as in Section \ref{sect:4d4holomorphictwist}.
Then, $S_+ \otimes W$ decomposes under $\MU(2)$ as
\[
S_+\otimes W\cong L\oplus L^*\oplus \det(L)\oplus \CC \cdot Q_0 \oplus \CC \cdot \Bar{Q}_0 \oplus \det(L)^{-1}.
\]

Consider the family of supercharges $Q_t = Q_0 + t \Bar{Q}_0 \in S_+\otimes W$ of rank $(1, 0)$ contained in the scalar summands above.

\begin{theorem}
The twist of 4d $\cN=4$ super Yang--Mills with respect to the family $Q_t$ of square-zero supercharges is perturbatively equivalent to the holomorphic Hodge theory $\map(\CC^2_{\rm Dol}, B \fg_{\rm Hod})$. 
Moreover, this equivalence is $\MU(2)$-equivariant.
\label{thm:4d4Atwist}
\end{theorem}
\begin{proof}
The 4d $\cN=4$ super Yang--Mills theory viewed as a $\cN=2$ theory is 4d $\cN=2$ Yang--Mills theory with matter valued in $U=T^*\fg$. 
Under this identification $\alpha\colon \U(1)\rightarrow G_R$ defined above coincides with $\alpha_0$ from Section \ref{sect:4d_2_holomorphictwist}. 
From Theorem \ref{thm:4dDonaldsontwist} we obtain that the twist is $\Sect\left(\CC^2, \left((T^*\fg \otimes K_{\CC^2}^{1/2})\ham \fg\right)_{\Hod}\right)$ as a $\ZZ$-graded theory. 

The twisting homomorphism used in Section \ref{sect:4d_2_holomorphictwist} differs from the twisting homomorphism defined above. 
In particular, this equivalence is not $\U(2)$-equivariant. In the present case the fields decompose in the same fashion as in the proof of Theorem \ref{thm:4d4holomorphictwist} which are precisely the fields of $\map(\CC^2_{\rm Dol}, B \fg_{\rm Hod})$.
\end{proof}

\subsubsection{Rank \texorpdfstring{$(1, 1)$}{(1,1)} Partially Topological Twist}
\label{sect:4d4partialtwist}

Next we consider the twist with respect to a rank $(1,1)$ supercharge. As in Section \ref{sect:4d_2_11} the twist is compatible with the group $G = \Spin(2, \RR) \times \Spin(2, \RR)$. However, we will use a different twisting homomorphism. We denote each factor by $\Spin(2, \RR)_i$, $i=1,2$, and by $S_{\pm, i}$ the semi-spin representations of the factor $\Spin(2, \RR)_i$.

The twisting homomorphism is
\[
\phi \colon \Spin(2, \RR)_1 \times \Spin(2, \RR)_2 \to G_R=\SL(W),
\]
under which $W$ splits as
\begin{equation}
W\cong (S_{+, 1}\otimes S_{+, 2}\oplus S_{-, 1}\otimes S_{-, 2})\oplus (S_{+, 1}\otimes S_{-, 2}\oplus S_{-, 1}\otimes S_{+, 2}).
\label{eq:4d411decomposition}
\end{equation}

In this case $S_+\otimes W$ and $S_-\otimes W^*$ have two-dimensional trivial $G$-subrepresentations. Any scalar rank $(1, 1)$ supercharge is square-zero. We choose a homomorphism $\alpha\colon \U(1)\rightarrow G_R$ under which the first two summands in \eqref{eq:4d411decomposition} have weight $1$ and the last two summands have weight $-1$. This makes the chosen rank $(1, 1)$ supercharge have weight $1$.

\begin{theorem}
The rank (1, 1) twist of 4d $\cN=4$ super Yang--Mills on $M=\RR^4$ is perturbatively equivalent to the generalized BF theory with the space of fields $T^*[-1] \map(\CC_{\rm Dol}\times \RR^2_{\mathrm{dR}}, B \fg)$. Moreover, the equivalence is $\Spin(2, \RR)\times \Spin(2, \RR)$-equivariant.
\label{thm:4d4partialtwist}
\end{theorem}
\begin{proof}
The 4d $\cN=4$ super Yang--Mills theory viewed as a $\cN=2$ theory is 4d $\cN=2$ Yang--Mills theory with matter valued in $U=T^*\fg$. By Theorem \ref{thm:4d_11_twist} we obtain that the twist is equivalent to $T^*[-1]\map(\CC\times \RR^2_{\mathrm{dR}}, \fg/\fg)$ as a $\ZZ$-graded theory. Let us now analyze the $\Spin(2, \RR)\times \Spin(2, \RR)$-action.

By construction the twisting homomorphism $\Spin(2,\RR)\times \Spin(2, \RR)\rightarrow G_R$ defined by \eqref{eq:4d411decomposition} factors as $\Spin(2, \RR)\times \Spin(2, \RR)\subset \MU(2)\rightarrow G_R$, where the latter map is the twisting homomorphism used in Section \ref{sect:4d4holomorphictwist}. Therefore, we have to restrict the fields used in that section to $\Spin(2, \RR)\times \Spin(2, \RR)$. But by Theorem \ref{thm:4d4holomorphictwist} the fields belong to $T^*[-1]\map(\CC_{\rm Dol}\times \CC_{\rm Dol}, B\fg)$ whose underlying graded $\Spin(2, \RR)\times \Spin(2, \RR)$-equivariant bundle coincides with that of $T^*[-1]\map(\CC_{\rm Dol}\times \RR^2_{\mathrm{dR}}, B\fg)$.
\end{proof}

\subsubsection{Rank \texorpdfstring{$(2,1)$}{(2,1)} Topological Twist}

Next we look at the case of the twist by a rank $(2,1)$ supercharge. 
As in many cases so far, it will be useful to consider a family of supercharges which are generically of rank $(2,1)$. Consider the twisting homomorphism $\phi\colon \Spin(2; \RR)\times \Spin(2; \RR)\to G_R$ from Section \ref{sect:4d4partialtwist}.

Consider a family of scalar square-zero supercharges
\begin{equation}
Q_{s, t} = Q_0 + t Q_1 + sQ_2,
\label{eq:4d412Hodgefamily}
\end{equation}
where $Q_0,Q_1$ are rank $(1, 0)$ supercharges and $Q_2$ is a rank $(0, 1)$ supercharge, so that $Q_0 + t Q_1$ is the family of square-zero supercharges from Section \ref{sect:4d4Atwist} and $Q_0 + Q_2$ is the rank $(1, 1)$ square-zero supercharge from Section \ref{sect:4d4partialtwist}. We will now calculate the twist with respect to the family $Q_0 + tQ_1 + Q_2$.

\begin{theorem}
The twist of the 4d $\cN=4$ super Yang--Mills theory with respect to the family $Q_{1, t}$ of square-zero supercharges of Equation \eqref{eq:4d412Hodgefamily} is perturbatively equivalent to the theory $\map\left(\CC_{\rm Dol} \times \RR^2_{\mr{dR}}, B\fg_{\Hod}\right)$. Moreover, this equivalence is $\spin(2;\RR) \times \spin(2;\RR)$-equivariant.
\end{theorem}
\begin{proof}
The idea of the proof will be to eliminate fields as in the proof of Theorem \ref{thm:4dDonaldsontwist}, but keeping track of the $s$-dependence.

Let $S(s,t)$ be the action functional of the 4d $\cN=4$ super Yang--Mills theory twisted by $Q_0 + tQ_1 + sQ_2$. Then we have
\begin{equation}\label{eqn:(2,1)2}
S(s, t) = S(s, 0) + \left(S^{(1)}(tQ_1) + 2 S^{(2)}(Q_0 + tQ_1) + S^{(2)}(tQ_1)\right) + 2S^{(2)}(tQ_1 + sQ_2),
\end{equation}
where $S^{(1)} = S^{(1)}_{\rm gauge}, S^{(2)} = S^{(2)}_{\rm gauge}$ are the functionals from Section \ref{sect:gaugemultipletSUSY} encoding the infinitesimal actions by supersymmetry. Here the middle three terms in the parentheses are the $t$-dependent terms in Proposition \ref{4d_donaldson_susy_prop} and $S^{(2)}(tQ_1 + sQ_2)$ is proportional to $\tilde{\chi}^*$.

We will now repeat the simplifications performed in the proof of Theorem \ref{thm:4dDonaldsontwist}. We first perform a change of variables sending $\{\chi^*, \chi, \wt\chi^*, \wt\chi\}$ to $\{\chi^*-t\wt\chi^*, \chi, \wt\chi^*, \wt\chi + t\chi\}$. Then we perform the following field eliminations:
\begin{itemize} 
\item Using Proposition \ref{prop:integrateoutfield} we set $\chi = 0$ and $\chi^* - t\wt\chi^*$ to a certain value.

\item Using Proposition \ref{prop:BRSTdoublet} we set $\rho^* = 0$, $A_{1, 0} = 0$ and $\rho, A_{1, 0}^*$ to certain values.

\item Using Proposition \ref{prop:BRSTdoublet} we set $\Tilde{a} = 0$, $\tilde{\chi}^*=0$ and $\Tilde{\chi}+t\chi, \Tilde{a}^*$ to certain values.

\item Using Proposition \ref{prop:BRSTdoublet} we set $\Tilde{\nu}^* = 0$, $\nu = 0$ and $\Tilde{\nu},\nu^*$ to certain values.
\end{itemize}

The last term $S^{(2)}(tQ_1 + sQ_2)$ in Equation \eqref{eqn:(2,1)2} is proportional to $\Tilde{\chi}^*$, therefore it disappears upon applying the third step above. 
Applying all the remaining steps, the first term $S(s,0)$ becomes the action functional of the $(1,1)$ twist upon setting $s=1$, see Theorem \ref{thm:4d4partialtwist}. 
Finally, the term in parentheses in Equation \eqref{eqn:(2,1)2} agrees with the $t$-dependent terms in Equation \eqref{eqn:(2,0)}, (where now the fields are adjoint valued). 
We have already seen that the $t$-dependent terms give rise to the desired Hodge family, so this completes the proof. 
\end{proof}

\subsubsection{Rank \texorpdfstring{$(2, 2)$}{(2,2)} Topological Twist}
\label{sect:4dqgltwist}

Consider a rank $(2, 2)$ supercharge $Q\in S_+\otimes W\oplus S_-\otimes W^*$. It defines embeddings $S_+^*\hookrightarrow W$ and $S_-^*\hookrightarrow W^*$ and the square-zero condition is that their images pair to zero. In other words, we have a short exact sequence
\[0\longrightarrow S_+^*\longrightarrow W\longrightarrow S_-\longrightarrow 0.\]
The semi-spin representations $S_\pm$ carry volume forms induced by scalar spinorial pairings. Moreover, $W$ has a canonical volume form since it is the semi-spin representation of $\Spin(6, \CC)\cong \SL(4, \CC)$. Comparing these volume forms under the above exact sequence gives an invariant $s\in\CC^\times$ of a rank $(2, 2)$ square-zero supercharge. Moreover, $\Spin(6, \CC)$-orbits of rank $(2, 2)$ square-zero supercharges are parameterized by this invariant.

Let $N_\RR=\RR^4$ equipped with a Euclidean structure and $N=N_\RR\otimes_\RR\CC$ its complexification. We consider the 8-dimensional Euclidean vector space $N$ which carries a complex half-density. By the results of Section \ref{sect:8dtopologicaltwist} we obtain a family $Q_t$ of 8d square-zero supercharges. Its dimensional reduction to 4 dimensions also gives a family of 4d square-zero supercharges. Then from Theorem \ref{thm:8dHodgereduction} we obtain the following statement.

\begin{theorem}
The twist of 4d $\cN=4$ super Yang--Mills with respect to the family $Q_t$ of square-zero supercharges is perturbatively equivalent to the topological Hodge theory $\map(\RR^4_{\mathrm{dR}}, B\fg_{\Hod})$.
\label{thm:4d422Hodgetwist}
\end{theorem}

Let us now rewrite the family $Q_t$ in 4-dimensional terms. Consider the Kapustin--Witten twisting homomorphism $\phi_{\mathrm{KW}}\colon \Spin(4, \RR)\subset \Spin(4, \RR)\times \Spin(2, \RR)\subset \Spin(6, \CC)$ under which $W$ decomposes as
\[W\cong S_+\oplus S_-.\]
In this case the spinorial representation becomes
\[\Sigma\cong (S_+\otimes S_+)\oplus (S_+\otimes S_-)\oplus (S_-\otimes S_+)\oplus (S_-\otimes S_-).\]
In particular, there are two scalar supercharges $Q_+$ and $Q_-$ given by the volume forms on $S_+$ and $S_-$ respectively. We may then consider a family of supercharges
\[Q = u Q_+ + i v Q_-\]
for $u, v\in\CC$. If $u,v\neq 0$ we obtain a rank $(2, 2)$ supercharge. In this case the map $Q\colon S_+^*\cong S_+\rightarrow W$ is given by multiplication by $u$ and the map $Q\colon W\rightarrow S_-$ is given by multiplication by $iv$. Therefore, its $s$-invariant is
\[s = -\frac{u^2}{v^2}.\]

\begin{remark}
The family $uQ_+ + iv Q_-$ of square-zero supercharges is the same family studied by Kapustin and Witten, see \cite[Section 3.1]{KapustinWitten}.
\end{remark}

These supercharges are related to $Q_t$ as follows. Let $S_+^8$ be the semi-spin representation of $\Spin(8, \CC)$ and $S_\pm$ the semi-spin representations of $\Spin(4, \CC)$ as before. Under the embedding
\[\Spin(4, \CC)\subset \Spin(4, \CC)\times \Spin(4, \CC)\subset \Spin(8, \CC)\]
$S_+^8$ splits as
\[S_+^8\cong (S_+\otimes S_+)\oplus (S_-\otimes S_-),\]
so $Q_+, Q_-\in S_+^8$. We then obtain
\[Q_0 = Q_+ + Q_-,\qquad \Bar{Q}_0 = Q_+ - Q_-.\]
Therefore, the $s$-invariant of the family $Q_t$ is
\begin{equation}
s = \frac{(1+t)^2}{(1-t)^2}.
\label{eq:4d4sinvariant}
\end{equation}

\begin{corollary}
The rank $(2, 2)$ twist of 4d $\cN=4$ super Yang--Mills for $s=1$ is perturbatively equivalent to the topological BF theory $T^*[-1]\map(\RR^4_{\mathrm{dR}}, B\fg)$.
\end{corollary}
\begin{proof}
The supercharge $Q_0$ has $s$-invariant $s=1$. By Theorem \ref{thm:4d422Hodgetwist} the twist by $Q_0$ is perturbatively equivalent to the specialization of the theory $\map(\RR^4_{\mathrm{dR}}, B\fg_{\Hod})$ at $t=0$. By Proposition \ref{prop:Hodgetheoryspecialization} the latter is isomorphic to the topological BF theory $T^*[-1]\map(\RR^4_{\mathrm{dR}}, B\fg)$.
\end{proof}

\begin{corollary}
The rank $(2, 2)$ twist of 4d $\cN=4$ super Yang--Mills for $s\neq 1$ is perturbatively trivial.
\end{corollary}
\begin{proof}
For any $s\neq 1$ we may find $t\neq 0$ solving \eqref{eq:4d4sinvariant}. But by Proposition \ref{prop:Hodgetheoryspecialization} the specialization of the topological Hodge theory $\map(\RR^4_{\mathrm{dR}}, B\fg_{\Hod})$ at $t\neq 0$ is perturbatively trivial.
\end{proof}

\section{Dimension 3}
The odd part of the $3$-dimensional supersymmetry algebra is
\[
\Sigma \iso S \otimes W, 
\]
where $S$ is the two-dimensional complex spin representation of $\Spin(3;\CC)$, and where $W$ is a vector space equipped with a bilinear pairing.  

The maximal supersymmetric gauge theory has $\mc N= \dim(W) = 8$.  
There are $\mc N=4$ super Yang--Mills theories for every choice $U$ of a complex symplectic representation of the gauge group, and $\mc N=2$ super Yang--Mills theories for every choice $R$ of arbitrary complex representation of the gauge group. 
Finally, there are $\cN=1$ super Yang--Mills theories in 3 dimensions, but there are no square-zero supercharges with that amount of supersymmetry. In dimension $3$, much like we saw in dimensions $5$ and $7$, all twisted theories can be obtained by dimensional reduction from theories one dimension higher.

\subsection{\texorpdfstring{$\cN = 2$}{N=2} Super Yang--Mills Theory}
\label{sect:3d_2_section}

Fix a gauge group $G$ and a representation $R$.  The 3d $\mc N=2$ super Yang--Mills theory arises by dimensional reduction from 4d $\mc N=1$ super Yang--Mills theory with an $R$-valued chiral multiplet.  In this case, $W = \CC^2$ equipped with a nondegenerate symmetric bilinear pairing. The R-symmetry group is $G_R=\CC^\times$, acting on $W$ with weights $1$ and $-1$.

This theory admits a unique twist up to equivalence:
\begin{itemize}
 \item A square zero supercharge $Q \ne 0 \in \Sigma$ has two invariant directions.  There is a twisting homomorphism $\phi = \MU(1)\cong \U(1)\hookrightarrow G_R$, so the twisted theory carries a $U(1)$-action. The twist is $\ZZ$-graded.
\end{itemize}

\subsubsection{Minimal Twist}
\label{sect:3dminimaltwist}
A square-zero supercharge $Q$ has rank 1, i.e. $Q = q \otimes w$ for some $w \in W$. We use the twisting homomorphism $\phi\colon \MU(1)\xrightarrow{\det^{1/2}} \U(1)\hookrightarrow G_R$.

As in Section \ref{sect:9dminimaltwist}, the specification of $q$ is equivalent to the choice of a one-dimensional subspace $N_\RR\subset V_\RR$ and a complex structure on $V_\RR/N_\RR$ together with a complex half-density.  Note that in one dimension the choice of a complex half-density is equivalent to a choice of spin structure.

\begin{theorem} \label{3d_minimal_twist_thm}
The minimal twist of the 3d $\cN=2$ super Yang--Mills theory with Lie algebra $\fg$ with matter valued in a $\fg$-representation $R$ is perturbatively equivalent to the generalized BF theory with the space of fields $T^*[-1]\mr{Map}(\CC \times \RR_{\mr{dR}}, R/\gg)$. Moreover, this equivalence is $\mr U(1)$-equivariant.
\end{theorem}

\begin{proof}
By Theorem \ref{4d_minimal_twist_thm} the twist of 4d $\cN=1$ super Yang--Mills on $L\times N$ by $Q$ is perturbatively equivalent to the holomorphic BF theory with the space of fields $T^*[-1]\map(L\times N, R / \fg)$. By Proposition \ref{prop:BFholomorphicreduction} we get that the dimensional reduction of the holomorphic BF theory on $L\times N$ along $\Re\colon N\rightarrow N_\RR$ is isomorphic to the holomorphic BF theory with the space of fields $T^*[-1]\map(L\times N_\RR, R/\fg)$.
\end{proof}

\subsection{\texorpdfstring{$\cN = 4$}{N=4} Super Yang--Mills Theory} \label{3d_4_section}
Next, consider the 3d $\mc N=4$ supersymmetric Yang--Mills theory with matter valued in a symplectic $G$-representation $U$. The $R$-symmetry group is $G_R = \Spin(4;\CC)$, acting on $W$ by the vector representation.

In the $\mc N=4$ supersymmetry algebra there are now three non-trivial orbits of square-zero supercharges.  An element $Q \in S \otimes W$ gives rise to a map $S^* \to W$; $Q$ squares to zero if its image is totally isotropic.  The classification of orbits includes the rank of this map.
\begin{itemize}
 \item Rank 1.  
 In this case $Q$ is minimal, with 2 invariant directions.  
 Such supercharges lie in a subalgebra isomorphic to the $\mc N=2$ supersymmetry algebra and are unique up to equivalence. They admit a twisting homomorphism and a $\ZZ$-grading.
 \item Rank 2. Such supercharges are topological.  
A rank 2 supercharge defines a Lagrangian subspace of $W$, and therefore an orientation.  
The $G_R = \Spin(4;\CC)$-action factors through an $\SO(4)$ action on $W$, and so preserves orientation, so there are two $G_R$ orbits corresponding to the two choices of orientation. We refer to these as the A twist and the B twist, distinguished by whether they admit a $\ZZ$-grading:
\begin{enumerate}
\item
An A-twist supercharge admits a twisting homomorphism $\phi \colon \U(1) \to G_R$ and a $\ZZ$-grading $\alpha \colon \U(1) \to G_R$.
\item
A B-twist supercharge admits the diagonal twisting homomorphism $\phi' \colon \SU(2) \to \SU(2) \times \SU(2) \to G_R$.  
This twist is only $\ZZ/2\ZZ$-graded.
\end{enumerate}
\end{itemize}

The distinction via twisting homomorphisms and $\ZZ$-gradings follows by identifying the twists as dimensional reductions from 4d $\mc N=2$.
 
\begin{lemma} \label{3d_4_orbits_lemma}
A rank $(2,0)$ square-zero supercharge in the 4d $\mc N=2$ supersymmetry algebra restricts to an A-twisting supercharge in 3d $\mc N=4$.  Likewise, a rank $(1,1)$ square zero supercharge in 4d $\mc N=2$ restricts to a B-twisting supercharge in 3d $\mc N=4$.
\end{lemma}

\begin{proof}
Let $W_4$ be the complex two-dimensional auxiliary space of the 4d $\cN=2$ supersymmetry algebra.
The projection from the 4d $\mc N=2$ supertranslation algebra to the 3d $\mc N=4$ supertranslation algebra induces an isomorphism $W_4 \oplus W_4^* \to W$ of representations of the group $\spin(3;\CC)$.  
This splits the fundamental representation $W$ into the sum of two Lagrangians, defining an orientation on $W$. 
A rank $(2,0)$ supercharge induces the Lagrangian subspace $W_4 \sub W$, which is oriented.  This supercharge admits a compatible homomorphism $\alpha \colon \U(1) \to G_R = \SL(2;\CC) \times \SL(2;\CC)$ given by the embedding into the second factor. This is the A-twist in our classification above.

A rank $(1,1)$ supercharge induces a Lagrangian subspace of the form $L \oplus L^* \sub W$, where $L$ is a 1-dimensional subspace of $W_4$.  This subspace has the opposite orientation, so corresponds to the B-twist in our classification above.  
\end{proof}

\subsubsection{Minimal Twist}
\label{sect:3d_4_minimal_twist}
There is a unique twisting homomorphism $\phi \colon \MU(1) \to \Spin(4;\CC)$ given by the restriction of the 4d $\cN=2$ twisting homomorphism for the minimal twist as in Theorem \ref{thm:4d2holomorphictwist} to the subgroup $\MU(1) \subset \MU(2)$. 

To incorporate the $\ZZ$-grading, we use the homomorphism $\alpha \colon \mr U(1) \to \SU(2)_+ \times \SU(2)_- \iso G_R$. This homomorphism coincides with the dimensional reduction of $\alpha_0$ from Section \ref{sect:4d_2_holomorphictwist}. The corresponding $\U(1)$-action on the 3d auxiliary space $W$ under $\alpha$ coincides with the $\U(1)$-action on the 4d auxiliary space $W_4$ under the isomorphism $W\cong W_4\oplus W_4^*$. Indeed, both have weights $(1,1,-1,-1)$.  Therefore, our equivalence is compatible with the $\ZZ$-grading in Theorem \ref{thm:4d2holomorphictwist} induced by $\alpha_0$.

\begin{theorem} \label{3d_4_minimal_twist_thm}
The minimal twist of 3d $\cN=4$ super Yang--Mills on $\CC\times \RR$ is perturbatively equivalent to the generalized BF theory with the space of fields $T^*[-1]\Sect(\CC \times \RR_{\mr{dR}}, (U \otimes K^{1/2}_\CC) \ham \gg))$. Moreover, the equivalence is $\U(1)$-equivariant.
\end{theorem}
\begin{proof}
The statement follows by applying the dimensional reduction (Theorem \ref{prop:BFholomorphicreduction}) to Theorem \ref{thm:4d2holomorphictwist} calculating the holomorphic twist of the 4d $\mc N=2$ super Yang--Mills theory on $\CC\times\CC$, where we dimensionally reduce along the projection $\Re\colon\CC\rightarrow \RR$ in the second factor.
\end{proof}

\subsubsection{Topological A-Twist}
\label{sect:3d_4_A_twist}
Let $L=\CC$ equipped with a Hermitian structure and a complex half density, $N_\RR=\RR$ equipped with a Euclidean structure and $N=N_\RR\otimes_\RR\CC$ its complexification. Consider the 4-dimensional spacetime $V_\RR=L\times N$. Under the projection $\Re\colon N\rightarrow N_\RR$ the family $Q_t$ of 4-dimensional square-zero supercharges given by \eqref{eq:4dHodgefamily} dimensionally reduces to a family of 3-dimensional square-zero supercharges which at $t\neq 0$ are topological at and $t=0$ have 2 invariant directions. Since they admit a compatible $\ZZ$-grading, at $t=0$ we obtain the holomorphic twist and at $t\neq 0$ we obtain the topological A-twist. Therefore, from Theorem \ref{thm:4dDonaldsontwist} we obtain the following statement.

\begin{theorem}
The twist of the 3d $\cN=4$ super Yang--Mills theory with respect to the family $Q_t$ of square-zero supercharges is perturbatively equivalent to the generalized Hodge theory $\mr{Sect}(\CC \times \RR_{\mr{dR}}, ((U \otimes K^{1/2}_\CC)\ham \gg)_{\mr{Hod}})$. Moreover, this equivalence is $\U(1)$-equivariant.
\label{3d_4_A_twist_thm}
\end{theorem}

\begin{corollary}
The topological A-twist of the 3d $\cN=4$ super Yang--Mills theory is perturbatively trivial.
\end{corollary}

\subsubsection{Topological B-Twist}
\label{sect:3d_4_B_twist}
We consider $V_\RR=\RR^3$ equipped with a Euclidean structure and as before let $V=V_\RR\otimes_\RR\CC$ be its complexification. $V$ carries a Hermitian structure and a half-density, so by the results of Section \ref{sect:6dholomorphictwist} we obtain a square-zero supercharge $Q$. Under the projection $\Re\colon V\rightarrow V_\RR$ the supercharge $Q$ dimensionally reduces to the topological B-supercharge in 3 dimensions. Therefore, from Theorem \ref{thm:6dCSreduction} we obtain the following statement.

\begin{theorem} \label{3d_4_B_twist_thm}
The rank $2$ B-twist of the 3d $\mc N=4$ super Yang--Mills theory is perturbatively equivalent to the generalized Chern--Simons theory with the space of fields $\map(\RR^3_{\mr{dR}}, U\ham \gg)$. Moreover, the equivalence is $\Spin(3; \RR)$-equivariant.
\end{theorem}

\subsection{\texorpdfstring{$\cN = 8$}{N=8} Super Yang--Mills Theory} \label{3d8section}
Finally, consider the maximally supersymmetric Yang--Mills theory in dimension 3.
This theory is the dimensional reduction of 10d $\cN = (1,0)$ supersymmetric Yang--Mills theory.  
The $R$-symmetry group is $G_R = \Spin(7;\CC)$, where $W$ is the $8$-dimensional spin representation.

In the $\mc N=8$ supersymmetry algebra the classification of twists is the same as we saw in 3d $\mc N=4$.  
There are three orbits: one consisting of rank 1 supercharges, and two orbits of rank 2 supercharges.  We can see this in the following way.
\begin{lemma}
There are two distinct $\spin(3;\CC) \times G_R$-orbits of square-zero supercharges of rank $2$ in the 3d $\mc N=8$ supersymmetry algebra: the generic orbit and the special orbit.
\label{lm:3dN8twoorbits}
\end{lemma}

\begin{proof}
Choose a symplectic basis $\langle s, s' \rangle$ for $S$, and let $Q = s \otimes w + s' \otimes w'$ be a square-zero rank 2 element of $S \otimes W$.  Let $V_7$ denote the fundamental representation of $G_R$.  As in Section \ref{9d_section}, the element $w \in W$ is equivalent to the data of a maximal isotropic subspace $L \sub V_7$, together with a choice of a half-density.  This element $w$ is stabilized, in particular, by a copy of the metalinear group $\ML(L)$.
Under the group $\ML(L)$ the auxiliary space $W$ decomposes as
\[W \iso \left(\CC \oplus L \oplus \wedge^2 L \oplus \wedge^3 L \right) \otimes \det(L)^{-1/2},\]
(see also \cite[Section 4.7]{ElliottSafronov}), with $w$ lying in the last summand.  Under this decomposition, split the remaining element $w' = (v_0, v_1, v_2, v_3)$.  Then
\begin{itemize}
 \item $v_0 = 0$, because the square-zero condition implies, in particular, that $(w,w') = 0$ with respect to the scalar spinor pairing on $W$, here given by the wedge pairing.
 \item Without loss of generality $v_3 = 0$, since under the action of $\spin(3;\CC)$, $s \otimes w + s' \otimes w' \sim s \otimes w + s' \otimes (w'-w)$.
 \item If $v_1 = 0$ then $v_2 \ne 0$, and all choices of non-zero $v_2$ are in the same orbit under $\SL(L) \sub \stab(w) \sub G_R$.
 \item If $v_1 \ne 0$ then without loss of generality $v_2 = 0$, using the action by wedge product of $L \sub \stab(w) \sub G_R$.  Finally all choices of non-zero $v_1$ are likewise in the same orbit under $\SL(L) \sub \stab(w) \sub G_R$.  The stabilizer of $w$ acts on the space $v_1 \ne 0$, so these latter two cases comprise two inequivalent orbits.
\end{itemize}
\end{proof}

The classification of twists therefore takes the following form.
\begin{itemize}
 \item Rank 1.  In this case $Q$ is minimal, with 2 invariant directions.  Such supercharges come from the $\mc N=2$ supersymmetry algebra.  They admit a twisting homomorphism from $\mr U(1)$ and a $\ZZ$-grading.
 \item Rank 2 twists. These twists are topological, and come from the $\mc N=4$ supersymmetry algebra. They admit a twisting homomorphism from $\Spin(3, \RR)$ and a $\ZZ$-grading. There are two such:
 \begin{enumerate}
 \item A-twist (the generic rank 2 orbit).
 \item B-twist (the special rank 2 orbit).
 \end{enumerate}
\end{itemize}

\subsubsection{Minimal Twist}
\label{sect:3d8minimal_twist}

The 3d $\cN=8$ supersymmetric Yang--Mills theory is obtained by a dimensional reduction of the 4d $\mc N=4$ supersymmetric Yang--Mills theory. Therefore, from Theorem \ref{thm:4d4holomorphictwist} we obtain the following statement.

\begin{theorem}  \label{3d_8_minimal_twist_thm}
The minimal twist of the 3d $\cN=8$ super Yang--Mills theory on $M=\CC\times \RR$ is perturbatively equivalent to the generalized BF theory with space of fields $T^*[-1]\mr{Map}(\CC_{\rm Dol} \times \RR_{\mr{dR}}, \fg/\fg)$. Moreover, the equivalence is $\U(1)$-equivariant.
\end{theorem}

\subsubsection{Topological Twists}

Consider the dimensional reduction of the family $Q_t$ of rank $(2, 2)$ square-zero supercharges in the 4d $\cN=4$ supersymmetry algebra from Section \ref{sect:4dqgltwist}. Since this is a family of topological supercharges in 4 dimensions, it dimensionally reduces to a family of topological supercharges in 3 dimensions.

\begin{theorem}
The twist of the 3d $\cN=8$ super Yang--Mills theory with respect to the family $Q_t$ of square-zero supercharges is perturbatively equivalent to the theory $T^*[-1]\map(\RR^3_{\mathrm{dR}}, B\fg_{\Hod})$. Moreover, the equivalence is $\Spin(3, \RR)$-equivariant.
\label{thm:3d8Hodgetwist}
\end{theorem}

\begin{proof}
The claim follows by dimensional reduction (Proposition \ref{CS_to_BF_diml_red_prop}) from the corresponding statement in 4 dimensions (Theorem \ref{thm:4d422Hodgetwist}) which calculates the twist of the 4d $\cN=4$ super Yang--Mills theory with respect to that family to be $\map(\RR^4_{\mathrm{dR}}, B\fg_{\Hod})$.
\end{proof}

As we see from Theorem \ref{thm:3d8Hodgetwist}, at $t\neq 0$ the twist is perturbatively trivial while at $t=0$ it is not. By Lemma \ref{lm:3dN8twoorbits} there are only two orbits of topological supercharges, so the case $t\neq 0$ (the generic orbit) must be the A-twist and the case $t=0$ (the special orbit) must be the B-twist.

\begin{corollary}
The topological A-twist of the 3d $\cN=8$ super Yang--Mills theory is perturbatively trivial.
\label{cor:3dN8Atwist}
\end{corollary}

\begin{corollary}
The topological B-twist of the 3d $\cN=8$ super Yang--Mills theory is perturbatively equivalent to $T^*[-1]\map(\RR^3_{\mathrm{dR}}, \fg/\fg)$. Moreover, the equivalence is $\Spin(3, \RR)$-equivariant.
\label{cor:3dN8Btwist}
\end{corollary}

\section{Dimension 2}
The odd part of the $2$-dimensional supersymmetry algebra is 
\[
\Sigma = S_+ \otimes W_+ \oplus S_- \otimes W_- ,
\]
where $W_\pm$ are complex vector spaces equipped with symmetric nondegenerate bilinear pairings. 
The complex semi-spin representations $S_\pm$ are 1-dimensional, where $\spin(2;\CC) \iso \CC^\times$ acts with weight $\pm \frac{1}{2}$. 
There is an independent pairing $\Gamma_\pm \colon S_\pm^{\otimes 2} \to V_2 \iso \CC_1 \oplus \CC_{-1}$ for each chirality, where $\Spin(2;\CC)$ acts on $\CC_1, \CC_{-1}$ with weights $1, -1$ respectively.

There are two classes of twisted supersymmetric gauge theory that we will consider in two dimensions.  First, we have theories with $(\mc N, \mc N)$ supersymmetry corresponding to $\cN = \dim(W_+) = \dim(W_-)$; these arise via dimensional reduction from supersymmetric gauge theories in higher dimensions. Namely, we have the 2d $\cN=(1, 1)$, 2d $\cN=(2, 2)$, 2d $\cN=(4, 4)$ and 2d $\cN=(8, 8)$ super Yang--Mills theories. 
The 2d $\cN=(1, 1)$ supersymmetry algebra does not admit square-zero supercharges, so we will not consider it. We additionally have gauge theories with chiral, i.e. $(\cN_+, 0)$ supersymmetry, which we have constructed in Sections \ref{sect:2dchiral} and \ref{sect:SYM}. 
We will address twists for these two classes of theories in turn.

\subsection{\texorpdfstring{$\cN=(2,2)$}{N=(2,2)} Super Yang--Mills Theory} \label{sect:2d(2,2)}
First, consider the 2d $\mc N=(2,2)$ supersymmetric Yang--Mills theory.  The R-symmetry group is $$G_R = \ZZ/2\ZZ \ltimes (\spin(2;\CC) \times \spin(2;\CC)) \iso \ZZ/2 \ltimes ( \CC^\times \times \CC^\times)$$ with the $\Spin(2; \CC)$-factors acting by their vector representation on $W_+ \cong \CC^2$ and $W_- \cong \CC^2$ respectively, and with $\ZZ/2\ZZ$ acting on both $W_+$ and $W_-$ by $(a,b) \mapsto (-a, -b)$. 

\vspace{-10pt}
\paragraph{Fields:} We can describe the BRST fields of $\mc N=(2,2)$ super Yang--Mills by restricting the fields in dimension 3 from Section \ref{sect:3d_2_section}, or equivalently the 4d fields from Section \ref{sect:4d_1_section}, to representations of the group $\spin(2;\CC)$.  
In any case, the fields we obtain are
\begin{itemize}
 \item gauge bosons $A \in \Omega^1(\RR^2; \gg)$, and a pair of scalars $(a, \Tilde{a}) \in \Omega^0(\RR^2; \gg \oplus \gg)$.
 \item matter bosons $(\Bar{\phi}, \phi) \in \Omega^0(\RR^2 ; R \oplus R^*)$.
 \item gauge fermions $(\lambda_+ \otimes u_+,\lambda_- \otimes u_-) \in \Omega^0(\RR^2 ; \Pi(S_+ \otimes W_+ \oplus S_- \otimes W_-) \otimes \gg)$.
 \item matter fermions $(\psi^+_+,\psi^+_-, \psi^-_+, \psi^-_-) \in \Omega^0(\RR^2 ; \Pi(S_+ \oplus S_-) \otimes R \oplus \Pi(S_+ \oplus S_-) \otimes R^*)$.
 \item a ghost field $c \in \Omega^0(\RR^2; \gg)[1]$.
\end{itemize}

In the $\mc N=(2,2)$ supersymmetry algebra there are three classes of non-trivial orbits of square-zero supercharges.  
\begin{itemize}
 \item Square-zero supercharges of rank $(1,0)$ or $(0,1)$, which are holomorphic.  
There is a compatible twisting homomorphism from $\mr U(1)$, and a compatible twisting datum $\alpha \colon \mr U(1) \to \spin(2;\RR)$ acting with weight $1$ on $S_+$ and weight $-1$ acting on $S_-$.  
 \item Square-zero supercharges of rank $(1,1)$ are topological, and split into four orbits under the action of $\spin(2;\CC) \times (\CC^\times \times \CC^\times)$.  Indeed, we can identify a square-zero supercharge of rank $(1,1)$ as a pair of vectors~
 
\noindent$((\lambda, \pm i \lambda), (\mu, \pm i \mu)) \in W_+ \oplus W_-$ with $\lambda, \mu \ne 0$.  By acting by $G_R$ we can set $\lambda = \mu = 1$, leaving four orbits, represented by the supercharges $Q_A = ((1,i),(1,i)), Q_A^\dagger = ((1,-i),(1,-i)), Q_B = ((1,i),(1,-i))$ and $Q_B^\dagger = ((1,-i),(1,i))$.  The $\ZZ/2\ZZ$-action swaps the two A supercharges and the two B supercharges, leaving two orbits under $\spin(2;\CC) \times G_R$.  
\begin{enumerate}
\item The A-twist is compatible with the twisting homomorphism $\phi_{A} \colon \spin(2;\CC) \to \spin(2;\CC) \times \spin(2;\CC)$ with weights $(1,1)$.
\item The B-twist is compatible with the twisting homomorphism $\phi_{B} \colon \spin(2;\CC) \to \spin(2;\CC) \times \spin(2;\CC)$ with weights $(1,-1)$.  
\end{enumerate}
Moreover, the A-supercharges admit a compatible homomorphism $\alpha_A = \phi_B\colon \U(1)\rightarrow G_R$, and the B-supercharges admit a compatible homomorphism $\alpha_B = \phi_A\colon \U(1)\rightarrow G_R$.
\end{itemize}

The calculation of the twists here is similar to what we saw in 4d $\mc N=2$ supersymmetry.  
The holomorphic twist and the B-twist arise by a dimensional reduction from twists of the 3d $\cN=2$ supersymmetric Yang--Mills theory.
On the contrary, the A-twist as a deformation of the holomorphic twist does not arise as a dimensional reduction. 

\subsubsection{Holomorphic Twist} \label{sect:2d22minimaltwist}

First, we record the holomorphic twist of the 2d $\cN=(2,2)$ supersymmetric Yang--Mills theory. The holomorphic twist is $\ZZ$-graded using $\alpha = \phi_A$.

\begin{theorem} \label{2d_minimal_twist_thm}
The holomorphic twist of the 2d $\cN=(2,2)$ super Yang--Mills theory with matter valued in a complex $\fg$-representation $R$ is perturbatively equivalent to the holomorphic BF theory with the space of fields $T^*[-1]\mr{Map}(\CC, T[1](R/\gg))$. Moreover, this equivalence is $\mr U(1)$-equivariant.
\end{theorem}
\begin{proof}
The statement follows by applying dimensional reduction (Proposition \ref{prop:BFdeRhamreduction}) to Theorem \ref{3d_minimal_twist_thm} calculating the minimal twist of 3d $\cN=2$ super Yang--Mills on $\CC\times \RR_{\mathrm{dR}}$, where we dimensionally reduce along the projection $\RR\rightarrow \pt$.
\end{proof}

\subsubsection{Topological A-Twist} \label{sect:2d22Atwist}

To deform the holomorphic twist to the topological A-twist, obtaining a Hodge deformation, we use similar techniques to those of Section \ref{sect:4d2Donaldson}.  We first analyze the supersymmetry action.  Consider the 1-parameter family of supercharges
\begin{equation} \label{eq:2dHodgefamily}
Q_t = Q_0 + tQ',
\end{equation}
where $Q_0 = (1,i) \in S_+$ is a holomorphic supercharge, and $Q' = (1,i) \in S_-$.  This family of supercharges is compatible with the twisting homomorphism $\phi_A$ -- the map $\U(1) \to \spin(2;\CC) \times \spin(2;\CC)$ with weights $(1,1)$. They admit a compatible homomorphism $\alpha=\phi_B\colon \U(1)\rightarrow G_R$ with weights $(1,-1)$.  Let us first decompose our fields according to the twisting homomorphism $\phi_A$:

\begin{itemize}
\item gauge bosons $A_{1,0} \in \Omega^{1,0}(\CC; \fg)$, $A_{0,1} \in \Omega^{0,1} (\CC ; \fg)$, $a \in \Omega^0(\CC ; \fg)[2]$, $\Tilde{a} \in \Omega^0(\CC ; \fg)[2]$. 
\item gauge fermions $\lambda_0 \in \Omega^{0}(\CC ; \fg)[-1]$, $\chi_{0,1} \in \Omega^{0,1}(\CC ; \fg) [1]$, $\lambda_{1,0} \in \Omega^{1,0}(\CC ; \fg)[1]$, $\Tilde{\lambda}_0 \in \Omega^0(\CC ; \fg) [3]$.
 \item matter bosons $\gamma_0 \in \Omega^0(\CC ; R)$, $\phi \in \Omega^0(\CC ; R^*)$.
 \item matter fermions $\gamma_{0,1} \in \Omega^{0,1}(\CC ; R)[-1]$, $\Tilde{\psi}_0 \in \Omega^0(\CC ; R)[1]$, $\beta_{1,0} \in \Omega^{1,0}(\CC ; R^*)[-1]$, $\psi \in \Omega^0(\CC ; R^*)[1]$; 
 \item a ghost field $c \in \Omega^0(\RR^2; \fg)[1]$.
\end{itemize}

\begin{prop}
Suppose $Q_t$ is the rank $(1,1)$ supercharge of \eqref{eq:2dHodgefamily}.
The $\U(1)$ decomposition of the functionals $S_{\rm gauge}^{(1)}, S^{(1)}_{\rm matter}, S^{(2)}_{\rm gauge}, S^{(2)}_{\rm matter}$ (see \eqref{eq:gaugeI1},  \eqref{eq:gaugeI2}, \eqref{eq:matterI1}, \eqref{eq:matterI2}) in terms of the fields of 2d $\cN=(2,2)$ super Yang--Mills theory are
\begin{align*}
S_{\rm gauge}^{(1)}(Q_t) &= \int \dvol \left(-(\lambda_{1,0}, A_{1,0}^*) - (\Tilde{\lambda}_0, \Tilde{a}^*) + \frac 12(F_{1,1} \lambda_0^* + (\ol \dd_{A_{0,1}}a, \chi_{0,1}^*)) \right)  \\
&+  \int \dvol t \left( -(\chi_{0,1}, A_{0,1}^*) - (\lambda_0, \Tilde{a}^*) + \frac 12(F_{1,1} \Tilde{\lambda}_0^*+ (\partial_{A_{1,0}} \Tilde{a}, \lambda_{1,0}^*) )   \right)  \\
S^{(2)}_{\rm gauge} (Q_t) &= \int \dvol \left(t (\lambda_0^*, \Tilde{\lambda}_0^*) - \frac 14(\lambda_0^* + t \Tilde{\lambda}_0^*)^2 + t(a, c^*)\right)  \\
S_{\rm matter}^{(1)}(Q_t) &= \int \dvol \left(\left( \psi_0 \phi^* + \frac 12 (\ol \dd_{A_{0,1}} \gamma_0, \gamma_{0,1}^*) \right) +t\left( (\Tilde{\psi}_0, \gamma_0^*) + \frac 12 (\dd_{A_{1,0}}  \phi, \beta_{1,0}^*)  \right)\right)  \\
S_{\rm matter}^{(2)}(Q_t) &= \int \dvol \frac t2 \left(\beta_{1,0}^*, \gamma_{0,1}^*\right).
\end{align*}
\end{prop}

\begin{theorem} \label{2d_2_A_twist_thm}
The twist of the 2d $\mc N=(2,2)$ supersymmetric Yang--Mills theory with matter valued in a complex $\fg$-representation $R$ with respect to the family of square-zero supercharges $Q_t$ is perturbatively equivalent to the Hodge family with the space of fields $T^*[-1]\mr{Map}(\CC, (R/\gg)_{\mr{Hod}})$.  
This equivalence is $\mr U(1)$-equivariant.
\end{theorem}

\begin{proof}
Observe that the quadruple of fields $\{\lambda_0^*, \lambda_0, \Tilde{\lambda}_0^*, \Tilde{\lambda}_0\}$ has the same Poisson brackets as the quadruple $\{\lambda_0^* - t \Tilde{\lambda}_0^*, \lambda_0, \lambda_0^*, \Tilde{\lambda}_0 + t \lambda_0\}$. 
Therefore, we may eliminate the fields $\lambda_0^*-t\wt\lambda_0^*, \chi$ using Proposition \ref{prop:integrateoutfield}. 
We then have trivial BRST doublets $\{\wt\lambda_0 + t\lambda_0, \wt a\}$, $\{\lambda_{1,0}, A_{1,0}\}, \{\psi_0, \phi\}$ which may be eliminated using Proposition \ref{prop:BRSTdoublet}. We are left with the action

\[
S_{\mr{BF}} + \int \dvol t \bigg( -(\chi_{0,1}, A_{0,1}^*) + (a, c^*) + (\Tilde{\psi}_0, \gamma_0^*) + \frac{1}{2} (\beta_{1,0}^*, \gamma_{0,1}^*) \bigg) 
\]
Here $S_{\mr{BF}}$ is the action functional of holomorphic BF theory as in the result of the minimal twist of $\cN=(2,2)$, see Theorem \ref{2d_minimal_twist_thm}, obtained by setting $t = 0$. 

We can identify the linearized BV complex with the following diagram:
\[\xymatrix@R=0.4cm@C=0.2cm{
\ul{-2} & \ul{-1} & \ul{0} & \ul{1} & \ul{2} & \ul{3} \\
&\Omega^0(\CC; \gg)_{c} \ar[r] &\Omega^{0,1}(\CC; \gg)_{A_{0,1}} & & \Omega^{1,0}(\CC; \gg^*)_{\chi_{0,1}^*} \ar[r] &\Omega^{1,1}(\CC; \gg^*)_{a^*} \\
\Omega^0(\CC; \gg)_{a} \ar[r] \ar@{-->}[ur] &\Omega^{0,1}(\CC; \gg)_{\chi_{0,1}} \ar@{-->}[ur] && \Omega^{1,0}(\CC; \gg^*)_{A_{0,1}^*} \ar[r] \ar@{-->}[ur] &\Omega^{1,1}(\CC; \gg^*)_{c^*} \ar@{-->}[ur]\\
&\Omega^0(\CC; R^*)_{\Tilde{\psi}} \ar[r] \ar@{-->}[dr] &\Omega^{0,1}(\CC; R^*)_{\beta_{1,0}^*} \ar@{-->}[dr] &&&\\
&& \Omega^0(\CC; R^*)_{\gamma_0} \ar[r]  &\Omega^{0,1}(\CC; R^*)_{\gamma_{0,1}} &&\\
&&\Omega^{1,0}(\CC; R)_{\gamma_{0,1}^*} \ar[r] \ar@{-->}[dr] &\Omega^{1,1}(\CC;  R)_{\gamma_0^*} \ar@{-->}[dr] &&\\
&&& \Omega^{1,0}(\CC; R)_{\beta_{1,0}} \ar[r]  &\Omega^{1,1}(\CC; R)_{\Tilde{\psi}_0^*}. &\\
}\]
Here, the solid arrows represent the linearized BV operator of the minimal twist, and the dotted arrows represent the $t$-dependent terms.
This is exactly the deformation to $T^*[-1]\mr{Map}(\CC, (R/\gg)_{\mr{Hod}})$.
\end{proof}

\begin{corollary}
The topological A-twist of the 2d $\cN=(2,2)$-supersymmetric Yang--Mills theory is perturbatively trivial.
\end{corollary}

\subsubsection{Topological B-Twist} \label{sect:2d22Btwist}

Finally, there is the B-twist of 2d $\cN=(2,2)$ supersymmetry.
This twist arises from the twist of 3d $\cN=2$ via dimensional reduction along the holomorphic direction.
Indeed, let $W_3$ be the two-dimensional auxiliary space in the 3d $\mc N=2$ supertranslation algebra.  
As a $\spin(2;\CC)$-representation, the 3d $\cN=2$ spinorial representation decomposes as $S_+ \otimes W_3 \oplus S_- \otimes W_3$.
Generically, rank one square zero elements in 3d $\mc N=2$ supersymmetry define rank $(1,1)$ square zero elements in the 2d $\mc N=(2,2)$ supersymmetry algebra.  
There is a square-zero supercharge in 3d $\mc N=2$ compatible with the identity twisting homomorphism $\spin(2;\CC) \to \spin(2;\CC)$.  This becomes the B-twisting homomorphism of $\spin(2;\CC) \to \spin(2;\CC) \times \spin(2;\CC)$ of weight $(1,-1)$ in the 2d $\mc N=(2,2)$-algebra.

\begin{theorem} \label{2d_2_B_twist_thm}
The topological B-twist of the 2d $\mc N=(2,2)$ super Yang--Mills theory with matter valued in a $\fg$-representation $R$ is perturbatively equivalent to the topological BF theory with the space of fields $T^*[-1]\mr{Map}(\RR^2_{\mr{dR}}, R/\gg)$.  This equivalence is $\SO(2)$-equivariant.
\end{theorem}

\begin{proof}
The theory is obtained as the dimensional reduction of the minimal twist of 3d $\mc N=2$ supersymmetric Yang--Mills theory along the holomorphic direction.
Combining Theorem \ref{3d_minimal_twist_thm} and Proposition \ref{prop:BFholomorphicreduction} we obtain the desired result.
\end{proof}

\subsection{\texorpdfstring{$\cN=(4,4)$}{N=(4,4)} Super Yang--Mills Theory} \label{sect:2d(4,4)}
Next, we consider $\mc N=(4,4)$ supersymmetric Yang--Mills theory.
This theory is obtained as the dimension reduction of 3d $\mc N=4$ supersymmetric Yang--Mills theory.  
No new twists arise, i.e. every square-zero supercharge sits inside an $\mc N=(2,2)$ subalgebra.

The R-symmetry group is $G_R = \SU(2) \times \spin(4;\CC)$ whose action on $\Sigma = S_+ \otimes W_+ \oplus S_- \otimes W_-$ can be described as follows.
We can identify $W_\pm = S^4_\pm \otimes W_+^6$, where $W_+^6$ is the auxiliary space of 6d $\cN=(1,0)$ supersymmetry, and $S^4_\pm$ are the semi-spin representations of $\spin(4;\CC)$. 
The group $\spin(4;\CC)$ acts on $S^4_\pm$ in the natural way, and $\SU(2)$ acts on $W_+^6 \cong \CC^2$ as the fundamental representation.

In addition to rank $(1,0)$ supercharges, which are holomorphic, there are square-zero rank $(1,1)$ supercharges.  
As in Section \ref{sect:2d(2,2)}, these supercharges split into two orbits.

\begin{prop}
There are two $G_R$-orbits in the space of rank $(1,1)$ square-zero supercharges in the $\mc N=(4,4)$ supersymmetry algebra. The generic orbit consists of A-supercharges and the special orbit consists of B-supercharges.
\label{prop:2d44twoorbits}
\end{prop}

\begin{proof}
Decompose the 6d semi-spin representation $S_+^6$ into $S_+^4 \oplus S_-^4$ as a representation of $\spin(4;\CC)$. 
A rank $(1,1)$ square-zero supercharge can be identified with the data of a null vector in $S_+^4$, a null vector in $S_-^4$, and a pair of null vectors in $W_+^6$, all non-zero.  
There are only two $G_R$-orbits in this space of quadruples of null vectors: either the two null vectors in $W_+^6$ are collinear, or they are distinct.  
These are the B- and A-supercharges respectively.
\end{proof}

Upon dimensional reduction from 3d $\mc N=4$, the $B$-supercharge reduces to the 2d $\mc N=(4,4)$ $B$-supercharge.  
Indeed, the 3d $B$-supercharge squares to zero as an element of the 6d $\mc N=(1,0)$ supersymmetry algebra, which means that in this 6d algebra it has rank 1.  
Therefore the two vectors in $W_+^6$ discussed above must be collinear. 

\subsubsection{Holomorphic Twist} \label{sect:2d44minimaltwist}
The holomorphic twist is obtained as the dimensional reduction of the minimal twist of 3d $\mc N=4$ supersymmetric Yang--Mills theory.
By Theorem \ref{3d_4_minimal_twist_thm} and Proposition \ref{CS_to_BF_diml_red_prop} we have the following. 

\begin{theorem}
The holomorphic twist of 2d $\mc N=(4,4)$ super Yang--Mills theory with matter valued in a symplectic $\fg$-representation $U$ is perturbatively equivalent to a holomorphic BF theory, with moduli space given by $T^*[-1]\Sect(\CC, T[1]((U\otimes K_\CC^{1/2}) \ham \gg))$. 
This equivalence is $\mr U(1)$-invariant.
\label{thm:2d44holomorphictwist}
\end{theorem}

\subsubsection{Topological Twists}
\label{sect:2d44Atwist}
\label{sect:2d44Btwist}

Let $V_\RR=\RR^2$ equipped with a Euclidean structure and $V=V_\RR\otimes_\RR\CC$ its complexification. By the results of Section \ref{sect:4d2Donaldson} we obtain a family $Q_t$ of square-zero supercharges in the 4-dimensional supersymmetry algebra. Dimensionally reducing this family along $\Re\colon V\rightarrow V_\RR$, we obtain a family of square-zero supercharges in the 2-dimensional supersymmetry algebra. It is easy to see that it is a family of topological supercharges.

\begin{theorem}
The twist of the 2d $\cN=(4, 4)$ super Yang--Mills theory with matter valued in a symplectic $\fg$-representation $U$ with respect to the family $Q_t$ of square-zero supercharges is perturbatively equivalent to the theory $\map(\RR^2_{\mathrm{dR}}, (U\ham \fg)_{\Hod})$. Moreover, the equivalence is $\SO(2, \RR)$-equivariant.
\label{thm:2d44Hodgetwist}
\end{theorem}
\begin{proof}
The claim follows by dimensional reduction (Proposition \ref{CS_diml_red_prop}) from the corresponding statement in 4 dimensions (Theorem \ref{thm:4dDonaldsontwist}) which calculates the twist of the 4d $\cN=2$ super Yang--Mills theory with respect to that family to be $\Sect(\CC^2, ((U\otimes K^{1/2}_{\CC^2})\ham \fg)_{\Hod})$.
\end{proof}

As we see from Theorem \ref{thm:2d44Hodgetwist}, at $t\neq 0$ the twist is perturbatively trivial while at $t=0$ it is not. By Proposition \ref{prop:2d44twoorbits} there are only two orbits of topological supercharges, so the case $t\neq 0$ (the generic orbit) must be the A-twist and the case $t=0$ (the special orbit) must be the B-twist.

\begin{corollary}
The topological A-twist of the 2d $\cN=(4, 4)$ super Yang--Mills theory is perturbatively trivial.
\label{cor:2d44Atwist}
\end{corollary}

\begin{corollary}
The topological B-twist of the 2d $\cN=(4, 4)$ super Yang--Mills theory is perturbatively equivalent to the theory $T^*[-1]\map(\RR^2_{\mathrm{dR}}, U\ham \fg)$. Moreover, the equivalence is $\SO(2, \RR)$-equivariant.
\label{cor:2d44Btwist}
\end{corollary}

\subsection{\texorpdfstring{$\cN=(8,8)$}{N=(8,8)} Super Yang--Mills Theory} \label{sect:2d(8,8)}
Next we consider the $\mc N=(8,8)$ supersymmetric Yang--Mills theory.
No new twists arise, i.e. every square-zero supercharge sits inside an $\mc N=(2,2)$ subalgebra.

The R-symmetry group in the $\mc N=(8,8)$ case is $G_R = \spin(8;\CC)$, acting on $W_+$ and $W_-$ by its two semi-spin representations.  The classification of square-zero supercharges is identical to the classification we saw in the $\mc N=(4,4)$ case.  Rank $(1,0)$ square zero supercharges, and rank $(1,1)$ square zero supercharges split into two orbits in the following way.

\begin{prop}
There are two $G_R$-orbits in the space of rank $(1,1)$ square-zero supercharges in the $\mc N=(8,8)$ supersymmetry algebra. The generic orbit consists of A-supercharges and the special orbit consists of B-supercharges.
\end{prop}

\begin{proof}
We classify $\spin(8;\CC)$-orbits in the space of pairs of non-zero null-vectors $w_+ \in W_+$ and $w_- \in W_-$.  Since $\spin(8;\CC)$ acts transitively on the possible choices of $w_+$, it remains for us to understand the action of the stabilizer $\stab(w_+) \sub \spin(8;\CC)$ on the space of null vectors $w_-$.  The element $w_+$ is equivalent to the data of a Lagrangian subspace $L \sub V_8$, along with a half-density.  As a representation of the subgroup $\mr{ML}(L)$ of $\stab(w_+)$, the two semi-spin representations decompose as
\begin{align*}
 W_+ &\iso (\CC \oplus \wedge^2 L \oplus \wedge^4 L) \otimes \det(L)^{-1/2} \\
 W_- &\iso (L \oplus \wedge^3 L) \otimes \det(L)^{-1/2},
\end{align*}
with $w_+ \in \wedge^4 L$.  With respect to this decomposition, say $w_- = (v_1,v_3)$.  If $v_1 \ne 0$, then we can act by $\wedge^2 L \sub \stab(w_+)$ to make $v_3 = 0$.  Under the action of $\SL(L)$ all such non-zero $v_1$ are in the same orbit.  Likewise if $v_1 = 0$ then $v_3 \ne 0$ and we can act by $\SL(L)$ to see that all such non-zero $v_3$ are in the same orbit.  There are, therefore, two orbits once again, with one degenerating to the other.
\end{proof}

The 2d $\cN=(8, 8)$ super Yang--Mills theory may be considered as a 2d $\cN=(4, 4)$ super Yang--Mills theory with matter valued in the symplectic $\fg$-representation $U=T^*\fg$. So, all the computations in this section follow from the corresponding computations in Section \ref{sect:2d(4,4)}.

\subsubsection{Holomorphic Twist} \label{sect:2d88minimaltwist}

The 2d $\cN=(8, 8)$ super Yang--Mills theory is a dimensional reduction of the 3d $\cN=8$ super Yang--Mills theory. Therefore, from Theorem \ref{3d_8_minimal_twist_thm} we obtain the following statement.

\begin{theorem}
The holomorphic twist of the 2d $\mc N=(8,8)$ super Yang--Mills theory is perturbatively equivalent to a holomorphic BF theory with the space of fields given by $T^*[-1]\mr{Map}(\CC_{\mathrm{Dol}}, T[1] \fg/\fg)$. This equivalence is $\mr U(1)$-equivariant.
\end{theorem}

\subsubsection{Topological Twists}
\label{sect:2d88Atwist}
\label{sect:2d88Btwist}

Consider the family of topological supercharges $Q_t$ from Section \ref{sect:2d44Atwist}. From Theorem \ref{thm:2d44Hodgetwist} we obtain the following statement.

\begin{theorem}
The twist of the 2d $\cN=(8, 8)$ super Yang--Mills theory with respect to the family $Q_t$ of square-zero supercharges is perturbatively equivalent to the theory $\map(\RR^2_{\mathrm{dR}}, T^*(\fg/\fg)_{\Hod})$. Moreover, the equivalence is $\SO(2, \RR)$-equivariant.
\label{thm:2d88Hodgetwist}
\end{theorem}

\begin{corollary}
The topological A-twist of the 2d $\cN=(8, 8)$ super Yang--Mills theory is perturbatively trivial.
\label{cor:2d88Atwist}
\end{corollary}

\begin{corollary}
The topological B-twist of the 2d $\cN=(8, 8)$ super Yang--Mills theory is perturbatively equivalent to the theory $T^*[-1]\map(\RR^2_{\mathrm{dR}}, T^*(\fg/\fg))$. Moreover, the equivalence is $\SO(2, \RR)$-equivariant.
\label{cor:2d88Btwist}
\end{corollary}

\subsection{\texorpdfstring{$\cN=(2,0)$}{N=(2,0)} Super Yang--Mills Theory} \label{sect:2d(2,0)}

We turn to the $\cN=(2,0)$ supersymmetric Yang--Mills theory.
In the supersymmetry algebra, we have $W_-=0$ and $W_+$ is a complex two-dimensional vector space equipped with a symmetric pairing. 
The supersymmetric matter consists of the $\cN = (2,0)$ chiral multiplet with values in a $\fg$-representation $R$. 
The R-symmetry group is $G_R = \SO(2;\CC)$ which acts on $W_+$ by the defining representation. 

The fields of the untwisted theory are:
\begin{itemize}
\item a gauge boson $A \in \Omega^1(\RR^2;\fg)$.
\item gauge fermions $\lambda \in C^\infty(\RR^2 ; \Pi S_+ \otimes W_+ \otimes \fg)$. 
\item matter bosons $\phi \in C^\infty(\RR^2 ; R)$ and $\Bar{\phi} \in C^\infty(\RR^2 ; R^*)$;
\item matter fermions $\psi \in C^\infty(\RR^2 ; \Pi S_- \otimes R^*)$  and $\Bar{\psi} \in C^\infty(\RR^2 ; \Pi S_- \otimes R)$.
\item a ghost field $c \in \Omega^0(\RR^2 ; \fg)[1]$.
\end{itemize}
The field $\lambda$ transforms in the defining representation of $G_R = \SO(2;\CC)$.
The fields $\psi, \Bar{\psi}$ have weights $-1,+1$ respectively.

The theory admits a unique twist by the following class of supercharge.
\begin{itemize}
\item Elements $Q \in S_+ \otimes W_+$ of rank $1$. 
Such supercharges are automatically square-zero and are holomorphic.
We take $\alpha \colon U(1) \hookrightarrow G_R$ to be the standard embedding.
There is a compatible twisting homomorphism
\[
\phi \colon \MU(1) \xto{\det^{1/2}} \U(1) \hookrightarrow G_R .
\]
\end{itemize}

\subsubsection{Holomorphic Twist} \label{sect:2d20minimaltwist}

Choose a complex structure $L \subset V$. 
Under the embedding $MU(1) \hookrightarrow \Spin(2; \CC)$, the semi-spin representations decompose as
\[
S_+ = \det(L)^{1/2} \;\; , \;\; S_- = \det(L)^{-1/2} .
\]

Note that under the twisting homomorphism $\phi = \det^{1/2}$, we have $W_+ = \det(L)^{1/2} \oplus \det(L)^{-1/2}$.

The fields decompose under the twisting homomorphism as:
\begin{itemize}
\item gauge bosons $A_{1,0} \in \Omega^{1,0}(\CC ; \fg)$, $A_{0,1} \in \Omega^{0,1}(\CC ; \fg)$.
\item gauge fermions $\lambda_0 \in \Omega^0(\CC ; \fg)[-1]$, $\lambda_{1,0} \in \Omega^{1,0}(\CC ; \fg) [1]$.
\item matter bosons $\phi \in \Omega^0(\CC ; R^*)$, $\gamma_0 \in \Omega^0(\CC ; R)$.
\item matter fermions $\psi \in \Omega^{0}(\CC ; R^*)[1]$, $\gamma_{0,1} \in \Omega^{0,1}(\CC ; \fg)[-1]$.
\item a ghost field $A_0 \in \Omega^0(\CC ; \fg)[1]$.
\end{itemize}

The action functional decomposes as follows:
\begin{align*}
S_{\mr{gauge}} &= \int \dvol \bigg(-(F^{2, 0}, F^{0, 2}) - \frac{1}{4}(\Lambda F_{1, 1})^2 + \frac 12 (\lambda_0 , \dbar_{A_{0,1}} \lambda_{1,0})  \bigg)    \\
S_{\mr{matter}} &= \int \dvol \bigg( (\dd_{A_{1,0}}\Bar{\phi} , \ol \dd_{A_{0,1}} \gamma_0) + (\dd_{A_{1,0}}\gamma_0, \ol \dd_{A_{0,1}} \Bar{\phi}) + (\psi_0 , \partial_{A_{1,0}} \gamma_{0,1}) + ([\lambda_{1,0}, \psi_0], \gamma_{0,1})   \bigg) \\
S_{\mr{anti}} &= \int \dvol \bigg( (\dd_{A_{1, 0}} A_0 , A_{1,0}^*) + (\ol\dd_{A_{0, 1}} A_0, A_{0,1}^*) + ([\lambda_{1,0}, A_0], \lambda_{1,0}^*) +  [\lambda_{0}, A_0] \lambda_{0}^*\\
&\qquad + \frac{1}{2}[A_0, A_0]A_0^* + [\gamma_0, A_0] \gamma_0^* + [\Bar{\phi}, A_0] \Bar{\phi}^* + [\gamma_{0,1}, A_0] \gamma_{0,1}^* + [\psi_0, A_0] \psi_{0}^*\bigg) \\
S^{(1)}_{\mr{gauge}} &=  \int \dvol \bigg( - (\lambda_{1,0}, A_{1,0}^*) \bigg) \\
S^{(1)}_{\mr{matter}} &=  \int \dvol \bigg( (\psi_{0}, \Bar{\phi}^*) + \frac 12(\dbar_{A_{0,1}} \gamma_0, \gamma_{0,1}^*) \bigg)\\
S^{(2)}_{\mr{gauge}} &= \int \dvol \bigg( -\frac{1}{4} (\lambda_0^*)^2 \bigg) .
\end{align*}

\begin{theorem} \label{thm:2d(2,0)}
The minimal twist of 2d $\cN=(2,0)$ super Yang--Mills with matter valued in a $\fg$-representation $R$ is perturbatively equivalent to holomorphic BF theory coupled to the holomorphic $\beta\gamma$ system with moduli space $T^*[-1] \map(\CC , R / \fg)$.
This equivalence is $\U(1)$-equivariant.
\end{theorem}
\begin{proof}
First, we eliminate the field $\lambda_{0}$ using Proposition \ref{prop:integrateoutfield}.  
We then observe that the action includes the terms $\int  (\lambda_{1,0} , A_{1,0}^*)$ and  $\int (\phi, \psi_0^*)$.  
Thus, the two pairs $(\lambda_{1,0}, A_{1,0})$ and $(\phi, \psi_0)$ form BRST doublets, 
which can be eliminated using Proposition \ref{prop:BRSTdoublet}.  

The twisted theory is therefore perturbatively equivalent to the theory with BV action 
\begin{align*}
& \int \bigg(\dbar_{A_{0,1}} A_{0} \wedge A_{0,1}^* + (\dbar_{A_{0,1}} \gamma_{0} \wedge \gamma_{0,1}^*) \bigg) + \dvol \bigg( \frac{1}{2}[A_0, A_0]A_0^*+ [\gamma_0, A_0] \gamma_0^* + ([\gamma_{0,1}, A_0] , \gamma_{0,1}^*)  \bigg).
\end{align*}
This is the action functional of the required theory, where $B_{1,0} = A_{0,1}^*$, $B_{1,1} = A_{0}^* \dvol$ comprise the anti-fields of holomorphic BF theory and $\beta_{1,0} = \gamma_{0,1}^*$,  $\beta_{1,1} = \gamma_0^* \dvol$ comprise the anti-fields of the $\beta\gamma$ system.
\end{proof}

\subsection{\texorpdfstring{$\cN=(4,0)$}{N=(4,0)} Super Yang--Mills Theory} \label{sect:2d(4,0)}

Next, we consider the $\cN=(4,0)$ supersymmetric Yang--Mills theory.
In the supersymmetry algebra, we have $W_- = 0$ and $W_+$ is a complex four-dimensional vector space equipped with a nondegenerate symmetric bilinear pairing.
The supersymmetric matter consists of the $\cN = (4,0)$ hypermultiplet with values in a complex symplectic $\fg$-representation $U$. 
The R-symmetry group is $G_R = \SL(2;\CC) \times \GL(1;\CC)$.
As a $G_R$-representation, $W_+ = \Tilde{W} \oplus \Tilde{W}$ where $\Tilde{W}$ is the defining representation of $\SL(2;\CC)$ and where $\GL(1;\CC)$ acts by weight $(1,-1)$ with respect to this decomposition. 

\begin{itemize}
\item a gauge boson $A \in \Omega^1(\RR^2;\fg)$.
\item gauge fermions $(\lambda_-, \lambda_+) \in C^\infty(\RR^2 ; S_+ \otimes (\Tilde{W} \oplus \Tilde{W}) \otimes \fg)$. 
\item matter bosons $\phi \in C^\infty(\RR^2 ; \Tilde{W} \otimes U)$.
\item matter fermions $(\psi_-, \psi_+) \in C^\infty(\RR^2 ; S_- \otimes (U\oplus U))$.
\item a ghost field $c \in \Omega^0(\RR^2 ; \fg)[1]$.
\end{itemize}

Under $\GL(1; \CC) \subset G_R$, the pairs of fields $(\lambda_-, \lambda_+)$ and $(\psi_-, \psi_+)$ have weights $(-1,+1)$ 
The field $\phi$ has weight zero.

The theory admits a unique twist:
\begin{itemize}
\item Elements $Q \in S_+ \otimes W_+$ of rank $1$. 
Such supercharges are automatically square-zero and holomorphic.
There is a twisting homomorphism $\phi : \U(1) \to G_R$ and twisting datum $\alpha : \U(1) \to G_R$. 
\end{itemize}

\subsubsection{Holomorphic twist} \label{sect:2d40minimaltwist}

Choose a complex structure $L \subset V$. 
Under the embedding $MU(1) \hookrightarrow \Spin(2; \CC)$, the semi-spin representations decompose as
\[
S_+ = \det(L)^{1/2} \;\; , \;\; S_- = \det(L)^{-1/2} .
\]
We choose a twisting homomorphism $\phi\colon \U(1) \to G_R$ under which $\Tilde{W} = \det(L)^{1/2} \oplus \det(L)^{-1/2}$
so that
\[
S_+ \otimes W_+ = \CC Q \oplus \CC \oplus L^{\oplus 2} .
\]

The $\ZZ$-grading is induced by a natural embedding $\alpha \colon \U(1) \hookrightarrow \GL(1;\CC) \to G_R$. 

Further, fields decompose under the twisting homomorphism as:
\begin{itemize}
\item gauge bosons $A_{1,0} \in \Omega^{1,0}(\CC ; \fg)$, $A_{0,1} \in \Omega^{0,1}(\CC ; \fg)$;
\item gauge fermions $\lambda^+ \in \Omega^0(\CC ; \fg)[-1]$, $\lambda^+_{1,0} \in \Omega^{1,0}(\CC ; \fg) [-1]$,  $\lambda^- \in \Omega^0(\CC ; \fg)[1]$, $\lambda^-_{1,0} \in \Omega^{1,0}(\CC ; \fg) [1]$;
\item matter bosons $\gamma \in \Omega^0(\CC ; K_{\CC}^{1/2} \otimes U)$, $\nu \in \Omega^0(\CC ; K_{\CC}^{-1/2} \otimes U)$;
\item matter fermions $\psi_- \in \Omega^{0}(\CC ; K_\CC^{-1/2} \otimes U)[1]$, $\psi_+ \in \Omega^{0}(\CC ; K_\CC^{-1/2} \otimes U)[-1]$;
\item a ghost field $c \in \Omega^0(\CC ; \fg)[1]$.
\end{itemize}

The action functional decomposes as follows:
\begin{align*}
S_{\mr{gauge}} &= \int \dvol \bigg(-(F^{2, 0}, F^{0, 2}) - \frac{1}{4}(\Lambda F_{1, 1})^2 \bigg) + \lambda^+_0  \dbar_{A_{0,1}} \lambda^-_{1,0} + \lambda^-_{0} , \dbar_{A_{0,1}} \lambda^+_{1,0}    \\
S_{\mr{matter}} &= \int \dvol \bigg( (\dd_{A_{1,0}} \nu , \ol \dd_{A_{0,1}} \gamma) + (\dd_{A_{1,0}}\gamma, \ol \dd_{A_{0,1}} \nu) + (\psi_- , \partial_{A_{1,0}} \gamma) + (\psi_+ , \partial_{A_{1,0}} \nu)  + \\
& \qquad + ([\lambda^+, \gamma], \psi_-) + ([\lambda_{1,0}^+, \nu], \psi_-) + ([\lambda^-, \gamma], \psi_+) + ([\lambda^-_{1,0}, \nu], \psi_-) \bigg) \\
S_{\mr{anti}} &= \int \dvol \bigg( (\dd_{A_{1, 0}} c , A_{1,0}^*) + (\ol\dd_{A_{0, 1}} c, A_{0,1}^*) \\ & \qquad + ([\lambda^+_{1,0}, c], \lambda_{1,0}^{+*}) + + ([\lambda^-_{1,0}, c], \lambda_{1,0}^{-*}) +  [\lambda^+_{0}, c] \lambda_{0}^{+*}+ [\lambda^-_{0}, c] \lambda_{0}^{-*}\\
&\qquad + \frac{1}{2}[c, c]c^* + [\gamma, c] \gamma^* + [\nu, c] \nu^* + [\psi_-^*, c] \psi_-^* + [\psi_+, c] \psi_+^*\bigg) \\
S^{(1)}_{\mr{gauge}} &=  \int \dvol \bigg( - (\lambda^-_{1,0}, A_{1,0}^*) \bigg) \\
S^{(1)}_{\mr{matter}} &=  \int \dvol \bigg( (\psi_{-}, \nu^*) + \frac 12(\dbar_{A_{0,1}} \gamma, \psi_+^*) \bigg)\\
S^{(2)}_{\mr{gauge}} &= \int \dvol \bigg( -\frac{1}{4} (\lambda^{-*})^2 \bigg) .
\end{align*}

\begin{theorem} \label{thm:2d(4,0)}
The holomorphic twist of the 2d $\cN=(4,0)$ super Yang--Mills theory with matter valued in a symplectic $\fg$-representation $U$ is perturbatively equivalent to the holomorphic $BF$ theory with the space of fields $T^*[-1] \Sect(\CC , (U \otimes K^{1/2}_{\CC}) \ham \fg)$. This equivalence is $\U(1)$-equivariant.
\end{theorem}

\begin{proof}
First, we eliminate the field $\lambda^-$ using Proposition \ref{prop:integrateoutfield}.  
We then observe that the action includes the terms $\int  (\lambda^-_{1,0} , A_{1,0}^*)$ and  $\int (\psi_-, \nu^*)$.  
Thus, the two pairs $(\lambda^-_{1,0}, A_{1,0})$ and $(\nu, \psi_-)$ form BRST doublets, 
which can be eliminated using Proposition \ref{prop:BRSTdoublet}.  

The twisted theory is therefore perturbatively equivalent to the theory with BV action 
\begin{align*}
& \int \bigg(\dbar_{A_{0,1}} c \wedge A_{0,1}^* + \lambda^+ \dbar_{A_{0,1}} \lambda^-_{1,0} + \dbar_{A_{0,1}} \gamma \wedge \psi_+^* \bigg) + \dvol \bigg( \frac{1}{2}[c, c]c^*+ ([\lambda_{1,0}^-, c] , \lambda_{1,0}^{-*}) + [\lambda^+, c] \lambda^{+*} + [\gamma, c] \gamma^* + [\psi_+, c] \psi_+^* \bigg).
\end{align*}
This is the action functional of the required theory.
\end{proof}

\subsection{\texorpdfstring{$\cN=(\cN_+,0)$}{N=(N+,0)} Super Yang--Mills Theory}

We consider pure $\cN=(\cN_+,0)$ super Yang--Mills for a Lie algebra $\fg$, where $\cN_+ \geq 2$.
The spinorial representation is $\Sigma = S_+ \otimes W_+$ where $W_+$ is the $\cN_+$-dimensional auxiliary space equipped with a nondegenerate symmetric bilinear pairing.

\begin{itemize}
\item $\fg$-valued bosons: a gauge field $A \in \Omega^1(\RR^2;\fg)$.
\item $\fg$-valued fermions: a spinor $\lambda \in C^\infty(\RR^2 ; S_+ \otimes W_+ \otimes \fg)$. 
\end{itemize}

The theory admits a unique twist by the following class of supercharge.
\begin{itemize}
\item Elements $Q \in S_+ \otimes W_+$ of rank $1$. 
Such supercharges are automatically square-zero, holomorphic, and lie in a $\cN = (2,0)$ subalgebra.
We take the twisting datum $\alpha$ and twisting homomorphism $\phi$ to factor through those of Section \ref{sect:2d(2,0)}.
\end{itemize}

\subsubsection{Holomorphic Twist} \label{sect:2dN0minimaltwist}

Choose a complex structure $L \subset V$. 
Under the embedding $U(1) \hookrightarrow \spin(2;\CC)$, the semi-spin representations decompose as
\[
S_+ = \det(L)^{1/2} \;\; , \;\; S_- = \det(L)^{-1/2} .
\]

Note that under the twisting homomorphism $\phi = \det^{1/2}$, we have $W_+ = \det(L)^{1/2} \oplus \det(L)^{-1/2} \oplus \CC^{\cN_+ - 2}$.

The fields decompose under the twisting homomorphism as:
\begin{itemize}
\item gauge bosons $A_{1,0} \in \Omega^{1,0}(\CC ; \fg)$, $A_{0,1} \in \Omega^{0,1}(\CC ; \fg)$;
\item gauge fermions $\lambda_0 \in \Omega^0(\CC ; \fg)[-1]$, $\lambda_{1,0} \in \Omega^{1,0}(\CC ; \fg) [1]$, $\Tilde{\lambda} \in \Omega^{0} (\CC ; K^{1/2}_\CC \otimes \fg)$;
\item a ghost field $A_0 \in \Omega^0(\CC ; \fg)[1]$.
\end{itemize}

The action functional decomposes as follows:
\begin{align*}
S_{\mr{gauge}} &= \int \dvol \bigg(-(F^{2, 0}, F^{0, 2}) - \frac{1}{4}(\Lambda F_{1, 1})^2 + \frac 12 (\lambda_0 , \dbar_{A_{0,1}} \lambda_{1,0}) + \frac 12 (\Tilde{\lambda}, \dbar_{A_{0,1}} \Tilde{\lambda}) \bigg)    \\
S_{\mr{anti}} &= \int \dvol \bigg( (\dd_{A_{1, 0}} A_0 , A_{1,0}^*) + (\ol\dd_{A_{0, 1}} A_0, A_{0,1}^*) + ([\lambda_{1,0}, A_0], \lambda_{1,0}^*) +  [\lambda_{0}, A_0] \lambda_{0}^*\\
&\qquad + [\Tilde{\lambda}, A_0] \Tilde{\lambda}^* + \frac{1}{2}[A_0, A_0]A_0^* \bigg) \\
S^{(1)}_{\mr{gauge}} &=  \int \dvol \bigg( - (\lambda_{1,0}, A_{1,0}^*) \bigg) \\
S^{(2)}_{\mr{gauge}} &= \int \dvol \bigg( -\frac{1}{4} (\lambda_0^*)^2 \bigg) .
\end{align*}

\begin{theorem} \label{thm:2d(N,0)}
The holomorphic twist of 2d $\cN=(\cN_+,0)$ super Yang--Mills is perturbatively equivalent to the holomorphic $BF$ theory with the space of fields $T^*[-1]\mr{Sect}(\CC, (\fg^{\mc N_+ -2} \otimes K_{\CC}^{1/2}) / \gg)$. This equivalence is $\U(1)$-equivariant.
\end{theorem}

\begin{proof}
First, we integrate out the field $\lambda_{0}$ using Proposition \ref{prop:integrateoutfield}.   
We then observe that the action includes the term $\int  (\lambda_{1,0} , A_{1,0}^*)$.
Thus, the pair $(\lambda_{1,0}, A_{1,0})$ forms a BRST doublet, 
which can be eliminated using Proposition \ref{prop:BRSTdoublet}.  

The twisted theory is therefore perturbatively equivalent to the theory with BV action 
\begin{align*}
& \int \dvol \bigg((\dbar_{A_{0,1}} A_{0}, A_{0,1}^*) + (\dbar_{A_{0,1}} \Tilde{\lambda} ,  \Tilde{\lambda}) + \frac{1}{2}[A_0, A_0]A_0^*+ [ \Tilde{\lambda} , A_0]  \Tilde{\lambda}^* \bigg).
\end{align*}
This is the action functional of the required theory.
\end{proof}

\pagestyle{bib}
\printbibliography

\textsc{University of Massachusetts, Amherst}\\
\textsc{Department of Mathematics and Statistics, 710 N Pleasant St, Amherst, MA 01003}\\
\texttt{celliott@math.umass.edu}\\
\vspace{-5pt}

\textsc{Institut f\"ur Mathematik, Universit\"at Z\"urich}\\
\textsc{Winterthurerstrasse 190, 8057 Zurich, Switzerland}\\
\texttt{pavel.safronov@math.uzh.ch}\\
\vspace{-5pt}

\textsc{Northeastern University}\\
\textsc{Department of Mathematics, 360 Huntington Ave, Boston, MA 02115}\\
\texttt{brianwilliams.math@gmail.com}\\

\end{document}